%% file: arxiv.tex
\documentclass[hf]{ceurart}

\usepackage{cmap}

\usepackage{version}

\includeversion{conf}
\excludeversion{report}

\usepackage[dvipsnames]{xcolor} 
\usepackage{graphicx}
\usepackage{microtype}
\usepackage{amsmath, amsfonts, amssymb}
\usepackage{array}
\usepackage{multirow}
\usepackage{enumerate}
\usepackage[inline]{enumitem}
\usepackage{mdframed}
\usepackage[ruled, vlined]{algorithm2e}
\usepackage{listings}
\usepackage{times}

\usepackage{thmtools} 
\usepackage{thm-restate} 

\usepackage{pgfplots}
\pgfplotsset{compat=1.16} 
\usepgfplotslibrary{groupplots} 

\usepackage[skins]{tcolorbox} 

\usepackage{hyperref}
\usepackage{cleveref}
\crefname{appendix}{Appendix}{Appendices}
\Crefname{appendix}{Appendix}{Appendices}
\crefname{subappendix}{Appendix}{Appendices}
\Crefname{subappendix}{Appendix}{Appendices}
\crefname{subsubappendix}{Appendix}{Appendices}
\Crefname{subsubappendix}{Appendix}{Appendices}

\providecommand{\spnewtheorem}[4]{\newtheorem{#1}{#2}}
\input{macros.tex}
\definecolor{mygreen}{RGB}{0,180,0}
\renewcommand{\greensat}{{\color{mygreen} sat}}
\renewcommand{\greenunsat}{{\color{mygreen} unsat}}

\makeatletter
\@ifundefined{theorem} {\newtheorem{theorem}{Theorem}}{}
\@ifundefined{lemma} {\newtheorem{lemma}[theorem]{Lemma}}{}
\@ifundefined{corollary} {}{}
\@ifundefined{definition}{\newtheorem{definition}[theorem]{Definition}}{}
\@ifundefined{remark} {\newtheorem{remark}[theorem]{Remark}}{}
\@ifundefined{example} {\newtheorem{example}[theorem]{Example}}{}
\makeatother

\AtBeginDocument{%
  \ifcsname proof\endcsname%
    \renewenvironment{proof}{\noindent{\it Proof}\/:}{\qed\hspace*{1pt}}%
  \else%
    \newenvironment{proof}{\noindent{\it Proof}\/:}{\qed\hspace*{1pt}}%
  \fi%
}

\usepackage{tikz}
\usetikzlibrary{shapes, arrows.meta, positioning, calc}

\tikzset{
  term node/.style={
    draw,
    rounded corners,
    align=center,
    inner sep=4pt,
  fill=white
  },
  >={Stealth[]}
}

\lstdefinelanguage{smtlib}{
  morekeywords={
    set-logic, declare-datatypes, declare-const, assert, check-sat, select, Array, Int, not, =, nT, cons, id, arr
  },
  sensitive=true,
  morecomment=[l];,
  morestring=[b]",
}

\renewcommand{\yoni}[1]{{}}
\renewcommand{\yoniFixed}[1]{{}}
\renewcommand{\tomer}[1]{{}}
\renewcommand{\tomerFixed}[1]{{}}

\begin{document}
\setcounter{secnumdepth}{3}

\copyrightyear{2026}
\copyrightclause{Copyright for this paper by its authors.
	Use permitted under Creative Commons License Attribution 4.0
	International (CC BY 4.0).}

\conference{SMT 2026, 24th International Workshop on Satisfiability Modulo Theories}

\title{Automated Reasoning with Nested Datatypes}

\author[1]{Tomer Hakak}
\author[1]{Yoni Zohar}
\author[2]{Andrew Reynolds}
\author[3]{Clark Barrett}
\author[2]{Cesare Tinelli}

\ExplSyntaxOn
\RenewDocumentCommand \ceuraddress { }
{
	\bool_if:NTF \g_ceur_augr_bool
	{ }
	{
		\group_begin:
		\ceurAffSetup { type = normal }
		\l_ceur_aff_setup_tl
		\par \vskip\l_ceur_aff_before_dim
		\l_ceur_aff_align_tl
		\l_ceur_aff_size_tl
		\l_ceur_aff_shape_tl
		\l_ceur_aff_weight_tl
		\color{ \l_ceur_aff_color_tl }
		\seq_use:cn { g_ceur_aff\int_use:N \g_ceur_aaugr_int _seq }
		{ \quad }
		\par\vskip\l_ceur_aff_after_dim
		\group_end:
	}
}

\cs_set:Npn \__make_tbl_caption:nn #1#2
{
	\l_tbl_align_tl
	\skip_vertical:N \l_tbl_abovecap_skip
	{\parbox{\linewidth}
		{\rightskip=0pt\sffamily\small\textbf{\color{scolor}#1}\enskip #2\par\vskip4pt }}
	\skip_vertical:N \l_tbl_belowcap_skip
}
\ExplSyntaxOff

\address[1]{Bar-Ilan University, Ramat Gan, Israel}
\address[2]{The University of Iowa, Iowa City, USA}
\address[3]{Stanford University, Stanford, USA}

\begin{abstract}
	We introduce a theory of nested datatypes. The theory is obtained
	by restricting the naive combination of datatypes and arrays, so as to prevent
	non-standard models from emerging.
	A decision procedure for the theory is given and proven correct.
	Finally, we describe an implementation of the procedure,
	as well as an evaluation over both real-world and crafted benchmarks.
\end{abstract}

\maketitle

\section{Introduction}
\label{sec:intro}

The SMT-LIB theory of datatypes allows for a direct and convenient encoding
of common data structures, such as lists, trees, and more,
which makes it useful for software verification.
And indeed, the 2025 release of the SMT-LIB benchmark library includes
over 50K benchmarks that involve datatypes~\cite{SMTLib2025}.
%

The expressivity of datatypes is, however, limited. For example,
there are no operators to update an existing datatype value,
and there are no random-access operators.
Such limitations call for the combination of datatypes with theories for other
data structures, like arrays.

A naive combination of datatypes and arrays introduces a problem:
arrays bypass the principle of acyclicity, according to which a datatype value cannot occur in itself.
Indeed,
a datatype value could now occur as an element
of an array, where the array is a field of the same value.
Such behavior is undesirable, as it does not match the intuitive combination of arrays and datatypes,
which only intends to increase expressivity, not break acyclicity.
We refer to such models that admit cycles as {\em non-standard models}.
A more detailed description of this issue, both intuitively and formally, is given in
\Cref{ex:NO} and \Cref{fig:dt-ndt-ex2} below.

In this paper we define a theory of {\em nested datatypes}, obtained
by restricting the aforementioned naive combination to
exclude non-standard models.
%
%
We introduce a decision procedure for the new theory, based on
a reduction to the (naive) combination of datatypes, arrays and uninterpreted functions.
%
We have implemented a prototype of our reduction-based algorithm in the state-of-the-art
SMT-solver $\cvcfive$~\cite{cvc5},
and evaluated it on benchmarks arising from the Move prover verification
tool~\cite{MoveProver}, and on synthetic benchmarks that
control the ``hidden cycles'' in datatypes and arrays.
Our experimental evaluation shows promise.
On the organic benchmarks from the Move Prover, it is comparable to and sometimes outperforms an existing solver for nested datatypes implemented in $\zzz$~\cite{DBLP:conf/tacas/MouraB08}.
On the synthetic benchmarks, it outperforms $\zzz$.

In summary, our contributions are:
$(i)$~The introduction of a theory of nested datatypes;
$(ii)$~A reduction-based decision procedure for the new theory;
$(iii)$~An implementation in $\cvcfive$ and an evaluation using both real-world and crafted benchmarks.

\paragraph{Informal overview:}
We intuitively explain the problem using 
the first row of
\Cref{fig:dt-ndt-ex2}.
This motivating example will be formally introduced and developed as a running
example throughout the paper.
Here $x$ is a datatype value, while $a$ is an array whose elements are datatype
values.
The first row visually describes two simple constraints: $x$ is read
from array $a$ at index $i$, $a$ is a field of $x$.
Each constraint is harmless by itself.
Together, however, they describe a cycle: following the array field of $x$
leads to $a$, and reading $a$ at index $i$ leads back to $x$.
The naive combination of the theories of datatypes and arrays 
misses this because the cycle is split between the
theories.
The theory of Nested datatypes that we propose here reject the combined constraint.

\yoni{This informal overview is also supposed to explain the solution, that is, refer informally
to the pictures of \Cref{fig:dt-ndt-ex-tr}.}

\paragraph{Related work:}

The introduced theory of nested datatypes generalizes the SMT-LIB theory of
algebraic datatypes~\cite{SMTLib2017} (see~\cite{BST07,Oppen:1980:RRD:322203.322204} for prominent decision procedures and \cite{DBLP:conf/aaai/ShahMS24} for a more recent one).
%
The SMT solver $\zzz$~\cite{DBLP:conf/tacas/MouraB08} supports nested datatypes.
However, to the best of our knowledge, the literature lacks a precise definition of the supported theory and the implemented algorithm.

\paragraph{Outline:}
\Cref{sec:prelim} goes over the necessary notions from first-order logic and Satisfiability Modulo Theories.
In \Cref{sec:theoryndtsec} we introduce the theory of nested datatypes.
In \Cref{sec:theory:preprocessing} we present a decision procedure for this theory, and then in \Cref{sec:imp}
we describe an implementation and evaluation of it.
The correctness proofs for the decision procedure are quite involved and long.
They are included in the appendix.

\tikzset{
	font=\sffamily,
	basewrap/.style={draw=black, fill=blue!5, rounded corners=2pt},
	idbox/.style={draw, fill=yellow!50, minimum width=0.6cm, minimum height=0.6cm, font=\bfseries\color{blue}},
	outerid/.style={idbox, minimum size=1cm},
	outerbg/.style={basewrap, rounded corners=8pt},
	outerframe/.style={draw=black, thick},
	childwrap/.style={basewrap, minimum height=0.9cm},
	tile/.style={draw=black, fill=gray!50, minimum width=0.9cm, minimum height=0.6cm},
	divider/.style={draw=white, line width=0.5pt}
}

\section{Preliminaries}
\label{sec:prelim}

For natural numbers $0\leq k \leq n$, we write $[k,n]$ to denote $\{k, k+1, \ldots, n\}$.
For sets $X$ and $Y$,
$\almostC{X}{Y}$
is the set of {\em almost constant functions} from $X$ to $Y$, i.e., $\{ f: X \to Y \mid \exists c \in Y.\; \exists S \subseteq X.\; |S| < \infty \wedge \forall x \notin S,\; f(x) = c \}$.
Given set $X$ and element $e$,
$\constDomArr{X}{e}$ is the constant function whose domain is $X$ and
always returns $e$. When $X$ is clear from the context, we omit it, writing
$\constArr{{e}}$.

\begin{definition}[Well-founded relation]
	\label{prelim:wfr}
	A binary relation \( R \) on a set $X$ is \emph{well-founded}
	if there is no infinite
	sequence \( x_0, x_1, x_2, \ldots \) of elements of \( X \) such that \(
	x_{n+1} \, R \, x_n \) for every \( n \).
\end{definition}

\subsection{Signatures and Structures}
We briefly review the usual definitions of
many-sorted first-order logic with equality (see~\cite{enderton2001mathematical,10.1007/978-3-540-30227-8_53}).
%
For any set $S$, an {\em $S$-sorted set} $A$ is a function from $S$ to
$\powerset{X}\setminus\set{\emptyset}$ for some set $X$
such that $A(s)\cap A(s')=\emptyset$
whenever $s\neq s'$. We use $A_{s}$ to denote $A(s)$.

A {\em many-sorted signature} $\Sigma$ consists of
a set $\sorts{\Sigma}$ of {\em sorts},
a set $\func{\Sigma}$ of {\em function symbols},
and a set $\pred{\Sigma}$ of {\em predicate symbols}.
Function symbols have arities $\sigma_{1}\times\ldots\times\sigma_{n}\ra\sigma$,
and predicate symbols have arities $\sigma_{1}\times\ldots\times\sigma_{n}$,
with $\sigma_{1}\til\sigma_{n},\sigma\in\sorts{\Sigma}$.
For each sort $\sigma\in\sorts{\Sigma}$,
$\pred{\Sigma}$ includes an {\em equality symbol} $=_{\sigma}$ of arity
$\sigma\times\sigma$, denoted
$=$ when $\sigma$ is clear.
$\Sigma$ is called {\em finite} if $\sorts{\Sigma}$, $\func{\Sigma}$,
and $\pred{\Sigma}$ are finite.
We assume a $\sorts{\Sigma}$-sorted set
of {\em variables}.
Well-sorted terms, formulas, and literals are defined as usual.
For a $\Sigma$-formula $\formulaProof$ and a sort $\sigma$, we denote the set of
free variables in $\formulaProof$ of sort $\sigma$ by $\fv{\sigma}{\formulaProof}$.
This notation naturally extends to $\fv{S}{\formulaProof}$ when $S$ is a set of
sorts. When $S=\sorts{\Sigma}$ we may omit it.
Let $B \subseteq \func{\Sigma}$ be a set of function symbols.
We denote by $\Sigma_{\mid B}$ the signature with the same sorts as $\Sigma$, no predicate symbols (except $=_{\sigma}$ for $\sigma\in\sorts{\Sigma}$), and whose set of function symbols is $B$.

A $\Sigma$-literal is called {\em flat} if it has one of the following forms:
$x=f(x_{1}\til x_{n})$,
$P(x_{1}\til x_{n})$, or $\neg P(x_{1}\til x_{n})$
for some variables $x,y,x_{1}\til x_{n}$ and function and predicate symbols
$f$ and $P$ from $\Sigma$.
A (flat) {\em cube} is a conjunction of (flat) literals.

A {\em $\Sigma$-structure} consists of a $\sorts{\Sigma}$-sorted set $A$, and interpretations
to the function and predicate symbols of $\Sigma$.
We further require that $=_{\sigma}$ is interpreted
as the identity relation over $A_\sigma$ for every $\sigma\in\sorts{\Sigma}$.
A {\em $\Sigma$-interpretation} $\model$ is an extension of
a $\Sigma$-structure with interpretations to some set of variables.
For any $\Sigma$-term $\alpha$, $\alpha^{\model}$
denotes the interpretation of $\alpha$ in $\model$.
When $X$ is a set of $\Sigma$-terms, $X^{\model}=\set{x^{\model}\mid x\in X}$.
Similarly, $f^{\model}$ and $P^{\model}$ denote the interpretation
of $f$ and $P$ in $\model$, and we denote $A_\sigma$ by $\sigma^{\model}$ for every
$\sigma\in\sorts{\Sigma}$.
Given a signature $\Sigma'$ which is contained in $\Sigma$, we denote by $\model^{\Sigma'}$ the restriction of $\model$ to $\Sigma'$.

Satisfaction is denoted by $\models$.
A {\em $\Sigma$-theory} $T$ is a class of $\Sigma$-structures.
A $\Sigma$-interpretation
whose variable-free part is in $T$ is
a {\em $T$-interpretation}.
$\formulaProof$ is {\em $T$-satisfiable} (denoted $\models_T\formulaProof$)
if $\model\models\formulaProof$ for some $T$-interpretation $\model$.
$\formulaProof$ and $\formulaExUnsat$ are {\em $T$-equivalent} if they are satisfied by the same class of $T$-interpretations.

The $\Sigma$-theory of equality and uninterpreted functions (EUF) is
the class of all $\Sigma$-structures.
A signature $\Sigma$ is called an {\em arrays signature}
if
$\sorts{\Sigma}=\{\arrType,I,E\}$ for distinct sorts $\arrType$, $I$, and $E$,
$\func{\Sigma}=\{\select{\arrType}, \store{\arrType}\}$,
and $\pred{\Sigma}$ is empty.
The arity of $\select{\arrType}$ is
$\arrType\rightarrow E$, and the arity of
$\store{\arrType}$ is $(\arrType\times I \times E)\rightarrow \arrType$.
When $\arrType$ is clear from context, we may omit it, writing $\select{}$ and $\store{}$.
For an array signature $\Sigma$, we say that a $\Sigma$-structure $\model$ is an {\em arrays $\Sigma$-structure} if it satisfies
the axioms for arrays as specified in SMT-LIB.
The $\Sigma$-theory of arrays is the class of all
arrays $\Sigma$-structures~(see, e.g. \cite{BradleyManna2007}).

\subsection{The \smtlib Theory of Datatypes}
\label{sec:theory-dt}
We adapt the definition
from \cite{DBLP:journals/jar/ShengZRLFB22}
of
the SMT-LIB theory of datatypes~\cite{SMTLib2017}.

\subsubsection{Syntax}

\begin{definition}[Well-founded sort (DT)]
	\label{def:well-sorted-dt}
	Let $\Sigma$ be a signature and
	$S\subseteq\sorts{\Sigma}$.
	Define a sequence of functions $G_0,G_1,\ldots$ as follows:
	$G_{0}(\Sigma,S) =S$;
	$G_{i+1}(\Sigma,S) =
		G_{i}(\Sigma,S) \cup \{\sigma \mid \exists c\in\func{\Sigma}.c:\sigma_1 \times \ldots \times
		\sigma_n \ra \sigma, \sigma_1, \ldots, \sigma_n \in G_{i}(\Sigma,S) \}$.
	Then, a sort $\sigma\in\sorts{\Sigma}$ is {\em well-founded w.r.t. $S$}
	in $\Sigma$
	if $\sigma\in \bigcup_{i=1}^{\infty}G_{i}(\Sigma,S)$.
	We sometimes drop $S$, and just say ``well-founded'' when it is clear.
\end{definition}

\begin{definition}[Datatypes signature]
	\label{def:IDT_sig}
	A finite signature $\Sigma$ is called a {\em datatypes signature}
	if
	$\sorts{\Sigma}=\elemsorts_{\Sigma}\uplus\structsorts_{\Sigma}$ and
	$\func{\Sigma}=\constructors_{\Sigma}\uplus\selectors_{\Sigma}$,
	for some $\elemsorts_{\Sigma},\structsorts_{\Sigma},\constructors_{\Sigma},\selectors_{\Sigma}$,
	such that
	$\selectors_{\Sigma}=
		\{\sel{c}{i}:\sigma\ra\sigma_{i} \mid
		c\in\constructors_{\Sigma},
		c:\sigma_{1}\times\ldots\times\sigma_{n}\ra\sigma,
		1\leq i\leq n
		\}$ and
	$\predecessors_{\Sigma}=
		\{\is{c}:\sigma \mid
		c\in\constructors_{\Sigma},
		c:\sigma_{1}\times\ldots\times\sigma_{n}\ra\sigma
		\}$.
The elements of $\constructors_{\Sigma}$ are called {\em constructors},
those of  $\selectors_{\Sigma}$ {\em selectors},
and of $\predecessors_{\Sigma}$ {\em testers}.
	We further require that for every constructor $c$ with arity $\sigma_1 \times \ldots \times \sigma_n \ra \sigma$, $\sigma \in \structsorts_\Sigma$ and all sorts in $\structsorts_{\Sigma}$ are well-founded w.r.t. $\elemsorts_{\Sigma}$ in $\consig{\Sigma}{\Sigma}$.
\end{definition}

Selectors are often renamed, e.g.,
linked lists are often modeled using
constructors $\nil$ and $\cons$, where the two selectors for $\cons$
are named $\head$ and $\tail$, rather than
$\sel{\cons}{0}$ and
$\sel{\cons}{1}$.

We will use the signatures in \Cref{tab:sigs} as running examples.
The signature $\exsigdt$ is a pure datatypes signature with element sort $I$
and two struct sorts, $\arrType$ and $E$.
In $\exsigdt$, the sort $\arrType$ is a list sort for $E$-elements, built using
$\nil$ and $\cons$.
The sort $E$ is built using the constructor $\container$, whose selectors are $\id$ and $\children$; intuitively, $\id$ stores an index of
sort $I$, and $\children$ stores the list of children of sort $\arrType$.
The signature $\Sigma_1$ uses the same sort names for arrays from $I$ to $E$,
with array sort $\arrType$.
Thus, $\arrType$ is not meant to be both a datatype list sort and an array sort
inside one component theory: it is a list sort in $\exsigdt$ and an array sort in
$\Sigma_1$.
The signature $\Sigma_2^{-}$ keeps only the datatype constructor $\container$,
and treats sort $\arrType$ as an element sort for datatypes.
$\Sigma_2$ extends it with the constant constructor $\emp:E$, used as an empty
$E$-value.
Finally, $\exsignwf$ and $\exsigwf$ are obtained by combining the array
signature $\Sigma_1$ with $\Sigma_2^{-}$ and $\Sigma_2$, respectively.
The formulas in the last column will be used to illustrate how cycles arise
across the datatype and array components.

\begin{table}[t]
	\newcolumntype{M}[1]{>{\centering\arraybackslash}m{#1}}
	\centering
	{
		\scriptsize
		\setlength{\tabcolsep}{2pt}
		\renewcommand{\arraystretch}{1.15}
		\begin{tabular}{|M{1.05cm}|>{\raggedright\arraybackslash}m{3cm}|>{\raggedright\arraybackslash}m{3.3cm}|>{\raggedright\arraybackslash}m{3.25cm}|>{\raggedright\arraybackslash}m{3.65cm}|}\hline
			Name                                                                      &
			Sorts                                                                     &
			Symbols                                                                   &
			Symbol arity                                                              &
			Formula                                                                            \\\hline\hline
			                                                                          &   &  &  & \\

			\shortstack{
				$\exsigdt$
			}                                                                         &
			\shortstack{
				$\elemsorts_{\exsigdt}=\{I\}$                                                      \\
				$\structsorts_{\exsigdt}=\{\arrType,E\}$
			}                                                                         &
			\shortstack{
				$\constructors_{\exsigdt}=\{\container,\nil,\cons\}$                                                        \\
				$\selectors_{\exsigdt} = \{\id,\children,\head,\tail\}$                                                            \\
				$\id=\sel{\container}{1}$                                                          \\
				$\children=\sel{\container}{2}$                                                     \\
				$\head=\sel{\cons}{1}$                                                              \\
				$\tail=\sel{\cons}{2}$                                                              \\
				$\pred{\exsigdt}=\{\is{\container},\is{\nil},\is{\cons}\}$
			}                                                                         &
			\shortstack{
				$\container:I\times\arrType\to E$                                     \\
				$\nil:\arrType$                                                       \\
				$\cons:E\times\arrType\to\arrType$                                    \\
				$\id:E\to I$                                                          \\
				$\children:E\to\arrType$                                              \\
				$\head:\arrType\to E$                                                  \\
				$\tail:\arrType\to\arrType$                                           \\
				$\is{\container}:E$                                                    \\
				$\is{\nil}:\arrType$                                                   \\
				$\is{\cons}:\arrType$
			}                                                                         &
			\shortstack{
				$\psi_1:=\is{\cons}(a)\wedge x=\head(a)$                                           \\
				$\psi_2:=\is{\container}(x)\wedge a=\children(x)$                                  \\[2pt]
				$\zeroForm:=\is{\cons}(a)\wedge x=\head(a)$                                        \\
				$\qquad{}\wedge\is{\container}(x)\wedge a=\children(x)$
			}                                                                                  \\&&&&\\\hline
			                                                                          &   &  &  & \\

			\shortstack{
				$\Sigma_1$
			}                                                                         &
			\shortstack{
				$\sorts{\Sigma_1}=\{I,E,\arrType\}$
			}                                                                         &
			\shortstack{
				$\func{\Sigma_1}=\{\select{},\store{}\}$
			}                                                                         &
			\shortstack{
				$\select{}:\arrType\times I\to E$                                     \\
				$\store{}:\arrType\times I\times E\to\arrType$
			}                                                                         &
			\shortstack{
				$\psi_1' := x=\select{}(a,i)$
			}                                                                                  \\&&&&\\\hline
			                                                                          &   &  &  & \\

			\shortstack{
				$\Sigma_2^{-}$
			}                                                                         &
			\shortstack{
				$\elemsorts_{\Sigma_2^{-}}=\{I,\arrType\}$                                         \\
				$\structsorts_{\Sigma_2^{-}}=\{E\}$
			}                                                                         &
			\shortstack{
				$\constructors_{\Sigma_2^{-}}=\{\container\}$                                                                   \\
				$\selectors_{\Sigma_2^{-}} = \{\id,\children,\head,\tail\}$                                                       \\
				$\id=\sel{\container}{1}$                                                          \\
				$\children=\sel{\container}{2}$                                                     \\
				$\pred{\Sigma_2^{-}}=\{\is{\container}\}$
			}                                                                         &
			\shortstack{
				$\container:I\times\arrType\to E$                                     \\
				$\id:E\to I$                                                          \\
				$\children:E\to\arrType$                                              \\
				$\is{\container}:E$
			}                                                                         &
			\shortstack{
				$\psi_2:=\is{\container}(x)\wedge a=\children(x)$
			}                                                                                  \\&&&&\\\hline
			                                                                          &   &  &  & \\

			$\Sigma_2$                                                                &
			$\sorts{\Sigma_2}=\sorts{\Sigma_2^{-}}$                                   &
			\shortstack{
				$\func{\Sigma_2}=\func{\Sigma_2^{-}}\cup\{\emp\}$                                                  \\
				$\pred{\Sigma_2}=\pred{\Sigma_2^{-}}\cup\{\is{\emp}\}$
			}
			                                                                          &
			\shortstack{
				$\container:I\times\arrType\to E$                                     \\
				$\emp:E$                                                              \\
				$\id:E\to I$                                                          \\
				$\children:E\to\arrType$                                              \\
				$\is{\container}:E$                                                    \\
				$\is{\emp}:E$
			}
			                                                                          &
			$\psi_2:=\is{\container}(x)\wedge a=\children(x)$
			\\&&&&\\\hline

			\vspace{1em}$\exsignwf$\vspace{1em}                                       &
			\multicolumn{3}{c|}{\shortstack{$\exsignwf := \Sigma_1\cup\Sigma_2^{-}$}} &
			\multirow{2}{*}[-0.4em]{\shortstack{
				$\formulaExUnsat := \psi_1'\wedge\psi_2 $
			}}
			\\\cline{1-4}

			\vspace{1em}$\exsigwf$\vspace{1em}                                        &
			\multicolumn{3}{c|}{$\exsigwf := \Sigma_1\cup\Sigma_2$}                   &
			\\\hline
		\end{tabular}
	}
	\caption{Signatures and formulas.}
	\label{tab:sigs}
\end{table}

\begin{example}
	\label{ex:wfdt}
	Consider the datatypes signature $\exsigdt$ from \Cref{tab:sigs}.
	Sorts $A$ and $E$ are well-founded w.r.t. $\elemsorts_{\exsigdt}$ in $\consig{\exsigdt}{\exsigdt}$:
	$G_0(\consig{\exsigdt}{\exsigdt},\elemsorts_{\exsigdt})=\{I\}$,
	$G_1(\consig{\exsigdt}{\exsigdt},\elemsorts_{\exsigdt})=\{I,A\}$, and
	$G_2(\consig{\exsigdt}{\exsigdt},\elemsorts_{\exsigdt})=\{I,A,E\}=\sorts{\exsigdt}$.
	But, if we remove constructor $\nil$,
	then both sorts are no longer well-founded.
\end{example}

\subsubsection{Semantics}

\begin{definition}[trees]
	\label{def:trees}
	Given a signature $\Sigma$,
	a set
	$S\suq\sorts{\Sigma}$ and
	an $S$-sorted set $A$,
	the set of {\em $\Sigma$-trees} over $A$ of sort $\sigma
		\in\sorts{\Sigma}$ is denoted by
	$T_{\sigma}(\Sigma,A)$ and is inductively defined as follows:
	$T_{\sigma,0}(\Sigma,A) = A_\sigma$ if $\sigma\in S$ and
	$\emptyset$ otherwise; and
	$T_{\sigma,i+1}(\Sigma,A) = T_{\sigma,i}(\Sigma,A) \cup \{
		c(t_1,\dots,t_n) \mid c\in\func{\Sigma},c:\sigma_{1}\times\ldots\times\sigma_{n}
		\ra\sigma \in \func{\Sigma}, t_j \in T_{\sigma_j,i}(\Sigma,A),
		1\leq j\leq n\}$ for each
	$i \geq 0$.
	Finally, $T_{\sigma}(\Sigma,A) = \bigcup_{i \geq 0}
		T_{\sigma,i}(\Sigma,A)$.
\end{definition}

$T_{\sigma}(\Sigma,A)$ contains the ground $\sigma$-sorted terms over $\Sigma$ with $A$ as constants.

\begin{definition}[Datatypes structure]
	\label{datatypes-axioms}
	Let $\Sigma$ be a datatypes signature and $D$ an $\elemsorts$-sorted set.
	A $\Sigma$-structure $\model$ is a {\em datatypes $\Sigma$-structure generated by $D$} if,
	for every $c\in\constructors_{\Sigma}$ of arity $(\sigma_{1}\times\ldots\times\sigma_{n})\ra\sigma$
	and every $t_j \in \sigma_{j}^{\model}$:
	$(i)$~$\sigma^{\model} = T_{\sigma}(\consig{\Sigma}{\Sigma},D)$ for every $\sigma \in \sorts{\Sigma}$;
	$(ii)$~$c^{\model}(t_1 \til t_n) = c(t_1 \til t_n)$;
	$(iii)$~$\sel{c}{i}^{\model}( c(t_1 \til t_n) )=t_i$ for $1\leq i\leq n$; and
	$(iv)$~$\is{c}^{\model}=\set{ c(t_1 \til t_n) \mid t_1 \in \sigma_{1}^{\model} \til t_n \in \sigma_{n}^{\model} }$.
	The {\em $\Sigma$-theory of datatypes} $\tdt{\Sigma}$ is the class of datatypes $\Sigma$-structures.
\end{definition}

Selectors $\sel{c}{i}$ on terms with top constructor different than $c$ are interpreted arbitrarily, as in \smtlib.

\begin{remark}
	\label{rem:axioms}
	Equivalently, the theory of datatypes is characterized by the initial model (see \cite{meinke1992universal}) of axioms $\axtestpos$, $\axtestneg$, $\axsel$ from \Cref{fig:axioms} (see, e.g., \cite{BST07}); axiom $\axcyc$ also holds in the initial model (and in all datatypes structures).
\end{remark}

\begin{figure}[t]
	\centering
	\begin{tabular}{cccc}
		$\axtestpos$~$is_{c}(c(x_1,\ldots,x_n))$      &
		$\axtestneg$~$\neg is_{c}(d(x_1,\ldots,x_n))$ &
		$\axsel$~$s_{c,i}(c(x_1,\ldots,x_n))=x_i$     &
		$\axcyc$~$x\neq t$
	\end{tabular}
	\caption{Axioms for the theory of datatypes.
		In axiom $\axcyc$, $t$ is built solely by constructor symbols
		$CO_{\Sigma}=\constructors_{\Sigma}$, 
\yoni{i don't understand the last equality. Do we have CO (non-macro) as a notation defined anywhere?}
and $t$ properly contains $x$.}
	\label{fig:axioms}
\end{figure}

\begin{figure}[t]
	\centering
	\begin{tabular}{|c|c|c|c|}\hline
		Theory                              &

		\scalebox{0.7}{
			\begin{tikzpicture}[baseline={(current bounding box.center)}]

				\def\top{0.5}
				\def\bot{-0.5}
				\pgfmathsetmacro{\midy}{(\top + \bot)/2}
				\def\leftx{0}
				\def\numslots{1}
				\pgfmathsetmacro{\slotwidth}{2.6}
				\pgfmathsetmacro{\rightx}{\leftx + \numslots*\slotwidth}

				\fill[gray!50] (\leftx,\top) rectangle (\rightx,\bot);
				\draw[outerframe] (\leftx,\top) rectangle (\rightx,\bot);
				\node {\textbf{}};
				\pgfmathsetmacro{\labelx}{\leftx + 0.5*\slotwidth}
				\node at (\labelx,-0.8) {\textbf{a}};
				\node at (\labelx,-1.2) {\textbf{}};
				\node at (\leftx+0.15*\slotwidth, 0.05*\top) {\textbf{x}};
				\draw[divider] (\leftx+0.3*\slotwidth,\top) -- (\leftx+0.3*\slotwidth,\bot);

			\end{tikzpicture}
		}
		                                    &

		\scalebox{0.7}{
			\begin{tikzpicture}[baseline={(current bounding box.center)}]

				\def\top{0.5}
				\def\bot{-0.5}
				\pgfmathsetmacro{\midy}{(\top + \bot)/2}
				\def\leftx{1.2}
				\def\numslots{1}
				\pgfmathsetmacro{\slotwidth}{2.6}
				\pgfmathsetmacro{\rightx}{\leftx + \numslots*\slotwidth}

				\pgfmathsetmacro{\bgpadtop}{0.3}
				\pgfmathsetmacro{\bgpadbot}{0.5}
				\path[outerbg] (-0.5,\bot - \bgpadbot) rectangle (\rightx + 0.5,\top + \bgpadtop);

				\node[outerid, anchor=north west] (id) at (0,\top) {0};

				\fill[gray!50] (\leftx,\top) rectangle (\rightx,\bot);
				\draw[outerframe] (\leftx,\top) rectangle (\rightx,\bot);
				\node at ($(id.south)+(0,-0.3)$) {\textbf{}};
				\pgfmathsetmacro{\labelx}{\leftx + 0.5*\slotwidth}
				\node at (\labelx,-0.8) {\textbf{a}};
				\node at (\labelx - 0.7, -1.2) {\textbf{x}};
				\draw[divider] (\leftx+0.3*\slotwidth,\top) -- (\leftx+0.3*\slotwidth,\bot);

			\end{tikzpicture}
		}

		                                    &

		\scalebox{0.7}{
			\begin{tikzpicture}[baseline={(current bounding box.center)}]

				\def\top{0.5}
				\def\bot{-0.5}
				\pgfmathsetmacro{\midy}{(\top + \bot)/2}
				\def\leftx{1.2}
				\def\numslots{1}
				\pgfmathsetmacro{\slotwidth}{2.6}
				\pgfmathsetmacro{\rightx}{\leftx + \numslots*\slotwidth}

				\pgfmathsetmacro{\bgpadtop}{0.3}
				\pgfmathsetmacro{\bgpadbot}{0.5}
				\path[outerbg] (-0.5,\bot - \bgpadbot) rectangle (\rightx + 0.5,\top + \bgpadtop);

				\node[outerid, anchor=north west] (id) at (0,\top) {0};

				\fill[gray!50] (\leftx,\top) rectangle (\rightx,\bot);
				\draw[outerframe] (\leftx,\top) rectangle (\rightx,\bot);
				\node at ($(id.south)+(0,-0.3)$) {\textbf{}};
				\pgfmathsetmacro{\labelx}{\leftx + 0.5*\slotwidth}
				\node at (\labelx,-0.8) {\textbf{a}};
				\node at (\leftx+0.15*\slotwidth, 0.05*\top) {\textbf{x}};
				\node at (\labelx - 0.7, -1.2) {\textbf{x}};
				\draw[divider] (\leftx+0.3*\slotwidth,\top) -- (\leftx+0.3*\slotwidth,\bot);

			\end{tikzpicture}
		}
		\\\hline
		\multirow{2}{*}{$\tdt{\exsigdt}$}   & $\psi_1$  & $\psi_2$  & $\psi_1\wedge\psi_2$  \\
		                                    & \greensat & \greensat & \greenunsat           \\\hline
		\multirow{2}{*}{$\tpndt{\exsigwf}$} & $\psi_1'$ & $\psi_2$  & $\psi_1'\wedge\psi_2$ \\
		                                    & \greensat & \greensat & \redsat               \\\hline
		\multirow{2}{*}{$\tndt{\exsigwf}$}  & $\psi_1'$ & $\psi_2$  & $\psi_1'\wedge\psi_2$ \\
		                                    & \greensat & \greensat & \greenunsat           \\\hline
	\end{tabular}
	\caption{Formalization of cycles.
		$\psi_1$,
		$\psi_1'$, and
		$\psi_2$ are defined in \Cref{tab:sigs}.
	}
	\label{fig:dt-ndt-ex2}
\end{figure}

\begin{example}
	\label{ex:dt-semantics}
	Consider again $\exsigdt$ from \Cref{tab:sigs},
	and let $\model$ be a $\tdt{\exsigdt}$-interpretation
	with $I$ interpreted as the integers.
	Then, e.g., $\container(0,\nil)$ and
	$\container(0,\cons(\container(0,\nil),\nil))$ are elements of $E^{\model}$.
	The $\exsigdt$-formula $\zeroForm$ from \Cref{tab:sigs} is $\tdt{\exsigdt}$-unsatisfiable,
	as it violates $\axcyc$: \Cref{fig:dt-ndt-ex2}'s first row shows why -- $x$ cannot contain itself.
\end{example}

\subsection{Decidability and Combination}
\label{ex:decidableth}
\label{ex:stably-infinite}
\looseness=-1
A $\Sigma$-theory is {\em decidable} if the set of $\T$-satisfiable flat $\Sigma$-cubes is decidable;
arrays, EUF, and datatypes are decidable (see, e.g., \cite{BradleyManna2007}).
Given signatures $\Sigma_1, \Sigma_2$ and theories $T_1, T_2$ over them,
$\Sigma_1\cup\Sigma_2$ is the union of their sort, function, and predicate symbols,
and the {\em combination} $T_1\oplus T_2$ is the class of $\Sigma_1\cup\Sigma_2$-structures $\model$
with $\model^{\Sigma_1}\in T_1$ and $\model^{\Sigma_2}\in T_2$.
$T$ is \emph{stably infinite w.r.t.\ $S\subseteq\sorts{\Sigma}$} if every $T$-satisfiable
quantifier-free $\Sigma$-formula $\formulaProof$ has a $T$-interpretation $\model\models\formulaProof$
with $\lvert\sigma^{\model}\rvert$ infinite for all $\sigma\in S$.
EUF and arrays are stably infinite w.r.t.\ all their sorts (see, e.g., \cite{BradleyManna2007});
for a datatypes signature $\Sigma$, $\tdt{\Sigma}$ is stably infinite w.r.t.\ $\elemsorts_{\Sigma}$ (see, e.g., \cite{BST07}).
The Nelson-Oppen (NO) method~\cite{NelsonOppen,10.1007/978-3-540-30227-8_53} decides satisfiability
in the combination of stably infinite theories, by reasoning about {\em arrangements} (conjunctions of (dis)equalities) between shared variables.


\section{The Theory of Nested datatypes}
\label{sec:theoryndtsec}
In this section we define the theory of nested
datatypes.
There are two challenges in doing so, one syntactic, and the other semantic.
The first invites a new definition for a well-foundedness of a sort, that must take into account the existence of arrays.
The second requires to exclude unwanted structures from the defined theory.

\subsection{First challenge: well-foundedness}
\label{sec:ndtsig}

Going back to \Cref{tab:sigs}, notice that
$\zeroForm$ uses
the theory $\tdt{\exsigdt}$ of datatypes to represent a nested structure.
This is
natural, but sometimes not expressive enough.
For example, $\tdt{\exsigdt}$ does not have an operator
to update such lists, or a random-access reading operator.
Thus, it might be desirable to have the second field of datatype $E$
be an {\em array} of datatypes, rather than a list,
as the next example does.

\begin{example}
	\label{ex:wfndt}
	Consider signature $\exsignwf$, obtained as the union of
	$\Sigma_1$ and $\Sigma_2^{-}$,
	from \Cref{tab:sigs}.
	$\Sigma_1$ is an arrays signature, while
	$\Sigma_2^{-}$ is a datatypes signature.
\end{example}

Noticeably, $\exsignwf$ no longer has the
constructor $\nil$, that was available in $\exsigdt$.
There, it was used as an empty list constructor.
However, we have now replaced lists with arrays, and there are no
empty arrays.
But now we have a problem.
The well-foundedness condition from \Cref{def:well-sorted-dt} that was ensured
by $\nil$ in \Cref{ex:wfdt} is now lacking:
in terms of \Cref{def:well-sorted-dt}, $G_1$ will no longer include the sort
$A$, and so $G_2$ will not include the sort $E$.
Thus, replacing lists by arrays requires an appropriate
definition of well-foundedness,
that explicitly incorporates arrays sorts.

\begin{definition}[Well-founded sort (NDT)]
	\label{def:theory:well-founded}
	Let
	$\Sigma_1,\dots,\Sigma_n$ be arrays signatures, each with $\sorts{\Sigma_k} =
		\{\arrType_k, I_k, E_k\}$
	for each $1\leq k\leq n$, and let $\Sigma_{n+1}$
	be a datatypes signature.
	Let
	$\Sigma \;=\; \bigcup_{k=1}^{n+1} \Sigma_k$.
	%
	For a set $S \subseteq \sorts{\Sigma}$, set:
	$F_{0}(\Sigma,S) =S$;
	$F_{i+1}(\Sigma,S) =
		F_{i}(\Sigma,S) \cup \{\sigma \mid \exists
		c\in\func{\Sigma}.c:\sigma_1 \times \ldots \times
		\sigma_n \ra \sigma, \sigma_1, \ldots, \sigma_n \in F_{i}(\Sigma,S) \} \cup
		\{\arrType_k \mid E_k,I_k \in F_{i}(\Sigma,S)\}$.
	A sort $\sigma$
	is \emph{well-founded in $\Sigma$}
	w.r.t. $S$
	if \( \sigma \;\in\; \bigcup_{i\ge 0}\!
	F_i\Bigl(\Sigma,\;S\Bigr) \).

\end{definition}

Notice that while in \Cref{def:well-sorted-dt} well-foundedness is defined for datatypes signatures,
in \Cref{def:theory:well-founded} it is defined for signatures obtained by a union
of a datatypes signature with a number of array signatures.

\begin{example}
	\label{ex:formula}
	\label{ex:well-founded-sig}
	Using \Cref{def:theory:well-founded}, we see that sorts $\arrType$ and $E$ are not well-founded in signature $\consig{\exsignwf}{\Sigma_2^{-}}$ w.r.t. $\{I\}$:
	indeed,
	$F_{0}(\consig{\exsignwf}{\Sigma_2^{-}}, \{I\})=\{I\}$ for every $i$.
	However, denote by $\Sigma_2$ the signature obtained
	from $\Sigma_2^{-}$ by adding a nullary constructor $\emp$ of sort $E$, and let $\exsigwf=\Sigma_1\cup\Sigma_2$, as described in \Cref{tab:sigs}.
	We now get that both $E$ and $\arrType$ are well-founded
	in $\consig{\exsigwf}{\Sigma_2}$ w.r.t. $\{I\}$.
\end{example}

We shall also require all array sorts to be distinct from one another
and from the datatype sorts.
We further assume that the nesting of arrays and datatypes only
involves the element sort of arrays, not the index sort;
and that there is no nesting between arrays.
This is formalized as follows:

\begin{notation}
	\label{not:ndt-sig}
	Let $\Sigma_1, \dots, \Sigma_n$ be array signatures,
	s.t. $\sorts{\Sigma_j} = \{\arrType_j, I_j,
		E_j\}$ for all $1 \leq j \leq n$.
	Let
	$\Sigma_{n+1}$
	be a signature s.t.
	$\Sigma_1,\ldots,\Sigma_{n+1}$ are pairwise disjoint.
	For $\Sigma := \bigcup_{i=1}^{n+1} \Sigma_i$,
	$\arrays_{\Sigma} := \{\arrType_j \mid 1 \leq j \leq
		n\}$.
\end{notation}

\begin{definition}[NDT signature]
	\label{sec:theory:ndt-sig}
	Let $\Sigma_1,\ldots,\Sigma_n$ be arrays signatures,
	each with
	$\sorts{\Sigma_i} = \{\,E_i,\, I_i,\, \arrType_i\}$,
	and
	let $\Sigma_{n+1}$ be a datatypes signature
	and denote
	$\sorts{\Sigma_{n+1}} = \elemsorts_{\Sigma_{n+1}} \uplus \structsorts_{\Sigma_{n+1}}$.
	Let $\Sigma = \bigcup_{k=1}^{n+1}\Sigma_k$.
	Suppose the sorts in $\arrays_{\Sigma}\,\cup\, \structsorts_{\Sigma_{n+1}}$ are well-founded w.r.t.
	$\sorts{\Sigma}\,\setminus\,(\arrays_{\Sigma}\,\cup\, \structsorts_{\Sigma_{n+1}})$,
	$|\arrays_\Sigma|=n$,
	$\{I_1,\ldots,I_n\} \;\cap\; (\structsorts_{\Sigma_{n+1}} \cup \arrays_{\Sigma}) \;=\; \emptyset$,
	and
	$\structsorts_{\Sigma_{n+1}} \,\cap\, \arrays_{\Sigma} \;=\; \emptyset$.
	Then $\Sigma$ is called the \emph{NDT signature} of
	$\Sigma_1,\ldots,\Sigma_{n+1}$.
\end{definition}

\begin{example}
	\label{ex:ndt-sig}
	Both $A$ and $E$ in $\exsigwf$ from \Cref{tab:sigs} are well founded w.r.t. $\{I\}$,
	$|\arrays_{\exsigwf}|=1$, $\{I\} \cap (\structsorts_{\Sigma_2} \cup \arrays_{\exsigwf}) = \emptyset$, and
	$\structsorts_{\Sigma_2} \,\cap\, \arrays_{\exsigwf}=\emptyset$.
	Therefore, $\exsigwf$ can be called the NDT signature of
	$\Sigma_1,\Sigma_2$.
\end{example}


\subsection{Second challenge: non-standard interpretations}
\label{sec:ndtsemantics}

Let us start with the most natural candidate as semantics for nested datatypes,
obtained by simply taking the theory combination of datatypes and arrays.

\begin{definition}[Pre-NDT theory]
	\label{sec:theory:pre-ndt-theory}
	Let $T_1,\ldots,T_n$ be array theories with signatures
	$\Sigma_1,\ldots,\Sigma_n$.
	Let $T_{n+1}$ be a datatypes theory with signature $\Sigma_{n+1}$.
	Assume that $\Sigma$ is the NDT signature of $\Sigma_{1},\ldots,\Sigma_{n+1}$.
	Then $\tpndt{\Sigma} := \bigoplus_{k=1}^{n+1}T_k$
	is the \emph{$\Sigma$-theory of pre-NDTs}.
\end{definition}

As the next example shows, this theory exhibits unexpected behavior.

\begin{example}
	\label{ex:NO}
	Let
	$T_1$ be the array theory over $\Sigma_1$ and $T_2$ the $\Sigma_2$-theory of datatypes,
	where $\Sigma_1$ and $\Sigma_2$ are as in \Cref{tab:sigs}.
	Then $\tpndt{\exsigwf} = T_1 \oplus T_2$ is the $\exsigwf$-theory of pre-NDTs.
	Consider the formula
	$\formulaExUnsat$
	from \Cref{tab:sigs},
	obtained from $\zeroForm$ by replacing the datatype selector
	with an array select.
	$\formulaExUnsat$
	is
	$\tpndt{\exsigwf}$-satisfiable, even though it can be viewed as a variant
	of the $\exsigdt$-formula $\zeroForm$,
	which is unsatisfiable in the theory of $\exsigdt$-datatypes.
	This is unsatisfactory.
	The intended meaning of $\formulaExUnsat$ is depicted in the last column
	of \Cref{fig:dt-ndt-ex2} -- the same meaning as $\zeroForm$, just
	using arrays.
	Its $\tpndt{\exsigwf}$-satisfiability can be shown using the
	Nelson-Oppen method:
	it suffices to show that
	$\psi_1'$ is $T_1$-satisfiable
	and $\psi_2$ is $T_2$-satisfiable,
	as for each sort they share at most one variable,
	and so only trivial arrangements exist.
\end{example}

We have seen that $\formulaExUnsat$ from \Cref{tab:sigs}
is satisfiable in the theory $\tpndt{\Sigma}$, even though
intuitively, it should not be.
The next definition identifies the problem with that
interpretation, and allows us to restrict our attention
to standard interpretations.
This is done by defining a relation in $\tpndt{\Sigma}$-structures,
and require its well-foundedness.
The relation connects input and output sorts of both arrays and datatypes.
Its formal definition is quite involved, but intuitively it simply
relates terms to their sub-terms and arrays to their elements.
With this relation, we can finally define the theory of nested datatypes that we will adopt.

\begin{definition}[NDT relation and theory]
	\label{sec:theory:nested-relation}
	\label{sec:theory:ndt-sigma-structure}
	Let $T_1,\ldots,T_n$ be the array theories of signatures
	$\Sigma_1,\ldots,\Sigma_n$ and
	$T_{n+1}$ the theory of $\Sigma_{n+1}$-datatypes.
	Let $\Sigma$ be the NDT signature of $\Sigma_{1},\ldots,\Sigma_{n+1}$.
	Further assume that $\func{\Sigma_{n+1}} = \constructors_{\Sigma_{n+1}} \uplus \selectors_{\Sigma_{n+1}}$,
	and $\sorts{\Sigma_i} = \{\arrType_i, I_i, E_i\}$
	for each $1\leq i\leq n$.
	Let $\model$ be a $\tpndt{\Sigma}$-structure.
	The {\em NDT-relation of $\model$} (denoted $\rel{\model}$) is
	defined as follows:
	$(\alpha,\beta)\in\rel{\model}$ if and only if one of the following
	conditions holds:
	$(i)$~$\alpha=t_j$ and $\beta=c(t_1,\ldots,t_n)$ for some
	$c\in\constructors_{\Sigma_{n+1}}$ with arity
	$(\sigma_{1}\times\ldots\times\sigma_{m})\ra\sigma$
	such that
	$t_1\in\sigma_1^\model,\ldots,t_m\in\sigma_m^{\model}$
	and $1\leq j\leq m$;
	$(ii)$~$\alpha=\select{\arrType_i}^{\model}(a,j)$ and $\beta=a$
	for some
	$1\leq i\leq n$, $a\in {\arrType_i}^{\model}$
	and $j\in I_i^{\model}$.
	We say
	that $\model$ is an \emph{NDT $\Sigma$-structure} if
	$\rel{\model}$ is well-founded.\footnote{The term {\em well-founded} has been overloaded.
		\Cref{tab:wdsum} summarizes its various usages.}
	The theory of $\Sigma$-nested datatypes, denoted $\tndt{\Sigma}$
	is the class of all NDT $\Sigma$-structures.
\end{definition}

\begin{table}[t]
	\centering
	\begin{tabular}{|c|c|c|}\hline
		Notion                                    & Definition                     & Usage                             \\\hline
		Well-founded relations                    & \Cref{prelim:wfr}              & \Cref{sec:theory:nested-relation} \\
		Well-founded sort in datatypes signatures & \Cref{def:well-sorted-dt}      & \Cref{def:IDT_sig}                \\
		Well-founded sort in NDT signatures       & \Cref{def:theory:well-founded} & \Cref{sec:theory:ndt-sig}         \\
		\hline
	\end{tabular}
	\caption{Summary of definitions and usages of the term ``well-founded'' in this paper.}
	\label{tab:wdsum}
\end{table}

\Cref{sec:theory:ndt-sigma-structure} fixes the problem raised in \Cref{ex:NO}:

\begin{example}
	\label{ex:varphi-unsat}
	The formula $\formulaExUnsat$ from
	\Cref{tab:sigs}
	is $\tndt{\exsigwf}$-unsatisfiable, as mentioned in
	\Cref{fig:dt-ndt-ex2}.
\end{example}

\section{Decision Procedure for Nested Datatypes}
\label{sec:theory:preprocessing}
In this section we introduce a decision procedure for the theory
of nested datatypes.
In what follows,
let $\Sigma$ be the NDT signature of $\Sigma_1,\ldots,\Sigma_{n+1}$,
where $\Sigma_1,\dots,\Sigma_n$ are arrays signatures and $\Sigma_{n+1}$ is a datatypes signature.
Also, let $\formulaProof$ be a $\Sigma$-formula.
Our procedure will determine the
$\tndt{\Sigma}$-satisfiability of $\formulaProof$.

The procedure works as follows:
we define a new signature $\newSig{\Sigma}{\formulaProof}$ and a
$\newSig{\Sigma}{\formulaProof}$-formula $\trF(\formulaProof)$ such that
$\formulaProof$ is $\tndt{\Sigma}$-satisfiable
if and only if
$\trF(\formulaProof)$ is satisfiable in the
$\newSig{\Sigma}{\formulaProof}$-theory obtained by combining
the theories of datatypes, arrays, and uninterpreted functions
(but without the restriction to NDT structures).
The latter can be decided using existing SMT-solvers
for the resulting combination, in a Nelson-Oppen fashion.

\subsection{Translation of Signatures and Theories}
\label{sec:sigstrans}

We begin by defining new sorts.
The idea is to have the theories of arrays and datatypes separated,
so that datatypes will not have array sorts
as element sorts.
The connection between the datatypes and the arrays from the original signature will be maintained by uninterpreted functions.

For every sort $\sigma$ in $\sorts{\Sigma}$
we define a new sort $\trtype{\sigma}$.
For each array sort $\arrType$ in $\arrays_{\Sigma}$,
we introduce two additional sorts:
$\trarr{\arrType}$ and $\idx{\arrType}$.
In total, each sort $\arrType$ in $\arrays_{\Sigma}$
gives rise to three new sorts:
$\trtype{\arrType}$, $\trarr{\arrType}$ and $\idx{\arrType}$;
and each other sort $\sigma\in\sorts{\Sigma}\setminus\arrays_{\Sigma}$
gives rise to a single new sort $\trtype{\sigma}$.
The new sorts are summarized in~\Cref{tab:new-sorts}.
Intuitively, for each non-array sort $\sigma$,
we create a copy of it named $\trtype{\sigma}$.
But if $\sigma$ is an array sort,
$\trtype{\sigma}$ is no longer a copy of it.
Instead, it is a datatypes sort that is meant to represent
the original array sort.
Additionally, $\trarr{\sigma}$ is an array sort, and its
element sort may represent a datatype, but not a nested one.
Thus, every array sort $\sigma$ is represented by two sorts:
$\trsig{\sigma}$ and $\trarr{\sigma}$.
The sort $\trsig{\sigma}$ will be used only for checking cycles,
whereas $\trarr{\sigma}$ is used for the regular array reasoning.
Finally, $\idx{\sigma}$ is an additional sort that serves
as a technical separator
between
datatypes with sort $\trtype{\sigma}$.
Then, the new signature $\newSig{\Sigma}{\formulaProof}$ is defined as follows.

\begin{table}[t]
	\begin{minipage}[t]{0.35\linewidth}
		\centering\footnotesize
		\begin{tabular}{|c|c|c|}\hline
			Sort                                      & Set & Sorts   \\\hline
			$\sigma$                                  &
			$\sorts{\Sigma}\setminus\arrays_{\Sigma}$ &
			$\trtype{\sigma}$                                         \\\hline
			$\arrType_i$                              &
			$\arrays_{\Sigma}$                        &
			$\trtype{\arrType_i},\trarr{\arrType_i},\idx{\arrType_i}$ \\\hline
		\end{tabular}
		\caption{New sorts.}
		\label{tab:new-sorts}
	\end{minipage}%
	\hfill
	\begin{minipage}[t]{0.64\linewidth}
		\centering\footnotesize
		\begin{tabular}{|c|c|c|}\hline
			$\sorts{\trsig{\Sigma_i}}$ &
			$\{\trtype{I_i},\trtype{E_i},\trarr{\arrType_i}\}$                 \\\hline
			$\func{\trsig{\Sigma_i}}$  &
			$\{\select{\trarr{\arrType_i}}:(\trarr{\arrType}_i \times \trtype{I_i}) \ra \trtype{E_i},
				\store{\trarr{\arrType_i}}:(\trarr{\arrType}_i
				\times \trtype{I_i} \times \trtype{E_i}) \ra \trarr{\arrType}_i\}$ \\\hline
			$\pred{\trsig{\Sigma_i}}$  & $\emptyset$                           \\\hline
		\end{tabular}
		\caption{Signatures $\trsig{\Sigma_i}$.
			Arities are inline.}
		\label{tab:new-signature-1}
	\end{minipage}
\end{table}

\begin{definition}
	\label{def:new-sig}
	$\newSig{\Sigma}{\formulaProof} = \Bigl(\bigcup_{i=1}^{n} \trsig{\Sigma_i}\Bigr)
		\cup {\Sigma^\formulaProof_{n+1}}
		\cup \Sigma_{n+2}$,
	where for each
	$1\leq i\leq n$,
	$\trsig{\Sigma_i}$ is given in \Cref{tab:new-signature-1},
	$\Sigma^{\formulaProof}_{n+1}$ is given in \Cref{fig:dt-sig-trans},
	and $\Sigma_{n+2}$ is given in \Cref{fig:one-to-one-sig}.
\end{definition}

Notice that the $n+2$ signatures
$\trsig{\Sigma_1},\ldots,\trsig{\Sigma_n},
	\Sigma^{\formulaProof}_{n+1},\Sigma_{n+2}$
may share sorts, but are otherwise disjoint.
Let us elaborate on $\newSig{\Sigma}{\formulaProof}$ by focusing on each
of the constructed signatures separately.
We start with signatures
$\trsig{\Sigma_1},\ldots,\trsig{\Sigma_n}$,
given in \Cref{tab:new-signature-1}.
Let $1\leq i\leq n$.
$\trsig{\Sigma_i}$ is an arrays signature,
whose index sort is
$\trtype{I_i}$,
its element sort is
$\trtype{E_i}$, and its array sort is
$\trarr{A_i}$.
As an arrays signature, it has two function symbols
$\select{\trarr{\arrType_i}},
	\store{\trarr{\arrType_i}}$,
with the appropriate arities.

Next, we describe ${\Sigma^\formulaProof_{n+1}}$.
It relies on the number of terms used
under selects and stores in $\formulaProof$.

\begin{definition}
	\label{def:free-vars-array-set}
	For a sort $\arrType_i \in \arrays_\Sigma$,
	and a flat $\Sigma$-cube $\formulaProof$,
	let $\indexSet{\arrType_i}{\formulaProof}$
	be the set of variables $x$
	in $\fv{I_i}{\formulaProof}$ such that
	$\select{\arrType_i}(a,x)$ or
	$\store{\arrType_i}(a,x,v)$ occur in $\formulaProof$
	for some $a$, $x$ and $v$.
	$\indexCar{\arrType_i}{\formulaProof}$ is its size.
\end{definition}

\begin{example}
	\label{ex:free-vars-array-set}
	Referring to $\exsigwf$-formula $\formulaExUnsat$ from \Cref{tab:sigs},
	we have $\indexSet{\arrType}{\formulaExUnsat}=\{i\}$ and $\indexCar{\arrType}{\formulaExUnsat}=1$.
\end{example}

$\Sigma^{\formulaProof}_{n+1}$ is defined in \Cref{fig:dt-sig-trans}.
The idea is to mimic
$\Sigma_{n+1}$, while avoiding arrays completely.
Let us start with its sorts.
%
For every $\sigma \in \structsorts_{\Sigma_{n+1}}$,
we have $\trtype{\sigma} \in
	\structsorts_{\Sigma^{\formulaProof}_{n+1}}$, and for every sort $\sigma \in
	\elemsorts_{\Sigma_{n+1}} \cup \{E_i \mid i \in [1,n]\}$ that is not in
$\arrays_{\Sigma}$ or $\structsorts_{\Sigma_{n+1}}$,
we have $\trtype{\sigma} \in \elemsorts_{\Sigma^{\formulaProof}_{n+1}}$.
For every $\arrType_i$, where $i \in [1,n]$,
we have $\trtype{\arrType_i} \in \structsorts_{\Sigma^{\formulaProof}_{n+1}}$.

For each constructor $c : (\rho_1 \times \ldots \times \rho_m) \ra \sigma$ in $\constructors_{\Sigma_{n+1}}$,
we create a corresponding constructor
$\trsig{c} : (\trtype{\rho_1} \times \ldots \times \trtype{\rho_m}) \ra \trtype{\sigma}$
in $\constructors_{\Sigma^{\formulaProof}_{n+1}}$.
Intuitively, $\trsig{c}$ is a copy of $c$ that operates on the transformed sorts.
For every $1\leq i\leq n$, we create a single constructor $\arrCons{i}$ of arity
$(\idx{\arrType_i} \times {\trtype{E_i} \times \ldots \times \trtype{E_i}}) \ra \trtype{\arrType_i}$.
Intuitively, this constructor represents an array using a finite tuple:
the first argument (of sort $\idx{\arrType_i}$) is a unique identifier for the array,
and the remaining $\indexCar{\arrType_i}{\formulaProof}$ arguments
store the elements that can be accessed using the indices from $\indexSet{\arrType_i}{\formulaProof}$.
The reason for the first argument is to ensure that if two arrays
are different, due to indices not in $\indexSet{\arrType_i}{\formulaProof}$,
their translations will remain different.
We can always assume w.l.o.g that $\indexCar{\arrType_i}{\formulaProof} \geq 1$;
if it is zero, we add a literal $\select{\arrType_i}(a,j) = v$ to $\formulaProof$ with fresh variables $a$, $j$, and $v$.
Once the sorts and constructors of $\Sigma^{\formulaProof}_{n+1}$ are defined,
the selectors and testers are induced as usual for datatypes.
Notice,
$\Sigma^{\formulaProof}_{n+1}$ itself is a datatypes signature:

\begin{restatable}{lemma}{sigisdtsigtr}
	\label{lem:well-foundedness-n+1}
	${\Sigma^\formulaProof_{n+1}}$ is a datatypes signature.
\end{restatable}

The last signature that is left to discuss is $\Sigma_{n+2}$,
presented in \Cref{fig:one-to-one-sig}.
Its purpose is to relate
the new sorts $\trarr{\arrType}$ and $\trtype{\arrType}$ to one another.
For each $\arrType\in\arrays_{\Sigma}$,
we have both $\trarr{\arrType}$ and $\trtype{\arrType}$
as sorts in $\Sigma_{n+2}$,
as well as two function symbols,
$f_{\arrType}$ and $g_{\arrType}$,
with arities
$\trtype{\arrType} \rightarrow \trarr{\arrType}$
and
$ \trarr{\arrType}\rightarrow \trtype{\arrType}$

\begin{table}[t]
	\hspace*{-1cm}%
	\begin{minipage}[t]{0.6\linewidth}
		\centering\footnotesize
		\begin{tabular}{|c|c|}\hline
			$\elemsorts_{{\Sigma^\formulaProof_{n+1}}}$                                              &
			$\{\trtype{\sigma}\mid \sigma\in\elemsorts_{\Sigma_{n+1}}\setminus \arrays_{\Sigma}\}\cup$                                                                                                  \\ & $
				\{\idx{\arrType_i}\mid i\in [1,n]\} \cup$                                                                                                                                                   \\
			                                                                                         & $\{\trsig{E_i} \mid i \in [1,n], E_i \notin (\arrays_\Sigma \cup \structsorts_{\Sigma_{n+1}})\}$
			\\\hline
			$\structsorts_{{\Sigma^\formulaProof_{n+1}}}$
			                                                                                         &
			$\{\trtype{\sigma}\mid
				\sigma\in\structsorts_{\Sigma_{n+1}}\cup\arrays_{\Sigma}\}$
			\\\hline
			$\constructors_{{\Sigma^\formulaProof_{n+1}}}$                                           &
			$\{\trsig{c}\mid c\in\constructors_{\Sigma_{n+1}}\}\cup
				\{\arrCons{i}\mid 1\leq i\leq n\}$
			\\\hline
			$\selectors_{{\Sigma^\formulaProof_{n+1}}},\predecessors_{{\Sigma^\formulaProof_{n+1}}}$ &
			As induced by $\constructors_{{\Sigma^\formulaProof_{n+1}}}$
			\\\hline
			\multirow{2}{*}{Arities}                                                                 &
			$\trsig{\container}:
				(\trtype{\sigma_{1}}\times\ldots\times\trtype{\sigma_{m}})\rightarrow\trtype{\sigma}$                                                                                                       \\ &
			$\arrCons{i}:
				(\idx{\arrType_{i}}\times\underbrace{\trtype{E_i}\times\ldots\times\trtype{E_i}}_{\indexCar{\arrType_i}{\formulaProof}})\rightarrow\trtype{\arrType_{i}}$
			\\\hline
		\end{tabular}
		\caption{The datatypes signature $\Sigma^{\formulaProof}_{n+1}$.
		}
		\label{fig:dt-sig-trans}
	\end{minipage}%
	\hfill
	\begin{minipage}[t]{0.39\linewidth}
		\centering\footnotesize
		\begin{tabular}{|c|c|}\hline
			$\sorts{\Sigma_{n+2}}$ &
			$\{\trtype{\arrType_i}\mid 1\leq i \leq n\}\cup$ \\ & $
				\{\trarr{\arrType_i}\mid 1\leq i \leq n\}$       \\\hline
			$\func{\Sigma_{n+2}} $ &
			$\{f_{\arrType_i}\mid 1\leq i \leq n\}\cup$      \\ & $\{g_{\arrType_i}\mid 1\leq i \leq n\}$ \\\hline
			$\pred{\Sigma_{n+2}} $ & $\emptyset$             \\\hline
			Arities                &
			$f_{\arrType}:
				\trtype{\arrType}\ra\trarr{\arrType}$,
			$g_{\arrType}:
				\trarr{\arrType}\ra\trtype{\arrType}$.
			\\\hline
		\end{tabular}
		\caption{Signature $\Sigma_{n+2}$.
		}
		\label{fig:one-to-one-sig}
	\end{minipage}
\end{table}

\begin{example}
	\label{ex:new-sig-sigma-n+1}
	Recall $\exsigwf$
	and $\formulaExUnsat$ from~\Cref{tab:sigs}.
	$\trsig{\Sigma_1}$,
	${\Sigma_2^\formulaExUnsat}$, and
	$\Sigma_3$ are
	given in \Cref{tab:sigs-tran}.
	$\trsig{\container}$ and $\trsig{\emp}$
	are copies of
	$\container$ and $\emp$.
	The reason we only add $\arrCons{1}$ is
	that $\exsigwf$ has one arrays signature -- $\Sigma_1$.
	$\arrCons{1}$ represents the array selects.
	Its first argument, of sort $\idx{A}$, is a unique identifier.
	Its remaining arguments represent the elements of the array that are accessed
	from relevant indices.
	As shown in \Cref{ex:free-vars-array-set}, there is one relevant index.
	Thus, $\arrCons{1}$ has one more input sort,
	$\trtype{E}$, representing that element.
	%

\end{example}

\begin{table}[t]
	\newcolumntype{M}[1]{>{\centering\arraybackslash}m{#1}}
	\centering
	{
		\scriptsize
		\begin{tabular}{|M{.8cm}|M{3.6cm}|M{5.5cm}|M{3.7cm}|}\hline
			Name                                                              &
			Sorts                                                             &
			Symbols                                                           &
			Formula                                                                                                                    \\\hline\hline
			                                                                  &                        &   &                           \\

			\shortstack{
				$\trsig{\Sigma_1}$
			}                                                                 &
			\shortstack{
				$\sorts{\trsig{\Sigma_1}} = \{\trarr{\arrType}, \trtype{I}, \trtype{E}\}$
			}                                                                 &
			\shortstack{
				$\func{\trsig{\Sigma_1}} = \{\select{\trarr{\arrType}}, \store{\trarr{\arrType}}\}$
			}                                                                 &
			\shortstack{
				$\tra(\psi_1') := \trsig{x}=\select{}(\trarr{a},\trsig{i})$
			}
			\\&&&\\\hline
			                                                                  &                        &   &                           \\

			$\Sigma_2^{\formulaExUnsat}$                                      &
			\shortstack{
				$\sorts{{\Sigma_2^\formulaExUnsat}} = \elemsorts_{\Sigma_2^\formulaExUnsat} \uplus
					\structsorts_{\Sigma_2^\formulaExUnsat}$                                                                                   \\
				$\elemsorts_{\Sigma_2^\formulaExUnsat} = \{ \trsig{I}, \idx{\arrType} \}$                                                  \\
				$\structsorts_{\Sigma_2^\formulaExUnsat} = \{ \trsig{\arrType}, \trsig{E} \}$
			}
			                                                                  &
			\shortstack{
				$\func{\Sigma_2^\formulaExUnsat} = \constructors_{\Sigma_2^\formulaExUnsat} \cup
					\selectors_{\Sigma_2^\formulaExUnsat}$                                                                                     \\
				$\constructors_{\Sigma_2^\formulaExUnsat} = \{ \trsig{\container}, \trsig{\emp}, \arrCons{1} \}$                           \\
				$\selectors_{\Sigma_2^\formulaExUnsat} = \{ \trsig{\id}, \trsig{\children}, \sel{\arrCons{1}}{1}, \sel{\arrCons{1}}{2} \}$ \\~\\
				\\~\\
				$\pred{{\Sigma_2^\formulaExUnsat}} = \{ \is{\trsig{\container}}, \is{\trsig{\emp}}, \is{\arrCons{1}} \}$
			}
			                                                                  &
			\shortstack{
				$\tra(\psi_2) :=$                                                                                                          \\ $\is{\trsig{\container}}(\trsig{x})\wedge \trsig{a} = \trsig{\children}(\trsig{x}) $
			}
			\\&&&\\\hline

			$\Sigma_3$                                                        &
			$\sorts{\Sigma_3} = \{\trarr{\arrType}, \trsig{\arrType}\}$       &
			\shortstack{
				$\func{\Sigma_3} = \{f_{\arrType}, g_{\arrType}\}$
			}                                                                 &
			$ \trarr{a} = f_{\arrType}(\trsig{a})$
			\\\hline
			                                                                  &
			\multicolumn{2}{|c|}{}                                            &                                                        \\
			$\newSig{\exsigwf}{\formulaExUnsat}$                              &
			\multicolumn{2}{|c|}{
			$\trsig{\Sigma_1} \cup {\Sigma_2^\formulaExUnsat} \cup \Sigma_3$} &
			$\tra(\formulaExUnsat)=\tra(\psi_1')\wedge\tra(\psi_2)$                                                                    \\
			                                                                  & \multicolumn{2}{|c|}{} &                               \\
			\hline
			\shortstack{Sorts                                                                                                          \\and\\Arities} &
			\multicolumn{3}{>{\centering\arraybackslash}m{12.8cm}|}{
			\shortstack{
				$\trsig{\container}:\trtype{I} \times \trtype{A} \ra \trtype{E}$,
				$\trsig{\emp}:\trtype{E}$,
				$\arrCons{1}:\idx{A} \times \trtype{E} \ra \trtype{A}$                                                                     \\
				$\select{\trarr{\arrType}}:(\trarr{\arrType} \times \trtype{I}) \ra \trtype{E}$,
				$\store{\trarr{\arrType}}:(\trarr{\arrType} \times \trtype{I} \times \trtype{E}) \ra \trarr{\arrType}$                     \\
				$f_{\arrType}:\trtype{\arrType} \rightarrow \trarr{\arrType}$,
				$g_{\arrType}:\trarr{\arrType} \rightarrow \trtype{\arrType}$,
				$\trsig{x}:\trsig{E}$, $\trsig{i}:\trsig{I}$, $\trsig{a}:\trsig{\arrType}$, $\trarr{a}:\trarr{\arrType}$
			}
			}                                                                                                                          \\\hline
		\end{tabular}
		\normalsize
	}
	\caption{Translations for
		\Cref{tab:sigs}.
		$\trsig{\id}$  and
		$\trsig{\children}$ resp. stand for
		$\sel{\trsig{\container}}{1}$ and
		$\sel{\trsig{\container}}{2}$.
	}
	\label{tab:sigs-tran}
\end{table}

We now define our target theory,
as a combination of arrays, datatypes and EUF.

\begin{definition}
	\label{def:new-theory}
	For $1\leq i\leq n$, let $\trsig{T_i}$ be the $\trsig{\Sigma_i}$-theory
	of arrays,
	$T_{n+1}^\formulaProof$ the
	${\Sigma^\formulaProof_{n+1}}$-theory of datatypes and $T_{n+2}$ the
	$\Sigma_{n+2}$-theory of EUF.
	Define the following
	$\newSig{\Sigma}{\formulaProof}$-theory:
	\(
	\newT{\formulaProof} = \Bigl(\biguplus_{i=1}^{n} \trsig{T_i}\Bigr) \uplus T_{n+1}^\formulaProof \uplus T_{n+2}.
	\)
\end{definition}

\begin{example}
	\label{ex:preprocessing-signature}
	For $\exsigwf$ and formula $\formulaExUnsat$ of~\Cref{tab:sigs}, we obtain
	$\newSig{\exsigwf}{\formulaExUnsat}$, given in
	\Cref{tab:sigs-tran}.
	$\trsig{T_1}$ is the array theory of $\Sigma_1$,
	$\trsig{T_2}$ is the datatype theory of $\Sigma_2$,
	and $T_3$ is the EUF theory of $\Sigma_3$.
	\(
	\newT{\formulaExUnsat} = \trsig{T_1} \uplus T_2^\formulaExUnsat \uplus T_3.
	\)
\end{example}

\begin{remark}
	It might be more natural to represent each array variable $a$
	with its own constructor $c_a$, thus omitting the first
	argument (and the sort $\idx{\arrType}$ altogether).
	However, we want to make sure that if our input formula entails
	an equality $a=b$ between two array variables $a$ and $b$,
	then the datatype terms that represent them will also be equal.
	But if these datatypes are constructed using distinct constructors,
	they will surely be distinct.
	In contrast, having an identifier as an argument to the constructor,
	allows us to keep the resulting datatypes equal.
\end{remark}

\subsection{Translation of Formulas and Terms}
\label{sec:transforms}
We define a function $\tra$ from  $\Sigma$ to
$\newSig{\Sigma}{\formulaProof}$.
For that,
we introduce fresh variables
in~\Cref{tab:new-vars}.
For each variable $u$ of sort $\sigma \in \sorts{\Sigma}$, we define a new variable $\trsig{u}$ of sort $\trtype{\sigma}$.
For each variable $a$ of sort $\arrType_i \in \arrays_{\Sigma}$,
we also define a new variable $\trarr{a}$ of sort $\trarr{\arrType_i}$ and $\indexCar{\arrType_i}{\formulaProof}$ variables $\tempvar{a}{1}, \ldots, \tempvar{a}{\indexCar{\arrType_i}{\formulaProof}}$ of sort $\trtype{E_i}$.
This last set of variables will be used in the next section.
The translation matches $\Sigma$-terms to $\newSig{\Sigma}{\formulaProof}$-terms as follows.

\begin{table}[t]
	\centering
	\begin{tabular}{|c|c|c|c|}
		\hline
		Original Variable    & Original Sort                                    & New Variable                                                                & New Sort           \\\hline
		$u$                  & $\sigma \in \sorts{\Sigma}$                      & $\trsig{u}$                                                                 & $\trsig{\sigma}$   \\\hline
		\multirow{2}{*}{$a$} & \multirow{2}{*}{$\arrType_i \in \arrays_\Sigma$} & $\trarr{a}$                                                                 & $\trarr{\arrType}$ \\\cline{3-4}
		                     &                                                  & $\tempvar{a}{1}, \ldots, \tempvar{a}{\indexCar{\arrType_i}{\formulaProof}}$ & $\trsig{E_i}$      \\\hline
	\end{tabular}
	\caption{Fresh variables introduced for the $\newSig{\Sigma}{\formulaProof}$ signature.}
	\label{tab:new-vars}
\end{table}


\begin{definition}
	\label{sec:preprocessing-tr}
	The function $\tra$
	is given in
	\Cref{tab:literal-transformations}
	for flat $\Sigma$-literals.
	For conjunctions of flat literals, it is naturally extended to
	$\tra(\ell_1\wedge\ldots\wedge\ell_n)=\tra(\ell_1)\wedge\ldots\wedge\tra(\ell_n)$.
\end{definition}

\begin{example}
	\label{ex:preprocessing-tr}
	\Cref{tab:sigs-tran} shows the translation of $\formulaExUnsat$
	from \Cref{tab:sigs}.
	\Cref{fig:dt-ndt-ex-tr} presents illustrations of the translation in the first three rows.
	Even though we showed in~\Cref{ex:varphi-unsat} that
	$\formulaExUnsat$ is not $\tndt{\Sigma}$-satisfiable,
	$\tra(\formulaExUnsat)$ is
	$\newT{\formulaExUnsat}$-satisfiable,
	which suggests that $\tra$ alone is not sufficient.
\end{example}

\begin{table}[t]
	\centering\footnotesize
	\begin{tabular}[t]{|l|l|}
		\hline
		\textbf{Literal} $\literal$   & $\tra(\literal)$                                                                 \\[6pt] \hline
		$x = y$                       & $\trsig{x} = \trsig{y}$                                                          \\[4pt] \hline
		$x \neq y$                    & $\trsig{x} \neq \trsig{y}$                                                       \\[4pt] \hline
		$x = \select{\arrType}(a,i)$  & $\trsig{x} = \select{\trarr{\arrType}} \bigl(\trarr{a},\trsig{i}\bigr)$          \\[4pt] \hline
		$b = \store{\arrType}(a,i,v)$ & $\trarr{b} = \store{\trarr{\arrType}} \bigl(\trarr{a},\trsig{i},\trsig{v}\bigr)$ \\[4pt] \hline
	\end{tabular}\hspace{1em}%
	\begin{tabular}[t]{|l|l|}
		\hline
		\textbf{Literal} $\literal$ & $\tra(\literal)$                                               \\[6pt] \hline
		$x = \sel{c}{i}(y)$         & $\trsig{x} = \sel{\trsig{c}}{i}\bigl(\trsig{y}\bigr)$          \\[4pt] \hline
		$x = c(t_1 \til t_n)$       & $\trsig{x} = \trsig{c}\bigl(\trsig{t_1}\til \trsig{t_n}\bigr)$ \\[4pt] \hline
		$\is{c}\bigl(x\bigr)$       & $\is{\trsig{c}}\bigl(\trsig{x}\bigr)$                          \\[4pt] \hline
		$\neg \is{c}\bigl(x\bigr)$  & $\neg \is{\trsig{c}}\bigl(\trsig{x}\bigr)$                     \\[4pt] \hline
	\end{tabular}
	\caption{Translation function $\tra$.}
	\label{tab:literal-transformations}
\end{table}

\begin{figure}[t]
	\centering
	\renewcommand{\arraystretch}{1.1}
	\setlength{\tabcolsep}{6pt}
	\setlength{\extrarowheight}{2pt}
	\newcommand{\rowpad}{\rule[-22pt]{0pt}{44pt}}
	\newcommand{\xstarpic}[1][]{%
		\begin{tikzpicture}[baseline={(current bounding box.center)}]
			\def\top{0.5}
			\def\bot{-0.5}
			\def\leftx{1.2}
			\pgfmathsetmacro{\slotwidth}{2.6}
			\pgfmathsetmacro{\rightx}{\leftx + \slotwidth}
			\path[outerbg] (-0.5,\bot - 0.5) rectangle (\rightx + 0.5,\top + 0.3);
			\node[outerid, anchor=north west] (id) at (0,\top) {0};
			\fill[gray!50] (\leftx,\top) rectangle (\rightx,\bot);
			\draw[outerframe] (\leftx,\top) rectangle (\rightx,\bot);
			\pgfmathsetmacro{\labelx}{\leftx + 0.5*\slotwidth}
			\node at (\labelx,-0.8) {\textbf{$\trsig{a}$}};
			\node at (\labelx - 0.7, -1.2) {\textbf{$\trsig{x}$}};
			#1
			\draw[divider] (\leftx+0.3*\slotwidth,\top) -- (\leftx+0.3*\slotwidth,\bot);
		\end{tikzpicture}%
	}
	\newcommand{\outerpic}[1]{%
		\begin{tikzpicture}[baseline={(current bounding box.center)}]
			\def\top{1.3}
			\def\bot{-1.3}
			\def\leftx{0}
			\pgfmathsetmacro{\slotwidth}{11.0}
			\pgfmathsetmacro{\rightx}{\leftx + \slotwidth}
			\fill[gray!50] (\leftx,\top) rectangle (\rightx,\bot);
			\draw[outerframe] (\leftx,\top) rectangle (\rightx,\bot);
			\pgfmathsetmacro{\labelx}{\leftx + 0.5*\slotwidth}
			\node at (\labelx,\bot-0.3) {\textbf{$\trarr{a}$}};
			\node at (\leftx+0.25*\slotwidth, 0) {#1};
			\draw[divider] (\leftx+0.5*\slotwidth,\top) -- (\leftx+0.5*\slotwidth,\bot);
		\end{tikzpicture}%
	}
	\begin{tabular}{|c|c|}\hline
		\textbf{Formula} & \textbf{Illustration}                                \\\hline
		\rowpad \shortstack{$\tra(\psi_1')$                                     \\[4pt] \greensat} &
		\scalebox{0.7}{\outerpic{\textbf{$\trsig{x}$}}}
		\\\hline
		\rowpad \shortstack{$\tra(\psi_2)$                                      \\[4pt] \greensat} &
		\scalebox{0.7}{\xstarpic}
		\\\hline
		\rowpad \shortstack{$\tra(\formulaExUnsat)$                             \\[4pt] \greensat} &
		\scalebox{0.7}{\outerpic{\xstarpic}}
		\\\hline
		\rowpad \shortstack{$\tra(\formulaExUnsat)\wedge\Lone{\formulaExUnsat}$ \\[4pt] \greenunsat} &
		\scalebox{0.7}{\outerpic{\xstarpic[\node at (\leftx+0.15*\slotwidth, 0.05*\top) {\textbf{$\trsig{x}$}};]}}
		\\\hline
	\end{tabular}
	\caption{Formalization of cycles in $\newT{\formulaExUnsat}$.
		$\tra(\psi_1'),
			\tra(\psi_2),
			\tra(\formulaExUnsat)$ are in \Cref{tab:sigs-tran}.
		$\Lone{\formulaExUnsat}$ is in \Cref{ex:L0}.
	}
	\label{fig:dt-ndt-ex-tr}
\end{figure}

\subsection{Lemmas}
\label{sec:lemmas}
To fix the issue raised by \Cref{ex:preprocessing-tr}, we augment $\tra$ with lemmas, that are
given in \Cref{fig:notation-summary}.

%
Recall that each constructor $\arrCons{j}$ has an arity based on the set
$\indexSet{\arrType_j}{\formulaProof}$.
%
In order to properly apply the constructors $\arrCons{j}$s, we fix
an ordering for the set  $\indexSet{\arrType_j}{\formulaProof}$ by
fixing a bijection:

\begin{definition}
	\label{def:indOrder}
	Given $1\leq j\leq n$ and $\formulaProof$,
	$\indOrder{\arrType_j}{\formulaProof}$
	denotes a fixed bijection
	from
	$\indexSet{\arrType_j}{\formulaProof}$ to $
		\{1, \ldots, \indexCar{\arrType_j}{\formulaProof}\}$.
	When $\arrType_j$ and $\formulaProof$ are clear from the context,
	we omit them and simply write $\indOrder{}{}$.
\end{definition}

\begin{table}[t]
	\centering
	\begin{tabular}{|l|p{0.75\textwidth}|}
		\hline
		Notation                        & Definition \\ \hline
		\(\Lzero{a}{i}{\formulaProof}\) &
		\(\displaystyle
		\begin{array}{l}
			\Bigl\{
			\tempvar{a}{\indOrder{\arrType_j}{\formulaProof}(i)} = \select{\trarr{\arrType_j}}(\trarr{a}, \trsig{i})~~~,~~~
			\tempvar{a}{\indOrder{\arrType_j}{\formulaProof}(i)}= \sel{\arrCons{j}}{\indOrder{\arrType_j}{\formulaProof}(i)+1}(\trsig{a})
			\Bigr\}
		\end{array}
		\)                                           \\ \hline
		\(\Lone{\formulaProof}\)        &
		\(\displaystyle
		\bigcup_{x = \select{\arrType}(a,i) \in \formulaProof}
		\Lzero{a}{i}{\formulaProof}
		\)                                           \\ \hline
		\(\Ltwo{\formulaProof}\)        &
		\(\displaystyle
		\bigcup_{b = \store{\arrType}(a, j, v) \in \formulaProof}~~
		\bigcup_{i \in \indexSet{\arrType}{\formulaProof}}
		\Bigl(
		\Lzero{a}{i}{\formulaProof} \cup \Lzero{b}{i}{\formulaProof}
		\Bigr)
		\)                                           \\ \hline
		\(\Lthree{\formulaProof}\)      &
		\(\displaystyle
		\bigcup_{\arrType \in \arrays_\Sigma}~~
		\bigcup_{a \in \fv{\arrType}{\formulaProof}}
		\Bigl\{
		\trarr{a} = f_{\arrType}(\trsig{a}),\quad
		\trsig{a} = g_{\arrType}(\trarr{a})
		\Bigr\}
		\)                                           \\ \hline
	\end{tabular}
	\caption{Lemmas \(\allArr{a}{\formulaProof}\), \(\Lone{\formulaProof}\), \(\Ltwo{\formulaProof}\), and \(\Lthree{\formulaProof}\).
		The sort of $a$ is $\arrType_j$.
	}
	\label{fig:notation-summary}
\end{table}

$\Lzero{a}{i}{\formulaProof}$ ensures that the value
obtained by a select on the array variable $\trarr{a}$ of sort
$\trarr{\arrType_j}$ at index $\trsig{i}$ agrees with the corresponding field of $\trsig{a}$.
It utilizes variables
$\tempvar{a}{1}, \ldots, \tempvar{a}{\indexCar{\arrType_i}{\formulaProof}}$
to keep the translation flat.
%
$\Lone{\formulaProof}$ applies $\LzeroF$ to every select literal in
$\formulaProof$, while $\Ltwo{\formulaProof}$ applies it for all indices
occurring in the formula and every array variable that appears under a store.

The motivation for $\Lthree{\formulaProof}$ is to ensure that
for every array sort $A$,
there is a one to one mapping between the interpretations of free variables of sort $\trarr{A}$ and those of sort $\trtype{A}$.
$f_{\arrType}$ and $g_{\arrType}$ are asserted to be inverses of one another
when restricted to the relevant domain, thus ensuring the 1-1 correspondence.

We now combine these sets of literals to obtain the full set of lemmas
associated with the formula $\formulaProof$.
This allows us to define the final translation function $\trF$.

\begin{definition}
	\label{sec:preprocessing-lemmas-reduced}
	\label{sec:preprocessing-full-alg}
	Let \(\formulaProof\) be a flat cube in an NDT signature \(\Sigma\).
	The set of lemmas associated with \(\formulaProof\) is
	\(
	\reduced{\formulaProof} = \Lone{\formulaProof} \cup \Ltwo{\formulaProof} \cup \Lthree{\formulaProof}.
	\)
	The translation $\trF$ is defined by:
	\(
	\trF(\formulaProof) = \tra(\formulaProof)
	\land \bigwedge \reduced{\formulaProof}.
	\)
\end{definition}

\begin{example}
	\label{ex:L0}
	\label{ex:L2}
	Consider $\formulaExUnsat$ from \Cref{tab:sigs}.
	$
		\Lzero{a}{i}{\formulaExUnsat} = \{
		\tempvar{a}{1} = \select{}(\trarr{a}, \trsig{i}),
		\tempvar{a}{1} = \sel{\arrCons{1}}{2}(\trsig{a})
		\}$.
	Further,
	\(
	\Lone{\formulaExUnsat} = \Lzero{a}{i}{\formulaExUnsat}
	\).
	Since there are no stores,   \(
	\Ltwo{\formulaExUnsat} = \emptyset.
	\)
	Thanks to $\Lone{\formulaExUnsat}$, we obtain the desired
	$\newT{\formulaExUnsat}$-unsatisfiability,
	as illustrated in \Cref{fig:dt-ndt-ex-tr}. Notice that unsatisfiability
	follows directly from the theory of datatypes, as there is a cycle that does not go through
	any array:
	indeed, $\trsig{x}$ and $\trsig{a}$ are of sorts
	$\trsig{E}$ and $\trsig{\arrType}$, respectively, and both are datatype sorts that belong to
	$\structsorts_{\Sigma_2^\formulaExUnsat}$.
	This is in contrast to the last two rows of \Cref{fig:dt-ndt-ex2}, where $x$ is a datatype variable
	and $a$ is an array variable.
\end{example}


\subsection{Correctness}
\label{sec:correctness}

The next theorem shows that our translation is correct, in the sense that it preserves
(un)satisfiability, assuming a stable infiniteness condition.
\begin{theorem}
	\label{thm:correctness-of-translation}
	Suppose $T_{n+1}$ is
	stably infinite w.r.t. $\big(\bigcup_{i=1}^{n}\sorts{\Sigma_{i}}\big) \cap \sorts{\Sigma_{n+1}}$.
	Then
	$\formulaProof$ is $\tndt{\Sigma}$-satisfiable
	iff
	$\trF(\formulaProof)$ is $\newT{\formulaProof}$-satisfiable.
	The left-to-right direction holds without
	assuming stable infiniteness.
\end{theorem}

Given that the translation is correct and is easily computable, it is only left to show that the target theory
is decidable, which holds under the same stable infiniteness assumptions.

\begin{restatable}{theorem}{Decidability}
	\label{thm:decidability-of-new-theory}
	Suppose $T_{n+1}$ is stably infinite w.r.t.
	$\big(\bigcup_{i=1}^{n}\sorts{\Sigma_{i}}\big) \cap \sorts{\Sigma_{n+1}}$
	Then $\newT{\formulaProof}$ is decidable.
\end{restatable}


\begin{example}
	Consider $\Sigma_1$, $\Sigma_2$, $\exsigwf$ from \Cref{tab:sigs};
	by \Cref{ex:ndt-sig}, $\exsigwf$ is the NDT signature of $\Sigma_1$ and $\Sigma_2$.
	Let $T_1, T_2$ be the array theory over $\Sigma_1$ and the $\Sigma_2$-theory of datatypes;
	$T_2$ is stably infinite w.r.t. $\{I,A\}$ (see \Cref{ex:stably-infinite}).
	By \Cref{thm:correctness-of-translation,thm:decidability-of-new-theory}, $\tndt{\exsigwf}$ is decidable.
\end{example}

\section{Implementation and Evaluation}
\label{sec:imp}
A prototype implementation of the translation is
 implemented as a preprocessing pass in the $\cvcfive$ SMT
solver~\cite{cvc5}. Although the reduction is defined over flat cubes, the
implementation handles arbitrary input formulas by applying the translation
function $\trF$ recursively to all subterms.

\subsection{Setup}
All experiments were conducted on a machine running Ubuntu 22.04 via WSL2 (Windows Subsystem for Linux~2, kernel \texttt{6.6.87.2-microsoft-standard-WSL2}), equipped with a 13th Gen Intel Core i7-1355U processor (4 cores, 8 logical CPUs) and 11\,GiB of RAM.
A time limit of 10 seconds was imposed per benchmark.
We used $\zzz$ version~4.15.6 with default options, and our fork of $\cvcfive$, invoked with the options to support our implementation.\footnote{An artifact that includes the source code, binaries, benchmarks, scripts, and results is available in \url{https://zenodo.org/records/20055878}.}

\subsection{Benchmarks}

\begin{figure}
	\[
		\formulaEval{k} := \bigwedge_{i=1}^k \Bigl(a_i = \children(x_i) \land x_{i \bmod k + 1} = \select{}(a_{i}, j) \land \is{\container}(x_i)\Bigr).
	\]

	\[
		\formulaSatEval{k} := \bigwedge_{i=1}^k \Bigl(a_i = \children(x_i) \land \is{\container}(x_i)\Bigr)
		\land
		\bigwedge_{i=1}^{k-1} x_{i \bmod k + 1} = \select{}(a_{i}, j).
	\]
	\caption{Hand-crafted formulas.}
	\label{fig:hand-crafted-form}
\end{figure}

We evaluated our reduction on two independent benchmark suites:

\medskip
\noindent
{\bf Synthetic cycle benchmarks:}
formulas that systematically generalize $\formulaExUnsat$ from \Cref{tab:sigs}.
For each $k\geq 1$, the formula $\formulaEval{k}$ is given in
\Cref{fig:hand-crafted-form}.
It encodes the existence of a cycle of size $2k$, where half of the cycle goes through arrays and half through datatypes.
Notice that
$\formulaEval{1}$ is the same as $\formulaExUnsat$ from \Cref{tab:sigs}.

For every $k$, $\formulaEval{k}$ is not satisfiable in the theory of nested datatypes.
In order to test our implementation on satisfiable formulas, we generate {\em
		satisfiable variants} by removing a single constraint from $\formulaEval{k}$,
which corresponds to breaking the cycle. These
are formulas $\formulaSatEval{k}$ from \Cref{fig:hand-crafted-form}.

We have generated $\formulaEval{k}$ for $1\leq k\leq 500$.
The satisfiable instances were much more difficult for the solvers,
and hence
we generated $\formulaSatEval{k}$ for $1\leq k\leq 50$.
This smaller bound is already beyond
the largest $k$ either solver can solve.

\medskip
\noindent
{\bf Move Prover benchmarks:}
verification benchmarks that originate from the Move Prover~\cite{MoveProver}.
The Move Prover~\cite{MoveProver} is a formal verification tool for smart contracts written in the Move language~\cite{Move}.
It works by encoding smart contracts, as well as verification conditions for them, in the Boogie intermediate verification language~\cite{Boogie}, which then generates SMT-LIB queries that are sent to SMT-solvers.\footnote{
	We used commit
	\texttt{ac0b70a656bfb63eb7e569444b7a3a5ed318d6de} of \cite{MoveProverRepo},
	and commit
	\texttt{044e533713346840b55e53c33082911fd077d299} of
	\cite{BoogieRepo}.}
In order to correctly model persistent data structures, the encoding
employs the combination of datatypes and arrays, where the intended meaning is that of nested datatypes.\footnote{In fact, the work presented in this paper was motivated and initiated by the needs of the Move Prover.}

We considered 15 tests
from the Move Library.
These tests use features that our translation does not support, like quantifier reasoning, incremental solving and constant arrays.
Therefore, we performed the following steps using a script (available in the artifact):
\begin{enumerate*}
	\item Deleted top-level commands that use unsupported features, namely
	quantifiers and constant arrays.
	\item Inlined all incremental \verb!check-sat! calls by tracking the \texttt{push} and
	\texttt{pop} commands and generating SMT-LIB files that contain all relevant
	definitions and assertions in each context.
	\item Removed options specific to $\zzz$ and Boogie labels, which are
	automatically produced by Boogie.
\end{enumerate*}
\yoni{what are Boogie labels? This is not clear.}

Concretely, across the 15 original benchmarks, the preprocessing considered
13,313 top-level SMT-LIB commands and removed 4029 of them:
2832 containing quantifiers, 236 containing constant arrays, 944 $\zzz$
options, and 17 Boogie label commands.
The remaining 9284 commands were then split according to the incremental
contexts.
This way, we generated 76 non-incremental quantifier-free SMT-LIB files (51 satisfiable and 25
unsatisfiable), with an average of about 5 non-incremental benchmarks for each
original incremental benchmark.

In the Move prover, some nested datatypes include arrays whose
index sort is a (nested) datatype, violating
the assumption
$\structsorts_{\Sigma_{n+1}} \,\cap\, \arrays_{\Sigma} \;=\; \emptyset$
of \Cref{sec:theory:ndt-sig}.
This assumption was introduced to simplify the presentation,
and
our implementation supports such nested datatypes.
The reason is that nested datatypes appearing in array indices are handled as
a separate NDT signature from the one used for array elements.
Thus, cycles through array elements are still detected by the reduction
described in this paper, while the nested structure of the index sort is
solved by an independent NDT component.
Sorts that cannot participate in such cycles, such as integers, may still be
shared between these components; the remaining interaction is then handled by
standard theory-combination techniques.
\yoni{This last paragraph is really unclear.}

While the Move prover benchmarks originate from a real-world application, 
they represent only a part of that application, as unsupported features are removed.
Still, including them is relevant because the resulting benchmarks represent 
the basic quantifier-free non-incremental reasoning necessary (though not sufficient)
for the original benchmarks.

\subsection{Results}

\begin{table}[t]
	\centering
	\footnotesize
	\begin{tabular}{|l|r||r|r||r|r|}
		\hline
		\multirow{2}{*}{Set} & \multirow{2}{*}{\#}
		                     & \multicolumn{2}{c||}{solved}
		                     & \multicolumn{2}{c|}{time (s)}                                                         \\
		\cline{3-6}
		                     &
		                     & \multicolumn{1}{c|}{$\cvcfive$} & \multicolumn{1}{c||}{$\zzz$}
		                     & \multicolumn{1}{c|}{$\cvcfive$} & \multicolumn{1}{c|}{$\zzz$}                         \\\hline\hline
		$\formulaEval{k}$    & 500                             & 500                          & 469 & 299.8 & 1035.5 \\\hline
		$\formulaSatEval{k}$ & 50                              & 25                           & 1   & < 1   & < 1    \\\hline\hline
		MPS                  & 51                              & 51                           & 51  & < 1   & 2.5    \\\hline
		MPU                  & 25                              & 25                           & 25  & < 1   & < 1    \\\hline
	\end{tabular}
	\caption{Results. MPS and MPU denote satisfiable and unsatisfiable Move Prover benchmarks.}
	\label{tab:results}
\end{table}


\Cref{tab:results} summarizes the results, in terms of number of benchmarks that were solved within the time limit, and the total runtimes (in seconds) over the benchmarks that
were commonly solved
by both $\cvcfive$ and $\zzz$.
$\cvcfive$ is able to solve all benchmarks from the unsatisfiable synthetic set, and $\zzz$ can solve most of them. On the commonly solved benchmarks, $\cvcfive$ is faster, and for most instances is able to solve the formula in under 2 seconds.
The satisfiable synthetic benchmarks are harder for both solvers. 
\yoni{reviewer C reqyested "an intuition of why that is the case".}
$\cvcfive$ solves half of them, while $\zzz$ only solves one.
Both solvers are able to solve all Move Prover benchmarks quickly, though $\cvcfive$ is slightly faster on the satisfiable ones.

\section{Conclusion}
We have introduced and studied a theory of nested datatypes, that avoids
non-standard models.
A decision procedure was presented, which is based
on a translation to the standard combination between arrays and datatypes,
augmented with uninterpreted functions.
Our empirical evaluation supports the decision procedure and shows its effectiveness
on both real-world and synthetic benchmarks.

Next immediate steps is to expand the theoretical and practical support of our framework
to incremental solving and quantifiers.
More generally,
he desired combination of arrays and datatypes is just a particular instance of
a wider problem with naive theory combination. Indeed, similar undetected cycles
can appear when replacing arrays with sequences, sets, and other data structures.
A more general solution is therefore needed.
We are hopeful that the particular case of arrays that we have solved here
can serve as a basis for a more general solution.

\newpage
\bibliography{polite}

\newpage
\appendix
\crefalias{section}{appendix}
\crefalias{subsection}{subappendix}
\crefalias{subsubsection}{subsubappendix}
\section*{Appendix}


In this appendix we prove the main theoretical results
of the paper, namely
\Cref{thm:correctness-of-translation} (correctness) and \Cref{thm:decidability-of-new-theory} (decidability).
For \Cref{thm:correctness-of-translation}, we split
the proof for soundness and completeness.
Soundness is proven in \Cref{app:sec:soundness}, and individually stated as \Cref{thm:preprocessing-correctness}.
Completeness is proven in \Cref{app:sec:completeness}, and individually stated as \Cref{thm:preprocessing-correctness-2}.
The proof of \Cref{thm:decidability-of-new-theory} is presented in \Cref{app:sec:decproof}.
%

We shall make use of the following classical result:

\begin{theorem}[Nelson--Oppen %
		\cite{NelsonOppen,10.1007/978-3-540-30227-8_53}]
	\label{thm:nelson-oppen}
	Let $T_1$ and $T_2$ be decidable $\Sigma_1$- and $\Sigma_2$-theories, respectively,
	such that:
	\begin{enumerate*}[label=(\roman*)]
		\item $\func{\Sigma_1} \cap \func{\Sigma_2} = \emptyset$ and
		$\pred{\Sigma_1} \cap \pred{\Sigma_2} = \emptyset$; and
		\item $T_1$ and $T_2$ are both stably infinite w.r.t.\ $S_0=\sorts{\Sigma_1}\cap\sorts{\Sigma_2}$.
	\end{enumerate*}
	Then $T_1\oplus T_2$ is decidable.
\end{theorem}

\section{Soundness}
\label{app:sec:soundness}
In what follows,
let $\Sigma$ be the NDT signature of $\Sigma_1,\ldots,\Sigma_{n+1}$,
where $\Sigma_1,\dots,\Sigma_n$ are arrays signatures and $\Sigma_{n+1}$ is a datatypes signature.
We prove the following soundness theorem:
\begin{theorem}
	\label{thm:preprocessing-correctness}
	For every flat \(\Sigma\)-cube \(\formulaProof\),
	if $\formulaProof$ is $\tndt{\Sigma}$-satisfiable,
	then $\trF(\formulaProof)$ is $\newT{\formulaProof}$-satisfiable.
\end{theorem}

As a motivating example, the formula $\formulaExUnsat$ from \Cref{tab:sigs} is $\tndt{\Sigma}$-unsatisfiable,
and thus cannot be used for the current section.
We therefore use the following $\tndt{\Sigma}$-satisfiable example instead.

\begin{example}
	\label{ex:proof-formula}
	\Cref{tab:tndtsatex} includes a $\tndt{\exsigwf}$-satisfiable formula
	$\formulaExSat$.
	The translated formula $\trF(\formulaExSat)$, according to~\Cref{sec:preprocessing-full-alg}, is
	also presented in \Cref{tab:tndtsatex}.
	Here, $\trsig{x},\trsig{y}$ and $\tempvar{a}{1}$ are variables of sort $\trsig{E}$,
	$\trsig{a}$ is a variable of sort $\trsig{\arrType}$,
	$\trsig{j},\trsig{k}$ are of sort $\trsig{I}$,
	and $\trarr{a}$ is an array variable of sort $\trarr{\arrType}$.
	We shall later see that indeed this translated formula is
	$\newT{\formulaProof}$-satisfiable.
\end{example}

Let us now consider an example of an NDT $\exsigwf$-interpretation.

\begin{example}
	\label{ex:intended-model}
	An example for a $\exsigwf$-interpretation $\modelS$
	is given in \Cref{tab:exInter}.
	Notice that $\modelS\models\formulaExSat$.
	For illustration, \Cref{tab:exInter} also includes
	interpretations for variables
	that do not occur in $\formulaExUnsat$.
	It can be shown that
	$E^{\modelS} = \{\emp\} \cup \{\container(n, y) \mid n \in \mathbb{N},\, y
		\in \arrType^{\modelS}\}$,
	and $\arrType^{\modelS} = \almostC{I^{\modelS}}{E^{\modelS}}$.
	$y^{\modelS}$ can be thought of as a container with id 2 and infinitely
	many copies of $\emp$ in an integer indexed array.
	%
	It can be shown that $\modelS$ is an NDT-$\exsigwf$-structure.

\end{example}

\begin{table}[t]
	\newcolumntype{M}[1]{>{\centering\arraybackslash}m{#1}}
	\centering
	\begin{tabular}{|M{3cm}|M{10cm}|}\hline
		$\formulaExSat$       &
		$
			x = \select{}(a,j)\;\land\;
			y = \container(k,a)
		$
		\\\hline
		$\trF(\formulaExSat)$ &
		\shortstack{
			$
				\trsig{x} = \select{}(\trarr{a},\trsig{j})
				\land
				\trsig{y} = \trsig{\container}(\trsig{k},\trsig{a})
				\land$                  \\
			$
				\tempvar{a}{1} = \select{}(\trarr{a}, \trsig{j})
				\land
				\tempvar{a}{1} = \sel{\arrCons{1}}{2}(\trsig{a})
				\land$                  \\ $
				\trarr{a} = f_{\arrType}(\trsig{a})
				\land
				\trsig{a} = g_{\arrType}(\trarr{a})
			$
		}
		\\\hline
	\end{tabular}
	\caption{A $\tndt{\exsigwf}$-formula and its translation.
		Note that $\newSig{\exsigwf}{\formulaExSat}$ is the same as $\newSig{\exsigwf}{\formulaExUnsat}$,which is
		given in \Cref{tab:sigs-tran}.
	}
	\label{tab:tndtsatex}
\end{table}

\begin{table}[t]
	\newcolumntype{M}[1]{>{\centering\arraybackslash}m{#1}}
	\centering
	\begin{tabular}{|M{2cm}|M{10cm}|}\hline
		\shortstack{
			Domains
		}          &
		\shortstack{
			$I^{\modelS} = \mathbb{N}$                        \\
			$E^{\modelS} = \bigcup_{i \in \mathbb{N}}
				\subdomain{E}{i}$                                 \\
			$\arrType^{\modelS} =
				\bigcup_{i \in \mathbb{N}} \subdomain{A}{i}$
			\\\\
			where:                                            \\\\
			$
				\subdomain{A}{0} = \subdomain{E}{0} = \emptyset
			$                                                 \\
			$\subdomain{E}{i+1} = \subdomain{E}{i} \cup \{\emp\} \cup \{\container(n,
				y) \mid n \in \mathbb{N},\, y \in \subdomain{A}{i}\}$,
			$\subdomain{A}{i+1} = \almostC{\mathbb{N}}{\subdomain{E}{i}}$
		}
		\\\hline
		Functions  &
		\shortstack{
			$\select{}$ and $\store{}$ are defined as         \\
			read and update
			operations on functions.                          \\\\
			$\container^{\modelS}(n,y)=\container(n,y)$       \\
			$\id^{\modelS}(\container(n,y))=n$                \\
			$\children^{\modelS}(\container(m,y))=y$          \\
			$\id^{\modelS}(\emp)=0$                           \\
			$\children^{\modelS}(\emp)=\constArr{{\emp}}$
		}
		\\\hline
		Predicates &
		\shortstack{
		$\is{\container}^{\modelS}(\container(n,y))=true$ \\
		$\is{\container}^{\modelS}(\emp)=false$           \\
		$\is{\emp}^{\modelS}(\emp)=true$                  \\
		$\is{\emp}^{\modelS}(\container(n,y))=false$
		}
		\\\hline
		Variables  &
		\shortstack{
			$x^{\modelS} = \emp$,
			$a^{\modelS} = \constArr{{\emp}}$,
			$i^{\modelS} = 1$,                                \\
			$y^{\modelS} = \container(2,\constArr{{\emp}})$,
			and
			$j^{\modelS} = 2$.
		}
		\\\hline
	\end{tabular}
	\caption{A $\exsigwf$-interpretation $\modelS$.
	}
	\label{tab:exInter}
\end{table}

\paragraph{Roadmap.}
Let $\model$ be a $\tndt{\Sigma}$-interpretation that satisfies $\formulaProof$. We will define a $\newT{\formulaProof}$-interpretation $\newM$ that satisfies $\trF(\formulaProof)$.
We define \(\newM\) (\Cref{app:sec:define-domains} --
\ref{app:sec:interpretation-functions}), and prove that it is a
$\newT{\formulaProof}$-interpretation
(Section~\ref{app:sec:newM-newT-interpretation}).
Then, we show that \(\newM\)
satisfies every translated literal
(Section~\ref{app:sec:satisfaction-literals}) and every lemma
(Section~\ref{app:sec:satisfaction-lemmas}).

\subsection{Defining the Domains}
\label{app:sec:define-domains}

\begin{table}[t]
	\centering
	\renewcommand{\arraystretch}{1.2}
	\begin{tabular}{|c|c|c|}
		\hline
		\textbf{New sort}    & \textbf{Original sort}                                                                & \textbf{Domain chosen in $\newM$} \\ \hline\hline
		$\idx{\arrType_i}$   & $\arrType_i \in \arrays_\Sigma$
		                     & $\idx{\arrType_i}^{\newM}= \arrType_i^{\model}$                                                                           \\ \hline
		\multirow{2}{*}{$\trsig{\tau}$}
		                     & $\tau \in \sorts{\Sigma} \setminus (\arrays_\Sigma \cup \structsorts_{\Sigma_{n+1}})$
		                     & $\trsig{\tau}^\newM= \tau^{\model}$                                                                                       \\ \cline{2-3}
		                     & $\tau \in \structsorts_{\Sigma_{n+1}} \cup \arrays_\Sigma$
		                     & $\displaystyle
			\trsig{\tau}^{\newM} = T_{\trsig{\tau}}
			\bigl(\consig{\Sigma_{n+1}^{\formulaProof}}{\Sigma_{n+1}^{\formulaProof}},\,B\bigr)$                                                             \\ \hline
		$\trarr{\arrType_i}$ & $\arrType_i \in \arrays_\Sigma$
		                     & $\displaystyle
			\trarr{\arrType_i}^{\newM}= \trsig{I_i}^{\newM}\!\to\!\trsig{E_i}^{\newM}$                                                                       \\
		\hline
	\end{tabular}
	\caption{Domains assigned to the sorts of $\newSig{\Sigma}{\formulaProof}$.
		Here $B$ is the $\elemsorts_{\newSig{\Sigma}{\formulaProof}}$-sorted set with
		$B_\rho=\rho^{\newM}$ for every
		$\rho\in\elemsorts_{\newSig{\Sigma}{\formulaProof}}$.}
	\label{tab:domains-newM}
\end{table}

The domains assigned to the sorts of $\Sigma$ by $\newM$ are given in \Cref{tab:domains-newM}.
We have three cases to consider:
\begin{enumerate}
	\item If $\sigma=\idx{\arrType_i}$ for some $i\in [1,n]$
	      then $\sigma$ is interpreted in $\newM$ exactly in the same way $\arrType_i$
	      is interpreted in $\model$.
	      Intuitively, this means that the ``identifier" of each array is itself.
	\item If $\sigma=\trsig{\tau}$ for some sort $\tau$, there are two sub-cases. If $\tau$ is neither in $\structsorts_{\Sigma_{n+1}}$ nor in $\arrays_{\Sigma}$, then $\trsig{\tau} \notin \structsorts_{\Sigma_{n+1}^\formulaProof}$ and we interpret $\trsig{\tau}$ in $\newM$ in the same way $\tau$ is interpreted in $\model$.
	      Otherwise, $\tau$ is in $\arrays_{\Sigma}$, or in $\structsorts_{\Sigma_{n+1}}$, and thus $\trsig{\tau} \in \structsorts_{\Sigma_{n+1}^\formulaProof}$, and we interpret $\trsig{\tau}$ in $\newM$ as the set of trees defined in \Cref{def:trees}.
	\item If $\sigma=\trarr{\arrType_i}$ for some $i\in [1,n]$, we interpret $\trarr{\arrType_i}$ in $\newM$ as the set of functions from $\trsig{I_i}^\newM$ to $\trsig{E_i}^\newM$.
\end{enumerate}

\begin{table}[t]
	\newcolumntype{M}[1]{>{\centering\arraybackslash}m{#1}}
	\centering
	\begin{tabular}{|M{2cm}|M{11cm}|}\hline
		\shortstack{
			Domains
		}          &
		\shortstack{
		$\trsig{I}^{\newMS} = I^{\model} = \mathbb{N}$                                                                    \\
		${\idx{\arrType}}^{\newMS} = A^\model$                                                                            \\
		$\trsig{E}^{\newMS} = T_{\trsig{E}}(\consig{\Sigma_{2}^\formulaExSat}{\Sigma_{2}^\formulaExSat}, B)$              \\
		$\trsig{\arrType}^{\newMS} =T_{\trsig{\arrType}}(\consig{\Sigma_{2}^\formulaExSat}{\Sigma_{2}^\formulaExSat}, B)$ \\
		$\trarr{A}^{\newMS} = \trsig{I}^\newMS \to \trsig{E}^\newMS$
		}
		\\\hline
		Functions  &
		\shortstack{
		$\select{}$ and $\store{}$ are defined as                                                                         \\
		read and update
		operations on functions.                                                                                          \\\\
		$\trsig{\container}^{\newMS}(n,y)=\trsig{\container}(n,y)$                                                        \\
		$\trsig{\id}^{\newMS}(\trsig{\container}(n,y))=n$                                                                 \\
		$\trsig{\children}^{\newMS}(\trsig{\container}(m,y))=y$                                                           \\
		$\trsig{\id}^{\newMS}(\trsig{\emp})=0$                                                                            \\
		$\trsig{\children}^{\newMS}(\trsig{\emp})=\constArr{{\trsig{\emp}}}$                                              \\
		$f_{\arrType}^{\newMS}(\arrCons{1}(a,e)) = \applyTag(a)$                                                          \\
		$g_{\arrType}^{\newMS}(t)=
			\begin{cases}
				\applyStar(a)               & \text{if } a\in\arrType^{\model}\text{ and } \applyTag(a)=t, \\[4pt]
				\defelem{\trtype{\arrType}} & \text{otherwise}.
			\end{cases}$
		}
		\\\hline
		Predicates &
		\shortstack{
		$\is{\trsig{\container}}^{\newMS}(\trsig{\container}(n,y))=true$                                                  \\
		$\is{\trsig{\container}}^{\newMS}(\emp)=false$                                                                    \\
		$\is{\emp}^{\newMS}(\emp)=true$                                                                                   \\
		$\is{\emp}^{\newMS}(\trsig{\container}(n,y))=false$
		}
		\\\hline
		Variables  &
		\shortstack{
			\(\trsig{x}^{\newMS} = \applyStar(x^\model) = \trsig{\emp}\)                                                      \\
			\(\trsig{a}^{\newMS} = \applyStar(a^\model) = \arrCons{1} (a^\model, \trsig{\emp})\)                              \\
			\(\trarr{a}^{\newMS} = \applyTag(a^\model) = \constArr{\trsig{\emp}}\)                                            \\
			\(\trsig{j}^{\newMS} = \applyStar(j^\model) = 1\)                                                                 \\
			\(\trsig{y}^{\newMS} = \applyStar(y^\model) = \trsig{\container}(2,\trsig{a}^{\newMS})\)                          \\
			\(\trsig{k}^{\newMS} = \applyStar(k^\model) = 2\)                                                                 \\
			\(\tempvar{a}{1}^{\newMS} = \applyStar\Bigl(\select{}^\model(a^\model,1)\Bigr) = \trsig{\emp}\)
		}
		\\\hline
	\end{tabular}
	\caption{A $\newSig{\exsigwf}{\formulaExSat}$-interpretation $\newMS$.
		$B$ is the $\elemsorts_{\newSig{\exsigwf}{\formulaExSat}}$-sorted set with\\
		$B_{\rho} := \rho^\newMS$ for every
		$\rho\in\elemsorts_{\newSig{\exsigwf}{\formulaExSat}}$.
	}
	\label{tab:extransinter}
\end{table}

\begin{example}
	\label{ex:define - domains}
	Consider the formula $\formulaExSat$ and its satisfying interpretation $\modelS$ from~\Cref{ex:proof-formula}.
	The sorts of $\newSig{\exsigwf}{\formulaExSat}$ are interpreted in $\newMS$
	as described in \Cref{tab:extransinter}.
\end{example}

\subsection{Interpretation of the Variables}
\label{app:variablessoundnessint}

In order to define the interpretation of the variables, we define two mappings between the domains of the sorts in $\Sigma$ and the domains of the sorts in $\newSig{\Sigma}{\formulaProof}$.

The mapping $\applyStar$ will be used to define the interpretation of all variables in $\trF(\formulaProof)$ whose sorts are not in $\{\trarr{\arrType_i} \mid i\in[1,n]\}$, and the mapping $\applyTag$ will be used to define the interpretation of the variables of sorts $\trarr{\arrType_i}$ for $i\in[1,n]$.

\subsubsection{Defining \(\applyStar\)}

In order to define the first mapping, $\applyStar$, we first need to define the following partial function:
\begin{definition}[Height]
	\label{def:height}
	We define the partial function $\height:\bigcup_{\sigma\in\sorts{\Sigma}} \sigma^\model \ra \mathbb{N}$ in the following way:
	For each $t \in \bigcup_{\sigma\in\sorts{\Sigma}} \sigma^\model$,
	\begin{enumerate}[label=(\alph*)]
		\item If \(t = c(t_1,\dots,t_m)\) for some $c \in \constructors_{\Sigma_{n+1}}$ and $t_1,\ldots,t_m$ and $\height$ is defined over $t_1,\ldots,t_m$, we thus define:
		      $\height(t) \;=\; 1 + \max\{\height(t_1),\ldots,\height(t_m)\}.$
		\item If $t \in \arrType_j^\model$ for some $j \in [1,n]$.
		      Then, let
		      \(
		      \iota_1,\dots,\iota_{\indexCar{\sigma}{\formulaProof}}
		      \)
		      be the elements of \(\indexSet{\sigma}{\formulaProof}\),
		      listed in the order determined by \(\indOrder{\sigma}{\formulaProof}\).
		      Then, we define:
		      \(
		      \height(t) \;=\; 1 + \max\{\height(\select{\sigma}^\model\!(t,\iota_1^\model)),\ldots,\height(\select{\sigma}^\model\!(t,\iota_{\indexCar{\sigma}{\formulaProof}}^\model))\}.
		      \)
		      If $\height$ is defined over all the elements of the set above.
		\item If \(\sigma\notin (\structsorts_{\Sigma_{n+1}}\cup \arrays_\Sigma)\), then
		      define:
		      $\height(t)=0$
	\end{enumerate}

\end{definition}

\begin{example}
	\label{ex:height-example}
	Consider $\modelS$ from \Cref{tab:exInter}
	Then, we have:

	\begin{align*}
		\height(x^{\modelS}) & = \height(\emp) = 1 + \max\emptyset = 1,                                                                                \\
		\height(a^{\modelS}) & = \height(\constArr{\emp}) = 1 + \max\bigl\{ \height(x^{\modelS}) \bigr\} = 1 + \max\{1\} = 2,                          \\
		\height(y^{\modelS}) & = \height(\cons(2, \constArr{\emp})) = 1 + \max\bigl\{ \height(2), \height(a^{\modelS}) \bigr\} = 1 + \max\{0, 2\} = 3, \\
		\height(i^{\modelS}) & = \height(j^{\modelS}) = 0.
	\end{align*}

\end{example}

\begin{lemma}
	\label{lem:height-total-function}
	The function \(\height\) is a total function over \(\bigcup_{\sigma\in\sorts{\Sigma}} \sigma^\model\).
\end{lemma}

\begin{proof}
	Suppose, towards a contradiction, that there exists some $t \in \bigcup_{\sigma\in\sorts{\Sigma}} \sigma^\model$ such that $\height(t)$ is not defined.
	We know that $t\notin \structsorts_{\Sigma_{n+1}}\cup \arrays_\Sigma$, since otherwise $\height(t)$ would be defined by the third case of~\Cref{def:height}.
	Thus, there are 2 possibilities:
	\begin{enumerate}
		\item $t$ is of sort $\sigma \in \structsorts_{\Sigma_{n+1}}$.
		      Then, since $\model^{\Sigma_{n+1}}$ is a datatypes $\Sigma_{n+1}$-interpretation, $t$ must be of the form $c(t_1,\ldots,t_m)$ for some constructor $c\in\constructors_{\Sigma_{n+1}}$ and $t_1,\ldots,t_m \in \sigma^\model$.
		      So there must be some $t_j$ such that $\height(t_j)$ is not defined.
		      Note that $(t_j, t) \in \rel{\model}$.
		\item $t$ is of sort $\arrType_j$ for some $j\in[1,n]$.
		      Then, there must be some $i\in\indexSet{\arrType_j}{\formulaProof}$ such that $\height(\select{\arrType_j}^\model(t,i^\model))$ is not defined.
		      Note that $(\select{\arrType_j}^\model(t,i^\model), t) \in \rel{\model}$.
	\end{enumerate}

	In either case, there exists some $s\in\bigcup_{\sigma\in\sorts{\Sigma}} \sigma^\model$ such that $(s,t)\in\rel{\model}$ and $\height(s)$ is not defined.
	We can apply this argument again and get an infinitely descending sequence of elements in $\rel{\model}$, which contradicts the well-foundedness of $\rel{\model}$.
\end{proof}

Now, we can use the function $\height$ to define the mapping $\applyStar$.
\begin{definition}
	\label{def:applyTag}
	For every $k \in \mathbb{N}$, we define the partial function
	\[
		\applyPartStar{k} \;\colon\;
		\bigcup_{\sigma \in \sorts{\Sigma}} \{ x \in \sigma^\model \mid \height(x) \leq k \}
		\;\longrightarrow\;
		\bigcup_{\sigma \in \sorts{\Sigma}} \trsig{\sigma}^{\newM}
	\]
	In the following way:
	For $k = 0$, we define $\applyPartStar{0}(t) = t$ for every $t \in \sigma^\model$ such that $\height(t) = 0$.
	For $k > 0$, we define $\applyPartStar{k}(t)$ as follows:
	\begin{enumerate}[label=(\alph*)]
		\item If $\height(t) < k$, then we define $\applyPartStar{k}(t) = \applyPartStar{k-1}(t)$.
		\item If $\height(t) = k$, then we define $\applyPartStar{k}(t)$ as follows:
		      \begin{enumerate}[label=(\roman*)]
			      \item If \(t = c(t_1,\dots,t_m)\) for some $c:\sigma_1 \times \ldots \times \sigma_m \ra \sigma \in \constructors_{\Sigma_{n+1}}$ and $\forall i\in[1,m]$, $\applyPartStar{k}(t_i)$ is defined and in $\trsig{\sigma_i}^{\newM}$, then we define:
			            $\applyPartStar{k}(t) \;=\; \trsig{c}\bigl(\applyPartStar{k}(t_1),\,\dots,\,\applyPartStar{k}(t_m)\bigr).$
			      \item If $t \in \arrType_j^\model$ for some $j \in [1,n]$, and for every $i\in\indexSet{\arrType_j}{\formulaProof}$, $\applyPartStar{k}( \select{\arrType_j}^\model(t,i^\model))$ is defined and in $\trsig{E_j}^{\newM}$, then we define:
			            \[
				            \applyPartStar{k}(t) \;=\; \arrCons{j} \bigl(\,
				            t,\,
				            \applyPartStar{k}(\select{\arrType_j}^\model(t,\iota_1^\model)),\dots,
				            \applyPartStar{k}(\select{\arrType_j}^\model(t,\iota_{\indexCar{\arrType_j}{\formulaProof}}^\model))
				            \bigr).
			            \]
		      \end{enumerate}
	\end{enumerate}

	We then define:
	\(
	\applyStar = \bigcup_{k=0}^{\infty} \applyPartStar{k}.
	\)

\end{definition}

\begin{lemma}
	\label{lem:applyStar-total}
	$\applyStar: \bigcup_{\sigma \in \sorts{\Sigma}} \sigma^\model \ra \bigcup_{\sigma \in \sorts{\Sigma} }\trsig{\sigma}^\newM$ is a total function.
	In addition,
	$\forall x \in \sigma^\model. \applyStar(x) \in \trsig{\sigma}^\newM$ for every $\sigma\in\sorts{\Sigma}$.
\end{lemma}
\begin{proof}
	First note that for every $k \in \mathbb{N}$, $\applyPartStar{k}$ is a partial function over $\bigcup_{\sigma \in \sorts{\Sigma}} \{ x \in \sigma^\model \mid \height(x) \leq k \}$, and that for every $j,k \in \mathbb{N}$, if $j < k$, then $\applyPartStar{j} \subseteq \applyPartStar{k}$.
	Thus, $\applyStar$ is a partial function over $\bigcup_{\sigma \in \sorts{\Sigma}} \sigma^\model$.

	In order to show it is a total function, we first use induction on \(k\) to show that \(\applyPartStar{k}\) is a total function over \(\bigcup_{\sigma \in \sorts{\Sigma}} \{ x \in \sigma^\model \mid \height(x) \leq k \}\) and that $\forall \sigma \in \sorts{\Sigma}.\forall x \in \sigma^\model.$ if $\height(x) \leq k$, then $\applyPartStar{k}(x) \in \trsig{\sigma}^\newM$.

	$\forall x \in \sigma^\model. \applyPartStar{k}(x) \in \trsig{\sigma}^\newM$

	\textbf{Base case:} For \(k=0\), let \(t\in \sigma^\model\) such that \(\height(t)=0\).
	Then, by~\Cref{def:height}, \(\sigma\notin (\structsorts_{\Sigma_{n+1}}\cup \arrays_\Sigma)\). Hence, we have \(\applyPartStar{0}(t) = t\) and since $\sigma^\model = \trsig{\sigma}^\newM$, we have that \(\applyPartStar{0}(t) \in \trsig{\sigma}^\newM\).

	\textbf{Inductive step:} Assume that \(\applyPartStar{k}\) is a total function over \(\bigcup_{\sigma \in \sorts{\Sigma}} \{ x \in \sigma^\model \mid \height(x) \leq k \}\) and that $\forall \sigma \in \sorts{\Sigma}.\forall x \in \sigma^\model.$ if $\height(x) \leq k$, then $\applyPartStar{k}(x) \in \trsig{\sigma}^\newM$.
	Let there be \(t\in \sigma^\model\) such that \(\height(t) \leq k+1\).
	If \(\height(t) < k + 1\), then by the inductive hypothesis, we know that \(\applyPartStar{k}(t)\) is defined and in \(\trsig{\sigma}^\newM\).
	Otherwise, by~\Cref{def:height}, we have two cases:
	\begin{enumerate}[label=(\alph*)]
		\item If \(t \in \sigma^\model\) where \(\sigma\in\structsorts_{\Sigma_{n+1}}\), then since $\model^{\Sigma_{n+1}}$ is a datatypes $\Sigma_{n+1}$-interpretation, we must have
		      \(t = c(t_1,\dots,t_m)\)
		      for some $c$ and $t_1,\ldots,t_m$. By the definition of the $\height$ function, we know that $\height(t_i) \leq k$ for every $i\in[1,m]$.
		      By the inductive hypothesis, we have that $\applyPartStar{k+1}(t_i) = \applyPartStar{k}(t_i)$ and in $\trsig{\sigma_i}^{\newM}$ for every $i\in[1,m]$.
		      Thus,
		      $$\applyPartStar{k+1}(t) \;=\; \trsig{c}\bigl(\applyPartStar{k+1}(t_1),\,\dots,\,\applyPartStar{k+1}(t_m)\bigr).$$
		      Since $\trsig{c} \in \constructors_{\newSig{\Sigma}{\formulaProof}}$, we have that
		      $\applyPartStar{k+1}(t) \in \trsig{\sigma}^\newM$.
		\item If \(t \in \arrType_j^\model\) for some $j \in [1,n]$, then by the definition of the $\height$ function, we know that $\height(\select{\arrType_j}^\model(t,i^\model)) \leq k$ for every $i\in\indexSet{\arrType_j}{\formulaProof}$.
		      By the inductive hypothesis, we have that $\applyPartStar{k+1}(\select{\arrType_j}^\model(t,i^\model)) = \applyPartStar{k}(\select{\arrType_j}^\model(t,i^\model))$ and in $\trsig{E_j}^{\newM}$ for every $i\in\indexSet{\arrType_j}{\formulaProof}$.
		      Thus,
		      $$\applyPartStar{k+1}(t) \;=\; \arrCons{j} \bigl(\,
			      t,\,
			      \applyPartStar{k+1}(\select{\arrType_j}^\model(t,\iota_1^\model)),\dots,
			      \applyPartStar{k+1}(\select{\arrType_j}^\model(t,\iota_{\indexCar{\arrType_j}{\formulaProof}}^\model))
			      \bigr).$$
		      Since $\arrCons{j} \in \constructors_{\newSig{\Sigma}{\formulaProof}}$, we have that
		      $\applyPartStar{k+1}(t) \in \trsig{\sigma}^\newM$.
	\end{enumerate}

	Let there be $t \in \sigma^\model$ we now show that $\applyStar(t)$ is defined and in $\trsig{\sigma}^\newM$.
	Since~\Cref{lem:height-total-function} shows that $\height$ is a total function over $\bigcup_{\sigma\in\sorts{\Sigma}} \sigma^\model$, we have that $\forall t \in \bigcup_{\sigma\in\sorts{\Sigma}} \sigma^\model.\; \exists k\in\mathbb{N}.\;\height(t) = k$. Thus, $\applyStar(t) = \applyPartStar{k}(t)$ and $\applyPartStar{k}(t) \in \trsig{\sigma}^\newM$.
\end{proof}

\begin{remark}
	\label{rem:applyStar-easier-representation}
	Note that the function \(\applyStar\) satisfies the following properties:
	For each $t \in \sigma^\model$ and every $\sigma\in\sorts{\Sigma}$, we have:
	\begin{enumerate}[label=(\alph*)]
		\item If \(t = c(t_1,\dots,t_m)\) for some $c \in \constructors_{\Sigma_{n+1}}$,
		      \(
		      \applyStar(t) \;=\; \trsig{c}\bigl(\applyStar(t_1),\,\dots,\,\applyStar(t_m)\bigr).
		      \)
		\item If $\sigma = \arrType_j$ for some $j \in [1,n]$.
		      Then, let
		      \(
		      \iota_1,\dots,\iota_{\indexCar{\sigma}{\formulaProof}}
		      \)
		      be the elements of \(\indexSet{\sigma}{\formulaProof}\),
		      listed in the order determined by \(\indOrder{\sigma}{\formulaProof}\).
		      Then,
		      \[
			      \applyStar(t)
			      \;=\;
			      \arrCons{j} \bigl(\,
			      t,\,
			      \applyStar(\,\select{\sigma}^\model\!(t,\iota_1^\model)),\dots,
			      \applyStar(\,\select{\sigma}^\model\!(t,\iota_{\indexCar{\sigma}{\formulaProof}}^\model))
			      \bigr).
		      \]
		\item If \(\sigma\notin (\structsorts_{\Sigma_{n+1}}\cup \arrays_\Sigma)\), then
		      $\applyStar(t)=t$.
	\end{enumerate}
\end{remark}

\begin{example}
	\label{ex:gamma}
	Consider the following elements in the interpretations of the sorts in~\Cref{tab:exInter}:
	\begin{enumerate}
		\item \(f_1 := \constArr{\emp} \in A^{\modelS}\).
		\item \(x_1 := \container(1, f_1) \in E^{\modelS}\).
		\item \(f_2 :=\constArr{x_1}\in A^{\modelS}\).
		\item \(x_2 = \container(2, f_2)\in E^{\modelS}\).
	\end{enumerate}

	Then,
	\begin{enumerate}
		\item \(\applyStar(f_1) = \arrCons{1}\bigl(f_1, \applyStar(\emp)\bigr) = \arrCons{1}\bigl(f_1, \trsig{\emp}\bigr)\).
		\item \(\applyStar(x_1) = \trsig{\container}\bigl(\applyStar(1), \applyStar(f_1)\bigr)\) =
		      \(\trsig{\container}\bigl(1, \arrCons{1}\bigl(f_1, \trsig{\emp}\bigr)\bigr)\).
		\item \(\applyStar(f_2) = \arrCons{1}\bigl(f_2, \applyStar(x_1)\bigr) = \arrCons{1}\bigl(f_2, \trsig{\container}\bigl(1, \arrCons{1}\bigl(f_1, \trsig{\emp}\bigr)\bigr)\bigr)\).
		\item \(\applyStar(x_2) = \trsig{\container}\bigl(\applyStar(2), \applyStar(f_1)\bigr)\) =
		      \(\trsig{\container}\bigl(2, \arrCons{1}\bigl(f_2, \trsig{\container}\bigl(1, \arrCons{1}\bigl(f_1, \trsig{\emp}\bigr)\bigr)\bigr))\).
	\end{enumerate}
\end{example}

\begin{lemma}
	\label{lem:applyStar-injective}
	$\applyStar$ is injective.
\end{lemma}

\begin{proof}
	We show first that if \(t_1\in\sigma_1^\model\) and \(t_2\in\sigma_2^\model\) with \(\sigma_1\neq\sigma_2\), then \(\applyStar(t_1)\neq\applyStar(t_2)\).
	Indeed, by~\Cref{lem:applyStar-total}
	\(\applyStar(t_1)\in\trsig{\sigma_1}^{\newM}\) and \(\applyStar(t_2)\in\trsig{\sigma_2}^{\newM}\),
	and the domains for different new sorts are disjoint (see \Cref{tab:domains-newM}).

	It remains to prove that for each fixed sort \(\sigma\), whenever \(t_1,t_2\in\sigma^\model\) and \(t_1\neq t_2\), then
	\(\applyStar(t_1)\neq\applyStar(t_2)\). We proceed by induction on
	\(
	N \;:=\;\max\bigl\{h(t_1),\,h(t_2)\bigr\},
	\)

	\medskip\noindent
	\textbf{Base case (\(N=0\)).}
	According to the definition of \(\height\), if \(N=0\), then $\sigma \notin (\structsorts_{\Sigma_{n+1}}\cup \arrays_\Sigma)$, and thus $\applyStar(t_1) = t_1 \neq t_2 = \applyStar(t_2)$.

	\medskip\noindent
	\textbf{Inductive step.}
	Assume the claim holds for all pairs whose maximum depth is \(<N\). Let \(t_1,t_2\in\sigma^\model\) with \(\max\{h(t_1),h(t_2)\}=N>0\). There are two cases to consider:
	\begin{enumerate}
		\item If $\sigma \in \structsorts_{\Sigma_{n+1}}$, then there are two options:
		      \begin{enumerate}[label=(\alph*)]
			      \item If \(t_1=c(s_1,\dots,s_m)\) and \(t_2=\hat c(r_1,\dots,r_k)\) with \(c\neq\hat c\), then according to~\Cref{rem:applyStar-easier-representation}, we have
			            \[
				            \applyStar(t_1)=\trsig{c}(\dots)
				            \;\neq\;
				            \trsig{\hat c}(\dots)
				            =\applyStar(t_2),
			            \]
			            since distinct constructors yield disjoint images.

			      \item Otherwise both use the same constructor \(c\), say
			            \(t_1=c(s_1,\dots,s_m)\) and \(t_2=c(r_1,\dots,r_m)\), but \(t_1\neq t_2\). Then for some \(j\), \(s_j\neq r_j\) and
			            \(h(s_j),h(r_j)<N\). By the inductive hypothesis
			            \(\applyStar(s_j)\neq\applyStar(r_j)\). Thus, due to~\Cref{rem:applyStar-easier-representation},
			            \(\applyStar(t_1) = \trsig{c}(\ldots,\applyStar(s_j),\ldots) \neq
			            \trsig{c}(\ldots,\applyStar(r_j),\ldots) = \applyStar(t_2)\).
		      \end{enumerate}

		\item If \(\sigma=\arrType_j\in\arrays_\Sigma\), then
		      \[
			      \applyStar(t_1)=\arrCons{j}\bigl(t_1,\ldots\bigr) \neq
			      \arrCons{j}\bigl(t_2,\ldots\bigr) = \applyStar(t_2)
		      \]
	\end{enumerate}
\end{proof}

We can now define the function $\applyTag$ that will be used to define the interpretation of some of the variables in $\trF(\formulaProof)$.

\subsubsection{Defining \(\applyTag\)}

We define a function
\(
\applyTag:
\bigcup_{i=1}^{n} \arrType_i^\model
\to
\bigcup_{i=1}^{n} \trarr{\arrType_i}^{\newM}
\)
as follows:
for each \(t\in \arrType_j^\model\) (with \(j\in [1,n]\)),
$\applyTag(t)\in \trarr{\arrType_j}^{\newM}= \trsig{I_j}^\newM \ra \trsig{E_j}^\newM$ is a function
defined as follows:
\(\forall i\in
\trsig{I_j}^{\newM} = I_j^\model\):
\[
	\applyTag(t)(i) = \applyStar\Bigl(\select{\arrType_j}^\model(t, i)\Bigr).
\]

\begin{lemma}
	\label{lem:applyTag-well-defined}
	For every \(t\in \arrType_j^\model\) (with \(j\in [1,n]\)),
	\(\applyTag(t) \in \trarr{\arrType_j}^{\newM}\)
\end{lemma}
\begin{proof}
	Let \(t\in \arrType_j^\model\) (with \(j\in [1,n]\)).
	Then note that every element in \(\trsig{I_j}^{\newM}\) is a function from \(\trsig{I_j}^{\newM}\) to \(\trsig{E_j}^{\newM}\).
	Let there be \(i\in \trsig{I_j}^{\newM}\).
	Since $\trsig{I_j}^{\newM} = I_j^\model$, we have that \(i\in I_j^\model\). We know that \(\select{\arrType_j}^\model(t, i)\) is defined and in \(E_j^\model\) since \(t\in \arrType_j^\model\) and \(i\in I_j^\model\). ~\Cref{lem:applyStar-total} shows that \(\applyStar(\select{\arrType_j}^\model(t, i)) \in \trsig{E_j}^{\newM}\).
	Thus, \(\applyTag(t)(i) = \applyStar(\select{\arrType_j}^\model(t, i))\) is defined and in \(\trsig{E_j}^{\newM}\).
\end{proof}

\begin{example}
	\label{ex:applyTag}
	Consider $a^{\modelS} = \constArr{\emp}$ from~\Cref{tab:exInter}.
	Then $\applyTag(a^{\modelS})$ is a function from $\trsig{I_1}^{\newM} = \mathbb{N}$ to $\trsig{E_1}^{\newM}$, defined as follows:
	\(\forall n \in \mathbb{N}. \applyTag(a^{\modelS})(n) = \applyStar(\select{\arrType_1}^{\modelS}(a^{\modelS}, n)) = \applyStar(\emp) = \trsig{\emp}\).
	Thus, \(\applyTag(a^{\modelS}) = \constArr{\trsig{\emp}}\)
\end{example}

\begin{lemma}
	\label{lem:applyTag-injective}
	\( \applyTag\) is injective.
\end{lemma}

\begin{proof}
	Let \(t_1 \in \arrType_j^{\model}\) and \(t_2 \in \arrType_k^{\model}\) with \(j,k \in [1,n]\), such that \(t_1 \neq t_2\).

	\textbf{Case 1:} \(j \neq k\).
	Then \(t_1\in\arrType_j^\model\) and \(t_2\in\arrType_k^\model\), so
	\(
	\applyTag(t_1) \in \trarr{\arrType_j}^{\newM}
	\land
	\applyTag(t_2) \in \trarr{\arrType_k}^{\newM}.
	\)
	These are disjoint domains by construction, so \(\applyTag(t_1)\neq \applyTag(t_2)\).

	\textbf{Case 2:} \(j = k\), but \(t_1 \neq t_2\).
	Then \(t_1, t_2 \in \arrType_j^{\model}\). Since \(\model\) satisfies extensionality for arrays, there exists some index \(i \in I_j^{\model} = \trsig{I_j}^{\newM}\) such that
	\[
		\select{\arrType_j}^{\model}(t_1, i) \neq \select{\arrType_j}^{\model}(t_2, i).
	\]
	By the definition of \(\applyTag\) and the injectivity of \(\applyStar\) (Lemma~\ref{lem:applyStar-injective}),
	\begin{align*}
		\applyTag(t_1)(i)
		= \applyStar\bigl(\select{\arrType_j}^\model(t_1, i)\bigr) \neq \applyStar\bigl(\select{\arrType_j}^\model(t_2, i)\bigr)
		= \applyTag(t_2)(i)
	\end{align*}

	Hence \(\applyTag(t_1)\neq\applyTag(t_2)\) in both cases. This proves that
	\(\applyTag\) is injective.
\end{proof}

\subsubsection{Interpretation of the Variables, using $\applyStar$ and $\applyTag$}

The interpretation of the variables introduced during the translation is summarized in~\Cref{tab:new-vars-interpret}.
Given a variable $x$ for $\newSig{\Sigma}{\formulaProof}$ we have three possibilities.
If $x=\trsig{u}$ for some variable $u$ of a sort in $\Sigma$,
the interpretation of $x$ in $\newM$ is set to be the output
of $\applyStar$ on the interpretation of $u$ in $\model$;
If $x=\trarr{a}$ for some variable $a$ of an array sort of $\Sigma$,
the interpretation of $x$ in $\newM$ is set to be the output
of $\applyTag$ on the interpretation of $a$ in $\model$; finally,
if $x=\tempvar{a}{k}$ for some
variable $a$ of an array sort $A_i$ of $\Sigma$ and
$1\leq k\leq \indexCar{\arrType_i}{\formulaProof}$,
then $x$ is interpreted in $\newM$ using the following considerations.
First, let $j$ such that $k=\indOrder{\arrType_i}{\formulaProof}(j)$
(or, equivalently, $j=\indOrder{\arrType_i}{\formulaProof}^{-1}(k)$).
The reason that every $\tempvar{a}{k}$ is interpreted this way is that we want to ensure that every possible application of $\Lzero{a}{i}{\formulaProof}$ is satisfied, as seen in~\Cref{app:sec:satisfaction-lemmas}.

\begin{table}[t]
	\centering
	\begin{tabular}{|c|c|c|c|c|}
		\hline
		\shortstack{Original                                                                                                                                                        \\ Variable} &
		\shortstack{Original                                                                                                                                                        \\ Sort}
		                     & \shortstack{New                                                                                                                                      \\ Variable} & \shortstack{New \\ Sort} & \shortstack{Interpretation \\ in $\newM$} \\\hline
		$u$                  & $\sigma \in \sorts{\Sigma}$                          & $\trsig{u}$        & $\trsig{\sigma}$ & $\trsig{u}^{\newM}= \applyStar\bigl(u^{\model}\bigr)$ \\\hline
		\multirow{2}{*}{$a$} & \multirow{2}{*}{$\arrType_i \in \arrays_\Sigma$}
		                     & $\trarr{a}$                                          & $\trarr{\arrType}$
		                     & $\trarr{a}^{\newM}= \applyTag\bigl(a^{\model}\bigr)$                                                                                                 \\\cline{3-5}
		                     &
		                     & $\tempvar{a}{k}$
		                     & $\trsig{E_i}$
		                     & $\displaystyle
			\tempvar{a}{k}^{\newM}=
			\applyStar\!\bigl(
			\select{\arrType_i}^{\model}\!\bigl(
			a^{\model},
			(\indOrder{\arrType_i}{\formulaProof}^{-1}(k))^{\model}
			\bigr)
			\bigr)$                                                                                                                                                                     \\
		\hline
	\end{tabular}
	\caption{Fresh variables introduced for the signature $\newSig{\Sigma}{\formulaProof}$ and their interpretations in $\newM$. Here $k$ ranges over $1,\dots,\indexCar{\arrType_i}{\formulaProof}$.}
	\label{tab:new-vars-interpret}
\end{table}

\begin{example}
	\label{ex:interpretation-variables}
	Consider the formula $\trF(\formulaExSat)$ from~\Cref{tab:tndtsatex}.
	Its variable interpretations in $\newMS$ are given in \Cref{tab:extransinter},
	and rely on \Cref{ex:gamma,ex:applyTag}.
\end{example}

\subsection{Interpretation of Functions and Predicates}
\label{app:sec:interpretation-functions}

\begin{table}[t]
	\centering
	\begin{tabular}{|l|l|}
		\hline
		Array operators & standard read/update interpretation \\\hline
		Constructors    & Fixed by the theory                 \\\hline
		Testers         & Fixed by the theory                 \\\hline
		Selectors       &
		\shortstack{
			$
				\sel{\trsig{c}}{i}^{\newM}(t)=
				\begin{cases}
					t_i                                            & \text{if }
					t=\trsig{c}(t_1,\dots,t_m)\text{ for some }t_1,\ldots,t_m                     \\[4pt]
					\applyStar\!\bigl(\sel{c}{i}^{\model}(s)\bigr) &
					\text{otherwise, if there is } s\in\sigma^{\model}\text{ s.t.}\applyStar(s)=t \\[4pt]
					\defelem{\trtype{\sigma_{i}}}                  & \text{otherwise}
				\end{cases}
			$}                                                    \\\hline
		UF symbols      &
		\shortstack{
			$
				f_{\arrType_{j}}^{\newM}\bigl(\arrCons{j}(t_1,\dots,t_k)\bigr)
				=\applyTag(t_1)
			$                                                     \\
			$
				g_{\arrType_{j}}^{\newM}(t)=
				\begin{cases}
					\applyStar(a)                   & \text{if } a\in\arrType_{j}^{\model}\text{ and } \applyTag(a)=t, \\[4pt]
					\defelem{\trtype{\arrType_{j}}} & \text{otherwise}.
				\end{cases}
			$}                                                    \\\hline
	\end{tabular}
	\caption{Interpretation of function and predicate symbols,
		where \(\defelem{\trtype{\sigma_{i}}}\) is an arbitrary (but fixed) element of
		\(\trtype{\sigma_{i}}^{\newM}\).
	}
	\label{tab:funprednewM}
\end{table}

The interpretation of function and predicate symbols
in $\newM$ is defined in \Cref{tab:funprednewM} and
described as follows.

\paragraph{Array operations.}
The symbols \(\select{\trarr{\arrType_{j}}}\) and
\(\store{\trarr{\arrType_{j}}}\) are interpreted in the standard way,
as read and update operations over functions.

\paragraph{Datatype symbols.}
The interpretation of
constructors and testers is fixed, as we require
$\newM^{\trsig{{\Sigma_{n+1}^{\formulaProof}}}}$ to be a datatypes structure.
As for selectors, their interpretation is fixed when applied
to terms that are constructed by their corresponding constructors.
When applied to different constructors, we maintain the connection
to $\model$ as induced by $\applyStar$, or set them
arbitrarily if no such connection is induced.
Injectivity of \(\applyStar \) guarantees that this definition is unambiguous.

\paragraph{Uninterpreted functions.}
Let $1\leq j\leq n$.
We define
$f_{\arrType_{j}}^{\newM}:
	\trtype{\arrType_j}^{\newM}\ra\trarr{\arrType_j}^{\newM}$
and
$g_{\arrType_{j}}^{\newM}:
	\trarr{\arrType_j}^{\newM}\ra\trtype{\arrType_j}^{\newM}$
as follows.
For $f_{\arrType_{j}}$,
each element of $\trtype{\arrType_j}^{\newM}$
has the form $\arrCons{j}(t_1,\ldots,t_k)$,
where $t_1\in{\idx{\arrType_{j}}}^{\newM}=\arrType_j^{\model}$.
We then take $t_1$ and pass it through $\applyTag$.
For $g_{\arrType_{j}}$,
given $t\in\trarr{\arrType_j}$ we consider two cases.
In the first, there exists $a\in\arrType_j^{\model}$
with $\applyTag(a)=t$, and then we return
$\applyStar(a)$.
If there is no such $a$, we return an arbitrary but fixed element
\(\defelem{\trtype{\arrType_{j}}}\) of
\(\trtype{\arrType_{j}}^{\newM}\).
Injectivity of \(\applyTag \) guarantees that the definition of \(g_{\arrType_{j}}^{\newM}\) is unambiguous.

\subsection{$\newM$ is a $\newT{\formulaProof}$-interpretation}
\label{app:sec:newM-newT-interpretation}
Now that the construction of $\newM$ is done, we prove that
it is indeed a $\newT{\formulaProof}$-interpretation.
\begin{lemma}
	\label{lem:newM-newT-interpretation}
	$\newM$ is a $\newT{\formulaProof}$-interpretation.
\end{lemma}
\begin{proof}
	Recall that
	\(
	\newT{\formulaProof}
	\;=\;
	\Bigl(\biguplus_{j=1}^{n} \trsig{T_j}\Bigr)
	\;\uplus\;
	T_{n+1}^\formulaProof
	\;\uplus\;
	T_{n+2}\,.
	\)
	For each \(j\in[1,n]\), the reduct \(\newM^{\trsig{\Sigma_j}}\) interprets \(\select{\trarr{\arrType_j}}\) and \(\store{\trarr{\arrType_j}}\) in the standard way for an arrays structure.
	On \(\Sigma_{n+1}^\formulaProof\), \(\newM\) uses the usual tree‐based interpretation of
	constructors, selectors and testers (see Definition~\ref{datatypes-axioms}). By
	construction, each selector and tester in \(\Sigma_{n+1}^\formulaProof\) obeys the requirements of datatypes structures.
	Finally, on \(\Sigma_{n+2}\) we only have the symbols \(f_{\arrType_j}\) and \(g_{\arrType_j}\),
	and there is no restriction on their interpretation.
\end{proof}

\subsection{Satisfaction of the Translated Literals}
\label{app:sec:satisfaction-literals}
\begin{lemma}
	\label{lem:translit}
	Let $l$ be a $\Sigma$-literal and let $\tra(l)$ be its translation.
	If $\model \models l$ then $\newM \models \tra(l)$.
\end{lemma}

\begin{proof}
	We go over all possible shapes of $l$.

	\begin{enumerate}
		\item\textbf{\(x = y\).}
		      Since $x^{\model}=y^{\model}$,
		      \(\applyStar(x^{\model})=\applyStar(y^{\model})\),
		      hence \(\trsig{x}^{\newM}=\trsig{y}^{\newM}\).
		\item\textbf{\(x \neq y\).}
		      $x^{\model}\neq y^{\model}$ and injectivity of $\applyStar$ gives \(\trsig{x}^{\newM} = \applyStar(x^\model) \neq \applyStar(y^\model) = \trsig{y}^{\newM}\).
		\item\textbf{\(x = \select{\arrType}(a,i)\).}
		      In $\model$ we have
		      \(x^{\model}= \select{\arrType}^{\model}(a^{\model},i^{\model})\).
		      Therefore
		      \[
			      \select{\trarr{\arrType}}^{\newM}(\trarr{a}^{\newM}, \trsig{i}^{\newM})
			      = \trarr{a}^{\newM}(\trsig{i}^{\newM})
			      = \applyStar\Bigl(\select{\arrType}^\model(a^\model,i^\model)\Bigr)
			      = \applyStar(x^\model) = \trsig{x}^{\newM}.
		      \]
		      The first equality is by the definition of $\select{\trarr{\arrType}}^{\newM}$, the second one is by the definition of $\trarr{a}^{\newM}$, the third one is by the substitution of $\select{\arrType}^\model(a^\model,i^\model)$ with $x^\model$, and the last one is by the definition of $\trsig{x}^{\newM}$.
		\item\textbf{\(b = \store{\arrType_l}(a,i,v)\).}
		      As $b^\model = \store{\arrType_l}^\model(a^\model,i^\model,v^\model)$, it follows that:
		      \begin{enumerate}
			      \item $\select{\arrType_l}^\model(b^\model,i^\model)=v^\model$, and
			      \item for all $j\in I_l^\model\setminus\{i^\model\}$, $\select{\arrType_l}^\model(b^\model,j)=\select{\arrType_l}^\model(a^\model,j)$.
		      \end{enumerate}
		      Then, by the new interpretation we have:
		      \begin{enumerate}
			      \item $\trarr{b}^{\newM}(\trsig{i}^{\newM})
				            = \applyStar\bigl(\select{\arrType_l}^\model(b^\model,i^\model)\bigr)
				            = \applyStar(v^\model) = \trsig{v}^{\newM}
				            = \store{\trarr{\arrType_l}}^{\newM}(\trarr{a}^{\newM}, \trsig{i}^{\newM}, \trsig{v}^{\newM})(\trsig{i}^\newM)$,
			      \item For every $j\in \trsig{I_l}^\newM\setminus\{i^\newM\}$,
			            $\trarr{b}^{\newM}(j)
				            = \applyStar\bigl(\select{\arrType_l}^\model(b^\model,j)\bigr)
				            = \applyStar\bigl(\select{\arrType_l}^\model(a^\model,j)\bigr)
				            = \select{\trarr{\arrType_l}}^{\newM}(\trarr{a}^{\newM}, j)
				            =\trarr{a}^{\newM}(j)$,
		      \end{enumerate}
		      Hence, $\forall j \in \trsig{I_l}^\newM.\trarr{b}^{\newM}(j) = \store{\trarr{\arrType_l}}^{\newM}(\trarr{a}^{\newM}, \trsig{i}^{\newM}, \trsig{v}^{\newM})(j)$,
		      and thus, \(\trsig{b}^{\newM} = \store{\trarr{\arrType_l}}^{\newM}(\trarr{a}^{\newM}, \trsig{i}^{\newM}, \trsig{v}^{\newM})\).

		\item\textbf{\(x = \sel{c}{i}(y)\).}
		      As \( x^\model = \sel{c}{i}^\model(y^\model) \), we distinguish two cases:
		      \begin{enumerate}
			      \item \textbf{Case 1:} If \( y^\model \) is an application of \( c \), then we can denote
			            \(
			            y^\model = c(t_1, \dots, t_m)
			            \).
			            Since $\model^{\Sigma_{n+1}}$ is a datatypes $\Sigma_{n+1}$-structure, and by~\Cref{datatypes-axioms}:
			            \(
			            x^\model = t_i
			            \).
			            Consequently, according to \Cref{rem:applyStar-easier-representation}
			            \[
				            \trsig{y}^{\newM} = \applyStar(y^\model) = \trsig{c}\bigl(\applyStar(t_1), \dots, \applyStar(t_m)\bigr),
			            \]
			            and applying the new selector yields
			            \[
				            \sel{\trsig{c}}{i}^{\newM}(\trsig{y}^{\newM}) = \applyStar(t_i) = \applyStar(x^\model) = \trsig{x}^{\newM}.
			            \]

			      \item \textbf{Case 2:} If \( y^\model \) is not an application of \( c \), then
			            \[
				            \sel{\trsig{c}}{i}^{\newM}(\trsig{y}^{\newM})
				            = \applyStar\Bigl(\sel{c}{i}^\model\bigl(y^\model\bigr)\Bigr)
				            = \applyStar\Bigl(x^\model\Bigr)
				            = \trsig{x}^{\newM}
			            \]
		      \end{enumerate}

		      In both cases, the equality
		      \[
			      \trsig{x}^{\newM} = \sel{\trsig{c}}{i}^{\newM}(\trsig{y}^{\newM})
		      \]
		      holds.

		\item\textbf{\(x = c(t_1,\dots,t_m)\).}
		      Since $x^\model = c(t_1^\model, \ldots, t_m^\model)$, \Cref{rem:applyStar-easier-representation} suggests that
		      \[
			      \trsig{x}^{\newM}
			      = \applyStar(x^\model)
			      = \applyStar(c(t_1^\model, \ldots, t_m^\model))
			      = \trsig{c}\bigl(\applyStar(t_1^\model), \ldots, \applyStar(t_m^\model)\bigr)
			      = \trsig{c}\bigl(\trsig{t_1}^{\newM}, \ldots, \trsig{t_m}^{\newM}\bigr).
		      \]

		\item\textbf{\(\is{c}(x)\) or \(\neg\is{c}(x)\).}
		      Whether $x^{\model}$ is (or is not) a $c$-term is preserved by
		      $\applyStar$, so the corresponding tester
		      $\is{\trsig{c}}$ (or its negation) is satisfied in $\newM$.
	\end{enumerate}
\end{proof}

Hence every literal of $\trF(\formulaProof)$ is true in $\newM$, so
$\newM\models\trF(\formulaProof)$.

\subsection{Satisfaction of the Lemmas}
\label{app:sec:satisfaction-lemmas}

\subsubsection{Satisfaction of \(\Lone{\formulaProof}\) and \(\Ltwo{\formulaProof}\)}

Fix an array variable \(a\in\fv{\arrType_j}{\formulaProof}\) (for some \(j\in[1,n]\)).
Let
\( \iota_1,\dots,\iota_{\indexCar{\arrType_j}{\formulaProof}}
\quad\text{be the indices in }
\indexSet{\arrType_j}{\formulaProof}
\text{ ordered by }
\indOrder{\arrType_j}{\formulaProof}.
\)
For each \(k\in[1,\indexCar{\arrType_j}{\formulaProof}]\) the block
\(
\Lzero{a}{\iota_k}{\formulaProof}
=
\{\,
\tempvar{a}{k}=\select{\trarr{\arrType_j}}(\trarr{a},\trsig{\iota_k}),
\;
\tempvar{a}{k}=\sel{\arrCons{j}}{\,k+1}(\trsig{a})
\,\}
\)
is satisfied in \(\newM\):

\[
	\tempvar{a}{k}^\newM
	= \applyStar\Bigl(\select{\arrType_j}^\model\bigl(a^\model, \iota_k^\model\bigr)\Bigr)
	= \trarr{a}^{\newM}(\trsig{\iota_k}^\newM)
	= \select{\trarr{\arrType_j}}^{\newM}\bigl(\trarr{a}^{\newM}, \trsig{\iota_k}^\newM\bigr),
\]
Also, the interpretation of $\trsig{a}$ in $\newM$ is defined as follows:

\[
	\trsig{a}^{\newM}
	= \arrCons{j}\Bigl(
	a^\model,
	\applyStar\Bigl(\select{\arrType_j}^\model\bigl(a^\model, \iota_1^\model\bigr)\Bigr),
	\ldots,
	\applyStar\Bigl(\select{\arrType_j}^\model\bigl(a^\model, \iota_{\indexCar{\arrType_j}{\formulaProof}}^\model\bigr)\Bigr)
	\Bigr).
\]
Hence, it follows that
\[
	\tempvar{a}{k}^\newM
	=\applyStar\Bigl(\select{\arrType_j}^\model\bigl(a^\model, \iota_k^\model\bigr)\Bigr)
	=\sel{\arrCons{j}}{k+1}^{\newM}(\trsig{a}^{\newM})
\]
and thus $\newM \models \Lzero{a}{\iota_k}{\formulaProof}$.

Note that
\[
	\Lone{\formulaProof} \cup \Ltwo{\formulaProof} \subseteq \bigcup_{a \in \fv{\arrType_j}{\formulaProof}} \bigcup_{i \in \indexSet{\arrType_j}{\formulaProof}} \Lzero{a}{i}{\formulaProof}.
\]

Thus, we have shown that for every literal \(l \in \Lone{\formulaProof} \cup \Ltwo{\formulaProof}\), we have $\newM \models l$.

\subsubsection{Satisfaction of \(\Lthree{\formulaProof}\)}
Let \(a \in \fv{\arrType_j}{\formulaProof}\) for some $j \in [1,n]$.
Note that:
\[
	f_{\arrType_j}^{\newM}(\trsig{a}^{\newM})
	= f_{\arrType_j}^{\newM}(\arrCons{j}(a^\model, \ldots))
	= \applyTag(a^\model)
	= \trarr{a}^{\newM}
\]

And:
\[
	g_{\arrType_j}^{\newM}(\trarr{a}^{\newM})
	= g_{\arrType_j}^{\newM}(\applyTag(a^\model))
	= \applyStar(a^\model)
	= \trsig{a}^{\newM}
\]

Thus, $\forall l \in \Lthree{\formulaProof}. \newM \models l$.

\section{Completeness}
\label{app:sec:completeness}

In what follows,
$\Sigma$ continues to be the NDT signature of $\Sigma_1,\ldots,\Sigma_{n+1}$,
where $\Sigma_1,\dots,\Sigma_n$ are arrays signatures and $\Sigma_{n+1}$ is a datatypes signature.
In contrast to \Cref{thm:preprocessing-correctness}, the
completeness result assumes that $T_{n+1}$
is stably infinite w.r.t.
$\big(\bigcup_{i=1}^{n}\sorts{\Sigma_{i}}\big) \cap \sorts{\Sigma_{n+1}}$,
a.k.a.
w.r.t. its {\em shared sorts}.

\begin{theorem}
	\label{thm:preprocessing-correctness-2}
	Let $\Sigma$ be the NDT signature of
	$\Sigma_1,\dots,\Sigma_{n+1}$ such that $T_{n+1}$ is stably infinite w.r.t. $\big(\bigcup_{i=1}^{n}\sorts{\Sigma_{i}}\big) \cap \sorts{\Sigma_{n+1}}$ and let
	$\formulaProof$ be a flat $\Sigma$-cube such that $\trF(\formulaProof)$ is
	$\newT{\formulaProof}$-satisfiable.
	Then $\formulaProof$ is $\tndt{\Sigma}$-satisfiable.
\end{theorem}

The main idea of this section is to generalize the construction of the interpretation presented in~\Cref{tab:exInter}.

\begin{table}[t]
	\newcolumntype{M}[1]{>{\centering\arraybackslash}m{#1}}
	\centering
	\begin{tabular}{|M{2cm}|M{11cm}|}
		\hline
		\textbf{Domains}    &
		\begin{tabular}{@{}l@{}}
			$\trsig{I}^{\hat{\newM}} = \mathbb{N}$                                                                             \\[3pt]
			${\idx{A}}^{\hat{\newM}} = \{\gamma\}$                                                                             \\[3pt]
			$\trsig{E}^{\hat{\newM}}
				= T_{\trsig{E}}\bigl(\consig{\newSig{\exsigwf}{\formulaExSat}}{\newSig{\exsigwf}{\formulaExSat}},\,B\bigr)$        \\[3pt]
			$\trsig{\arrType}^{\hat{\newM}}
				= T_{\trsig{\arrType}}\bigl(\consig{\newSig{\exsigwf}{\formulaExSat}}{\newSig{\exsigwf}{\formulaExSat}},\,B\bigr)$ \\[3pt]
			$\trarr{A}^{\hat{\newM}}
				= \trsig{I}^{\hat{\newM}}\;\to\;\trsig{E}^{\hat{\newM}}$
		\end{tabular}
		\\\hline
		\textbf{Functions}  &
		\begin{tabular}{@{}l@{}}
			\(\select{},\store{}\): standard read/update on \(\trarr{A}^{\hat{\newM}}\).      \\[3pt]
			\(\trsig{\container}^{\hat{\newM}}(n,y)=\trsig{\container}(n,y)\)                 \\[3pt]
			\(\trsig{\emp}^{\hat{\newM}}=\trsig{\emp}\)                                       \\[3pt]
			\(\trsig{\id}^{\hat{\newM}}(\trsig{\container}(n,y))=n\),\quad
			\(\trsig{\children}^{\hat{\newM}}(\trsig{\container}(m,y))=y\)                    \\[3pt]
			\(\trsig{\id}^{\hat{\newM}}(\trsig{\emp})=0\),\quad
			\(\trsig{\children}^{\hat{\newM}}(\trsig{\emp})=\constArr{\trsig{\emp}}\)         \\[6pt]
			\(f_{A}^{\hat{\newM}}\): constant function returning \(\trarr{a}^{\hat{\newM}}\). \\[3pt]
			\(g_{A}^{\hat{\newM}}\): constant function returning \(\trsig{a}^{\hat{\newM}}\).
		\end{tabular}
		\\\hline
		\textbf{Predicates} &
		\begin{tabular}{@{}l@{}}
			\(\is{\trsig{\container}}^{\hat{\newM}}(\trsig{\container}(n,y))=\text{true}\), \\[3pt]
			\(\is{\trsig{\container}}^{\hat{\newM}}(\trsig{\emp})=\text{false}\),           \\[3pt]
			\(\is{\trsig{\emp}}^{\hat{\newM}}(\trsig{\emp})=\text{true}\),                  \\[3pt]
			\(\is{\trsig{\emp}}^{\hat{\newM}}(\trsig{\container}(n,y))=\text{false}\).
		\end{tabular}
		\\\hline
		\textbf{Variables}  &
		\begin{tabular}{@{}l@{}}
			\(\trsig{x}^{\hat{\newM}} \;=\; \tempvar{a}{1}^{\hat{\newM}}
			= \trsig{\container}\bigl(3,\;\arrCons{1}(\gamma,\;\trsig{\emp})\bigr)\)                   \\[6pt]
			\(\trarr{a}^{\hat{\newM}}
			= \constArr{\trsig{x}^{\hat{\newM}}}
			= \constArr{\trsig{\container}(3,\arrCons{1}(\gamma,\trsig{\emp}))}\)                      \\[6pt]
			\(\trsig{j}^{\hat{\newM}} = \trsig{k}^{\hat{\newM}} = 1\)                                  \\[6pt]
			\(\trsig{a}^{\hat{\newM}}
			= \arrCons{1}\bigl(\gamma,\;\trsig{\container}(3,\arrCons{1}(\gamma,\trsig{\emp}))\bigr)\) \\[6pt]
			\(\trsig{y}^{\hat{\newM}}
			= \trsig{\container}\bigl(1,\;\arrCons{1}(\gamma,\trsig{\container}(3,\arrCons{1}(\gamma,\trsig{\emp})))\bigr)\)
		\end{tabular}
		\\\hline
	\end{tabular}
	\caption{A $\newSig{\exsigwf}{\formulaExSat}$–interpretation, $\hat{\newM}$, that satisfies $\trF(\formulaExSat)$, from \Cref{tab:tndtsatex}.
		Here $B$ is the $\elemsorts_{\newSig{\exsigwf}{\formulaExSat}}$–sorted set with $B_{\rho}=\rho^{\hat{\newM}}$ for each $\rho\in\elemsorts_{\newSig{\exsigwf}{\formulaExSat}}$.
	}
	\label{tab:exCompleteness}
\end{table}

\begin{example}
	\label{ex:proof-formula-2}
	Consider the formula $\formulaExSat$ from
	\Cref{tab:tndtsatex}, and its translation $\trF(\formulaExSat)$, also from
	\Cref{tab:tndtsatex}.
	We already know from \Cref{ex:proof-formula} that $\formulaExSat$ is
	$\tndt{\exsigwf}$-satisfiable,
	and from \Cref{thm:preprocessing-correctness} we get that
	$\trF(\formulaExSat)$ is $\newT{\formulaExSat}$-satisfiable.
	We reuse these formulas in order to provide examples
	for the steps of our completeness proofs.
	Since our starting point should be a
	$\newT{\formulaExSat}$-interpretation that satisfies $\trF(\formulaExSat)$,
	which may not be the one constructed in \Cref{app:sec:soundness},
	we define in \Cref{tab:exCompleteness} a new $\newT{\formulaExSat}$-interpretation,
	not obtained from the process of \Cref{app:sec:soundness},
	that satisfies $\trF(\formulaExSat)$.
	We shall later see that $\formulaExSat$ is $\tndt{\exsigwf}$-satisfiable, by constructing a satisfying $\tndt{\exsigwf}$-interpretation
	based on the one from \Cref{tab:exCompleteness}.
	In turn, this satisfying interpretation will be different
	from the one shown in \Cref{tab:exInter}.
\end{example}

\begin{table}[t]
	\newcolumntype{M}[1]{>{\centering\arraybackslash}m{#1}}
	\centering
	\begin{tabular}{|M{2cm}|M{11cm}|}
		\hline
		\textbf{Domains}    &
		\begin{tabular}{@{}l@{}}
			$\trsig{I}^{\newM} = \mathbb{N}$                                                                                                                                                                   \\[3pt]
			${\idx{A}}^{\newM} = \{\gamma, \delta\}$                                                                                                                                                           \\[3pt]
			$\trsig{E}^{\newM}
				= T_{\trsig{E}}\bigl(\consig{\newSig{\exsigwf}{\formulaExSat}}{\newSig{\exsigwf}{\formulaExSat}},\,B\bigr) = \{\trsig{\emp}\}$                                                                     \\[3pt]
			$\trsig{\arrType}^{\newM}
				= T_{\trsig{\arrType}}\bigl(\consig{\newSig{\exsigwf}{\formulaExSat}}{\newSig{\exsigwf}{\formulaExSat}},\,B\bigr) = \{\arrCons{1}(\gamma), \arrCons{1}(\delta)\}$                                  \\[3pt]
			$\trsig{S}^{\newM}
				= T_{\trsig{S}}\bigl(\consig{\newSig{\exsigwf}{\formulaExSat}}{\newSig{\exsigwf}{\formulaExSat}},\,B\bigr) = \{\trsig{\container}(\arrCons{1}(\gamma)), \trsig{\container}(\arrCons{1}(\delta))\}$ \\[3pt]
			$\trarr{A}^{\newM}
				= \trsig{I}^{\newM}\;\to\;\trsig{E}^{\newM}$
		\end{tabular}
		\\\hline
		\textbf{Functions}  &
		\begin{tabular}{@{}l@{}}
			\(\select{},\store{}\): standard read/update on \(\trarr{A}^{\newM}\).      \\[3pt]
			\(\trsig{\container}^{\newM}(y)=\trsig{\container}(y)\)                     \\[3pt]
			$\trsig{\emp}^{\newM}=\trsig{\emp}$                                         \\[3pt]
			$\arrCons{1}^{\newM}(u)=\arrCons{1}(u)$                                     \\[3pt]

			$\sel{\trsig{\container}}{1}^{\newM}(\trsig{\container}(y))=y$              \\[3pt]
			$\sel{\arrCons{1}}{1}^{\newM}(\arrCons{1}(u))=u$                            \\[3pt]
			\(f_{A}^{\newM}\): constant function returning an arbitrary constant value. \\[3pt]
			\(g_{A}^{\newM}\): constant function returning an arbitrary constant value.
		\end{tabular}
		\\\hline
		\textbf{Predicates} &
		\begin{tabular}{@{}l@{}}
			\(\is{\trsig{\container}}^{\newM}(\trsig{\container}(y))=\text{true}\), \\[3pt]
			\(\is{\trsig{\emp}}^{\newM}(\trsig{\emp})=\text{true}\),                \\[3pt]
			\(\is{\trsig{\emp}}^{\newM}(\trsig{\emp})=\text{true}\).
		\end{tabular}
		\\\hline
		\textbf{Variables}  &
		\begin{tabular}{@{}l@{}}
			\(\trsig{x}^{\newM} \;=\; \trsig{\container}(\arrCons{1}(\gamma))\) \\[6pt]
			\(\trsig{y}^{\newM} \;=\; \trsig{\container}(\arrCons{1}(\delta))\) \\[6pt]
		\end{tabular}
		\\\hline
	\end{tabular}
	\caption{A $\newSig{\exsigwf}{\formulaExSat}$–interpretation $\newM$ satisfying
		$\trsig{x}\neq\trsig{y}$.
		Here $B$ is the $\elemsorts_{\newSig{\exsigwf}{\formulaExSat}}$–sorted set with $B_{\rho}=\rho^{\newM}$ for each $\rho\in\elemsorts_{\newSig{\exsigwf}{\formulaExSat}}$.}
	\label{tab:exCompleteness-2}
\end{table}

\paragraph{Roadmap.}
Assume $\newM\models\trF(\formulaProof)$ for some arbitrary
$\newT{\formulaProof}$-interpretation~$\newM$.
The goal is, of course, to construct
a
$\tndt{\Sigma}$-interpretation~$\model$ that satisfies $\formulaProof$.
This is carried out in steps:

\begin{enumerate}[label=(\roman*)]
	\item \emph{Repairing $\newM$.}
	      Starting from the given $\newT{\formulaProof}$-interpretation~$\newM$,
	      \Cref{app:sec:comp-intuitive-model} builds a well-behaved $\newT{\formulaProof}$-interpretation that still satisfies $\trF(\formulaProof)$.

	\item \emph{Defining the domains.}
	      \Cref{app:sec:comp-define-domains} defines the interpretation of every
	      sort in the sought $\tndt{\Sigma}$-interpretation~$\model$.
	      In \Cref{app:sec:comp-properties-domains} we prove several useful properties of these domains.

	\item \emph{Guaranteeing infinite $|E_i^\model|$.}
	      ~\Cref{app:sec:comp-ei_infinite} shows that $\forall i \in [1,n]. |E_i^{\model}| = \infty$. This is needed in order to define the interpretation of the variables.
	\item \emph{Defining the interpretation of the symbols.}
	      In~\Cref{app:sec:comp-interpretation-vars}–\Cref{app:sec:comp-interpretation-funcs}, we define the interpretation of the functions, predicates, and variables in $\model$.
	\item \emph{Proving that $\model$ is a
		      $\tndt{\Sigma}$-interpretation.}
	      In~\Cref{app:sec:comp-model-interpretation}, we show that the
	      interpretation $\model$ is a $\tndt{\Sigma}$-interpretation.
	\item \emph{Checking literal satisfaction.}
	      Finally, in~\Cref{app:sec:comp-satisfaction-literal}, we show that the interpretation $\model$ satisfies every literal in $\formulaProof$.
\end{enumerate}

\subsection{Creating a well-behaved interpretation}
\label{app:sec:comp-intuitive-model}

Given a $\newT{\formulaProof}$-interpretation $\newM$,
that satisfies $\trF(\formulaProof)$,
our overall goal is to construct a
$\tndt{\Sigma}$-interpretation $\model$ that satisfies $\formulaProof$.
Our first step is to transform $\newM$ to another interpretation that
is more comfortable to work with.
Namely, we construct a $\newT{\formulaProof}$-interpretation that satisfies
$\trF(\formulaProof)$ and satisfies the following two requirements:
$(i)$~the domains of all the sorts that are not in $\structsorts_{\Sigma_{n+1}^\formulaProof} \cup \arrays_{\newSig{\Sigma}{\formulaProof}}$ ($\arrays_{newSig{\Sigma}} = \{\trarr{\arrType_i} \mid 1 \leq i \leq n\}$, see \Cref{not:ndt-sig}) are infinite; and
$(ii)$~the domains of array sorts consist of functions from
the domains of their index sorts to the domains of their
element sorts
(while this is typically the ``intended" semantics of arrays,
other interpretations can be built that satisfy the axioms of the theory).

We start with the following lemma, that establishes
a simple property of datatype structures.

\begin{lemma}
	\label{lem:elemsorts-set-containment}
	Let $\dtSig$ be a datatype signature, and
	let \(\bar B\) and \(\tilde B\) be \(\elemsorts_{\dtSig}\)-sorted sets with
	\(\bar B_\tau\subseteq\tilde B_\tau\) for every \(\tau\in\elemsorts_{\dtSig}\).
	Then for all \(\sigma\in\sorts{\dtSig}\)
	\(
	T_{\sigma}\bigl(\consig{\dtSig}{\dtSig},\bar B\bigr)
	\;\subseteq\;
	T_{\sigma}\bigl(\consig{\dtSig}{\dtSig},\tilde B\bigr).
	\)
\end{lemma}
\begin{proof}
	First, we denote $\widetilde{\dtSig}:=\consig{\dtSig}{\dtSig}$ for brevity.
	If \(\sigma\in\elemsorts\), then by definition, for every \(k\),
	\[
		T_{\sigma,k}(\widetilde{\dtSig},\bar B)
		= \bar B_\sigma
		\subseteq
		\tilde B_\sigma
		= T_{\sigma,k}(\widetilde{\dtSig},\tilde B).
	\]
	Now let \(\sigma\in\structsorts\). We argue by induction on \(k\) that \(
	T_{\sigma,k}\bigl(\widetilde{\dtSig},\bar B\bigr)
	\;\subseteq\;
	T_{\sigma,k}\bigl(\widetilde{\dtSig},\tilde B\bigr).
	\)

	\emph{Base (\(k=0\)).}
	\[
		T_{\sigma,0}(\widetilde{\dtSig},\bar B)
		= \emptyset
		\subseteq
		T_{\sigma,0}(\widetilde{\dtSig},\tilde B).
	\]

	\emph{Step.} Suppose
	\(T_{\tau,k}(\widetilde{\dtSig},\bar B)\subseteq T_{\tau,k}(\widetilde{\dtSig},\tilde B)\)
	for every \(\tau\). Then
	\[
		T_{\sigma,k+1}(\widetilde{\dtSig},\bar B)
		=
		T_{\sigma,k}(\widetilde{\dtSig},\bar B)
		\;\cup\;
		\bigl\{c(t_1,\dots,t_n)\mid c:\sigma_1\times\cdots\times\sigma_n\to\sigma,\;
		t_j\in T_{\sigma_j,k}(\widetilde{\dtSig},\bar B)\bigr\}.
	\]
	By the IH each \(t_j\in T_{\sigma_j,k}(\widetilde{\dtSig},\tilde B)\),
	so each \(c(t_1,\dots,t_n)\) lies in
	\(T_{\sigma,k+1}(\widetilde{\dtSig},\tilde B)\).
	Hence
	\(
	T_{\sigma}(\widetilde{\dtSig},\bar B)
	= \bigcup_{k\in \mathbb{N}} T_{\sigma,k}(\widetilde{\dtSig},\bar B)
	\subseteq \bigcup_{k\in \mathbb{N}} T_{\sigma,k}(\widetilde{\dtSig},\tilde B)
	= T_{\sigma}(\widetilde{\dtSig},\tilde B)
	\)
\end{proof}

\begin{table}[t]
	\centering
	\renewcommand{\arraystretch}{1.2}
	\begin{tabular}{|c|c|c|}
		\hline
		\textbf{Sort (in $\sorts{\newSig{\Sigma}{\formulaProof}}$)} & \textbf{Original sort (in $\sorts{\Sigma}$)}                                          & \textbf{Domain chosen in $\nextM$} \\ \hline\hline
		$\idx{\arrType_i}$                                          & $\arrType_i \in \arrays_\Sigma$
		                                                            & $\idx{\arrType_i}^{\nextM}=
			\begin{cases}
				\idx{\arrType_i}^{\newM} \cup \{ e_k^\idx{\arrType_i} \mid k \in \mathbb{N} \} & \text{if } |\idx{\arrType_i}^{\newM}| < \infty \\[6pt]
				\idx{\arrType_i}^{\newM}                                                       & \text{otherwise}
			\end{cases}$                                                                                                                                                             \\ \hline
		\multirow{2}{*}{$\trsig{\tau}$}
		                                                            & $\tau \in \sorts{\Sigma} \setminus (\arrays_\Sigma \cup \structsorts_{\Sigma_{n+1}})$
		                                                            & $\trsig{\tau}^\nextM = \begin{cases}
				\trsig{\tau}^{\newM} \cup \{ e_k^{\trsig{\tau}} \mid k \in \mathbb{N} \} & \text{if } |\trsig{\tau}^{\newM}| < \infty \\[6pt]
				\trsig{\tau}^{\newM}                                                     & \text{otherwise}
			\end{cases}$                                                                        \\ \cline{2-3}
		                                                            & $\tau \in \structsorts_{\Sigma_{n+1}} \cup \arrays_\Sigma$
		                                                            & $\displaystyle
			\trsig{\tau}^{\nextM} = T_{\trsig{\tau}}
			\bigl(\consig{\Sigma_{n+1}^{\formulaProof}}{\Sigma_{n+1}^{\formulaProof}},\,B\bigr)$                                                                                                     \\ \hline
		$\trarr{\arrType_i}$                                        & $\arrType_i \in \arrays_\Sigma$
		                                                            & $\displaystyle
			\trarr{\arrType_i}^{\nextM}= \trsig{I_i}^{\nextM}\!\to\!\trsig{E_i}^{\nextM}$                                                                                                            \\
		\hline
	\end{tabular}
	\caption{Interpretations of the sorts in $\nextM$.
	The set \(\{e_k^{\sigma} \mid k \in \mathbb{N}\}\)
	consists of countably many fresh
	elements.
	$B$ is the $\elemsorts_{\newSig{\Sigma}{\formulaProof}}$-sorted set with
	$B_\rho=\rho^{\nextM}$ for every
	$\rho\in\elemsorts_{\newSig{\Sigma}{\formulaProof}}$.
	}
	\label{tab:comp-domains-nextM}
\end{table}

\begin{lemma}
	\label{lem:infinitely-more-elements}
	Let \(\newM\) be a $\newT{\formulaProof}$-interpretation that satisfies a flat \(\newSig{\Sigma}{\formulaProof}\)-cube $\theta$.
	Then, there exists $\nextM$, a $\newT{\formulaProof}$-interpretation that satisfies $\theta$ such that:
	\begin{enumerate}
		\item \(
		      \forall \sigma \in \sorts{\newSig{\Sigma}{\formulaProof}} \setminus (\structsorts_{\Sigma_{n+1}^\formulaProof} \cup \arrays_{\newSig{\Sigma}{\formulaProof}}). |\sigma^{\nextM}| = \infty
		      \).
		\item $\forall i \in [1,n]. \trarr{\arrType_i}^{\nextM} = \trsig{I_i}^{\nextM} \to \trsig{E_i}^{\nextM}$.
		\item The $\select{\trarr{\arrType_i}}$ and $\store{\trarr{\arrType_i}}$ functions are interpreted in the usual way for every $i\in[1,n]$.
	\end{enumerate}
\end{lemma}

\begin{proof}
	W.l.o.g., we can assume that the domains in $\newM$ of all the sorts in
	$\sorts{\newSig{\Sigma}{\formulaProof}} \setminus
		(\structsorts_{\Sigma_{n+1}^\formulaProof} \cup
		\arrays_{\newSig{\Sigma}{\formulaProof}})$ are completely fresh and flat
	(i.e., they in particular do not contain any constructor symbols in
	them, or more specifically, we may assume that they are ``copies" of subsets of the natural numbers).
	In order to show the existence of such $\nextM$, we go through the following
	steps:
	\begin{enumerate}
		\item We define the domains of the sorts in \(\nextM\).
		\item We define the interpretations of the symbols in \(\nextM\).
		\item We show that \(\nextM\) is a $\newT{\formulaProof}$-interpretation.
		\item We show that every literal in \(\formulaProof\) is satisfied by \(\nextM\).
	\end{enumerate}

	\paragraph{1. Defining the domains.}
	The domains of the sorts in \(\newSig{\Sigma}{\formulaProof}\) are defined in \Cref{tab:comp-domains-nextM}.
	Let there be $\sigma \in \sorts{\newSig{\Sigma}{\formulaProof}}$. We have three cases to consider:
	\begin{enumerate}
		\item If $\sigma=\idx{\arrType_i}$ for some $i\in [1,n]$
		      then $\sigma$ is interpreted in $\nextM$ exactly in the same way it is interpreted in $\newM$ if its domain is already infinite.
		      Otherwise, we add infinitely many fresh elements to it, denoted as $\{e_k^{\idx{\arrType_i}} \mid k \in \mathbb{N}\}$.
		\item If $\sigma=\trsig{\tau}$ for some sort $\tau$, there are two sub-cases.
		      \begin{enumerate}
			      \item If $\tau$ is neither in $\structsorts_{\Sigma_{n+1}}$ nor in $\arrays_{\Sigma}$, then according to ~\Cref{fig:dt-sig-trans}, $\trsig{\tau} \notin \structsorts_{\Sigma_{n+1}^\formulaProof}$. We want to make sure that the cardinality of $\trsig{\tau}^\nextM$ is infinite, so if $\trsig{\tau}^\newM$ is finite, we add infinitely many fresh elements to it, denoted as $\{e_k^{\trsig{\tau}} \mid k \in \mathbb{N}\}$.
			            If $\trsig{\tau}^\newM$ is already infinite, we simply set $\trsig{\tau}^\nextM = \trsig{\tau}^\newM$.
			      \item Otherwise, $\tau$ is in $\arrays_{\Sigma}$, or in $\structsorts_{\Sigma_{n+1}}$, and thus, according to ~\Cref{fig:dt-sig-trans}, $\trsig{\tau} \in \structsorts_{\Sigma_{n+1}^\formulaProof}$, and we interpret $\trsig{\tau}$ in $\nextM$ as the set of trees defined in \Cref{def:trees}. (Note that the interpretations of all the sorts in $\elemsorts_{\newSig{\Sigma}{\formulaProof}}$ have already been defined in $\nextM$ as described in items 1 and 2a above, corresponding to the structure specified by \Cref{fig:dt-sig-trans}).
		      \end{enumerate}
		\item If $\sigma=\trarr{\arrType_i}$ for some $i\in [1,n]$, we interpret $\trarr{\arrType_i}$ in $\nextM$ as the set of functions from $\trsig{I_i}^\nextM$ to $\trsig{E_i}^\nextM$.
	\end{enumerate}

	By \Cref{lem:elemsorts-set-containment} we have $\forall \sigma \in \sorts{\Sigma}. \trsig{\sigma}^{\newM} \subseteq \trsig{\sigma}^{\nextM}$.

	\paragraph{2. Defining the interpretation of the variables.}
	For every $\sigma \in \sorts{\newSig{\Sigma}{\formulaProof}}$, we choose an arbitrary element $\defelem{\sigma} \in \sigma^{\nextM}$.

	We also create for every $i\in [1,n]$ a mapping
	\(
	\Psi_i: \trarr{\arrType_i}^{\newM} \ra \trarr{\arrType_i}^{\nextM}
	\).
	Let there be $x \in \trarr{\arrType_i}^{\newM}$.
	Note that $\trarr{\arrType_i}^{\nextM}$ is the set of functions from $\trsig{I_i}^{\nextM}$ to $\trsig{E_i}^{\nextM}$, and thus we can define $\Psi_i(x):\trsig{I_i}^{\nextM} \to \trsig{E_i}^{\nextM}$ in the following way:
	\[
		\forall j \in \trsig{I_i}^{\nextM}. \Psi_i(x)(j) =
		\begin{cases}
			\select{\trarr{\arrType_i}}^\newM(x,j) & \text{if } j \in \trsig{I_i}^{\newM} \\
			\defelem{\trsig{E_i}}                  & \text{otherwise}
		\end{cases}
	\]
	This is well defined since
	$\select{\trarr{\arrType_i}}^\newM(x,j) \in \trsig{E_i}^{\newM}
		\subseteq \trsig{E_i}^{\nextM}$ for every $j \in
		\trsig{I_i}^{\newM}$, and $\defelem{\trsig{E_i}} \in
		\trsig{E_i}^{\nextM}$.

	The mapping \(\Psi_i\) is an injection:
	let there be $a,b \in \trarr{\arrType_i}^{\newM}$ such that $a \neq b$.
	Then, there exists $j \in \trsig{I_i}^{\newM}$ such that $\select{\trarr{\arrType_i}}^{\newM}(a,j) \neq \select{\trarr{\arrType_i}}^{\newM}(b,j)$.
	Thus, $\Psi_i(a)(j) \neq \Psi_i(b)(j)$.

	The variables are interpreted in \(\nextM\) as follows:
	\begin{enumerate}
		\item For every array variable \(a\) of sort \(\trarr{\arrType_i}\), we set \(a^{\nextM} = \Psi_i(a^{\newM})\).
		\item For every variable of any other sort ($\sorts{\Sigma^\formulaProof} \setminus \arrays_{\Sigma^\formulaProof}$) , we set $x^\nextM = x^\newM$. This is well defined since we have shown that $\forall \sigma \in \sorts{\Sigma}. \trsig{\sigma}^{\newM} \subseteq \trsig{\sigma}^{\nextM}$.
	\end{enumerate}

	\paragraph{3. Defining the interpretation of the functions and predicates.}
	The function symbols in \(\nextM\) are interpreted as follows:
	\begin{enumerate}
		\item For every $i \in [1,n]$, the $\select{\trarr{\arrType_i}}$ function symbol is interpreted as
		      \[
			      \select{\trarr{\arrType_i}}^{\nextM}(a,j) =
			      a(j)
		      \]
		\item For every $i \in [1,n]$, the $\store{\trarr{\arrType_i}}$ function symbol is interpreted as
		      \[
			      \store{\trarr{\arrType_i}}^{\nextM}(a,j,v)(k) =
			      \begin{cases}
				      v    & \text{if } k = j \\
				      a(k) & \text{otherwise}
			      \end{cases}
		      \]
		\item The constructor symbols are interpreted according to the datatype theory.
		\item The selector symbols are interpreted as follows ($\sel{c}{i}:\trsig{\sigma} \to \trsig{\tau}$ when $\sigma, \tau \in \sorts{\Sigma}$):
		      \[
			      \sel{c}{i}^{\nextM}(x) = \begin{cases}
				      t_i                    & \text{if } x = c(t_1,\dots,t_n)                    \\
				      \sel{c}{i}^{\newM}(x)  & \text{otherwise, if } x \in \trsig{\sigma}^{\newM} \\
				      \defelem{\trsig{\tau}} & \text{otherwise}
			      \end{cases}
		      \]

		      Note that this function is well defined since $\trsig{\tau}^{\newM} \subseteq \trsig{\tau}^{\nextM}$ and thus $\sel{c}{i}^{\newM}(x) \in \trsig{\tau}^{\nextM}$ when $x \in \trsig{\sigma}^{\newM}$.
		\item For every function symbol \(f_{\arrType_i}\) of arity \(\trsig{\arrType_i} \to \trarr{\arrType_i}\), we define:
		      \[
			      f_{\arrType_i}^{\nextM}(x) =
			      \begin{cases}
				      \Psi_i(f_{\arrType_i}^{\newM}(x)) & \text{if } x \in \trsig{\arrType_i}^{\newM} \\
				      \defelem{\trarr{\arrType_i}}      & \text{otherwise}
			      \end{cases}
		      \]
		      Note that this function is well defined since $f_{\arrType_i}^{\newM}:\trsig{\arrType_i}^{\newM} \to \trarr{\arrType_i}^{\newM}$ and $\Psi_i: \trarr{\arrType_i}^{\newM} \to \trarr{\arrType_i}^{\nextM}$. Thus, we have that $\Psi_i(f_{\arrType_i}^{\newM}(x)) \in \trarr{\arrType_i}^{\nextM}$ for every $x \in \trsig{\arrType_i}^{\newM}$.

		\item For every function symbol \(g_{\arrType_i}\) of arity \(\trarr{\arrType_i} \to \trsig{\arrType_i}\), we set
		      \[ g_{\arrType_i}^{\nextM}(x) =
			      \begin{cases}
				      g_{\arrType_i}^{\newM}(\Psi_i^{-1}(x)) & \text{if } \Psi_i^{-1}(x) \neq \emptyset \\
				      \defelem{\trsig{\arrType_i}}           & \text{otherwise}
			      \end{cases}
		      \].

		      Note that this function is well defined.
		      First, we have already shown that $\Psi_i$ is injective, so $\Psi_i^{-1}(x)$ contains a single element, if any.
		      Second, since $\Psi_i \colon \trarr{\arrType_i}^{\newM} \to \trarr{\arrType_i}^{\nextM}$, we know that the restricted inverse
		      \[
			      \Psi_i^{-1}\!\!\upharpoonright_{\{y \in \trarr{\arrType_i}^{\nextM} \mid \Psi_i^{-1}(y) \neq \emptyset\}} \colon \{y \in \trarr{\arrType_i}^{\nextM} \mid \Psi_i^{-1}(y) \neq \emptyset\} \to \trarr{\arrType_i}^{\newM}
		      \]
		      is well defined.

		      Thus, as
		      $g_{\arrType_i}^{\newM}:\trarr{\arrType_i}^{\newM} \to
			      \trsig{\arrType_i}^{\newM}$ and $\trsig{\arrType_i}^{\newM} \subseteq
			      \trsig{\arrType_i}^{\nextM}$, we have that
		      $g_{\arrType_i}^{\newM}(\Psi_i^{-1}(x)) \in \trsig{\arrType_i}^{\nextM}$
		      when $\Psi_i^{-1}(x) \neq \emptyset$.

		\item The tester symbols are interpreted in the standard way for datatype structors:
		      \[
			      \forall c \in \constructors_{\newSig{\Sigma}{\formulaProof}}.\;
			      \is{c}^{\nextM}(x) =
			      \begin{cases}
				      true  & \text{if } x = c(\ldots) \\
				      false & \text{otherwise}
			      \end{cases}
		      \]
	\end{enumerate}
	By the way $\nextM$ interprets the $\select{\trarr{\arrType_j}}$ and $\store{\trarr{\arrType_j}}$ for every $j \in [1,n]$, we get that $\nextM^{\trsig{\Sigma_j}}$ is an array $\trsig{\Sigma_j}$-structure.
	Likewise, $\nextM^{\newSig{\Sigma}{\formulaProof}}$ is a datatype $\newSig{\Sigma}{\formulaProof}$-structure.

	\paragraph{3. Showing that $\nextM$ is a $\newT{\formulaProof}$-interpretation.}

	Recall that
	\(
	\newT{\formulaProof}
	\;=\;
	\Bigl(\biguplus_{j=1}^{n} \trsig{T_j}\Bigr)
	\;\uplus\;
	T_{n+1}^\formulaProof
	\;\uplus\;
	T_{n+2}\,.
	\)
	For each \(j\in [1,n]\), the reduct \(\nextM^{\trsig{\Sigma_j}}\) interprets
	\(\select{\trarr{\arrType_j}}\) and \(\store{\trarr{\arrType_j}}\) in the standard way for an arrays structure.
	On \(\Sigma_{n+1}^\formulaProof\), \(\nextM\) uses the usual tree‐based interpretation of
	constructors, selectors and testers (see Definition~\ref{datatypes-axioms}). By
	construction, each selector and tester in \(\Sigma_{n+1}^\formulaProof\) obeys the requirements of datatypes structures.
	Finally, on \(\Sigma_{n+2}\) we only have the symbols \(f_{\arrType_j}\) and \(g_{\arrType_j}\),
	and there is no restriction on their interpretation.

	\paragraph{4. Showing that $\nextM$ satisfies every literal in $\formulaProof$.}
	Let $\ell$ be a literal in $\formulaProof$.
	Suppose $\newM$ satisfies $\ell$.
	We prove $\nextM$ satisfies it,
	by considering all of its possible shapes.
	\begin{enumerate}
		\item For a literal of the form $x = y$, there are two cases:
		      \begin{enumerate}
			      \item If the sort of $x$ and $y$ is $\trarr{\arrType_i}$, then:
			            \[
				            x^{\nextM} = \Psi_i(x^{\newM}) = \Psi_i(y^{\newM}) = y^{\nextM}.
			            \]

			      \item If the sort of $x$ and $y$ is not $\trarr{\arrType_i}$, then we have that
			            \[
				            x^{\nextM} = x^{\newM} = y^{\newM} = y^{\nextM}.
			            \]
		      \end{enumerate}
		\item For a literal of the form $x \neq y$, there are two cases:
		      \begin{enumerate}
			      \item If the sort of $x$ and $y$ is $\trarr{\arrType_i}$, the injectivity of $\Psi_i$ ensures that:
			            \[
				            x^{\nextM} = \Psi_i(x^{\newM}) \neq \Psi_i(y^{\newM}) = y^{\nextM}.
			            \]

			      \item If the sort of $x$ and $y$ is not $\trarr{\arrType_i}$, then we have that
			            \[
				            x^{\nextM} = x^{\newM} \neq y^{\newM} = y^{\nextM}.
			            \]
		      \end{enumerate}
		\item For a literal of the form $x = \select{\trarr{\arrType_i}}(a,j)$, we have that
		      \[
			      x^{\nextM}
			      = x^{\newM}
			      = \select{\trarr{\arrType_i}}^{\newM}(a^{\newM},j^{\newM})
			      = a^{\nextM}(j^{\nextM})
			      = \select{\trarr{\arrType_i}}^{\nextM}(a^{\nextM},j^{\nextM}).
		      \]

		      The first equality holds since $x \in \fv{E_i}{\formulaExSat}$ and thus $x^{\nextM} = x^{\newM}$.
		      The second equality holds because $\newM$ satisfies the literal $x = \select{\trarr{\arrType_i}}(a,j)$. The third equality holds by the definition of array variables in $\nextM$.
		      The fourth equality holds by the interpretation of the $\select{\trarr{\arrType_i}}$ function symbol in $\nextM$.

		\item For a literal of the form $b = \store{\trarr{\arrType_i}}(a,i,v)$, we already know that
		      \[
			      b^{\newM} = \store{\trarr{\arrType_i}}^{\newM}(a^{\newM},i^{\newM},v^{\newM}).
		      \]
		      Thus:
		      \begin{enumerate}
			      \item $\select{\trarr{\arrType_i}}^{\newM}(b^{\newM},i^{\newM}) = v^{\newM}$.
			      \item $\forall k \in \trsig{I_i}^{\newM} \setminus \{i^{\newM}\}.
				            \select{\trarr{\arrType_i}}^{\newM}(b^{\newM},k) = \select{\trarr{\arrType_i}}^{\newM}(a^{\newM},k).$
		      \end{enumerate}
		      Let there be $k \in \trsig{I_i}^{\nextM}$.
		      Since, $I_i \in \sorts{\Sigma} \setminus (\structsorts \cup \arrays)$, we have that $\trsig{I_i}^{\newM} \subseteq \trsig{I_i}^{\nextM}$. Hence,
		      $\trsig{I_i}^{\nextM} = \{i^{\nextM}\} \cup (\trsig{I_i}^{\newM} \setminus \{i^{\nextM}\}) \cup (\trsig{I_i}^{\nextM} \setminus \trsig{I_i}^{\newM})$ and since $i^\nextM = i^\newM$ there are three cases:
		      \begin{enumerate}
			      \item If $k = i^{\nextM}$, then
			            \[
				            \store{\trarr{\arrType_i}}^{\nextM}(a^{\nextM},i^{\nextM},v^{\nextM})(k) = v^{\nextM} = v^{\newM} = \select{\trarr{\arrType_i}}^{\newM}(b^{\newM},i^\newM)
				            = b^{\nextM}(k).
			            \]
			            The first equality holds by the definition of $\store{\trarr{\arrType_i}}^{\nextM}$, the second equality holds since the $v^{\nextM} \in \trsig{E_i}^{\nextM}$, and thus $v^{\nextM} = v^{\newM}$, the third equality is the first equality derived from the fact that $b^{\newM} = \store{\trarr{\arrType_i}}^{\newM}(a^{\newM},i^{\newM},v^{\newM})$, and the last equality holds by the definition of $\select{\trarr{\arrType_i}}^{\nextM}$ and by substituting $k = i^{\nextM} = i^{\newM}$.
			      \item If $k \in \trsig{I_i}^{\newM} \setminus \{i^{\nextM}\}$, then
			            \[
				            \store{\trarr{\arrType_i}}^{\nextM}(a^{\nextM},i^{\nextM},v^{\nextM})(k) = a^{\nextM}(k) = \select{\trarr{\arrType_i}}^{\newM}(a^{\newM},k)
				            = \select{\trarr{\arrType_i}}^{\newM}(b^{\newM},k) = b^{\nextM}(k).
			            \]
			      \item If $k \in \trsig{I_i}^{\nextM} \setminus \trsig{I_i}^{\newM}$, then
			            \[
				            \store{\trarr{\arrType_i}}^{\nextM}(a^{\nextM},i^{\nextM},v^{\nextM})(k)
				            =a^{\nextM}(k)
				            = \defelem{\trsig{E_i}} = b^{\nextM}(k).
			            \]
		      \end{enumerate}
		      Thus, because of the extensionality axiom, we have that
		      \[
			      b^{\nextM} = \store{\trarr{\arrType_i}}^{\nextM}(a^{\nextM},i^{\nextM},v^{\nextM})
		      \]
		\item For a literal of the form $y = c(x_1,\ldots,x_n)$, we have that
		      \[
			      y^{\nextM}
			      = y^{\newM}
			      = c^{\newM}(x_1^{\newM},\ldots,x_n^{\newM})
			      = c^{\nextM}(x_1^{\nextM},\ldots,x_n^{\nextM})
		      \]

		\item For a literal of the form $x = \sel{c}{i}(y)$, we have that
		      \[
			      x^{\nextM}
			      = x^{\newM}
			      = \sel{c}{i}^{\newM}(y^{\newM})
			      = \sel{c}{i}^{\nextM}(y^{\nextM})
		      \]
		\item For a literal of the form $\is{c}(y)$ or $\neg\is{c}(y)$, we have that $y^{\nextM}$ is an application of a constructor \(c\) if and only if \(y^{\newM}\) is an application of a constructor \(c\).
		\item For a literal of the form $a = f_{\arrType_i}(x)$, we have that
		      \[
			      a^{\nextM}
			      = \Psi_i(a^{\newM})
			      = \Psi_i(f_{\arrType}^{\newM}(x^{\newM}))
			      = f_{\arrType}^{\nextM}(x^{\nextM}).
		      \]

		\item For a literal of the form $x = g_{\arrType_i}(a)$, we have that
		      \[
			      x^{\nextM}
			      = x^{\newM}
			      = g_{\arrType}^{\newM}(a^{\newM})
			      = g_{\arrType}^{\newM}(\Psi_i^{-1}(\Psi_i(a^{\newM})))
			      = g_{\arrType}^{\nextM}(a^{\nextM}).
		      \]

	\end{enumerate}
\end{proof}

\begin{table}[t]
	\newcolumntype{M}[1]{>{\centering\arraybackslash}m{#1}}
	\centering
	\begin{tabular}{|M{2cm}|M{11cm}|}
		\hline
		\textbf{Domains}    &
		\begin{tabular}{@{}l@{}}
			$\trsig{I}^{\newMC} = \mathbb{N}$                                                                                  \\[3pt]
			${\idx{\arrType}}^{\newMC} = \{\gamma\}\;\cup\;\{ e_k^{\idx{\arrType}} \mid k \in \mathbb{N}\}$                    \\[3pt]
			$\trsig{E}^{\newMC}
				= T_{\trsig{E}}\bigl(\consig{\newSig{\exsigwf}{\formulaExSat}}{\newSig{\exsigwf}{\formulaExSat}},\,B\bigr)$        \\[3pt]
			$\trsig{\arrType}^{\newMC}
				= T_{\trsig{\arrType}}\bigl(\consig{\newSig{\exsigwf}{\formulaExSat}}{\newSig{\exsigwf}{\formulaExSat}},\,B\bigr)$ \\[3pt]
			$\trarr{A}^{\newMC}
				= \trsig{I}^{\newMC}\;\to\;\trsig{E}^{\newMC}$
		\end{tabular}
		\\\hline
		\textbf{Functions}  &
		\begin{tabular}{@{}l@{}}
			\(\select{},\store{}\): standard read/update on \(\trarr{A}^{\newMC}\).        \\[3pt]
			\(\trsig{\container}^{\newMC}(n,y)=\trsig{\container}(n,y)\)                   \\[3pt]
			\(\trsig{\id}^{\newMC}(\trsig{\container}(n,y))=n\),\quad
			\(\trsig{\children}^{\newMC}(\trsig{\container}(m,y))=y\)                      \\[3pt]
			\(\trsig{\id}^{\newMC}(\trsig{\emp})=0\),\quad
			\(\trsig{\children}^{\newMC}(\trsig{\emp})=\constArr{\trsig{\emp}}\)           \\[6pt]
			\(f_{\arrType}^{\newMC}\): constant function returning \(\trarr{a}^{\newMC}\). \\[3pt]
			\(g_{\arrType}^{\newMC}\): constant function returning \(\trsig{a}^{\newMC}\).
		\end{tabular}
		\\\hline
		\textbf{Predicates} &
		\begin{tabular}{@{}l@{}}
			\(\is{\trsig{\container}}^{\newMC}(\trsig{\container}(n,y))=\text{true}\), \\[3pt]
			\(\is{\trsig{\container}}^{\newMC}(\trsig{\emp})=\text{false}\),           \\[3pt]
			\(\is{\trsig{\emp}}^{\newMC}(\trsig{\emp})=\text{true}\),                  \\[3pt]
			\(\is{\trsig{\emp}}^{\newMC}(\trsig{\container}(n,y))=\text{false}\).
		\end{tabular}
		\\\hline
		\textbf{Variables}  &
		\begin{tabular}{@{}l@{}}
			\(\trsig{x}^{\newMC} \;=\; \tempvar{a}{1}^{\newMC}
			= \trsig{\container}\bigl(3,\;\arrCons{1}(\gamma,\;\trsig{\emp})\bigr)\)                   \\[6pt]
			\(\trarr{a}^{\newMC}
			= \constArr{\trsig{x}^{\newMC}}
			= \constArr{\trsig{\container}(3,\arrCons{1}(\gamma,\trsig{\emp}))}\)                      \\[6pt]
			\(\trsig{j}^{\newMC} = \trsig{k}^{\newMC} = 1\)                                            \\[6pt]
			\(\trsig{a}^{\newMC}
			= \arrCons{1}\bigl(\gamma,\;\trsig{\container}(3,\arrCons{1}(\gamma,\trsig{\emp}))\bigr)\) \\[6pt]
			\(\trsig{y}^{\newMC}
			= \trsig{\container}\bigl(1,\;\arrCons{1}(\gamma,\trsig{\container}(3,\arrCons{1}(\gamma,\trsig{\emp})))\bigr)\)
		\end{tabular}
		\\\hline
	\end{tabular}
	\caption{A $\newSig{\exsigwf}{\formulaExSat}$–interpretation $\newMC$.
		Here $B$ is the $\elemsorts_{\newSig{\exsigwf}{\formulaExSat}}$–sorted set with $B_{\rho}=\rho^{\newMC}$ for every $\rho\in\elemsorts_{\newSig{\exsigwf}{\formulaExSat}}$.}
	\label{tab:InterpretationM}
\end{table}

\begin{example}
	\label{ex:infinitely-more-elements}
	Consider $\hat{\newM}$ from~\Cref{tab:exCompleteness}, which satisfies the formula $\trF(\formulaExSat)$ (\Cref{tab:tndtsatex}).
	According to~\Cref{lem:infinitely-more-elements}, we can construct a $\newSig{\exsigwf}{\formulaExSat}$-interpretation $\newMC$ that satisfies the same formula, and the three
	conditions specified in \Cref{lem:infinitely-more-elements}.
	This new interpretation $\newMC$ is defined in~\Cref{tab:InterpretationM}.
\end{example}

\begin{table}[t]
	\newcolumntype{M}[1]{>{\centering\arraybackslash}m{#1}}
	\centering
	\begin{tabular}{|M{4.5cm}|M{3.5cm}|M{6.5cm}|}
		\hline
		\textbf{Sort type}         &
		\textbf{Base case ($i=0$)} &
		\textbf{Inductive step}
		\\\hline
		$\displaystyle
			\sigma \in \sorts{\Sigma}
			\setminus(\arrays_\Sigma \cup \structsorts_{\Sigma_{n+1}})$
		                           &
		$\displaystyle
			\subdomain{\sigma}{0}
			= \trsig{\sigma}^{\newM}$
		                           &
		$\displaystyle
			\subdomain{\sigma}{i+1}
			= \trsig{\sigma}^{\newM}$
		\\[6pt]\hline
		$\displaystyle
			\sigma = \arrType_{j}\;(j\in [1,n])$
		                           &
		$\displaystyle
			\subdomain{\sigma}{0} = \emptyset$
		                           &
		$\displaystyle
			\subdomain{\sigma}{i+1}
			= \almostC{\subdomain{I_j}{i}}{\subdomain{E_j}{i}}$
		\\[6pt]\hline
		$\displaystyle
			\sigma \in \structsorts_{\Sigma_{n+1}}$
		                           &
		$\displaystyle
			\subdomain{\sigma}{0} = \emptyset$
		                           &
		$\displaystyle
			\subdomain{\sigma}{i+1}
			= \subdomain{\sigma}{i} \cup
			T_{\sigma}\!\bigl(\consig{\Sigma},\,B^{i}\bigr),\;
			\text{where }B^{i}_{\tau}=\subdomain{\tau}{i}\;
			(\tau\in\elemsorts)$
		\\[6pt]\hline
	\end{tabular}
	\caption{Iterative construction of the domains
		$\subdomain{\sigma}{i}$ for each sort $\sigma$ in $\model$.
		For each $\sigma$, $\sigma^\model =
			\bigcup_{i\in\mathbb{N}} \subdomain{\sigma}{i}$.}
	\label{tab:comp-domains}
\end{table}

\subsection{Defining the Domains of the $\tndt{\Sigma}$-interpretation}
\label{app:sec:comp-define-domains}
Recall that starting from a $\newT{\formulaProof}$ interpretation
$\newM$, we are constructing a $\tndt{\Sigma}$-interpretation $\model$.
Due to~\Cref{lem:infinitely-more-elements} we can assume that
the conditions it specifies hold.

We define the domains of $\model$ as follows.
For each sort $\sigma\in\sorts{\Sigma}$, we inductively define a sequence
$\{\subdomain{\sigma}{i}\}_{i\in\mathbb{N}}$
in \Cref{tab:comp-domains},
and then set
\( \sigma^\model = \bigcup_{i\in\mathbb{N}} \subdomain{\sigma}{i}
\).

\begin{example}
	Consider again the formula $\formulaExSat$ in~\Cref{tab:tndtsatex}.
	We have seen in \Cref{ex:infinitely-more-elements} that
	$\trF(\formulaExSat)$ is
	satisfied by the interpretation
	$\newMC$
	from \Cref{tab:InterpretationM}.
	The domains of the sorts in the above construction
	are the same as those in~\Cref{tab:exInter}
	(The interpretations of the variables may differ).
\end{example}

Our next goal is to show that the defined domains are indeed non-empty. In order to prove it we first show the following lemma:

\begin{lemma}
	\label{lem:subset-by-index}
	For every sort \(\sigma\in\sorts{\Sigma}\) and every \(i\in\mathbb{N}\),
	$
		\subdomain{\sigma}{i} \;\subseteq\; \subdomain{\sigma}{i+1}
	$.
\end{lemma}
\begin{proof}
	We consider three cases depending on the nature of \(\sigma\).

	\begin{enumerate}
		\item If \(\sigma \notin \structsorts_{\Sigma_{n+1}} \cup \arrays_{\Sigma}\), then by definition
		      $
			      \subdomain{\sigma}{i}
			      = \trsig{\sigma}^{\newM}
			      = \subdomain{\sigma}{i+1}
		      $,
		      so the inclusion is immediate.

		\item If \(\sigma \in \structsorts_{\Sigma_{n+1}}\), then
		      \(
		      \subdomain{\sigma}{i+1}
		      = \subdomain{\sigma}{i}
		      \;\cup\;
		      T_{\sigma}\bigl(\consig{\Sigma}, B^{i}\bigr),
		      \)
		      and hence \(\subdomain{\sigma}{i} \subseteq \subdomain{\sigma}{i+1}\).

		\item For all \(k \in [1,n]\), we use induction over $i \in \mathbb{N}$ to show that $\subdomain{\arrType_k}{i} \subseteq \subdomain{\arrType_k}{i+1}$.

		      \textbf{base:} if \(i = 0\), then
		      \(
		      \subdomain{\arrType_k}{0} = \emptyset,
		      \)

		      \textbf{inductive step:}
		      If \(i \ge 1\), then since
		      $\subdomain{E_k}{i-1} \subseteq \subdomain{E_k}{i}$
		      and
		      $\subdomain{I_k}{i-1} = \subdomain{I_k}{i}$,
		      we have
		      \[
			      \subdomain{\arrType_k}{i}
			      = \almostC{\subdomain{I_k}{i-1}}{\subdomain{E_k}{i-1}}
			      = \almostC{\subdomain{I_k}{i}}{\subdomain{E_k}{i-1}}
			      \;\subseteq\;
			      \almostC{\subdomain{I_k}{i}}{\subdomain{E_k}{i}}
			      = \subdomain{\arrType_k}{i+1}.
		      \]

		      The inclusion holds since enlarging the codomain $Y$ in $\almostC{X}{Y}$ preserves all almost constant functions.

	\end{enumerate}
\end{proof}

Now we can prove that the domains of the sorts in \(\Sigma\) are non-empty.

\begin{lemma}
	\label{lem:non-empty}
	For every sort \(\sigma\in\sorts{\Sigma}\), the domain \(\sigma^{\model}\) is nonempty.
\end{lemma}
\begin{proof}
	Let
	\(
	S \;=\;\sorts{\Sigma}\,\setminus\,\bigl(\arrays_{\Sigma}\cup\structsorts_{\Sigma_{n+1}}\bigr).
	\)
	Since all the sorts in \(\sorts{\Sigma}\) are well-founded with respect to \(S\), for
	each \(\sigma\in\sorts{\Sigma}\) there is some \(k\in\mathbb{N}\)
	with
	\(\sigma\in F_k(\Sigma,S)\)
	(see \Cref{def:theory:well-founded}).
	We prove by induction on \(k\) that
	\(
	\forall\sigma\in F_k(\Sigma,S).\quad \subdomain{\sigma}{k}\neq\emptyset.
	\)
	\begin{description}
		\item[Base case (\(k=0\)).] Then \(\sigma\in F_0(\Sigma,S)=S\), and by definition \(\subdomain{\sigma}{0}=\trsig{\sigma}^\newM\), which is nonempty, as $\newM$ is an interpretation.
		\item[Inductive step.] Assume the claim holds for \(k\). Take any \(\sigma\in F_{k+1}(\Sigma,S)\). By definition
		      \[
			      F_{k+1}
			      =F_k\;\cup\;
			      \bigl\{\sigma\mid\exists\,c:\sigma_1\times\cdots\times\sigma_n\to\sigma,\;\sigma_i\in F_k\bigr\}
			      \;\cup\;
			      \{\arrType_i\mid I_i,E_i\in F_k\}.
		      \]
		      There are three cases:
		      \begin{enumerate}
			      \item If \(\sigma\in F_k\), then \(\subdomain{\sigma}{k}\neq\emptyset\) by IH, so according to~\Cref{lem:subset-by-index}, \(\subdomain{\sigma}{k+1} \supseteq \subdomain{\sigma}{k} \neq\emptyset\).
			      \item If there exists a constructor \(c:\sigma_1\times\cdots\times\sigma_n\to\sigma\) with each \(\sigma_i\in F_k\), then IH gives
			            \(\subdomain{\sigma_i}{k}\neq\emptyset\). Hence
			            \(
			            \bigl\{c(t_1,\dots,t_n)\mid t_i\in\subdomain{\sigma_i}{k}\bigr\}
			            \;\subseteq\;
			            T_{\sigma}(\consig{\Sigma},B^k)
			            \;\subseteq\;
			            \subdomain{\sigma}{k+1},
			            \)
			            so \(\subdomain{\sigma}{k+1}\neq\emptyset\).
			      \item If \(\sigma=\arrType_i\) is an array sort where \(E_i, I_i \in F_k\), then IH gives \(\subdomain{I_i}{k}\neq\emptyset\) and \(\subdomain{E_i}{k}\neq\emptyset\). Therefore:
			            \(
			            \subdomain{\arrType_i}{k+1} = \almostC{\subdomain{I_i}{k}}{\subdomain{E_i}{k}}
			            \neq\emptyset
			            \).
		      \end{enumerate}
	\end{description}
	Hence for each \(\sigma\) there is some \(k\) with \(\subdomain{\sigma}{k}\neq\emptyset\), and so \(\sigma^{\model}=\bigcup_{k \in \mathbb{N}}\subdomain{\sigma}{k}\) is nonempty.
\end{proof}

\subsection{Properties of the Domains of the $\tndt{\Sigma}$-interpretation}
\label{app:sec:comp-properties-domains}

In order to show that \(\model\) is a \(\tndt{\Sigma}\)-interpretation, we need to show that the relation \(\rel{\model}\) is well-founded (see \Cref{sec:theory:nested-relation}).\footnote{Technically, $\model$ has not been fully defined yet, however, for proving well-foundedness, the interpretations of the domains, the constructors and array operators suffice.
}

We first prove the following lemmas, which provide a more intuitive understanding of the domains of the sorts in \(\model\).
\begin{lemma}
	\label{lem:domain-arr-intuition}
	For every \(i \in [1,n]\):
	\[
		\arrType_i^\model \;=\;\almostC{I_i^\model}{E_i^\model}.
	\]
\end{lemma}
\begin{proof}
	\emph{($\supseteq$)} Let there be \(t\in\almostC{I_i^\model}{E_i^\model}\), then its image is finite and lies in
	\(\bigcup_j\subdomain{E_i}{j}=E_i^\model\). \Cref{lem:subset-by-index} suggests that
	\(
	\forall j\in\mathbb{N}.\;\subdomain{E_i}{j}\subseteq \subdomain{E_i}{j+1}
	\),
	so there exists a \(k\) such that the image of $t$ lies in
	\(\subdomain{E_i}{k}\).
	Since \(I_i^\model=\subdomain{I_i}{k}\),
	we get \(
	t\in\almostC{\subdomain{I_i}{k}}{\subdomain{E_i}{k}}\;\subseteq\;\arrType_i^\model.
	\)

	\emph{($\subseteq$)} By definition
	\(
	\arrType_i^\model
	\;=\;\bigcup_{j\in\mathbb N}\almostC{\subdomain{I_i}{j}}{\subdomain{E_i}{j}},
	\)
	and each
	\(\subdomain{E_i}{j}\subseteq E_i^\model\),
	\(\subdomain{I_i}{j} = I_i^\model\),
	so \(\almostC{\subdomain{I_i}{j}}{\subdomain{E_i}{j}}\subseteq\almostC{I_i^\model}{E_i^\model}\).
	Hence \(\arrType_i^\model\subseteq\almostC{I_i^\model}{E_i^\model}\).

\end{proof}

\begin{lemma}
	\label{lem:domain-struct-intuition}
	Let \(B\) be the \(\elemsorts\)-sorted set with \(B_\tau=\tau^\model\) for every \(\tau\in\elemsorts\).
	Then for every \(\sigma\in\structsorts\).
	\(
	\sigma^\model \;=\; T_{\sigma}\bigl(\consig{\Sigma},B\bigr).
	\)
\end{lemma}
\begin{proof}
	\textbf{(\(\subseteq\)).}
	Recall \(\sigma^\model=\bigcup_{i\ge0}\subdomain{\sigma}{i}\) with
	\[
		\subdomain{\sigma}{i+1}=\subdomain{\sigma}{i}\cup
		T_{\sigma}\bigl(\consig{\Sigma},B^{i}\bigr),\qquad
		\forall \tau \in \elemsorts. B^{i}_\tau:=\subdomain{\tau}{i}.
	\]
	For every \(\tau\in\elemsorts\) we have \(B^{i}_\tau = \subdomain{\tau}{i} = \subseteq \bigcup_{j \in \mathbb{N}} \subdomain{\tau}{j} = B_\tau\)

	so according to~\Cref{lem:elemsorts-set-containment}
	\(T_{\sigma}(\consig{\Sigma},B^{i})\subseteq
	T_{\sigma}(\consig{\Sigma},B)\) for all \(k\).
	Taking the union over \(i\) yields
	\(\sigma^\model\subseteq T_{\sigma}(\consig{\Sigma},B)\).

	\textbf{(\(\supseteq\)).}
	For \(\sigma\in\structsorts\) we proceed by induction on \(k\) that \(T_{\sigma,k}(\consig{\Sigma}, B) \subseteq \sigma^\model\).

	\emph{Base \(k=0\).} \(T_{\sigma,0}=\emptyset\subseteq\sigma^\model\).

	\emph{Step.} Assume
	\(T_{\tau,k}(\consig{\Sigma},B)\subseteq\tau^\model\) for all \(\tau\).
	Let
	\(c(t_1,\dots,t_m)\in T_{\sigma,k+1}(\consig{\Sigma},B)\);
	each \(t_j\) lies in \(T_{\tau_j,k}(\consig{\Sigma},B)\)
	and hence in \(\tau_j^\model\) by the IH.
	Hence there exists $l_1, \ldots l_m$ such that
	\(
	\forall j\in\{1,\ldots,m\}.\ t_j \in T_{\tau_j}(\consig{\Sigma},B^{l_j})
	\)
	Denoting \(l=\max\{l_1,\ldots,l_m\}\), and since \(B^{l_j}_\tau\subseteq B^{l}_\tau\) for every \(\tau\in\elemsorts\), we have, according to~\Cref{lem:elemsorts-set-containment}
	\(
	\forall j\in\{1,\ldots,m\}.\ t_j \in T_{\tau_j}(\consig{\Sigma},B^{l})
	\)
	Therefore \(c(t_1,\dots,t_n)\in T_{\tau}(\consig{\Sigma},B^{l})\subseteq\sigma^\model\).
\end{proof}

\subsection{Infiniteness of Each $E_i^\model$}
\label{app:sec:comp-ei_infinite}

In order to prove that each \(E_i^\model\) is infinite, we first
define the directed graph \(\sigG{\Sigma}=(V,E)\) over the sorts of
\(\Sigma\), such that the edges go either from an output sort of a constructor into one of its input sorts, or from an array sort to its corresponding element sort.
Formally:
\begin{definition}
	\label{def:cyclic-sorts}
	Let \(\Sigma\) be an NDT signature.
	We define \(\sigG{\Sigma}=(V,E)\) to be the directed graph with
	\[
		V = \sorts{\Sigma},\quad
		E = \{(\sigma,\sigma_j)\mid \exists c\colon\sigma_1\times\cdots\times\sigma_k\to\sigma \in \constructors_{\Sigma_{n+1}},\; j\in[1,k]\}
		\;\cup\;\{(\arrType_i,E_i)\mid j\in [1,n]\}.
	\]
	A sort $\sigma$ is called \emph{cyclic} if there is a path with at least one edge in
	\(\sigG{\Sigma}\) that starts and ends with $\sigma$.
	A sort is called \emph{eventually cyclic} if there is a nonempty path in \(\sigG{\Sigma}\) that starts with it and ends with a cyclic sort.
\end{definition}
Note, that by definition every cyclic sort is also eventually cyclic.

\begin{example}
	\label{ex:cyclic-sorts}
	Consider $\exsigwf$ from \Cref{tab:sigs}.
	Then the nodes in the graph are \(\{I, E, A\}\) and the edges are:
	\begin{itemize}
		\item \(E\to A\) (from the constructor \(\container\))
		\item \(E\to I\) (from the constructor \(\container\))
		\item \(A\to E\) (from an array sort to its element sort)
	\end{itemize}
	Thus, the cyclic sorts are \(E\) and \(A\) (and \(I\) is not eventually cyclic as it has no outgoing edges).
\end{example}

\begin{lemma}
	\label{lem:cardinality-partially-ordered}
	Let there be $\sigma, \tau \in \sorts{\Sigma}$ such that there is a path in \(\sigG{\Sigma}\) from $\sigma$ to $\tau$. Then $|\sigma^\model| \geq |\tau^\model|$.
\end{lemma}
\begin{proof}
	Write the (finite) path in \(\sigG{\Sigma}\) from \(\sigma\) to \(\tau\) as
	\[
		\sigma=\rho_{0}\;\longrightarrow\;
		\rho_{1}\;\longrightarrow\;\dots\;\longrightarrow\;
		\rho_{l}=\tau
		\qquad(l\ge 0).
	\]

	It is sufficient to show that for every $i \in [0,l-1].|\rho_i^\model| \geq |\rho_{i+1}^\model|$.

	By the definition of the graph, the edge \((\rho_i,\rho_{i+1})\) can be justified due to one of the following two cases:
	\begin{enumerate}[label=(\arabic*), leftmargin=1.5em]
		\item There is a constructor \(c : \tau_1\,\dots\,\tau_k \to \rho_i\) such that \(\rho_{i+1} = \tau_j\) for some \(j\).
		\item \(\rho_i = \arrType_k\) and \(\rho_{i+1} = E_k\) for some \(k \in [1,n]\).
	\end{enumerate}

	In both cases, there exists an injective function
	\(f_{i} : \rho_{i+1}^\model \to \rho_i^\model \).
	\begin{enumerate}
		\item In case (1), we fix for each $ s \in [1,k]\setminus \{j\}$ some $ u_s \in \tau_s^\model$ and define:
		      \[
			      f_{i}(x) = c(u_1,\dots,u_{j-1},x,u_{j+1},\dots,u_k).
		      \]
		\item In case (2), we define \(f_{i}(x)\) as the constant array from \(I_k^\model\) to \(E_k^\model\) whose image is the singleton \(\{x\}\).
	\end{enumerate}
\end{proof}

\begin{definition}
	\label{def:dtIx}
	Let $\dtSig$ be a datatype signature and $\nextM$ be a $\dtSig$-interpretation.
	For every $x \in \sigma^\nextM$ for some $\sigma \in \sorts{\dtSig}$, we
	define the {\em minimum index of $x$}, denoted $\dtIx{x}$, as follows:
	\[
		\dtIx{x} = \min\{i \in \mathbb{N} \mid x \in T_{\sigma, i}(\consig{\dtSig}{\dtSig}, B)\}
	\]
	where $B$ is the $\elemsorts$-sorted set with $B_\tau = \tau^\nextM$ for every $\tau \in \elemsorts$.
\end{definition}

The underlying signature and interpretation in \Cref{def:dtIx} will always be clear from context.

We will make use of the following lemma concerning $\dtIx{\cdot}$ and we include its routine proof for completeness.

\begin{lemma}
	\label{lem:dtIx-depth-children}
	Let $\dtSig$ be a datatype signature and $\nextM$ a $\dtSig$-interpretation.
	If $x = c(x_1,\dots,x_r)\in\sigma^\nextM$ for some constructor
	$c:\sigma_1\times\cdots\times\sigma_r\to\sigma\in\constructors_{\dtSig}$, then
	\[
		\dtIx{x} \;=\; 1 + \max\{\dtIx{x_1},\dots,\dtIx{x_r}\}.
	\]
\end{lemma}
\begin{proof}
	Let's assume towards a contradiction that $\dtIx{x} = 0$. Since $\sigma \in \structsorts_{\dtSig}$, $T_{\sigma, 0}(\consig{\dtSig}{\dtSig}, B) = \emptyset$, in contradiction.

	Hence, $\dtIx{x} \geq 1$. Let's denote $\dtIx{x}$ as $k+1$ for some $k \in \mathbb{N}$ and $\hat{k} = \max\{\dtIx{x_1}, \ldots, \dtIx{x_r}\}$.
	By definition of $\dtIx{\cdot}$,
	$x \in T_{\sigma, k+1}(\consig{\dtSig}{\dtSig}, B) \setminus T_{\sigma, k}(\consig{\dtSig}{\dtSig}, B)$.

	Note that, by definition,
	\begin{align*}
		T_{\sigma, k+1}(\consig{\dtSig}{\dtSig}, B) \;=\; &
		T_{\sigma, k}(\consig{\dtSig}{\dtSig}, B) \;\cup                                            \\[4pt]
		                                                  & \Bigl\{\, c(t_1, \ldots, t_k) \;\Big|\;
		c : \sigma_1 \times \cdots \times \sigma_k \to \sigma
		\in \constructors_{\dtSig},                                                                 \\
		                                                  & \hspace{3.5cm}
		t_i \in T_{\sigma_i, k}(\consig{\dtSig}{\dtSig}, B)
		\Bigr\}.
	\end{align*}

	Since $x \notin T_{\sigma, k}(\consig{\dtSig}{\dtSig}, B)$, $x_i \in T_{\sigma_i, k}(\consig{\dtSig}{\dtSig}, B)$ for every $i \in [1,r]$. Since $\dtIx{x_i}$ is the minimum index $j$ such that $x_i \in T_{\sigma_i, j}(\consig{\dtSig}{\dtSig}, B)$ and $x_i \in T_{\sigma_i, k}(\consig{\dtSig}{\dtSig}, B)$,
	$\dtIx{x_i} \leq k$. Thus, by definition of $\hat{k}$, $\hat{k} \leq k$.

	Since $x_i \in T_{\sigma_i, \dtIx{x_i}}(\consig{\dtSig}{\dtSig}, B)$ for every $i \in [1,r]$,
	we know that $x_i \in T_{\sigma_i, \hat{k}}(\consig{\dtSig}{\dtSig}, B)$ for every $i \in [1,r]$.
	Hence, $x \in T_{\sigma, \hat{k} + 1}(\consig{\dtSig}{\dtSig}, B)$ and $k+1 =\dtIx{x} \leq 1 + \hat{k}$.

	Thus $k = \hat{k}$ and $\dtIx{x} = 1 + \max\{\dtIx{x_1}, \ldots, \dtIx{x_r}\}$.
\end{proof}

\begin{lemma}
	\label{lem:infinite-depth-eventually-cyclic}
	Let $\sigma\in\sorts{\Sigma}$ be \emph{eventually cyclic}.
	Then $|\sigma^{\model}| = \infty$.
\end{lemma}

\begin{proof}

	Because $\sigma$ is eventually cyclic, there is a path
	\[
		\sigma=\sigma_{k+\ell}\;\to\;\sigma_{k+\ell-1}\;\to\;\cdots\;\to\;\sigma_\ell
	\]
	ending at a node $\sigma_\ell$ that belongs to a directed cycle of some length $\ell$:
	\[
		\sigma_0\;\to\;\sigma_{1}\;\to\;\cdots\;\to\;\sigma_{\ell }
		\quad(\sigma_{\ell}=\sigma_0).
	\]

	Due to \Cref{lem:cardinality-partially-ordered}, it is sufficient to show that $|\sigma_\ell^{\model}| = \infty$ in order to prove that $|\sigma^{\model}| = \infty$.

	Consider the sorts \(\sigma_0, \sigma_1, \ldots, \sigma_{\ell}\) that form a directed cycle in \(\sigG{\Sigma}\).

	Let there be $m\in\mathbb{N}$, traversing the cycle $m$ additional times yields a path of length $M:=m\ell$.

	We use induction over $i \in [0,M]$ to define $x_{i} \in \sigma_{i}^{\model}$.

	\textbf{Base Case:} Since all domains are not empty, there exists $x_0 \in \sigma_0^{\model}$.

	\textbf{Inductive Step:} We define $x_{i+1} \in \sigma_{i+1}^{\model}$ depending on the edge $(\sigma_i, \sigma_{i+1})$ in $\sigG{\Sigma}$.

	The edge $(\sigma_{i+1},\sigma_{i})$ occurs in the graph $\sigG{\Sigma}$ because of one of two options (Since $\arrays \cap \structsorts = \emptyset$, only one of the two cases can hold):
	\begin{itemize}
		\item \emph{Constructor edge:}
		      there exists a constructor
		      \(c:\tau_1\times\!\cdots\!\times\tau_k\to\sigma_{i+1}\)
		      with $\tau_j=\sigma_{i}$ for some $1\leq j\leq k$.
		      As $\tau_1^\model, \ldots, \tau_k^\model$ are all non-empty, we can choose arbitrary $u_r \in \tau_r^\model$ for $r \in [1,k] \setminus \{j\}$.
		      We define:
		      \[
			      x_{i+1}:=c(u_1,\dots,u_{j-1},x_{i},u_{j+1},\dots,u_k) \in
			      T_{\sigma_{i+1}}(\consig{\Sigma},B) =
			      \sigma_{i+1}^{\model}
		      \]
		      Where the equality holds due to \Cref{lem:domain-struct-intuition} as $\sigma_{i+1} \in \structsorts$.
		\item \emph{Array edge:}
		      the edge $(\sigma_{i+1},\sigma_{i})$ occurs in the graph because there exists some $1\leq j \leq n$ such that $\sigma_{i+1}=\arrType_j$, $\sigma_{i}=E_j$.
		      We define $x_{i+1}$ as the constant array from $I_j^\model$ to $E_j^\model$ whose image is the singleton $\{x_{i}\}$.
	\end{itemize}

	We now show that $\{x_{j\ell}\mid j\in [0,M]\}$ are distinct.
	Consider two possibilities. If all the edges are constructor edges then we make use of $\dtIx{\cdot}$. Otherwise (i.e. at least one edge is an array edge) we make use of $\minIx{\cdot}$ (see \Cref{def:dtIx,not:height}).
	\begin{itemize}
		\item If all edges $(\sigma_0, \sigma_1), \ldots, (\sigma_{\ell - 1},\sigma_{\ell})$ are constructor edges, then by \Cref{lem:dtIx-depth-children} we have $\dtIx{x_{i+1}}>\dtIx{x_i}$ for every $i \in [0,M-1]$. Consequently, each full traversal of the cycle (of length $\ell$) strictly increases the minimum index, i.e.
		      \(
		      \forall k\in [0,M-1].\quad \dtIx{x_{(k+1)\ell}}>\dtIx{x_{k\ell}}.
		      \)
		      Therefore the elements $\{x_{k\ell}\mid 0\le k<m\}$ are distinct.
		\item Suppose some edge in the cycle is an array edge. Fix any traversal index \(k\in\{0,\dots,m-1\}\) and consider the \((k\ell)\)-th segment. For the specific position \(j\in[0,\ell-1]\) where \((\sigma_{j},\sigma_{j+1})\) is an array edge, \Cref{lem:minIx-select} gives a strict increase
		      \(
		      \minIx{x_{k\ell+j+1}}>\minIx{x_{k\ell+j}}.
		      \)
		      For every other step in the same traversal a constructor edge does not decrease the minimum index by \Cref{lem:minIx-constructor-application},
		      \(
		      \minIx{x_{k\ell+i+1}}\ge\minIx{x_{k\ell+i}},
		      \)
		      and an array edge again yields a strict increase by \Cref{lem:minIx-select}. Hence each full cycle raises the minimum index strictly, i.e.
		      \(
		      \forall k\in[0,M-1].\qquad \minIx{x_{(k+1)\ell}}>\minIx{x_{k\ell}}.
		      \)
		      It follows that
		      \(
		      \minIx{x_{m\ell}}>\minIx{x_{(m-1)\ell}}>\cdots>\minIx{x_0},
		      \)
		      so the elements \(x_{0},x_{\ell},\dots,x_{(m-1)\ell}\) have pairwise distinct \(\minIx{\cdot}\) values and therefore are distinct.
	\end{itemize}

	Since $m$ can be chosen arbitrarily large, we conclude that $|\sigma_\ell^{\model}| = \infty$.
\end{proof}

Now we can prove the intended lemma:

\begin{lemma}
	\label{lem:infinite-Ei}
	For every \(i\in [1,n]\) the domain \(E_i^\model\) is infinite.
\end{lemma}

\begin{proof}
	Assume, for the sake of contradiction, that there exists an index
	\(i_0\in [1,n]\) such that
	\(|E_{i_0}^\model|<\infty\).
	\paragraph{\textbf{Step 0. $\forall \sigma \in \sorts{\Sigma} \setminus (\arrays\cup\structsorts). |\sigma^\model| = \infty$.}} By construction,
	\(
	\sigma^\model
	=\trsig{\sigma}^{\newM},
	\)
	and we assume w.l.o.g. according to~\Cref{lem:infinitely-more-elements} that $\trsig{\sigma}^{\newM}$ is infinite.
	\paragraph{\textbf{Step 1. Single Signature and cyclicity.}}
	If the sort \(E_{i_0}\) does not belong to any other signature other than $\Sigma_{i_0}$, as all the sorts of $\Sigma_{i_0}$ are unique,
	\(E_{i_0}\in\sorts{\Sigma}\setminus(\arrays\cup\structsorts)\) and thus $|\sigma^\model|=\infty$ according to Step 0, contradicting our assumption.
	Hence \(E_{i_0}\) belongs to at least one other signature.

	\medskip

	Consider \(\sigG{\Sigma}=(V,E)\) as defined in
	\Cref{def:cyclic-sorts}.
	If \(E_{i_0}\) were \emph{eventually cyclic} in \(\sigG{\Sigma}\),
	\Cref{lem:infinite-depth-eventually-cyclic} would entail
	\(|E_{i_0}^{\model}|=\infty\).
	Therefore \(E_{i_0}\) is \emph{not} eventually cyclic.

	\paragraph{\textbf{Step 2. The reachable sub‑graph \(\sigG{\Sigma}'\) is a tree.}}
	Let
	\(\sigG{\Sigma}'=(V',E')\)
	be the sub‑graph of \(\sigG{\Sigma}\) that contains every sort reachable from
	\(E_{i_0}\).
	Because \(E_{i_0}\) is not eventually cyclic, \(\sigG{\Sigma}'\) is a DAG.
	Moreover, by \Cref{lem:cardinality-partially-ordered},
	every sort \(\sigma\in V'\) satisfies \(|\sigma^\model|<\infty\).

	\paragraph{\textbf{Step 3. Pick $E_{i_1}$.}}
	Note that there exists some $E_{i_1}$ from which no other sort in $\{E_i\mid i\in [1,n]\}$ is reachable.
	Assume towards a contradiction that there is no such sort, then from $E_{i_0}$ we can reach some $E^2 \in \{E_i\mid i\in [1,n]\}$, and from $E^2$ we can reach some $E^3 \in \{E_i\mid i\in [1,n]\}$ and so on.
	Thus when we look at $E^1 = E_{i_0}, E^2, \ldots, E^{n+1}$ the pigeonhole principle suggests that there must be a sort appearing twice, as there are only $n$ elements in $\{E_i\mid i\in [1,n]\}$, which contradicts the acyclicity of $\sigG{\Sigma}'$.

	We hereby denote as $E_{i_1}$ a sort in $V'$ from which no other sort in $\{E_i\mid i\in [1,n]\}$ is reachable.

	Let
	\(\sigG{\Sigma}''=(V'',E'')\)
	be the sub‑graph of \(\sigG{\Sigma}'\) that contains every sort reachable from
	\(E_{i_1}\).
	As a subgraph of \(\sigG{\Sigma}'\), \(\sigG{\Sigma}''\) is itself a DAG.

	\paragraph{\textbf{Step 4. $V'' \subseteq \structsorts$.}}

	Note that:
	\begin{itemize}
		\item Since $V'' \subseteq V'$ every sort \(\sigma\in V''\) satisfies \(|\sigma^\model|<\infty\).
		\item By definition of $E_{i_1}$, no other sort in $\{E_i \mid i\in [1,n]\}$ except $E_{i_1}$ itself belongs to $V''$.
	\end{itemize}
	Let there be $\sigma \in V''$. Since
	$V'' \subseteq \sorts{\Sigma}$ and $\sorts{\Sigma} = \structsorts_{\Sigma_{n+1}} \uplus \arrays_\Sigma \uplus (\sorts{\Sigma} \setminus (\arrays_\Sigma \cup \structsorts_{\Sigma_{n+1}}))$,
	$\sigma$ belongs to one of the three disjoint sets. ($\arrays_\Sigma$ and $\structsorts_{\Sigma_{n+1}}$ are disjoint by \Cref{sec:theory:ndt-sig}).
	\begin{enumerate}
		\item If \(\sigma \in \sorts{\Sigma}\setminus(\arrays_\Sigma\cup\structsorts_{\Sigma_{n+1}})\),
		      step 0 suggests that every such sort has an infinite domain, contradicting the fact that every sort in \(V''\) has a finite domain.
		\item If $\sigma = \arrType_k$ for some $k \in [1,n]$,
		      then its element sort \(E_k\) would be reachable from \(\arrType_k\) inside \(\sigG{\Sigma}''\),
		      contradicting the definition of $E_{i_1}$.
	\end{enumerate}
	Hence $\sigma \in \structsorts_{\Sigma_{n+1}}$ and thus \(V''\subseteq\structsorts_{\Sigma_{n+1}}\).

	\paragraph{\textbf{Step 5. The sorts in $V''$ must be interpreted as in \(\model\).}}

	For every \(T_{n+1}\)-interpretation \(\nextM\) and every
	\(\sigma\in V''\) we have \(\sigma^{\nextM} = \sigma^\model\).
	\smallskip\noindent

	In order to prove this claim, we first remember that according to \Cref{lem:domain-struct-intuition}, for every \(\sigma\in\structsorts\), $\sigma^\model = T_{\sigma}(\consig{\Sigma},B)$ where \(B\) is the \(\elemsorts\)-sorted set with \(B_\tau=\tau^\model\) for every \(\tau\in\elemsorts\).

	Since $\nextM$ is a \(T_{n+1}\)-interpretation, it follows that for every \(\sigma\in\structsorts\), $\sigma^{\nextM} = T_{\sigma}(\consig{\Sigma},\widehat{B})$ where \(\widehat{B}\) is the \(\elemsorts\)-sorted set with \(\widehat{B}_\tau=\tau^{\nextM}\) for every \(\tau\in\elemsorts\).

	We use induction to show that $\forall k \in \mathbb{N}.\forall \sigma \in V''. T_{\sigma, k}(\consig{\Sigma},B) = T_{\sigma, k}(\consig{\Sigma},\widehat{B})$
	(In other words: the elements of $B$ and $\widehat{B}$ do not occur in any of the terms of $T_{\sigma, k}(\consig{\Sigma},B)$ or $T_{\sigma, k}(\consig{\Sigma},\widehat{B})$ for every $\sigma \in V''$).

	\begin{enumerate}
		\item \textbf{Base case:} \(k=0\). Then
		      \(
		      T_{\sigma,0}(\consig{\Sigma},B) = \emptyset
		      = T_{\sigma,0}(\consig{\Sigma},\widehat{B}),
		      \)
		      since \(\sigma \in V'' \subseteq \structsorts\).

		\item \textbf{Step:} Assume
		      \(
		      T_{\tau, k}(\consig{\Sigma},B)
		      = T_{\tau, k}(\consig{\Sigma},\widehat{B})
		      \)
		      for every \(\tau \in V''\). Then, let \(\sigma \in V''\):
		      \(
		      T_{\sigma, k+1}(\consig{\Sigma},B)
		      = T_{\sigma, k}(\consig{\Sigma},B)
		      \cup
		      \bigl\{\, c(t_1,\ldots,t_m)
		      \;\bigm|\;
		      c:\tau_1 \times \cdots \times \tau_m \to \sigma \in \constructors_{\Sigma_{n+1}},\;
		      t_i \in T_{\tau_i, k}(\consig{\Sigma},B) \,\bigr\}.
		      \)
		      Since \(\tau_1, \ldots, \tau_m \in V''\), by the IH we have \(T_{\tau_i, k}(\consig{\Sigma},B) = T_{\tau_i, k}(\consig{\Sigma},\widehat{B})\) for every \(i \in [1,m]\).
		      Thus,
		      \(
		      T_{\sigma, k+1}(\consig{\Sigma},B)
		      = T_{\sigma, k+1}(\consig{\Sigma},\widehat{B}).
		      \)

		      Thus, $\sigma^\nextM = T_{\sigma}(\consig{\Sigma},\widehat{B}) = \cup_{i \in \mathbb{N}} T_{\sigma, i}(\consig{\Sigma},\widehat{B}) = \cup_{i \in \mathbb{N}} T_{\sigma, i}(\consig{\Sigma},B) = T_{\sigma}(\consig{\Sigma},B) = \sigma^\model$.
	\end{enumerate}

	\paragraph{\textbf{Step 6. Final contradiction.}}
	Since for every $T_{n+1}$-interpretation $\nextM$, $E_{i_1}^\nextM$ = $E_{i_1}^\model$ and $E_{i_1}^\model$ is finite, we get that $T_{n+1}$ is not stably infinite over $E_{i_1}$ even though $E_{i_1} \in \sorts{\Sigma_{n+1}} \cap \sorts{\Sigma_{i_1}} \subseteq \sorts{\Sigma_{n+1}} \cap \big(\sorts{\Sigma_1} \cup \ldots \cup \sorts{\Sigma_n}\big)$ contradicting the assumption in \Cref{thm:preprocessing-correctness-2}.
\end{proof}

\subsection{Defining the Interpretation of Variables}
\label{app:sec:comp-interpretation-vars}
We now define the interpretation of variables in $\model$.
For that, we first introduce several helper definitions.
$\intTerms{}$ (\Cref{def:translation-domain}) collects the interpretations of variables in $\newM$ and closes
it under the sub-term relation.
Then, we define functions $\tail$ (\Cref{def:tail}) that restricts the domains of constant arrays, and use it
to construct, for each $1\leq i\leq n$, a function $\assignTail{i}$ over subsets of $\intTerms{}$ (\Cref{def:comp-assign-tail}).
Finally, these are used to map $\intTerms{}$ to elements in the domains of $\model$ (\Cref{def:comp-translation-vars}), which
allows us to assign a value to each variable (\Cref{def:vars-interpretation}).

We start with the definition of $\intTerms{}$.
\begin{definition}
	\label{def:translation-domain}
	For each $i$, let $\intTerms{i}$ be defined as follows.
	\[
		\begin{aligned}
			\intTerms{0}   & = \bigl\{\,x^\newM \;\bigm|\; x\in\fv{\trsig{\sigma}}{\trF(\formulaProof)},\;\sigma\in\sorts{\Sigma}\bigr\}, \\[0.8em]
			\intTerms{i+1} & = \intTerms{i}
			\;\cup\;
			\Bigl\{\,t_j\;\Bigm|\;
			\begin{aligned}
				c \in\constructors_{\Sigma^\formulaProof_{n+1}},
				c(t_1,\dots,t_k)\in \intTerms{i}, j \in[1,k] \\
			\end{aligned}
			\Bigr\}
		\end{aligned}
	\]
	Then,
	$ \intTerms{} = \bigcup_{i\in\mathbb{N}} \intTerms{i}$.
\end{definition}

\begin{example}
	\label{ex:translation-domain}
	Consider the interpretation $\newMC$ in \Cref{tab:InterpretationM} and set $\formulaProof $ to be $\formulaExSat$ from \Cref{tab:tndtsatex}.
	Then the set \(\intTerms{0}\) consists of the following
	elements:
	$\trsig{\container}(3, {\arrCons{1}}(\gamma, \trsig{\emp}))$,
	${\arrCons{1}}(\gamma, \trsig{\container}(3, {\arrCons{1}}(\gamma, \trsig{\emp})))$, $1$ and
	$\trsig{\container}(1,{\arrCons{1}}(\gamma, \trsig{\container}(3, {\arrCons{1}}(\gamma, \trsig{\emp}))))$.
	Then the set \(\intTerms{1}\) adds
	$3$,
	${\arrCons{1}}(\gamma, \trsig{\emp})$ and
	$\gamma$
	to the elements of $\intTerms{0}$.
	The set \(\intTerms{} = \intTerms{2}\) is obtained
	by adding $\trsig{\emp}$.

\end{example}

\begin{lemma}
	\label{comp:lem:S-is-finite}
	$\intTerms{}$ is finite.
\end{lemma}
\begin{proof}
	First, we use induction to show that $\forall i \in \mathbb{N}. \intTerms{i}$ is finite.
	\begin{itemize}
		\item \textbf{Base case:} For $i=0$, $\intTerms{0}$ is finite by definition, as it contains only interpretations of free variables in $\trF(\formulaProof)$, which are finitely many.
		\item \textbf{Inductive step:} Assume $\intTerms{i}$ is finite. We show that $\intTerms{i+1}$ is finite.
		      By definition, $\intTerms{i+1} = \intTerms{i} \cup \{t_j \mid c : \sigma_1 \times \cdots \times \sigma_k \to \sigma \in \constructors_{\Sigma^\formulaProof_{n+1}}, c(t_1, \ldots, t_k) \in \intTerms{i}, j\in[1,k]\}$.
		      Since $\intTerms{i}$ is finite, the set $\{t_j \mid c : \sigma_1 \times \cdots \times \sigma_k \to \sigma \in \constructors_{\Sigma^\formulaProof_{n+1}}, c(t_1, \ldots, t_k) \in \intTerms{i}, j\in[1,k]\}$ is also finite, as it contains only finitely many elements for each term in $\intTerms{i}$.
	\end{itemize}

	Now, we show that there is some $\hat{l} \in \mathbb{N}$ such that $\intTerms{\hat{l}} = \intTerms{}$ and hence $\intTerms{}$ is finite.

	We define $\hat{l} = \max\{\dtIx{t} \mid t \in \intTerms{0}, t \in \sigma^\newM, \sigma \in \sorts{\Sigma_{n+1}^\formulaProof}\}$.

	We use induction on $i\in \{0, \ldots, \hat{l} - 1\}$ to show that $\forall s \in \intTerms{i+1} \setminus \intTerms{i}. \dtIx{s} \leq\hat{l} - i - 1$.

	\begin{description}
		\item[Base case \(i=0\).] We want to show that $\forall s \in \intTerms{1} \setminus \intTerms{0}. \dtIx{s} \leq\hat{l} - 1$.

		      Since $s \in \intTerms{1} \setminus \intTerms{0}$, there is some $x \in \intTerms{0}$ such that $x = c(x_1, \ldots, x_k)$ for some constructor $c : \sigma_1 \times \cdots \times \sigma_k \to \sigma$ and $s = x_j$ for some $j \in [1,k]$.
		      Thus, according to \Cref{lem:dtIx-depth-children} ,$\hat{l} \geq \dtIx{x} = 1 + \max\{\dtIx{x_1}, \ldots, \dtIx{x_k}\} \geq 1 + \dtIx{s}$, and therefore $\dtIx{s} \leq \hat{l} - 1$.
		\item[Inductive step.] Assume the claim holds for $i$, we now show that it holds for $i+1$.
		      Let $t \in \intTerms{i+2} \setminus \intTerms{i+1}$ Then there is some $x \in \intTerms{i+1}$ such that $x = c(x_1, \ldots, x_k)$ for some constructor $c : \sigma_1 \times \cdots \times \sigma_k \to \sigma$ and $t = x_j$ for some $j \in [1,k]$.
		      Assume towards a contradiction that $x \in \intTerms{i}$, then $t \in \intTerms{i+1}$, in contradiction.

		      Hence, $x \in \intTerms{i+1} \setminus \intTerms{i}$, and by the inductive hypothesis, $\dtIx{x} \leq \hat{l} - i - 1$.
		      Due to \Cref{lem:dtIx-depth-children}, $\dtIx{t} + 1 \leq \dtIx{x}$, thus $\dtIx{t} \leq \hat{l} - (i + 1) - 1$.
	\end{description}

	Let there be some $s \in \intTerms{\hat{l}} \setminus \intTerms{\hat{l} - 1}$.
	By the above, $\dtIx{s} \leq \hat{l} - \hat{l} = 0$.
	By definition of $\dtIx{\cdot}$, $s \in T_{\sigma, 0}(\consig{\Sigma^\formulaProof_{n+1}}{\Sigma^\formulaProof_{n+1}}, B)$ for some $\sigma \in \sorts{\dtSig}$.
	Since $T_{\sigma, 0}(\consig{\Sigma^\formulaProof_{n+1}}{\Sigma^\formulaProof_{n+1}}, B) \neq \emptyset$, we know that $\sigma \in \elemsorts_{\dtSig}$.

	That implies that no new elements can be added in $\intTerms{\hat{l}+1}$, hence $\intTerms{\hat{l}+1} =\intTerms{\hat{l}}$, and therefore $\intTerms{\hat{l}} = \intTerms{}$.
\end{proof}

We partition the elements of $\intTerms{} \cap \bigcup_{i\in [1,n]} \trsig{\arrType_i}^\newM$ into two disjoint sets: in one, we collect those elements that are the interpretations
of array variables that occur in store literals, and in the other,
the remaining elements reside.

\begin{definition}
	\label{def:direct-translation}
	The set $\storeElem{\formulaProof}{\newM}$ consists of the elements $s\in\intTerms{}\cap\bigcup_{i\in [1,n]}\trsig{\arrType_i}^{\newM}$ for which there exist some $a,b,i,x \in\fv{}{\formulaProof}$
	such that either
	$s=\trsig{a}^\newM$ or $s=\trsig{b}^{\newM}$,
	and
	$a=\store{\arrType_j}(b,i,x)$ is a literal
	of $\formulaProof$.

	The set
	$\nonStoreElem{\formulaProof}{\newM}$
	is its complement, that is,
	$\nonStoreElem{\formulaProof}{\newM} =\big( \intTerms{} \cap \bigcup_{i\in [1,n]} \trsig{\arrType_i}^\newM \big)\setminus \storeElem{\formulaProof}{\newM}$.
\end{definition}

Note that by definition, $\storeElem{\formulaProof}{\newM} \cup \nonStoreElem{\formulaProof}{\newM} = \intTerms{} \cap \bigcup_{i\in [1,n]} \trsig{\arrType_i}^\newM$ and $\storeElem{\formulaProof}{\newM} \cap \nonStoreElem{\formulaProof}{\newM} = \emptyset$.

\begin{example}
	\label{ex:store-nonstore}
	Consider the formula $\formulaExSat$ from \Cref{tab:tndtsatex}, and its translation $\trF(\formulaExSat)$, also from \Cref{tab:tndtsatex}.
	Consider as well the interpretation $\newM$ in \Cref{tab:exCompleteness}.
	We have a single array sort, $\arrType$. Thus:
	$\intTerms{} \cap \trsig{\arrType}^\newM = \{ {\arrCons{1}}(\gamma, \trsig{\emp}), {\arrCons{1}}(\gamma, \trsig{\container}(3, {\arrCons{1}}(\gamma, \trsig{\emp}))) \}$.

	Moreover, $\storeElem{\formulaExSat}{\newM} = \emptyset$ and $\nonStoreElem{\formulaExSat}{\newM} = \intTerms{} \cap \trsig{\arrType}^\newM$.
\end{example}

We shall make use of the following function $\tail$.
For each almost constant function $a\in\trarr{\arrType_i}^\newM$,
it removes from its domain the interpretations of variables
of sort $I_i$ that correspond to occurrences under a select or store operation in $\phi$, that is, we remove the interpretations of
$\indexSet{\arrType_i}{\formulaProof}$.
Formally:

\begin{definition}
	\label{def:tail}
	Let $a \in \trarr{\arrType_i}^\newM$ for some $i \in [1,n]$.
	Then we define:
	$\tail(a) = a |_{\trsig{I_i}^\newM \setminus \{\trsig{j}^\newM \mid j \in \indexSet{\arrType_i}{\formulaProof}\} }$.
\end{definition}

We then use $\tail$ in order to prove the existence
of another function $\assignTail{i}$, for each $1\leq i\leq n$.
This function takes an object $s$ from
the datatype
$\trsig{\arrType_i}^{\newM}$, such that
$a$ is also in $S$,
and returns some element in $E_i^\model$.

Recall that $\trsig{\arrType_i}$ is a datatype sort
that is intended to model the original array sort
$\arrType_i$.

The element returned by $\assignTail{i}$ belongs
to the interpretation of the corresponding element sort $E_i$
in $\model$.

We require that $\tail$ assigns two objects
$s_1$ and $s_2$ to the same element, if and only if
the following two conditions are met:
first,
both $s_1$ and $s_2$ are in
$\storeElem{\formulaProof}{\newM}$, that is,
they are interpretations of variables that occur in
a literal with an array store operation;
and second: recall that
$f_{\arrType_i}^\newM:\trtype{\arrType_i}^{\newM}\to\trarr{\arrType_i}^{\newM}$,
and thus it assigns $s_1$ and $s_2$ to two almost constant
functions (see \Cref{app:sec:comp-intuitive-model}).
These functions must agree on the restriction given by $\tail$.

The purpose of $\assignTail{i}$ is to ensure that if an array is obtained from another array
by a single store operation in the formula, then they will have the same $\tail$.

\begin{lemma}
	\label{def:comp-assign-tail}
	For every $i \in [1,n]$, there exists a function $\assignTail{i}: \intTerms{}
		\cap \trsig{\arrType_i}^\newM \to E_i^\model$, such that for every $s_1,
		s_2 \in \intTerms{} \cap \trsig{\arrType_i}^\newM$:
	\[
		\assignTail{i}(s_1) = \assignTail{i}(s_2)
		\Leftrightarrow
		\big( s_1, s_2 \in \storeElem{\formulaProof}{\newM} \wedge \tail({f_{\arrType_i}^\newM}(s_1)) = \tail({f_{\arrType_i}^\newM}(s_2)) \big)
	\]
\end{lemma}
\begin{proof}
	Define a relation $R$ on $\intTerms{} \cap \trsig{\arrType_i}^\newM$ by
	\[
		R = \{(s_1, s_2) \mid s_1, s_2 \in \storeElem{\formulaProof}{\newM},\
		\tail(f_{\arrType_i}^\newM(s_1)) = \tail(f_{\arrType_i}^\newM(s_2))\}
		\cup \{(s,s) \mid s \in \intTerms{} \cap \trsig{\arrType_i}^\newM \}.
	\]
	We check that $R$ is an equivalence relation:
	\begin{enumerate}
		\item \textbf{Reflexivity:} For every $s \in \intTerms{} \cap \trsig{\arrType_i}^\newM$, $(s,s) \in R$ by definition.
		\item \textbf{Symmetry:} Assume $(s_1, s_2) \in R$. If $s_1 = s_2$, then $(s_2, s_1) = (s_1, s_2) \in R$. Otherwise, $s_1, s_2 \in \storeElem{\formulaProof}{\newM}$ and $\tail({f_{\arrType_i}^\newM}(s_1)) = \tail({f_{\arrType_i}^\newM}(s_2))$. By symmetry of equality, $\tail({f_{\arrType_i}^\newM}(s_2)) = \tail({f_{\arrType_i}^\newM}(s_1))$, hence $(s_2, s_1) \in R$.
		\item \textbf{Transitivity:} Assume $(s_1, s_2), (s_2, s_3) \in R$. If $s_1 = s_2$ or $s_2 = s_3$, then by substitution we obtain $(s_1, s_3) \in R$. Otherwise, $s_1, s_2, s_3 \in \storeElem{\formulaProof}{\newM}$ and $\tail({f_{\arrType_i}^\newM}(s_1)) = \tail({f_{\arrType_i}^\newM}(s_2))$ and $\tail({f_{\arrType_i}^\newM}(s_2)) = \tail({f_{\arrType_i}^\newM}(s_3))$. By transitivity of equality, $\tail({f_{\arrType_i}^\newM}(s_1)) = \tail({f_{\arrType_i}^\newM}(s_3))$, hence $(s_1, s_3) \in R$.
	\end{enumerate}

	Thus $R$ partitions $\intTerms{} \cap \trsig{\arrType_i}^\newM$ into equivalence classes.
	Since this set is finite (\Cref{comp:lem:S-is-finite}) and $E_i^\model$ is infinite (\Cref{lem:infinite-Ei}), there exists an injection $f_{\text{class}}$ from the classes of $R$ into $E_i^\model$.

	For $s \in \intTerms{} \cap \trsig{\arrType_i}^\newM$, let $[s]_R$ be its equivalence class.
	Define
	\(
	\assignTail{i}(s) := f_{\text{class}}([s]_R).
	\)
	This satisfies the required property.
\end{proof}

\begin{example}
	\label{ex:comp:assign-tail}
	Consider the formula $\formulaExSat$ from \Cref{tab:tndtsatex}, and its translation $\trF(\formulaExSat)$, also from \Cref{tab:tndtsatex}.
	Consider as well the interpretation $\newMC$ in \Cref{tab:InterpretationM}.
	As mentioned in \Cref{ex:store-nonstore}, $\storeElem{\formulaExSat}{\newM} = \emptyset$ and $\nonStoreElem{\formulaExSat}{\newM} = \intTerms{} \cap \trsig{\arrType}^\newM = \{\arrCons{1}(\gamma, \trsig{\emp}), \arrCons{1}(\gamma, \trsig{\container}(3,\arrCons{1}(\gamma,\trsig{\emp})))\}$.
	Thus, $\assignTail{1}$ can be any injection from $\intTerms{} \cap \trsig{\arrType}^\newM$ to ${E^\model}$.
	For example, we can define it as follows:
	\begin{itemize}
		\item $\assignTail{1}(\arrCons{1}(\gamma, \trsig{\emp})) = \emp$.
		\item $\assignTail{1}(\arrCons{1}(\gamma, \trsig{\container}(3,\arrCons{1}(\gamma,\trsig{\emp})))) = \container(3, \constArr{\emp})$.
	\end{itemize}
\end{example}

\begin{definition}
	\label{def:associated-indices}
	For each $i \in [1,n]$ and $s \in \intTerms{} \cap \trsig{\arrType_i}^\newM$,
	define the set of associated indices as
	\begin{align*}
		\associatedIndices{s}{i}
		 & = \bigl\{\, j \in \indexSet{\arrType_i}{\formulaProof} \;\bigm|\;
		\exists a\in\fv{\arrType_i}{\formulaProof}.\;
		\exists x\in\fv{\trsig{E_i}}{\trF(\formulaProof)}.\;
		\trsig{a}^{\newM}=s \land {}                                         \\
		 & \qquad\qquad\qquad\qquad
		x=\select{\trarr{\arrType_i}}(\trarr{a},\trsig{j})\in\trF(\formulaProof)
		\bigr\}.
	\end{align*}
\end{definition}

In the following lemma we clarify the connection between $\indexSet{\arrType_i}{\formulaProof}$
and $\associatedIndices{s}{i}$.

\begin{lemma}
	\label{rem:associatedIndices-and-indexset}
	For every $i \in [1,n]$ it holds that $\indexSet{\arrType_i}{\formulaProof}=\bigcup_{s \in \intTerms{} \cap \trsig{\arrType_i}^\newM} \associatedIndices{s}{i} $.
\end{lemma}
\begin{proof}
	\textbf{($\subseteq$):} Let $j \in \indexSet{\arrType_i}{\formulaProof}$.
	Then there exists some $a \in \fv{\arrType_i}{\formulaProof}$ such that one of the following holds:
	\begin{itemize}
		\item There exists some $x \in \fv{\trsig{E_i}}{\trF(\formulaProof)}$ such that $x = \select{\arrType_i}(a, j) \in \formulaProof$.
		      In this case, $\select{\trarr{\arrType_i}}(\trarr{a}, \trsig{j}) = \trsig{x} \in \trF(\formulaProof)$.
		\item There exists some $x \in \fv{\trsig{E_i}}{\trF(\formulaProof)}$ and $b \in \fv{\arrType_i}{\formulaProof}$ such that $b = \store{\arrType_i}(a, j, x) \in \formulaProof$.
		      In this case, according to $\Ltwo{\formulaProof}$, $\tempvar{a}{\indOrder{\arrType_j}{\formulaProof}(i)} = \select{\trarr{\arrType_j}}(\trarr{a}, \trsig{i}) \in \trF(\formulaProof)$.
	\end{itemize}

	Hence, in both cases, $j \in \associatedIndices{\trsig{a}^\newM}{i}$, and thus $j \in \bigcup_{s \in \intTerms{} \cap \trsig{\arrType_i}^\newM} \associatedIndices{s}{i}$.

	\textbf{$\supseteq$:} Let there be $j \in \bigcup_{s \in \intTerms{} \cap \trsig{\arrType_i}^\newM} \associatedIndices{s}{i}$.
	Then there exists some $s \in \intTerms{} \cap \trsig{\arrType_i}^\newM$ such that $j \in \associatedIndices{s}{i}$.
	By definition of $\associatedIndices{s}{i}$, there exists some $a \in \fv{\arrType_i}{\formulaProof}$ such that $\trsig{a}^\newM = s$ and there exists some $x \in \fv{\trsig{E_i}}{\trF(\formulaProof)}$ such that $x = \select{\trarr{\arrType_i}}(\trarr{a}, \trsig{j}) \in \trF(\formulaProof)$.

	There are three cases to consider:
	\begin{itemize}
		\item $x = \select{\trarr{\arrType_i}}(\trarr{a}, \trsig{j}) \in \trF(\formulaProof)$ because it is a translation of a literal of the form $x = \select{\arrType_i}(a, j) \in \formulaProof$. Note, that in that case, there exists some $y \in \fv{E_i}{\formulaProof}$ such that $y = \select{\arrType_i}(a, j) \in \formulaProof$. Therefore, $j \in \indexSet{\arrType_i}{\formulaProof}$.
		      $j \in \indexSet{\arrType_i}{\formulaProof}$.
		\item $x = \select{\trarr{\arrType_i}}(\trarr{a}, \trsig{j}) \in \trF(\formulaProof)$ because of $\Lone{\formulaProof}$. Note, that in that case, there exists some $y \in \fv{E_i}{\formulaProof}$ such that $y = \select{\arrType_i}(a, j) \in \formulaProof$. Therefore, $j \in \indexSet{\arrType_i}{\formulaProof}$.
		\item $x = \select{\trarr{\arrType_i}}(\trarr{a}, \trsig{j}) \in \trF(\formulaProof)$ because of $\Ltwo{\formulaProof}$. Note, that in that case, there exists some $b,c \in \fv{\arrType_i}{\formulaProof}$ and some $y \in \fv{E_i}{\formulaProof}$ such that $b = \store{\arrType_i}(c, j, y) \in \formulaProof$. Therefore, $j \in \indexSet{\arrType_i}{\formulaProof}$.
	\end{itemize}

	Hence, in all cases, $j \in \indexSet{\arrType_i}{\formulaProof}$.
\end{proof}

\begin{example}
	\label{ex:associated-indices}
	Consider the formula $\formulaExSat$ from \Cref{tab:tndtsatex}, and its translation $\trF(\formulaExSat)$, also from \Cref{tab:tndtsatex}.
	Consider as well the interpretation $\newMC$ in \Cref{tab:InterpretationM}.
	Thus:
	\begin{itemize}
		\item $\associatedIndices{(\arrCons{1}(\gamma, \trsig{\emp}))}{1} = \emptyset$.
		\item $\associatedIndices{(\arrCons{1}(\gamma, \trsig{\container}(3,\arrCons{1}(\gamma,\trsig{\emp}))))}{1} = \{j\}$.
	\end{itemize}
\end{example}

Finally, we define a function $\trComp$,
whose goal is to map elements of $S$
(and in particular, interpretations of variables
that appear in $\formulaProof$)
to elements in the domains of $\model$.
For technical reasons, we first define it as
a partial function, but later
(\Cref{lem:comp-translation-total}) prove that it is total;
and, we define its range to be a superset of
the domains of $\model$.
Later on we show that for every $x \in \fv{\sigma}{\formulaProof}$, it hods that $\trComp(\trsig{x}^\newM) \in \sigma^\model$ (also in \Cref{lem:comp-translation-total}).
The definition relies on notation
$\indOrder{}{}$ from \Cref{def:indOrder}.

\begin{definition}
	\label{def:comp-translation-vars}
	We define the partial function \(\trComp\colon \intTerms{} \to \cup\{\sigma^\model \mid \sigma \in \sorts{\Sigma}\} \cup \big(\bigcup_{i \in [1,n]} \idx{\arrType_i}^\newM \big)\) as follows:
	\begin{enumerate}
		\item If $x \in \sigma^\newM$ for some $\sigma \notin \structsorts_{\Sigma^\formulaProof_{n+1}}$
		      and $x \in \cup\{\sigma^\model \mid \sigma \in \sorts{\Sigma}\} \cup \big(\bigcup_{i \in [1,n]} \idx{\arrType_i}^\newM \big)$, then $\trComp(x) = x$.

		\item If $x = \trsig{c}(t_1, \ldots, t_k)$ for some $c \colon \sigma_1 \times \cdots \times \sigma_k \to \sigma \in \constructors_{\Sigma_{n+1}}$,
		      $\trComp$ is defined for $t_1, \ldots, t_k$,
		      and $\forall i \in [1,k]. \trComp(t_i) \in \sigma_i^\model$
		      then
		      $\trComp(x) = c(\trComp(t_1), \ldots, \trComp(t_k))$.

		\item If $x = \arrCons{i}(t_1, \ldots, t_k)$ for some $i \in [1,n]$, $\trComp$ is defined over $t_1, \ldots, t_k$ and
		      $\forall j \in {2,\ldots,k}. \trComp(t_j) \in E_i^\model$,
		      then we define $\trComp(x)$ as follows:
		      \[
			      \forall j \in I_i^\model. \trComp(x)(j) =
			      \begin{cases}
				      \trComp(t_{\min\{\indOrder{}{}(k) \mid k \in \associatedIndices{x}{i}. \trsig{k}^\newM = j \}+1}) & \text{if } \exists k \in \associatedIndices{x}{i}. \trsig{k}^\newM = j \\
				      \assignTail{i}(x)                                                                                 & \text{otherwise}
			      \end{cases}
		      \]
	\end{enumerate}
\end{definition}

\begin{example}
	\label{ex:translation-function}
	Consider the formula \( \formulaExSat\) and its translation in \Cref{tab:tndtsatex} and $\newMC$ as defined in \Cref{tab:InterpretationM}.

	The set $\intTerms{}$ is given in \Cref{ex:translation-domain}, and is repeated here:
	\[
		\intTerms{} =
		\{
		\trsig{\container}(3, {\arrCons{1}}(\gamma, \trsig{\emp})),
		{\arrCons{1}}(\gamma, \trsig{\container}(3, {\arrCons{1}}(\gamma, \trsig{\emp}))),
		\trsig{\container}(1,{\arrCons{1}}(\gamma, \trsig{\container}(3, {\arrCons{1}}(\gamma, \trsig{\emp})))),
		3,
		{\arrCons{1}}(\gamma, \trsig{\emp}),
		\gamma,
		1,
		\trsig{\emp}
		\}
	\]
	Since the computation of $\trComp$ is done recursively, we order the elements of $\intTerms{}$ according to $\dtIx{\cdot}$.

	Note that we use the function $\assignTail{1}$ defined in \Cref{ex:comp:assign-tail} to compute $\trComp$ for the elements in $\intTerms{} \cap \trsig{\arrType}^\newM$.
	First, we note that $\idx{\arrType},\trsig{I} \notin \structsorts_{\Sigma^\formulaExSat_{n+1}}$.
	Thus, for every $x \in \intTerms{}$ such that $x \in {\idx{\arrType}}^\newM$ or $x \in I^\newM$, we have $\trComp(x) = x$.
	\begin{itemize}
		\item $\trComp(3) = 3$ (Note that $3 \in I^\model$)
		\item $\trComp(\gamma) = \gamma$
		\item $\trComp(1) = 1$
	\end{itemize}

	Since $\emp \in \constructors_{\Sigma_{2}}$, we have \(\trComp(\trsig{\emp}) = \emp\).

	Note that by \Cref{def:free-vars-array-set}, $\indexSet{\arrType}{\formulaExSat} = \{j\}$, $\trsig{j}^\newM
		= 1$ and $\indOrder{}{}(j) = 1$.
	Thus, $\min\{\indOrder{}{}(k) \mid k \in \indexSet{\arrType}{\formulaExSat}, \trsig{k}^\newM = 1\} = 1$.
	Hence, according to the definition of $\assignTail{1}$ in \Cref{ex:comp:assign-tail} we have:
	\[
		\trComp({\arrCons{1}}(\gamma, \trsig{\emp}))(i) =
		\begin{cases}
			\trComp(\trsig{\emp}) & \text{if } i = 1 \\
			\emp                  & \text{otherwise}
		\end{cases}
	\]
	So $\trComp({\arrCons{1}}(\gamma, \trsig{\emp})) = \constArr{\emp}$.

	Once again, since $\container \in \constructors_{\Sigma_{2}}$, we have:
	\(\trComp(\trsig{\container}(3, {\arrCons{1}}(\gamma, \trsig{\emp})))\):
	\[
		\trComp(\trsig{\container}(3, {\arrCons{1}}(\gamma, \trsig{\emp})))
		= \container(\trComp(3), \trComp({\arrCons{1}}(\gamma, \trsig{\emp})))
		= \container(3, \constArr{\emp})
	\]

	Similarly, we have:
	\[
		\trComp({\arrCons{1}}(\gamma, \trsig{\container}(3, {\arrCons{1}}(\gamma, \trsig{\emp}))))(i) =
		\begin{cases}
			\trComp(\trsig{\container}(3, {\arrCons{1}}(\gamma, \trsig{\emp}))) & \text{if } i = 1 \\
			\container(3, \constArr{\emp})                                      & \text{otherwise}
		\end{cases}
	\]
	So $\trComp({\arrCons{1}}(\gamma, \trsig{\container}(3, {\arrCons{1}}(\gamma, \trsig{\emp})))) = \constArr{\container(3, \constArr{\emp})}$.

	Finally, since $\container \in \constructors_{\Sigma_{2}}$, we have
	\begin{align*}
		 & \trComp(\trsig{\container}(1,{\arrCons{1}}(\gamma, \trsig{\container}(3, {\arrCons{1}}(\gamma, \trsig{\emp})))))     \\
		 & = \container(\trComp(1), \trComp({\arrCons{1}}(\gamma, \trsig{\container}(3, {\arrCons{1}}(\gamma, \trsig{\emp}))))) \\
		 & = \container(1, \constArr{\container(3, \constArr{\emp})})
	\end{align*}
\end{example}

\begin{lemma}
	\label{lem:comp-translation-total}
	The function $\trComp: \intTerms{} \to \cup\{\sigma^\model \mid \sigma \in \sorts{\Sigma}\} \cup \big(\bigcup_{i \in [1,n]} \idx{\arrType_i}^\newM\big)$ is total and $\forall \sigma \in \sorts{\Sigma}. \forall x \in \intTerms{} \cap \trsig{\sigma}^\newM. \trComp(x) \in \sigma^\model$.
\end{lemma}
\begin{proof}
	First note that for every $\sigma \in \sorts{\newSig{\Sigma}{\formulaProof}} \setminus \structsorts_{\Sigma^\formulaProof_{n+1}}$,
	and every $x \in \sigma^\newM$, there are 2 options:
	\begin{enumerate}
		\item $\sigma = \idx{\arrType_i}$ for some $i \in [1,n]$. In that case $\idx{\arrType_i}^\newM \subset \cup\{\sigma^\model \mid \sigma \in \sorts{\Sigma}\} \cup \big(\bigcup_{i \in [1,n]} \idx{\arrType_i}^\newM\big)$, and thus $\trComp(x) = x$.
		\item $\sigma = \trsig{\tau}$ for some $\tau \in \sorts{\Sigma} \setminus (\structsorts_{\Sigma_{n+1}} \cup \arrays_\Sigma)$. In that case, $\sigma^\newM = \tau^\model$, and thus, according to the first case in \Cref{def:comp-translation-vars}, $\trComp(x) = x \in \tau^\model$.
	\end{enumerate}

	We are left to prove the claim for the interpretations of the sorts of $\structsorts_{\Sigma_{n+1}^\formulaProof}$. We prove the claim by induction on $\dtIx{x}$.
	\begin{description}
		\item[Base case \(\dtIx{x} = 0\).] Since $\forall \sigma \in \structsorts_{\Sigma_{n+1}^\formulaProof}. T_{\sigma, 0}(\consig{\Sigma^\formulaProof_{n+1}}{\Sigma^\formulaProof_{n+1}}, B) = \emptyset$, the claim holds vacuously for all $x$ such that $\dtIx{x} = 0$.
		\item[Inductive step.] Assume the claim holds for all $x$ such that $\dtIx{x} \leq k$ (whether its sort is in $\structsorts_{\Sigma_{n+1}^\formulaProof}$ or not). We show that it holds for $\dtIx{x} = k+1$. Let $x \in \intTerms{}$ such that $\dtIx{x} = k+1$.
		      Then there are two cases:
		      \begin{itemize}
			      \item \(x = \trsig{c}(t_1, \ldots, t_m)\) for some $c : \sigma_1 \times \cdots \times \sigma_k \to \sigma \in \constructors_{\Sigma_{n+1}}$.
			            Note that since $\dtIx{x} = k+1$,
			            \Cref{lem:dtIx-depth-children} implies that $\dtIx{t_i} \leq
				            k$ for all $i \in [1,m]$. Thus, $\forall i\in [1,m]. \trComp(t_i)$ is
			            defined by the inductive hypothesis, and also $\trComp(t_i) \in
				            \sigma_i^\model$ for all $i \in [1,m]$.
			            Thus, $\trComp(x) =
				            c(\trComp(t_1), \ldots, \trComp(t_m)) \in \sigma^\model$.
			      \item \(x = \arrCons{i}(t_1, \ldots, t_k)\) for some $i \in [1,n]$.
			            Note that since $\dtIx{x} = k+1$, \Cref{lem:dtIx-depth-children} suggests that $\dtIx{t_j} \leq k$ for all $j \in [0,k]$. Thus $\forall j \in [0,k].\trComp(t_j)$ is defined and $\forall j \in [2,k]. \trComp(t_j) \in E_i^\model$.
			            Hence, according to the third case in the \Cref{def:comp-translation-vars}, $\trComp(x)$ is defined by the inductive hypothesis, and $\trComp(x) \in \almostC{I_i^\model}{E_i^\model} = \sigma^\model$ (the equality holds by \Cref{lem:domain-arr-intuition}).
		      \end{itemize}
	\end{description}
\end{proof}

\begin{lemma}
	\label{lem:comp-translation-injective}
	The function $\trComp$ is injective.
\end{lemma}
\begin{proof}
	Let there be $x,y \in \intTerms{}$ such that $x \neq y$.
	If $x \in \sigma_1^\newM$ and $y \in \sigma_2^\newM$ where $\sigma_1 \neq \sigma_2$, there are two options:
	\begin{itemize}
		\item WLOG, if $\sigma_1 \notin \structsorts_{\Sigma^\formulaProof_{n+1}}$, then $\trComp(x) = x$. Thus the only way to have $\trComp(x) = \trComp(y)$ is if $x = y$, contradicting our assumption.

		\item If $\sigma_1, \sigma_2 \in \structsorts_{\Sigma^\formulaProof_{n+1}}$, then we can denote $\sigma_1 = \trsig{\tau_1}$ and $\sigma_2 = \trsig{\tau_2}$ for some $\tau_1, \tau_2 \in \sorts{\Sigma}$.
		      Since $\trComp(x) \in \tau_1^\model$ and $\trComp(y) \in \tau_2^\model$, we have $\trComp(x) \neq \trComp(y)$.
	\end{itemize}
	Hence, we can assume $x,y \in \sigma^\newM$ for some $\sigma \in \sorts{\newSig{\Sigma}{\formulaProof}}$.

	If $\sigma \notin \structsorts_{\Sigma^\formulaProof_{n+1}}$, then $\trComp(x) = x \neq y = \trComp(y)$.

	We are left to prove the claim where $\sigma \in \structsorts_{\Sigma^\formulaProof_{n+1}}$.

	We prove our claim by induction on the $l = \max\{\dtIx{x}, \dtIx{y}\}$.

	\textbf{Base Case}: If $\dtIx{x} = \dtIx{y} = 0$, since $\sigma^\model = \emptyset$, the claim holds vacuously.

	\textbf{Inductive Case}: Assume the claim holds for all $x,y \in \intTerms{} \cap \tau^\newM$
	such that $\max\{\dtIx{x}, \dtIx{y}\} < k$ for every $\tau \in \structsorts_{\Sigma_{n+1}^\formulaProof}$ . We show that it holds for $\max\{\dtIx{x}, \dtIx{y}\} = k$.
	There are 2 cases to consider:

	\begin{itemize}
		\item If $x, y \in \trsig{\sigma}^\newM$ where $\sigma \in \structsorts_{\Sigma_{n+1}^\formulaProof}$, there are two subcases:
		      \begin{enumerate}
			      \item If $x = \trsig{c}(t_1, \ldots, t_m)$ and $y = \trsig{c}(s_1, \ldots, s_m)$ for some $c : \sigma_1 \times \cdots \times \sigma_k \to \sigma \in \constructors_{\Sigma_{n+1}}$, then by definition there exists some $j \in [1,k]$ such that $t_j \neq s_j$. According to the IH, $\trComp(t_j) \neq \trComp(s_j)$, and thus
			            \(
			            \trComp(x) = c(\trComp(t_1), \ldots, \trComp(t_m)) \neq c(\trComp(s_1), \ldots, \trComp(s_m)) = \trComp(y).
			            \)
			      \item If $x = \trsig{\hat{c}}(t_1, \ldots, t_m)$ and $y = \trsig{\tilde{c}}(s_1, \ldots, s_l)$, then
			            \(
			            \trComp(x) = \hat{c}(\trComp(t_1), \ldots, \trComp(t_m)) \neq \tilde{c}(\trComp(s_1), \ldots, \trComp(s_l)) = \trComp(y),
			            \)
			            since $\hat{c} \neq \tilde{c}$.
		      \end{enumerate}
		\item If $x, y \in \trsig{\arrType_i}^\newM$ for some $i \in [1,n]$, then we can denote $x = \arrCons{i}(t_1, \ldots, t_k)$ and $y = \arrCons{i}(s_1, \ldots, s_l)$.
		      There are 2 subcases:
		      \begin{enumerate}
			      \item If
			            \(
			            \neg \big( x, y \in \storeElem{\formulaProof}{\newM} \wedge \tail({f_{\arrType_i}^\newM}(x)) = \tail({f_{\arrType_i}^\newM}(y)) \big)
			            \)
			            according to \Cref{def:comp-assign-tail}, $\assignTail{i}(x) \neq \assignTail{i}(y)$. Hence, for every $k \in \trsig{I_i} \setminus \{\trsig{l}^\newM \mid l \in \indexSet{\arrType_i}{\formulaProof}\}$ we have that $\trComp(x)(k) = \assignTail{i}(x) \neq \assignTail{i}(y) = \trComp(y)$.

			      \item Otherwise, we have that $ \big( x, y \in \storeElem{\formulaProof}{\newM} \wedge \tail({f_{\arrType_i}^\newM}(x)) = \tail({f_{\arrType_i}^\newM}(y)) \big)$.

			            Since $x,y\in\storeElem{\formulaProof}{\newM}$ there are $a,b\in\fv{\arrType_i}{\formulaProof}$ with $\trsig{a}^\newM=x$ and $\trsig{b}^\newM=y$.

			            $\newM$ satisfies $\Lthree{\formulaProof}$:
			            \begin{itemize}
				            \item $\trsig{a} = g_{\arrType_i}(\trarr{a})$,
				            \item $\trsig{b} = g_{\arrType_i}(\trarr{b})$,
				            \item $\trarr{a} = f_{\arrType_i}(\trsig{a})$,
				            \item $\trarr{b} = f_{\arrType_i}(\trsig{b})$.
			            \end{itemize}
			            Thus, we have $\trarr{a}^\newM = f_{\arrType_i}^\newM(x)$ and $\trarr{b}^\newM = f_{\arrType_i}^\newM(y)$.

			            Suppose towards contradiction that $\trarr{a}^\newM = \trarr{b}^\newM$.

			            Using the first 2 equalities, we have:
			            \(
			            x = \trsig{a}^\newM = g_{\arrType_i}^\newM(\trarr{a}^\newM) = g_{\arrType_i}^\newM(\trarr{b}^\newM) = \trsig{b}^\newM = y,
			            \)
			            contradicting $x \neq y$.

			            Thus, $\trarr{a}^\newM \neq \trarr{b}^\newM$, meaning there exists some $j \in \trsig{I_i}^\newM$ such that
			            \(\trarr{a}^\newM(j) \neq \trarr{b}^\newM(j)\). Since $\tail({f_{\arrType_i}^\newM}(x)) = \tail({f_{\arrType_i}^\newM}(y))$, there exists some $k \in \indexSet{\arrType_i}{\formulaProof}$ such that $\trsig{k}^\newM = j$.

			            WLOG, for every $l \in \indexSet{\arrType_i}{\formulaProof}$ with $\trsig{l}^\newM = j$ we have $\indOrder{\arrType_i}{\formulaProof}(k)\le\indOrder{\arrType_i}{\formulaProof}(l)$.

			            Note, that since $x,y \in \storeElem{\formulaProof}{\newM}$ and according to the definition of $\Ltwo{\formulaProof}$ we know that $k \in \associatedIndices{x}{i}$ and $k \in \associatedIndices{y}{i}$.

			            Since $x, y \in \trsig{\arrType_i}^\newM$, since there is a single constructor whose target sort is $\trsig{\arrType_i}$, we can denote $x = \arrCons{i}(x_1, \ldots, x_m)$ and $y = \arrCons{i}(y_1, \ldots, y_m)$.

			            \begin{alignat*}{2}
				            \trComp(x)(j)
				             & \stackrel{(1)}{=}    & \; \trComp\bigl(x_{\indOrder{}{}(k)+1}\bigr)                                             \\
				             & \stackrel{(2)}{=}    & \; \trComp\bigl(\sel{\arrCons{i}}{\indOrder{}{}(k)+1}(x)\bigr)                           \\
				             & \stackrel{(3)}{=}    & \; \trComp\bigl(\sel{\arrCons{i}}{\indOrder{}{}(k)+1}(\trsig{a}^\newM)\bigr)             \\
				             & \stackrel{(4)}{=}    & \; \trComp\bigl(\select{\trarr{\arrType_i}}^\newM(\trarr{a}^\newM,\trsig{k}^\newM)\bigr) \\
				             & \stackrel{(5)}{=}    & \; \trComp\bigl(\select{\trarr{\arrType_i}}^\newM(\trarr{a}^\newM,j)\bigr)               \\
				             & \stackrel{(6)}{=}    & \; \trComp\bigl(\trarr{a}^\newM(j)\bigr)                                                 \\
				             & \stackrel{(7)}{\neq} & \; \trComp\bigl(\trarr{b}^\newM(j)\bigr)                                                 \\
				             & \stackrel{(8)}{=}    & \; \trComp\bigl(\select{\trarr{\arrType_i}}^\newM(\trarr{b}^\newM,j)\bigr)               \\
				             & \stackrel{(9)}{=}    & \; \trComp\bigl(\select{\trarr{\arrType_i}}^\newM(\trarr{b}^\newM,\trsig{k}^\newM)\bigr) \\
				             & \stackrel{(10)}{=}   & \; \trComp\bigl(\sel{\arrCons{i}}{\indOrder{}{}(k)+1}(\trsig{b}^\newM)\bigr)             \\
				             & \stackrel{(11)}{=}   & \; \trComp\bigl(\sel{\arrCons{i}}{\indOrder{}{}(k)+1}(y)\bigr)                           \\
				             & \stackrel{(12)}{=}   & \; \trComp(y_{\indOrder{}{}(k)+1})                                                       \\
				             & \stackrel{(13)}{=}   & \; \trComp(y)(j)
			            \end{alignat*}

			            and $\trComp(x) \neq \trComp(y)$.
		      \end{enumerate}

		      (1) According to the first subcase in the third case in the definition of $\trComp$ (See \Cref{def:comp-translation-vars})

		      (2) According to the interpretation of selectors in $\newM$ and since $x = \arrCons{i}(x_1, \ldots, x_m)$;

		      (3) Substitute \(x=\trsig{a}^\newM\);

		      (4) Since $k \in \indexSet{\arrType_i}{\formulaProof}$ and $\trsig{a}^\newM = x \in \storeElem{\formulaProof}{\newM}$, since $\newM$ satisfies $\Ltwo{\formulaProof}$, we have that $\sel{\arrCons{i}}{\indOrder{}{}(k)+1}(\trsig{a}^\newM) = \select{\trarr{\arrType_i}}^\newM(\trarr{a}^\newM,\trsig{k}^\newM)$;

		      (5) Substitute \(\trsig{k}^\newM = j\);

		      (6) According to the interpretation of $\select{\trarr{\arrType_i}}$ in $\newM$;

		      (7) Since \(\trarr{a}^\newM(j) \neq \trarr{b}^\newM(j)\), and due to the IH regarding injectivity of \(\trComp\) (Note that \(\dtIx{x_{\indOrder{}{}(k)+1}} < \dtIx{x}\) and \(\dtIx{y_{\indOrder{}{}(k)+1}} < \dtIx{y}\) by \Cref{lem:dtIx-depth-children});

		      (8) - (13) Reverse of (6) - (1).
	\end{itemize}
\end{proof}

\begin{definition}
	\label{def:vars-interpretation}
	For every $\sigma \in \sorts{\Sigma}$, and every \(x \in \fv{\sigma}{\formulaProof}\) its interpretation in $\model$ is defined as follows:
	\[
		x^\model = \trComp(\trsig{x}^\newM).
	\]
\end{definition}

Note that for every $\sigma \in \sorts{\Sigma}$, and every \(x \in \fv{\sigma}{\formulaProof}\), we have that $\trComp(\trsig{x}^\newM)$ is well defined and $\trComp(\trsig{x}^\newM) \in \sigma^\model$
by \Cref{lem:comp-translation-total}.

\subsection{Interpretation of Functions and Predicates}

\begin{table}[t]
	\centering
	\begin{tabular}{|l|l|}
		\hline
		Array operators & Standard read/update operations \\\hline
		Constructors    & Fixed by the theory             \\\hline
		Testers         & Fixed by the theory             \\\hline
		Selectors       &
		\shortstack{
			$
				\sel{c}{i}^{\model}(t)=
				\begin{cases}
					t_i                                                & \text{if }
					t=c(t_1,\dots,t_m)\text{ for some }t_1,\ldots,t_n                            \\[4pt]
					\trComp\!\bigl({\sel{\trsig{c}}{i}^\newM}(x)\bigr) &
					\text{otherwise, if there is } x \in \intTerms{}\text{ with } \trComp(x) = t \\[4pt]
					\defelem{\sigma}                                   & \text{otherwise}
				\end{cases}
			$}                                                \\\hline
	\end{tabular}
	\caption{Interpretation of function and predicate symbols,
		where \(\defelem{\sigma}\) is an arbitrary (but fixed) element of
		\(\sigma^{\model}\).
		Note that the interpretation of selectors is well defined as $\trComp$ is injective
		(\Cref{lem:comp-translation-injective}).
	}
	\label{tab:funpredComp}
\end{table}

\label{app:sec:comp-interpretation-funcs}
The interpretation of function and predicate symbols
in $\newM$ is defined in \Cref{tab:funpredComp} and
described as follows.

\paragraph{Array operations.}
The symbols \(\select{\arrType_{j}}\) and
\(\store{\arrType_{j}}\) are interpreted in the standard way,
as read and update operations over functions.

\paragraph{Datatype symbols.}
The interpretation of
constructors and testers is fixed, as we require
$\model^{\Sigma_{n+1}}$ to be a datatypes structure.
As for selectors, their interpretation is fixed when applied
to terms that are constructed by their corresponding constructors.
When applied to different constructors, we maintain the connection
to $\newM$ as induced by $\trComp$, or set them
arbitrarily if no such connection is induced.

\subsection{$\model$ is a $\tndt{\Sigma}$-interpretation}
\label{app:sec:comp-model-interpretation}

In order to show that $\model$ is a $\tndt{\Sigma}$-interpretation, we need to show that
the relation \(\rel{\model}\) is well-founded.

Recall that for every sort $\sigma$,
$\sigma^\model$ is defined to be
\( \bigcup_{i\in\mathbb{N}} \subdomain{\sigma}{i}
\).
Thus, each element of $\sigma^\model$ belongs
to $\subdomain{\sigma}{i}$ for some minimal $i$.
For the next lemma, as well as for several
subsequent lemmas, it will be useful to
have a special notation for that minimal $i$.

\begin{notation}
	\label{comp:model-height}
	\label{not:height}
	For every $\sigma \in \sorts{\Sigma}$ and every $x \in \sigma^\model$, we define $\minIx{x} = \min\{l \in \mathbb{N} \mid x \in \subdomain{\sigma}{l}\}$.
\end{notation}

In order to show that \(\rel{\model}\) is well-founded, we prove the following lemmas about the function \(\minIx{\cdot}\).

\begin{lemma}
	\label{lem:minIx-constructor-application}
	Let \(x=c(t_1,\dots,t_m)\) for some constructor \(c:\sigma_1\times\cdots\times\sigma_m\to\sigma\in\constructors_{\Sigma_{n+1}}\) and \(t_1\in\sigma_1^\model,\ldots,t_m\in\sigma_m^\model\). Then,
	\(
	\minIx{x} \geq \max\{\minIx{t_1},\ldots,\minIx{t_m}\}
	\).
\end{lemma}
\begin{proof}
	We can denote as \(k = \minIx{x}\) and \(k_j = \minIx{t_j}\) for every \(j \in [1,m]\).

	Since $\sigma \in \structsorts_{\Sigma_{n+1}}$, we know that $\subdomain{\sigma}{0}= \emptyset$ and hence \(k \geq 1\).
	By minimality of \(k\), we know that \(x \in \subdomain{\sigma}{k}\setminus \subdomain{\sigma}{k - 1}\).
	Since \(\subdomain{\sigma}{k}=\subdomain{\sigma}{k-1}\cup T_{\sigma}(\consig{\Sigma},B^{k-1})\), we have that \(x \in T_{\sigma}(\consig{\Sigma},B^{k-1})\). Hence, by the definition of \(T_{\sigma}(\consig{\Sigma},B^{k-1})\), we have that for every \(j \in [1,m]\), \(t_j \in \subdomain{\sigma_j}{k-1}\). Therefore, by the definition of \(\minIx{t_j}\), we have that \(k_j \leq k - 1 < k\) for every \(j \in [1,m]\), and hence \(
	\max\{\minIx{t_1},\ldots,\minIx{t_m}\} = \max\{k_1,\ldots,k_m\} \leq k
	\).
\end{proof}

\begin{lemma}
	\label{lem:minIx-select}
	Let \(x \in \arrType_i^\model\) for some \(i \in [1,n]\) and \(j \in I_i^\model\). Then,
	\(
	\minIx{\select{\arrType_i}^\model(x,j)} < \minIx{x}
	\).
\end{lemma}
\begin{proof}
	We denote \(k = \minIx{x}\).
	Since $\subdomain{\arrType_i}{0} = \emptyset$, we have that \(k \geq 1\).
	By the definition of the subdomains, we have $\subdomain{\arrType_i}{k}=\almostC{\subdomain{I_i}{k-1}}{\subdomain{E_i}{k-1}}$. Thus \(\select{\arrType_i}^\model(x,j) \in \subdomain{E_i}{k-1}\), and hence by the definition of \(\minIx{\cdot}\), we have that
	\(\minIx{\select{\arrType_i}^\model(x,j)} \leq k - 1 < k\).
\end{proof}

\begin{lemma}
	\label{lem:well-founded-rel}
	The relation \(\rel{\model}\) is well-founded.
\end{lemma}
\begin{proof}
	Suppose, toward contradiction, there is an infinite descending sequence
	\(
	x_0, x_1, x_2, \ldots
	\)
	such that $\forall i \in \mathbb{N}. x_{i+1} \rel{\model} x_i$.
	Let $\sigma_1,\sigma_2,\ldots$ be a sequence of sorts such that for each $i$, we have $x_i \in \sigma_i^\model$. We denote by $k_i = \minIx{x_i}$.
	According to \Cref{sec:theory:nested-relation}, there are two cases for each $i$:
	\begin{itemize}
		\item \(x_i=c(t_1,\dots,t_m)\) for some
		      \(t_1\in\subdomain{\sigma_1}{k_i-1},\ldots,
		      t_n\in\subdomain{\sigma_m}{k_i-1}\) and $x_{i+1}=t_j$ for some $j\in[1,m]$.
		\item \(x_{i+1}=\select{\sigma_i}^\model(x_i,j)\) for some \(j\in I_i^\model\).
	\end{itemize}

	\textbf{1. If $x_{i+1} \rel{\model} x_i$ by the first case then $k_{i+1} \leq k_i$:}
	If $x_{i+1} \rel{\model} x_i$ by the first case, then by \Cref{lem:minIx-constructor-application}, we have that \(k_{i} \geq \max\{\minIx{t_1},\ldots,\minIx{t_m}\}\). Since \(x_{i+1} = t_j\) for some \(j \in [1,m]\), we have that \(k_{i+1} = \minIx{t_j} \leq k_i\).

	\textbf{2. If $x_{i+1} \rel{\model} x_i$ by the second case then $k_{i+1} < k_i$:}
	We can assume that \(\sigma_i=\arrType_m\) for some \(m \in [1,n]\).
	By \Cref{lem:minIx-select}, we have that
	\(k_{i+1} = \minIx{\select{\arrType_m}^\model(x_i,j)} < \minIx{x_i} = k_i\).

	\textbf{3. There are infinitely many indices \(i\) for which the second case holds:}
	Suppose there are finitely many indices \(i\) for which the second case holds, then there is some index \(i_0\) such that for every \(i \ge i_0\), the first case holds.
	WLOG, we can assume that for every \(i \ge 0\), $x_{i+1} \rel{\model} x_i$ by the first case (by renaming the indices).
	Thus, for every \(i > 0\), the first case holds and $\sigma_i \in \structsorts_{\Sigma_{n+1}}$. According to \Cref{lem:domain-struct-intuition},
	\(\sigma_{0}^\model = T_{\sigma_{0}}(\consig{\Sigma},B) = \bigcup_{j \in \mathbb{N}} T_{\sigma_{0}, j}(\consig{\Sigma},B)\).

	Therefore, there is some \(j_0\) such that \(x_{0} \in T_{\sigma_{0}, j_0}(\consig{\Sigma},B)\).
	We prove by induction on \(i\) that for every \(i \ge 0\), \(x_{i} \in T_{\sigma_{i}, j_0 - i}(\consig{\Sigma},B)\).

	\emph{Base case \(i = 0\).} By the choice of \(j_0\), \(x_{0} \in T_{\sigma_{0}, j_0}(\consig{\Sigma},B)\).

	\emph{Inductive step.}
	Assume \(x_{i} \in T_{\sigma_{i_0+i}, j_0 - i}(\consig{\Sigma},B)\).
	We denote by $l = \min\{k \in \mathbb{N} \mid x_{i_0 + i} \in T_{\sigma_{i_0 + i}, k}(\consig{\Sigma},B)\}$. It is clear that \(l \leq j_0 - i\).

	Since $x_{i_0 + i} \in T_{\sigma_{i_0 + i}, l}(\consig{\Sigma},B) \setminus T_{\sigma_{i_0 + i}, l-1}(\consig{\Sigma},B)$, we get that $x_{i_0 + i} = c(t_1,\dots,t_m)$ for some constructor \(c:\sigma_1 \times \ldots \times \sigma_m \to \sigma_{i_0 + i}\) and \(t_1 \in T_{\sigma_1, l-1}(\consig{\Sigma},B),\ldots,t_m \in T_{\sigma_m, l-1}(\consig{\Sigma},B)\).
	There is some \(j \in [1,m]\) such that \(x_{i_0 + i + 1} = t_j\), and hence \(x_{i_0 + i + 1} \in T_{\sigma_j, l-1}(\consig{\Sigma},B)\).

	Note that $l-1 < l \leq j_0 - i$. Hence, \(x_{i_0 + i + 1} \in T_{\sigma_j, j_0 - (i+1)}(\consig{\Sigma},B)\).

	Thus, for $i = j_0$, since $\sigma_{j_0} \in \structsorts$ we have \(x_{j_0 } \in T_{\sigma_{i}, 0}(\consig{\Sigma},B) = \emptyset\), which is a contradiction.

	\textbf{4. Contradiction:}
	Therefore, there are infinitely many indices \(i\) for which the second case holds, which contradicts the well foundedness of the natural numbers with the standard less-than relation.
\end{proof}

We can now prove that $\model$ is indeed a $\tndt{\Sigma}$-interpretation.
\begin{lemma}
	\label{lem:model-Tndt-interpretation}
	$\model$ is a $\tndt{\Sigma}$-interpretation.
\end{lemma}
\begin{proof}
	We have already seen that $\rel{\model}$ is well founded in~\Cref{lem:well-founded-rel}.

	Hence, we only need to show that
	\begin{itemize}
		\item $\model^{\Sigma_j}$ is an arrays structure for $j \in [1,n]$.
		\item $\model^{\Sigma_{n+1}}$ is a datatype structure.
	\end{itemize}
	For each \(j\in [1,n]\), \Cref{lem:domain-arr-intuition} shows that
	$\arrType_j^\model$ is the standard functional interpretation of arrays, and we
	interpret
	\(\select{\arrType_j}\) and \(\store{\arrType_j}\) in the standard way for an arrays structure.
	As for \(\Sigma_{n+1}\), \Cref{lem:domain-struct-intuition} shows
	that $\forall \sigma \in \structsorts$, $\sigma^\model$ is
	interpreted in the standard way. Likewise,
	\(\model\) uses the usual tree‐based interpretation of
	constructors, selectors and testers (see Definition~\ref{datatypes-axioms}). By
	construction, each selector and tester in \(\Sigma_{n+1}^\formulaProof\) obeys the requirements of datatypes structures.
\end{proof}

\subsection{Satisfaction of the original formula}
\label{app:sec:comp-satisfaction-literal}
\begin{lemma}
	\label{lem:satisfactionbyA}
	$\model\models\formulaProof$.
\end{lemma}
\begin{proof}
	Recall that we are under the assumption that
	$\formulaProof$ is a flat $\Sigma$-cube.
	Let $\ell$ be a literal of $\formulaProof$.
	We prove that $\model\models\ell$ by considering
	every possible shape $\ell$ may take, and its corresponding translation
	(given in \Cref{tab:literal-transformations}).
	Note that since $\ell$ is a literal of $\formulaProof$,
	$\tra(\ell)$ is a literal of $\tra(\formulaProof)$ and is therefore
	satisfied by $\newM$.
	Further, the lemmas in
	$\Lone{\formulaProof}$, $\Ltwo{\formulaProof}$ and
	$\Lthree{\formulaProof}$ are also satisfied by
	$\newM$
	(see \Cref{fig:notation-summary}).
	\begin{enumerate}
		\item For $\ell$ of the form
		      \(
		      x = y
		      \),
		      we know that $\newM \models \trsig{x} = \trsig{y}$. Hence,
		      \(
		      x^\model = \trComp(\trsig{x}^\newM) = \trComp(\trsig{y}^\newM) = y^\model
		      \), and so $\model\models\ell$.
		\item For $\ell$ of the form
		      \(
		      x \neq y,
		      \)
		      we know that $\newM \models \trsig{x} \neq \trsig{y}$.
		      Hence, by the injectivity of $\trComp$ (\Cref{lem:comp-translation-injective}), we have:
		      \(
		      x^\model = \trComp(\trsig{x}^\newM) \neq \trComp(\trsig{y}^\newM) = y^\model
		      \), and so $\model\models\ell$.

		\item For $\ell$ of the form
		      \(
		      x = \select{\arrType_i}(a,j)
		      \),
		      we know that $\newM \models \trsig{x} =
			      \select{\trarr{\arrType_i}}(\trarr{a},\trsig{j})$.
		      According to \Cref{def:associated-indices}, $j\in\associatedIndices{\trsig{a}^\newM}{i}$.
		      Pick some
		      $k\in\associatedIndices{\trsig{a}^\newM}{i}$ with
		      $\trsig{k}^\newM=\trsig{j}^\newM$
		      which is minimal w.r.t.
		      $\indOrder{\arrType_i}{\formulaProof}$, that is,
		      for
		      every $l\in\associatedIndices{\trsig{a}^\newM}{i}$ with
		      $\trsig{l}^\newM=\trsig{j}^\newM$ we have
		      $\indOrder{\arrType_i}{\formulaProof}(k)\le\indOrder{\arrType_i}{\formulaProof}(l)$.

		      Notice that we must have $\trsig{a}^{\newM}=c_i(t_1,\ldots,t_m)$ for some $t_1,\ldots,t_m$,
		      as $c_i$ is the only constructor of output sort $\trsig{\arrType_i}$.
		      Now:
		      \begin{alignat*}{2}
			      \select{\arrType_i}^\model(a^\model,j^\model)                                                         \\
			       & \stackrel{(1)}{=} & \; a^\model(j^\model)                                                          \\
			       & \stackrel{(2)}{=} & \; \trComp(\trsig{a}^\newM)(\trComp(\trsig{j}^\newM))                          \\
			       & \stackrel{(3)}{=} & \; \trComp(\trsig{a}^\newM)(\trsig{j}^\newM)                                   \\
			       & \stackrel{(4)}{=} & \; \trComp(\trsig{t_{\indOrder{}{}(k)+1}}^{\newM})                             \\
			       & \stackrel{(5)}{=} & \; \trComp(\sel{\arrCons{i}}{\indOrder{}{}(k)+1}^\newM(\trsig{a}^\newM))       \\
			       & \stackrel{(6)}{=} & \; \trComp(\select{\trarr{\arrType_i}}^\newM(\trarr{a}^\newM,\trsig{k}^\newM)) \\
			       & \stackrel{(7)}{=} & \; \trComp(\select{\trarr{\arrType_i}}^\newM(\trarr{a}^\newM,\trsig{j}^\newM)) \\
			       & \stackrel{(8)}{=} & \; \trComp(\trsig{x}^\newM)                                                    \\
			       & \stackrel{(9)}{=} & \; x^\model
		      \end{alignat*}

		      (1) According to the interpretation of the select operation in $\model$ (see \Cref{tab:funpredComp}).

		      (2) According to the interpretation of $a$ and $j$ in $\model$ (\Cref{def:vars-interpretation}).

		      (3) By the first case in the definition of $\trComp$ (see
		      \Cref{def:comp-translation-vars}), as for every $i \in [1,n]$, according to \Cref{sec:theory:ndt-sig} we have that
		      $I_i \notin \arrays_\Sigma \cup \structsorts_{\Sigma_{n+1}}$
		      and therefore
		      $\trsig{I_i} \notin \structsorts_{\Sigma_{n+1}^\formulaProof}$ (see~\Cref{fig:dt-sig-trans}).

		      (4) Due to the first subcase in the third case in the definition of $\trComp$ (see \Cref{def:comp-translation-vars})

		      (5) Since $\trsig{a}^\newM = c_i(t_1,\ldots,t_m)$ and by the interpretation of selectors in $\newM$ (see \Cref{tab:funpredComp}).

		      (6) Since $\newM\models\Lone{\formulaProof}$.

		      (7) Since $\trsig{j}^\newM = \trsig{k}^\newM$.

		      (8) It holds that $\newM \models \trsig{x} =
			      \select{\trarr{\arrType_i}}(\trarr{a},\trsig{j})$.

		      (9) The last equality holds by the interpretation of $a$ in $\model$ (see \Cref{def:vars-interpretation}).

		\item For $\ell$ of the form
		      \(
		      b = \store{\arrType_i}(a,k,v)
		      \),
		      we know that $\newM \models \trarr{b} = \store{\trarr{\arrType_i}}(\trarr{a},\trsig{k},\trsig{v})$.
		      Since $\newM^{\trsig{\Sigma_i}}$
		      is an arrays structure, the followings hold:
		      \begin{align}
			      \newM \models \select{\trarr{\arrType_i}}(\trarr{b},\trsig{k}) = \trsig{v} \tag{$\ast$}\label{eq:array-read} \\
			      \forall j \in (\trsig{I_i}^\newM \setminus \{\trsig{k}^\newM\}). \select{\trarr{\arrType_i}}^\newM(\trarr{b}^\newM,j) = \select{\trarr{\arrType_i}}^\newM(\trarr{a}^\newM,j) \tag{$\ast\ast$}\label{eq:array-other}
		      \end{align}
		      In addition, since $\newM$ satisfies $\Ltwo{\formulaProof}$, we have:
		      \begin{align}
			      \forall j \in \indexSet{\arrType_i}{\formulaProof}. \select{\trarr{\arrType_i}}^\newM(\trarr{a}^\newM,\trsig{j}^\newM)
			      = \sel{\arrCons{i}}{\indOrder{}{}(j)+1}^\newM(\trsig{a}^\newM) \tag{$\ast\ast\ast$}\label{eq:array-a-index} \\ \notag
			      \forall j \in \indexSet{\arrType_i}{\formulaProof}. \select{\trarr{\arrType_i}}^\newM(\trarr{b}^\newM,\trsig{j}^\newM)
			      = \sel{\arrCons{i}}{\indOrder{}{}(j)+1}^\newM(\trsig{b}^\newM) \tag{$\ast\ast\ast\ast$}\label{eq:array-b-index}
		      \end{align}
		      Since $I_i^\model = \{k^\model\} \cup (\indexSet{\arrType_i}{\formulaProof}^\model \setminus \{k^\model\}) \cup (I_i^\model \setminus \indexSet{\arrType_i}{\formulaProof}^\model)$, in order to show that $\model \models b = \store{\arrType_i}(a,k,v)$, we need to show that:
		      \begin{enumerate}
			      \item $\store{\arrType_i}^\model(a^\model,k^\model,v^\model)(k^\model) = b^\model(k^\model)$.
			      \item For every $j \in \indexSet{\arrType_i}{\formulaProof}$ such that $j^\model \neq k^\model$. we have that $\store{\arrType_i}^\model(a^\model,k^\model,v^\model)(j^\model) = b^\model(j^\model)$.
			      \item $\forall r \in (I_i^\model \setminus \indexSet{\arrType_i}{\formulaProof}^\model). \store{\arrType_i}^\model(a^\model,k^\model,v^\model)(r) = b^\model(r)$.
		      \end{enumerate}

		      Note that since there is a single constructor $c_i$ whose target sort is $\trsig{\arrType_i}$, we have some
		      $a_1,\ldots,a_m,b_1,\ldots,b_m$ with:
		      \begin{equation}
			      \label{denote:a-and-b-star}
			      \trsig{a}^\newM = c_i(a_1,\ldots,a_m)
			      \land
			      \trsig{b}^\newM = c_i(b_1,\ldots,b_m)
		      \end{equation}

		      Finally, before showing the three cases, we observe that
		      according to the definition of $\Ltwo{\formulaProof}$, for every $l \in \indexSet{\arrType_i}{\formulaProof}$ we have that $\tempvar{a}{\indOrder{}{}(l)+1} = \select{\trarr{\arrType_i}}(\trarr{a}, \trsig{l}) \in \trF(\formulaProof)$ and $\tempvar{b}{\indOrder{}{}(l)+1} = \select{\trarr{\arrType_i}}(\trarr{b}, \trsig{l}) \in \trF(\formulaProof)$. Hence,
		      according to \Cref{def:associated-indices}, $l\in\associatedIndices{\trsig{a}^\newM}{i}$ and $l\in\associatedIndices{\trsig{b}^\newM}{i}$.

		      Hence, $\indexSet{\arrType_i}{\formulaProof} \subseteq
			      \associatedIndices{\trsig{a}^\newM}{i}$ and
		      $\indexSet{\arrType_i}{\formulaProof} \subseteq
			      \associatedIndices{\trsig{b}^\newM}{i}$, and thus we get
		      \begin{equation}
			      \label{equalitybetweensets}
			      \indexSet{\arrType_i}{\formulaProof} =
			      \associatedIndices{\trsig{a}^\newM}{i} =
			      \associatedIndices{\trsig{b}^\newM}{i}
		      \end{equation}

		      We now prove the aforementioned cases one by one:
		      \begin{enumerate}
			      \item Pick $t\in\indexSet{\arrType_i}{\formulaProof}$ with $\trsig{t}^\newM=\trsig{k}^\newM$ and minimal w.r.t. $\indOrder{\arrType_i}{\formulaProof}$, i.e., for every $l\in\indexSet{\arrType_i}{\formulaProof}$ with $\trsig{l}^\newM=\trsig{k}^\newM$ we have $\indOrder{\arrType_i}{\formulaProof}(t)\le\indOrder{\arrType_i}{\formulaProof}(l)$.

			            \begin{alignat*}{2}
				            b^\model(k^\model)                                                                                    \\
				             & \stackrel{(1)}{=} & \; \trComp(\trsig{b}^\newM)(\trComp(\trsig{k}^\newM))                          \\
				             & \stackrel{(2)}{=} & \; \trComp(\trsig{b}^\newM)(\trsig{k}^\newM)                                   \\
				             & \stackrel{(3)}{=} & \; \trComp(b_{\indOrder{}{}(t)+1})                                             \\
				             & \stackrel{(4)}{=} & \; \trComp(\sel{\arrCons{i}}{\indOrder{}{}(t)+1}^\newM(\trsig{b}^\newM))       \\
				             & \stackrel{(5)}{=} & \; \trComp(\select{\trarr{\arrType_i}}^\newM(\trarr{b}^\newM,\trsig{t}^\newM)) \\
				             & \stackrel{(6)}{=} & \; \trComp(\select{\trarr{\arrType_i}}^\newM(\trarr{b}^\newM,\trsig{k}^\newM)) \\
				             & \stackrel{(7)}{=} & \; \trComp(\trsig{v}^\newM)                                                    \\
				             & \stackrel{(8)}{=} & \; v^\model                                                                    \\
				             & \stackrel{(9)}{=} & \; \store{\arrType_i}^\model(a^\model,k^\model,v^\model)(k^\model)
			            \end{alignat*}
			            (1) According to the interpretation of $b$ and $k$ in $\model$ (\Cref{def:vars-interpretation}).

			            (2) By the first case in the definition of $\trComp$ (see
			            \Cref{def:comp-translation-vars}), as for every $i \in [1,n]$ we
			            have that $I_i \notin \arrays_\Sigma \cup
				            \structsorts_{\Sigma_{n+1}}$
			            (see \Cref{sec:theory:ndt-sig})
			            and therefore $\trsig{I_i} \notin
				            \structsorts_{\Sigma_{n+1}^\formulaProof}$
			            (see \Cref{fig:dt-sig-trans}).

			            (3) Due to \eqref{denote:a-and-b-star} and
			            \eqref{equalitybetweensets}
			            and the first subcase in the
			            third case in the definition of $\trComp$ (see
			            \Cref{def:comp-translation-vars}).

			            (4) According to \ref{denote:a-and-b-star} $\trsig{b}^\newM = c_i(b_1,\ldots,b_m)$ and by the interpretation of selectors in $\newM$ (see \Cref{tab:funpredComp}).

			            (5) By \Cref{eq:array-b-index}.

			            (6) Since $\trsig{t}^\newM = \trsig{k}^\newM$.

			            (7) By \Cref{eq:array-read}.

			            (8) This equality holds by the interpretation of $a$ in $\model$ (see
			            \Cref{def:vars-interpretation}).

			            (9) The last equality holds by the definition of the store operation in $\model$ (see \Cref{tab:funpredComp}).

			      \item Pick
			            $t\in\indexSet{\arrType_i}{\formulaProof}$ with
			            $\trsig{t}^\newM=\trsig{j}^\newM$ and minimal w.r.t. $\indOrder{\arrType_i}{\formulaProof}$, i.e., for every
			            $r\in\indexSet{\arrType_i}{\formulaProof}$ with
			            $\trsig{r}^\newM=\trsig{j}^\newM$ we have
			            $\indOrder{\arrType_i}{\formulaProof}(t)\le\indOrder{\arrType_i}{\formulaProof}(r)$.
			            By \eqref{equalitybetweensets}, $t,j \in
				            \associatedIndices{\trsig{a}^\newM}{i}$ and $t,j
				            \in \associatedIndices{\trsig{b}^\newM}{i}$.

			            \begin{alignat*}{2}
				            \store{\arrType_i}^\model(a^\model,k^\model,v^\model)(j^\model)                                        \\
				             & \stackrel{(1)}{=}  & \; a^\model(j^\model)                                                          \\
				             & \stackrel{(2)}{=}  & \; \trComp(\trsig{a}^\newM)(\trComp(\trsig{j}^\newM))                          \\
				             & \stackrel{(3)}{=}  & \; \trComp(\trsig{a}^\newM)(\trsig{j}^\newM)                                   \\
				             & \stackrel{(4)}{=}  & \; \trComp(a_{\indOrder{}{}(t)+1})                                             \\
				             & \stackrel{(5)}{=}  & \; \trComp(\sel{\arrCons{i}}{\indOrder{}{}(t)+1}^\newM(\trsig{a}^\newM))       \\
				             & \stackrel{(6)}{=}  & \; \trComp(\select{\trarr{\arrType_i}}^\newM(\trarr{a}^\newM,\trsig{t}^\newM)) \\
				             & \stackrel{(7)}{=}  & \; \trComp(\select{\trarr{\arrType_i}}^\newM(\trarr{b}^\newM,\trsig{t}^\newM)) \\
				             & \stackrel{(8)}{=}  & \; \trComp(\sel{\arrCons{i}}{\indOrder{}{}(t)+1}^\newM(\trsig{b}^\newM))       \\
				             & \stackrel{(9)}{=}  & \; \trComp(b_{\indOrder{}{}(t)+1})                                             \\
				             & \stackrel{(10)}{=} & \; \trComp(\trsig{b}^\newM)(\trsig{j}^\newM)                                   \\
				             & \stackrel{(11)}{=} & \; \trComp(\trsig{b}^\newM)(\trComp(\trsig{j}^\newM))                          \\
				             & \stackrel{(12)}{=} & \; b^\model(j^\model)
			            \end{alignat*}

			            (1) According to the interpretation of the store operation in $\model$ and since $j^\model \neq k^\model$ (see \Cref{tab:funpredComp}).

			            (2) According to the interpretation of $a$ and $k$ in $\model$ (\Cref{def:vars-interpretation}).

			            (3) By the first case in the definition of $\trComp$ (see
			            \Cref{def:comp-translation-vars}), as for every $i \in [1,n]$ we
			            have that $I_i \notin \arrays_\Sigma \cup \structsorts_{\Sigma_{n+1}}$
			            and therefore $\trsig{I_i} \notin
				            \structsorts_{\Sigma_{n+1}^\formulaProof}$.

			            (4) Due to \eqref{denote:a-and-b-star} $\trsig{a}^\newM = \arrCons{i}(a_1, \ldots a_m)$ and the first subcase in the third case in the definition of $\trComp$ (see \Cref{def:comp-translation-vars}).

			            (5) Since $\trsig{a}^\newM = c_i(a_1,\ldots,a_m)$.

			            (6) By \Cref{eq:array-a-index}.

			            (7) By \Cref{eq:array-other}.

			            (8) -- (12) are symmetric to steps (2)-(6), e.g., (8)
			            is explained as (6), (9) as (5), up to (12) as
			            (2). In (8), we rely on \Cref{eq:array-b-index}
			            rather than on \Cref{eq:array-a-index}.
			      \item Let there be $r \in (I_i^\model \setminus \indexSet{\arrType_i}{\formulaProof}^\model)$.
			            Then,
			            \begin{alignat*}{2}
				            \store{\arrType_i}^\model(a^\model,k^\model,v^\model)(r)  \\
				             & \stackrel{(1)}{=} & \;a^\model(r)                      \\
				             & \stackrel{(2)}{=} & \; \trComp(\trsig{a}^\newM)(r)     \\
				             & \stackrel{(3)}{=} & \; \assignTail{i}(\trsig{a}^\newM) \\
				             & \stackrel{(4)}{=} & \; \assignTail{i}(\trsig{b}^\newM) \\
				             & \stackrel{(5)}{=} & \; \trComp(\trsig{b}^\newM)(r)     \\
				             & \stackrel{(6)}{=} & \; b^\model(r)
			            \end{alignat*}

			            (1) According to the interpretation of the store operation in $\model$ and since $r \neq k^\model$ (see \Cref{tab:funpredComp}).

			            (2) According to the interpretation of $a$ in $\model$ (\Cref{def:vars-interpretation}).

			            (3) We first show the following three claims:
			            \begin{itemize}
				            \item $r \in \trsig{I_i}^\newM$: Since $I_i \notin
					                  \arrays_\Sigma \cup \structsorts_{\Sigma_{n+1}}$
				                  (\Cref{sec:theory:ndt-sig}), we have $\trsig{I_i} \notin
					                  \structsorts_{\Sigma_{n+1}^\formulaProof}$ by
				                  \Cref{fig:dt-sig-trans}.
				                  By \Cref{tab:comp-domains}, $\forall m \in \mathbb{N}. \subdomain{I_i}{m} = \trsig{I_i}^{\newM}$, so \(I_i^{\model} = \bigcup_{m \in \mathbb{N}} \subdomain{I_i}{m} = \trsig{I_i}^{\newM}\).
				                  Since $r \in I_i^\model$, we have $r \in \trsig{I_i}^\newM$.

				            \item $\{\trsig{l}^\newM \mid l \in \indexSet{\arrType_i}{\formulaProof}\} = \indexSet{\arrType_i}{\formulaProof}^\model$:
				                  For any $l \in \indexSet{\arrType_i}{\formulaProof}$,
				                  $l$ is of sort $I_i$.
				                  Since $\trsig{I_i} \notin
					                  \structsorts_{\Sigma_{n+1}^\formulaProof}$ and $\trsig{l}^\newM \in \trsig{I_i}^\newM = I_i^\model \subseteq \cup\{\sigma^\model \mid \sigma \in \sorts{\Sigma}\} \cup \big(\bigcup_{i \in [1,n]} \idx{\arrType_i}^\newM \big)$,
				                  the first
				                  case in the definition of $\trComp$
				                  (\Cref{def:comp-translation-vars})
				                  gives
				                  $\trComp(\trsig{l}^\newM) = \trsig{l}^\newM$.
				                  By \Cref{def:vars-interpretation}, $\trsig{l}^\newM = l^\model$.
				                  Hence, $\{\trsig{l}^\newM \mid l \in \indexSet{\arrType_i}{\formulaProof}\} =
					                  \{\trComp(\trsig{l}^\newM) \mid l \in \indexSet{\arrType_i}{\formulaProof}\} =
					                  \indexSet{\arrType_i}{\formulaProof}^\model$.

				            \item There exists no $l \in \associatedIndices{\trsig{a}^\newM}{i}$ such that $\trsig{l}^\newM = r$:
				                  Since $\associatedIndices{\trsig{a}^\newM}{i} \subseteq \indexSet{\arrType_i}{\formulaProof}$ (\Cref{def:associated-indices}), if such an $l$ existed, we would have $\trsig{l}^\newM \in \{\trsig{l}^\newM \mid l \in \indexSet{\arrType_i}{\formulaProof}\}$.
				                  However, by the previous bullet, $\{\trsig{l}^\newM \mid l \in \indexSet{\arrType_i}{\formulaProof}\} = \indexSet{\arrType_i}{\formulaProof}^\model$, and we know $r \notin \indexSet{\arrType_i}{\formulaProof}^\model$.
				                  Therefore, no such $l$ exists.
			            \end{itemize}

			            Now we can deduce that equality (3) stems from the second subcase in the third case of the definition of $\trComp$ (see \Cref{def:comp-translation-vars}).

			            (4) Note that $\trsig{a}^\newM, \trsig{b}^\newM \in
				            \storeElem{\formulaProof}{\newM}$
			            since $b = \store{\arrType_i}(a,k,v) \in
				            \formulaProof$
			            (see
			            \Cref{def:direct-translation}).

			            According to \Cref{eq:array-other} and the definition of $\tail$ (see \Cref{def:tail}), we have that
			            $\tail(\trarr{a}^\newM) = \tail(\trarr{b}^\newM)$.
			            Since $\newM$ satisfies $\Lthree{\formulaProof}$, we have that $\trarr{a}^\newM = f_{\arrType_i}^\newM(\trsig{a}^\newM)$ and $\trarr{b}^\newM = f_{\arrType_i}^\newM(\trsig{b}^\newM)$.

			            Thus, $\tail(f_{\arrType_i}^\newM(\trsig{a}^\newM)) =
				            \tail(\trarr{a}^\newM) = \tail(\trarr{b}^\newM) =
				            \tail(f_{\arrType_i}^\newM(\trsig{b}^\newM))$.

			            Hence,

			            \begin{itemize}
				            \item $\trsig{a}^\newM, \trsig{b}^\newM \in
					                  \storeElem{\formulaProof}{\newM}$ .
				            \item $\tail(f_{\arrType_i}^\newM(\trsig{a}^\newM)) =
					                  \tail(f_{\arrType_i}^\newM(\trsig{b}^\newM))$.
			            \end{itemize}

			            and according to the definition of $\assignTail{i}$, we have that $\assignTail{i}(\trsig{a}^\newM) = \assignTail{i}(\trsig{b}^\newM)$ (see \Cref{def:comp-assign-tail}).

			            (5) We have shown in step (3) that $r \in \trsig{I_i}^\newM$ and there exists no $l \in \associatedIndices{\trsig{b}^\newM}{i}$ such that $\trsig{l}^\newM = r$. Hence, by the second subcase in the third case of the definition of $\trComp$ (see \Cref{def:comp-translation-vars}), we get this equality.

			            (6) According to the interpretation of $b$ in $\model$ (\Cref{def:vars-interpretation}).
		      \end{enumerate}

		\item For $\ell$ of the form
		      \(
		      x = \sel{c}{i}(y),
		      \)
		      we know that $\newM \models \trsig{x} = \sel{\trsig{\container}}{i}(\trsig{y})$. We distinguish two cases:

		      \begin{enumerate}
			      \item \textbf{Case 1:} If \( \trsig{y}^\newM \) is an
			            application of \( \trsig{c} \), we can denote
			            $\trsig{y}^\newM = \trsig{c}(t_1, \ldots, t_m)$ for some $t_1,\ldots,t_m$.

			            \begin{alignat*}{2}
				            \sel{c}{i}^{\model}(y^{\model})                                                       \\
				             & \stackrel{(1)}{=} & \;\sel{c}{i}^{\model}(\trComp(\trsig{y}^\newM))                \\
				             & \stackrel{(2)}{=} & \;\sel{c}{i}^{\model}(\trComp(\trsig{c}(t_1, \ldots, t_m)))    \\
				             & \stackrel{(3)}{=} & \;\sel{c}{i}^{\model}(c(\trComp(t_1), \ldots, \trComp(t_m)))   \\
				             & \stackrel{(4)}{=} & \;\trComp(t_i)                                                 \\
				             & \stackrel{(5)}{=} & \;\trComp\bigl(\sel{\trsig{c}}{i}^\newM(\trsig{y}^\newM)\bigr) \\
				             & \stackrel{(6)}{=} & \; \trComp(\trsig{x}^{\newM})                                  \\
				             & \stackrel{(7)}{=} & \; x^{\model}
			            \end{alignat*}

			            (1) According to the interpretation of $y$ in $\model$ (see \Cref{def:vars-interpretation}).

			            (2) Substitute $\trsig{y}^\newM = \trsig{c}(t_1, \ldots, t_m)$.

			            (3) By the second case in the definition of $\trComp$ (see \Cref{def:comp-translation-vars}).

			            (4) By the interpretation of selectors in $\model$ (see \Cref{tab:funpredComp}).

			            (5) Since $\trsig{y}^\newM = \trsig{c}(t_1, \ldots, t_m)$ and $\newM^{\Sigma_{n+1}^\formulaProof}$ is a datatype structure, $\sel{\trsig{c}}{i}^\newM(\trsig{y}^\newM) = t_i$ (see \Cref{datatypes-axioms}).

			            (6)
			            Since $\newM \models \trsig{x} = \sel{\trsig{\container}}{i}(\trsig{y})$.

			            (7) The last equality holds by the interpretation of $x$ in $\model$ (see \Cref{def:vars-interpretation}).

			      \item \textbf{Case 2:} \( \trsig{y}^\newM \) is not
			            an application of \( \trsig{c} \), we can denote $\trsig{y} = \widetilde{c}(t_1, \ldots, t_m)$ for some constructor $\widetilde{c} \neq \trsig{c}$.

			            According to the definition of $\Sigma_{n+1}^\formulaProof$ (\Cref{fig:dt-sig-trans}), we can denote the target domain of those two constructors by $\trsig{\sigma}$ where either $\sigma \in \structsorts_{\Sigma_{n+1}}$ or $\sigma = \arrType_i$ for some $i \in [1,n]$. Since there are 2 constructors with target sort $\trsig{\sigma}$ in $\structsorts_{\Sigma_{n+1}^\formulaProof}$, we have that $\sigma \in \structsorts_{\Sigma_{n+1}}$ (as there is a single constructor whose target sort is $\trsig{\arrType_i}$). Hence, we can denote $\widetilde{c} = \trsig{\hat{c}}$ for some constructor $\hat{c} \neq c$ in $\constructors_{\Sigma_{n+1}}$ and $\trsig{y}^\newM = \trsig{\hat{c}}(t_1, \ldots, t_m)$. By the interpretation of $y$ in $\model$ (see \Cref{def:vars-interpretation}), we have that $y^\model = \trComp(\trsig{y}^\newM) = \hat{c}(\trComp(t_1), \ldots, \trComp(t_m))$.
			            \begin{alignat*}{2}
				            \sel{c}{i}^{\model}(y^{\model})                                                       \\
				             & \stackrel{(1)}{=} & \;\trComp\bigl(\sel{\trsig{c}}{i}^\newM(\trsig{y}^\newM)\bigr) \\
				             & \stackrel{(2)}{=} & \; \trComp(\trsig{x}^{\newM})                                  \\
				             & \stackrel{(3)}{=} & \; x^{\model}
			            \end{alignat*}

			            (1) Note that since $y^\model$ is an application
			            of $\hat{c}$ (and not $c$), and $y^\model =
				            \trComp(\trsig{y}^\newM)$, this equality stems
			            from the second case in the interpretation of
			            selectors in \Cref{tab:funpredComp}.

			            (2) Since $\newM \models \trsig{x} = \sel{\trsig{c}}{i}(\trsig{y})$.

			            (3) The last equality holds by the interpretation of $x$ in $\model$ (see \Cref{def:vars-interpretation}).
		      \end{enumerate}

		\item For $\ell$ of the form
		      \(
		      x = c(t_1, \ldots,t_m),
		      \)
		      We know that $\newM \models \trsig{x} = \trsig{c}(\trsig{t_1}, \ldots, \trsig{t_m})$ and thus
		      \[
			      x^{\model}
			      = \trComp(\trsig{x}^\newM)
			      = \trComp(\trsig{c}(\trsig{t_1}^{\newM}, \ldots, \trsig{t_m}^{\newM}))
			      = c(\trComp(\trsig{t_1}^{\newM}), \ldots, \trComp(\trsig{t_m}^{\newM}))
			      = c(t_1^\model, \ldots, t_m^\model)
		      \]

		\item For $\ell$ of the form $\is{c}(x)$, we know that $\newM \models \is{\trsig{c}}(\trsig{x})$. Thus, since $\newM^{\Sigma_{n+1}^\formulaProof}$ is a datatype structure, we have that $\trsig{x}^\newM = \trsig{c}(t_1, \ldots, t_m)$ for some $t_1,\ldots,t_m$. Hence,
		      \[
			      x^\model = \trComp(\trsig{x}^\newM) = \trComp(\trsig{c}(t_1, \ldots, t_m)) = c(\trComp(t_1), \ldots, \trComp(t_m))
		      \]
		      and so $\model \models \ell$.

		\item For $\ell$ of the form $\neg \is{c}(x)$, we know that $\newM \models \neg \is{\trsig{c}}(\trsig{x})$. Thus, since $\newM^{\Sigma_{n+1}^\formulaProof}$ is a datatype structure, we have that $\trsig{x}^\newM = \trsig{\hat{c}}(t_1, \ldots, t_m)$ for some $t_1,\ldots,t_m$ and $\hat{c} \neq c$. Hence,
		      \[
			      x^\model = \trComp(\trsig{x}^\newM) = \trComp(\trsig{\hat{c}}(t_1, \ldots, t_m)) = \hat{c}(\trComp(t_1), \ldots, \trComp(t_m))
		      \]
		      and so $\model \models \ell$.
	\end{enumerate}
\end{proof}

\section{Decidability}
\label{app:sec:decproof}
In this section, we prove
\Cref{thm:decidability-of-new-theory}.
Let $\Sigma$ be an NDT-signature and
$\formulaProof$ a flat $\Sigma$-cube, and
recall the graph $\sigG{\Sigma}$ from
\Cref{def:cyclic-sorts}.
We prove a lemma similar to \Cref{lem:cardinality-partially-ordered},
but for $\newT{\formulaProof}$-interpretations.

\begin{lemma}
	\label{lem:cardinality-partial-order-trsig}
	If there is a path from $\sigma$ to $\tau$ in the graph $\sigG{\Sigma}$ (see \Cref{def:cyclic-sorts}),
	then for every $\newT{\formulaProof}$-interpretation $\newM$, $|\trsig{\tau}^{\newM}| \leq |\trsig{\sigma}^{\newM}|$.
\end{lemma}
\begin{proof}
	In this proof we assume w.l.o.g. that for every $i$ in $[1,n]$, there exists a literal of the form $x = \select{\arrType_i}(a, j)$ in $\formulaProof$
	(we can always add such a literal with fresh variables without changing the $\tndt{\Sigma}$-satisfiability of $\formulaProof$).

	Write the (finite) path in \(\sigG{\Sigma}\) from \(\sigma\) to \(\tau\) as
	\[
		\sigma=\rho_{0}\;\longrightarrow\;
		\rho_{1}\;\longrightarrow\;\dots\;\longrightarrow\;
		\rho_{l}=\tau
		\qquad(l\ge 0).
	\]

	It is sufficient to show that for every $i \in [0,l-1].|\trsig{\rho_i}^\newM| \geq |\trsig{\rho_{i+1}}^\newM|$.

	By the definition of the graph, the edge \((\rho_i,\rho_{i+1})\) exists due to one of the following two cases:
	\begin{enumerate}[label=(\arabic*), leftmargin=1.5em]
		\item There is a constructor \(c : \tau_1\,\dots\,\tau_k \to \rho_i\) such that \(\rho_{i+1} = \tau_j\) for some \(j\).
		\item \(\rho_i = \arrType_m\) and \(\rho_{i+1} = E_m\) for some \(m \in [1,n]\).
	\end{enumerate}

	In both cases, there exists an injective function
	\(f_i : \trsig{\rho_{i+1}}^\newM \to \trsig{\rho_i}^\newM\):
	\begin{enumerate}
		\item In case (1), we fix for every $l \in [1,k]\setminus \{j\}$ some $u_l
			      \in \trsig{\tau_l}^\newM$
		      and define:
		      \[
			      f_{i}(x) = \trsig{c}(u_1,\dots,u_{j-1},x,u_{j+1},\dots,u_k).
		      \]
		\item In case (2), we fix $y \in {\idx{\arrType_m}}^\newM$ and define:
		      \[
			      f_{i}(x) = \arrCons{m}(y, x, \ldots, x).
		      \]
		      (Since by our assumption at the beginning of the proof, there is at least one literal of the form $x = \select{\arrType_i}(a, j)$ in $\formulaProof$, $\indexCar{\arrType_i}{\formulaProof} > 0$ and so the constructor $\arrCons{m}$ has at least one argument of the element sort $\trsig{E_m}$.%
		      )
	\end{enumerate}
\end{proof}

Next, we prove a lemma that is similar to \Cref{lem:infinite-depth-eventually-cyclic}.

\begin{lemma}
	\label{lem:trsig:infinite-depth-eventually-cyclic}
	Let $\sigma\in\sorts{\Sigma}$ be \emph{eventually cyclic} (see \Cref{def:cyclic-sorts}).
	For every $\newT{\formulaProof}$-interpretation $\newM$,
	$|\trsig{\sigma}^{\newM}| = \infty$.
\end{lemma}
\begin{proof}

	As before, we may assume without loss of generality that $\formulaProof$ contains a literal of the form $x = \select{\arrType_i}(a, j)$ for each $i$ in $[1,n]$.

	Since $\sigma$ is eventually cyclic, there exists in particular a sort $\hat{\tau}$ such that the edge $(\sigma,\hat{\tau})$ is in $\sigG{\Sigma}$.
	By \Cref{def:cyclic-sorts}, this implies $\sigma \in \structsorts_{\Sigma_{n+1}} \cup \arrays_\Sigma$, and according to \Cref{fig:dt-sig-trans}, we have $\trsig{\sigma} \in \structsorts_{\Sigma_{n+1}^\formulaProof} \subseteq \sorts{\Sigma_{n+1}^\formulaProof}$.

	Because $\trsig{\sigma} \in \sorts{\Sigma_{n+1}^\formulaProof}$ and $\newM^{\Sigma_{n+1}^{\formulaProof}}$ is a datatypes structure, we can compute $\dtIx{\cdot}$ (see \Cref{def:dtIx}) for elements in $\trsig{\sigma}^{\newM}$ under the interpretation $\newM^{\Sigma_{n+1}^{\formulaProof}}$.

	To prove the lemma, it suffices to show that $\forall N\in\mathbb{N}. \exists x \in \trsig{\sigma}^{\newM}.\, \dtIx{x} \ge N$.
	\textbf{1. Obtaining a long path.}
	Let $N\in\mathbb{N}$.
	We define a path of length at least $N$.
	Because $\sigma$ is eventually cyclic, there is a path
	\[
		\sigma=\tau_0\;\to\;\tau_1\;\to\;\cdots\;\to\;\tau_\ell
	\]
	ending at a node $\tau_\ell$ that belongs to a directed cycle
	\[
		\tau_\ell\;\to\;\tau_{\ell+1}\;\to\;\cdots\;\to\;\tau_{\ell+L}
		\quad(\tau_{\ell+L}=\tau_\ell).
	\]
	Choose $t\in\mathbb{N}$ so that
	\(\ell+tL\ge N\).
	Traversing the cycle $t$ additional times yields a (possibly repeating)
	path of length $T:=\ell+tL\ge N$:
	\[
		\sigma = \tau_0\;\to\;\tau_1\;\to\;\cdots\;\to\;\tau_T.
	\]

	For convenience, we represent this path with reverse indices as
	\[
		\sigma=\sigma_T\;\to\;\sigma_{T-1}\;\to\;\cdots\;\to\;\sigma_0 = \tau_T,
	\]
	where $\sigma_i = \tau_{{T - i}}$ for each $i \in [0,T]$.

	We now show that for every $i \in [0,T-1]$, $\trsig{\sigma_i} \in \sorts{\Sigma_{n+1}^\formulaProof}$, which ensures that $\dtIx{\cdot}$ is defined on $\trsig{\sigma_i}^{\newM}$. (We already know that $\trsig{\sigma_T} = \trsig{\sigma} \in \sorts{\Sigma_{n+1}^\formulaProof}$.)

	For every $i \in [0,T-1]$, the edge $(\sigma_{i+1},\sigma_{i})$ occurs in the graph $\sigG{\Sigma}$.
	By \Cref{def:cyclic-sorts}, this edge arises from one of two cases:
	\begin{itemize}
		\item There exists a constructor
		      \(c:\rho_1\times\!\cdots\!\times\rho_k\to\sigma_{i+1}\)
		      with $\rho_r=\sigma_{i}$ for some $1\leq r\leq k$. In this case, $\sigma_i \in \sorts{\Sigma_{n+1}}$, and according to \Cref{fig:dt-sig-trans}, $\trsig{\sigma_i} \in \sorts{\Sigma_{n+1}^\formulaProof}$.
		\item There exists some $1\leq m \leq n$ such that $\sigma_{i+1}=\arrType_m$ and $\sigma_{i}=E_m$. By our assumption at the beginning of the proof, $\formulaProof$ contains a literal $x = \select{\arrType_m}(a, j)$. Therefore, by \Cref{fig:dt-sig-trans}, the constructor $\arrCons{m}: \idx{\arrType_m} \times \trsig{E_m} \times \cdots \times \trsig{E_m} \to \trsig{\arrType_m}$ belongs to $\constructors_{\Sigma_{n+1}^\formulaProof}$. Consequently, $\trsig{\sigma_i} = \trsig{E_m} \in \sorts{\Sigma_{n+1}^\formulaProof}$.

	\end{itemize}

	\medskip
	\textbf{2. From a path to an element of large $\dtIx{\cdot}$.}

	We use induction over $i$ to show that $\forall i \in [0,T]. \exists x_{i} \in
		\trsig{\sigma_{i}}^{\newM}. \dtIx{x_{i}} \ge i$.
	When $i=T$, we get an element $x_T \in \trsig{\sigma_T}^{\newM} = \trsig{\sigma}^{\newM}$ with $\dtIx{x_T} \ge T \ge N$, which concludes the proof.

	\textbf{Base Case:} Since all domains are not empty, there exists $x_0 \in \trsig{\sigma_0}^{\newM}$ and by definition $\dtIx{x_0} \ge 0$.

	\textbf{Inductive Step:} Assume that for some $i\in[0,T-1]$
	there exists $x_{i} \in \trsig{\sigma_{i}}^{\newM}$ such that $\dtIx{x_{i}} \ge i$. We show that there exists $x_{i+1} \in \trsig{\sigma_{i+1}}^{\newM}$
	such that $\dtIx{x_{i+1}} \ge i+1$.

	The edge $(\sigma_{i+1},\sigma_{i})$ occurs in the graph $\sigG{\Sigma}$ because of one of two reasons:
	\begin{itemize}
		\item \emph{Constructor edge:}
		      there exists a constructor
		      \(c:\rho_1\times\!\cdots\!\times\rho_k\to\sigma_{i+1}\)
		      with $\rho_r=\sigma_{i}$ for some $1\leq r\leq k$.
		      As $\trsig{\rho_1}^\newM, \ldots, \trsig{\rho_k}^\newM$ are all non-empty, we can choose arbitrary $u_s \in \trsig{\rho_s}^\newM$ for $s \in [1,k] \setminus \{r\}$.
		      Thus,
		      \[
			      x_{i+1}:=\trsig{c}(u_1,\dots,u_{r-1},x_{i},u_{r+1},\dots,u_k) \in \trsig{\sigma_{i+1}}^{\newM}
		      \]

		      \begin{alignat*}{2}
			      \dtIx{x_{i+1}}                                                                                                                \\
			       & \stackrel{(1)}{=}   & \;1 + \max\{\dtIx{u_1}, \ldots, \dtIx{u_{r-1}}, \dtIx{x_{i}}, \dtIx{u_{r+1}}, \ldots, \dtIx{u_{k}}\} \\
			       & \stackrel{(2)}{\ge} & \; 1 + \dtIx{x_{i}}                                                                                  \\
			       & \stackrel{(3)}{\ge} & \; i + 1                                                                                             \\
		      \end{alignat*}

		      (1) According to \Cref{lem:dtIx-depth-children}.

		      (2) From the definition of $\max$.

		      (3) According to the IH $\dtIx{x_{i}} \ge i$.

		\item \emph{Array edge:}
		      the edge $(\sigma_{i+1},\sigma_{i})$ occurs
		      in the graph because
		      there exists some $1\leq m \leq n$ such that $\sigma_{i+1}=\arrType_m$, $\sigma_{i}=E_m$.
		      Fix some $r \in {\idx{\arrType_m}}^\newM$ and set
		      \[
			      x_{i+1}:=\arrCons{m}(r, x_{i},\dots,x_{i}) \in \trsig{\sigma_{i+1}}^{\newM}
		      \]
		      \begin{alignat*}{2}
			      \dtIx{x_{i+1}}                                                  \\
			       & \stackrel{(1)}{=}   & \;1 + \max\{\dtIx{r}, \dtIx{x_{i}}, \} \\
			       & \stackrel{(2)}{\ge} & \; 1 + \dtIx{x_{i}}                    \\
			       & \stackrel{(3)}{\ge} & \; i + 1                               \\
		      \end{alignat*}

		      (1) According to \Cref{lem:dtIx-depth-children}.

		      (2) From the definition of $\max$.

		      (3) According to the IH $\dtIx{x_{i}} \ge i$.
	\end{itemize}
\end{proof}

One more ingredient that is needed for the proof of
\Cref{thm:decidability-of-new-theory}
is a generalization of \Cref{thm:nelson-oppen} to multiple theories, which we prove here
for the sake of completeness.
Our proof uses Theorem 2.5 of \cite{JovanovicBarrett2010TR} (which is an extended version of \cite{JBLPAR}) and Theorem 9 of \cite{TinelliZarba2004} (the order-sorted Löwenheim-Skolem Theorem):

\begin{definition}
	\label{def:arrangement}
	Let $V$ be a set of variables and $\delta_V$ a conjunction
	of equality literals and negations of such literals (i.e., equalities and disequalities) between variables in $V$.
	We say that
	$\delta_V$ is an \emph{arrangement} of $V$ if there exists an equivalence relation $E$
	on $V$ such that $\delta_V = \wedge \{ x = y \mid (x,y) \in E \} \bigwedge \wedge \{ x \neq y \mid (x,y) \notin E\}$.
\end{definition}

\begin{theorem}[Many-sorted Löwenheim-Skolem Theorem]
	\label{thm:lowenheim-skolem}
	Given a many-sorted signature $\tilde{\Sigma}$, a $\tilde{\Sigma}$-theory $\tilde{T}$, a satisfiable $\tilde{T}$-formula $\Phi$, and a $\tilde{T}$-interpretation $\model$ satisfying $\Phi$, there exists a $\tilde{\Sigma}$-interpretation $\newM$ satisfying $\Phi$ such that $|\sigma^\model| \geq \aleph_0$ implies $|\sigma^\newM| = \aleph_0$, for each sort $\sigma \in \sorts{\tilde{\Sigma}}$.
\end{theorem}

\begin{lemma}[Thm 2.5 of \cite{JovanovicBarrett2010TR}]
	\label{lem:no-helper}
	Let $T_i$ be a $\Sigma_i$-theory, for $i = 1,2$, such that
	$\func{\Sigma_1} \cap \func{\Sigma_2} = \emptyset$ and
	$\pred{\Sigma_1} \cap \pred{\Sigma_2} = \emptyset$.
	Let $\Sigma = \Sigma_1 \cup \Sigma_2$,
	$T = T_1 \oplus T_2$, and
	$S = \sorts{\Sigma_1} \cap \sorts{\Sigma_2}$.
	For $i = 1,2$, let $\formulaProof_i$ be a flat $\Sigma_i$-cube, and let
	$V = \fv{}{\formulaProof_1} \cap \fv{}{\formulaProof_2}$.
	Suppose there exist a $T_1$-interpretation $\mathcal{A}$, a $T_2$-interpretation
	$\mathcal{B}$, and an arrangement
	$\delta_V$ of $V$ such that
	\begin{itemize}
		\item $\mathcal{A} \models \formulaProof_1 \wedge \delta_V$,
		\item $\mathcal{B} \models \formulaProof_2 \wedge \delta_V$, and
		\item $\lvert A_\sigma \rvert = \lvert B_\sigma \rvert$ for every $\sigma \in S$.
	\end{itemize}
	Then there exists a $T$-interpretation $\nextM$ such that
	\begin{itemize}
		\item $\nextM \models \formulaProof_1 \wedge \formulaProof_2 \wedge \delta_V$,
		\item $C_\sigma = A_\sigma$ for every $\sigma \in \sorts{\Sigma_1}$, and
		\item $C_\sigma = B_\sigma$ for every $\sigma \in \sorts{\Sigma_2} \setminus S$.
	\end{itemize}
\end{lemma}

\begin{theorem}[Nelson--Oppen combination for $n$ theories]
	\label{thm:nelson-oppen-n}
	Let $n \geq 2$ and, for each $i \in \{1,\dots,n\}$, let
	$\Sigma_i$ be a signature and
	$T_i$
	a
	$\Sigma_i$-theory. Assume:
	\begin{enumerate}[label=(\roman*)]
		\item For all $i \neq j$,
		      $\func{\Sigma_i} \cap \func{\Sigma_j} = \emptyset$ and
		      $\pred{\Sigma_i} \cap \pred{\Sigma_j} = \emptyset$ (only sorts may be shared).
		\item For every $i \in [1,n]$ we have that $T_i$ is stably infinite w.r.t. $\sorts{\Sigma_i}$.
		\item The quantifier-free satisfiability problem for each $T_i$ is decidable.
	\end{enumerate}
	Let $\Sigma = \bigcup_{i=1}^n \Sigma_i$ and let
	$T = T_1 \oplus \dots \oplus T_n$.
	Then the quantifier-free $\Sigma$-satisfiability problem for $T$ is decidable and $T$ is stably infinite w.r.t. $\sorts{\Sigma}$.
\end{theorem}
\begin{proof}
	We prove the claim by induction on $n$:

	\textbf{Base case} ($n=2$):
	According to \Cref{thm:nelson-oppen}, $T$ is decidable.
	We want to also show that $T$ is stably infinite w.r.t. $\sorts{\Sigma_1} \cup \sorts{\Sigma_2}$.
	Let $\formulaProof$ be a $\Sigma$-formula that is $T$-satisfiable.

	We need to show that there exists a $T$-interpretation $\nextM$ such that $\nextM \models \formulaProof$ and $\lvert \sigma^{\nextM} \rvert $ is infinite for every $\sigma \in \sorts{\Sigma_1} \cup \sorts{\Sigma_2}$.

	Since $\formulaProof$ is $T$-satisfiable, there exists a $T$-interpretation $\hat{\nextM}$ such that $\hat{\nextM} \models \formulaProof$.

	W.l.o.g, $\formulaProof=\formulaProof_1\wedge\formulaProof_2$, where
	$\formulaProof_1$ and $\formulaProof_2$ are conjunctions of $\Sigma_1$-literals and $\Sigma_2$-literals in $\formulaProof$, respectively.
	Let $V = \fv{}{\formulaProof_1} \cap \fv{}{\formulaProof_2}$ be the set of shared variables between $\formulaProof_1$ and $\formulaProof_2$.

	We can define the following arrangement $\delta_V$ of $V$:
	\[
		\delta_V = \wedge\{ x = y \mid \hat{\nextM} \models x = y \} \bigwedge \wedge\{ x \neq y \mid \hat{\nextM} \models x \neq y \}
	\]

	Note that $\hat{\nextM}^{\Sigma_1} \models \formulaProof_1 \wedge \delta_V$ and $\hat{\nextM}^{\Sigma_2} \models \formulaProof_2 \wedge \delta_V$.
	Since both $T_1$ and $T_2$ are stably infinite w.r.t.\ all their sorts, there exist:
	\begin{itemize}
		\item a $T_1$-interpretation $\newM$ such that $\newM \models \formulaProof_1 \wedge \delta_V$ and for every $\sigma \in \sorts{\Sigma_1}$, $\lvert \sigma^{\newM} \rvert$ is infinite, and
		\item a $T_2$-interpretation $\model$ such that $\model \models \formulaProof_2 \wedge \delta_V$ and for every $\sigma \in \sorts{\Sigma_2}$, $\lvert \sigma^{\model} \rvert$ is infinite.
	\end{itemize}
	By \Cref{thm:lowenheim-skolem}, we can further assume that $|\sigma^{\newM}| = \aleph_0$ for every $\sigma \in \sorts{\Sigma_1}$, and $|\sigma^{\model}| = \aleph_0$ for every $\sigma \in \sorts{\Sigma_2}$.
	In particular, for every $\sigma \in \sorts{\Sigma_1} \cap \sorts{\Sigma_2}$, we have $|\sigma^{\newM}| = \aleph_0 = |\sigma^{\model}|$.

	Thus, the conditions of \Cref{lem:no-helper} are satisfied.

	By \Cref{lem:no-helper}, there exists a $T$-interpretation $\nextM$ such that $\nextM \models \formulaProof_1 \wedge \formulaProof_2 \wedge \delta_V$, and for every $\sigma \in \sorts{\Sigma_1} \cup \sorts{\Sigma_2}$, $|\sigma^{\nextM}|$ is infinite. In particular, we also have $\nextM\models\formulaProof$.

	\textbf{Inductive step} ($n > 2$):
	Assume the claim holds for $k<n$ theories, and let $T_1, \dots, T_{n}$ be $\Sigma_1, \dots, \Sigma_{n}$-theories, respectively,
	such that conditions (i), (ii), and (iii) of the theorem statement hold.
	Let $\Sigma' = \bigcup_{i=1}^{n-1} \Sigma_i$ and
	$T' = T_1 \oplus \dots \oplus T_{n-1}$.
	By the induction hypothesis, $T'$ is decidable and
	stably infinite w.r.t.\
	$\bigcup_{i=1}^{n-1} \sorts{\Sigma_i}$.
	Applying the IH once more to $T'$ and $T_{n}$,
	we conclude that $T' \oplus T_{n} = T$ is decidable and
	stably infinite w.r.t.\
	$\bigcup_{i=1}^{n} \sorts{\Sigma_i}$.
\end{proof}

\medskip

We now prove two lemmas that will be useful in the proof of
\Cref{thm:decidability-of-new-theory}.
The first restricts the sorts with finite interpretations in $\newT{\formulaProof}$, while the second does so for sorts with infinite interpretations.

\begin{lemma}
	\label{dec:finite-implies-struct}
	Let $\nextM$ be a $\newT{\formulaProof}$-interpretation that satisfies
	\(
	\forall \sigma \in \sorts{\newSig{\Sigma}{\formulaProof}} \setminus (\structsorts_{\Sigma_{n+1}^\formulaProof} \cup \arrays_{\newSig{\Sigma}{\formulaProof}}). |\sigma^{\nextM}| = \infty
	\).
	For every $\sigma \in \sorts{\Sigma}$ if $\trsig{\sigma}^{\nextM}$ is finite then $\sigma \in \structsorts_{\Sigma_{n+1}}$.
\end{lemma}

\begin{proof}
	Let $\sigma \in \sorts{\Sigma}$ be such that $\trsig{\sigma}^{\nextM}$ is finite. We show that $\sigma \in \structsorts_{\Sigma_{n+1}}$ by ruling out all other possibilities.

	By \Cref{sec:theory:ndt-sig}, $\arrays_\Sigma \cap \structsorts_{\Sigma_{n+1}} = \emptyset$, so we can partition $\sorts{\Sigma}$ into three disjoint sets:
	\[
		\sorts{\Sigma} = (\sorts{\Sigma} \setminus (\structsorts_{\Sigma_{n+1}} \cup \arrays_{\Sigma})) \uplus \arrays_{\Sigma}\uplus \structsorts_{\Sigma_{n+1}}
	\]
	We show that $\sigma$ cannot belong to the first two sets.

	\textbf{Case 1:} Suppose $\sigma \in \sorts{\Sigma} \setminus (\structsorts_{\Sigma_{n+1}} \cup \arrays_{\Sigma})$.
	By \Cref{fig:dt-sig-trans}, we have $\trsig{\sigma} \notin \structsorts_{\Sigma_{n+1}^\formulaProof}$.
	Additionally, by \Cref{not:ndt-sig} and \Cref{tab:new-signature-1}, $\arrays_{\newSig{\Sigma}{\formulaProof}} = \{\trarr{\arrType_i} \mid i \in [1,n]\}$, which implies $\trsig{\sigma} \notin \arrays_{\newSig{\Sigma}{\formulaProof}}$.
	Thus, $\trsig{\sigma} \in \sorts{\newSig{\Sigma}{\formulaProof}} \setminus (\structsorts_{\Sigma_{n+1}^\formulaProof} \cup \arrays_{\newSig{\Sigma}{\formulaProof}})$.
	Since $\nextM$ satisfies the condition of the lemma, we have $|\trsig{\sigma}^{\nextM}| = \infty$, contradicting our assumption that $\trsig{\sigma}^{\nextM}$ is finite.

	\textbf{Case 2:} If $\sigma = \arrType_i$ for some $i \in [1,n]$, then since $\idx{\arrType_i} \in \sorts{\newSig{\Sigma}{\formulaProof}} \setminus (\structsorts_{\Sigma_{n+1}^\formulaProof} \cup \arrays_{\newSig{\Sigma}{\formulaProof}})$
	(see \Cref{fig:dt-sig-trans}),
	we have $|\idx{\arrType_i}^{\nextM}| = \infty$.
	Fix an element $t \in \trsig{E_i}^{\nextM}$. Then for every $j \in \idx{\arrType_i}^{\nextM}$, we have $\arrCons{i}(j,t,\ldots,t) \in \trsig{\arrType_i}^{\nextM}$ (if $\arrCons{i}$ only has one argument, we omit the $t$s).
	Therefore, $|\trsig{\arrType_i}^{\nextM}| = \infty$, which is a contradiction.
\end{proof}

\begin{lemma}
	\label{dec:S-hat-is-infinite}
	Let $\nextM$ be a $\newT{\formulaProof}$-interpretation in which
	for each
	$\sigma \in \sorts{\newSig{\Sigma}{\formulaProof}} \setminus (\structsorts_{\Sigma_{n+1}^\formulaProof} \cup \arrays_{\newSig{\Sigma}{\formulaProof}})$, $|\sigma^{\nextM}| = \infty$.
	Define $\widehat{S} = \sorts{\Sigma^\formulaProof_{n+1}} \cap \big(\sorts{\trsig{\Sigma_1}} \cup \ldots \cup \sorts{\trsig{\Sigma_n}} \cup \sorts{\Sigma_{n+2}}\big)$, then:
	$|\tau^{\nextM}| = \infty$
	for every $\tau \in \widehat{S}$.
\end{lemma}
\begin{proof}
	First note that $\widehat{S} \subseteq \sorts{\Sigma_{n+1}^\formulaProof}$.
	Thus, it suffices to consider $\tau \in \sorts{\Sigma_{n+1}^\formulaProof}$ with $|\tau^{\nextM}| < \infty$ and show that $\tau \notin \widehat{S}$.

	\paragraph{\textbf{Step 1. $\tau = \trsig{\sigma}$ for some $\sigma \in \structsorts_{\Sigma_{n+1}}$.}}

	By \Cref{fig:dt-sig-trans}, $\sorts{\Sigma_{n+1}^\formulaProof} = \{\trtype{\sigma}\mid \sigma\in\elemsorts_{\Sigma_{n+1}}\setminus \arrays_{\Sigma}\}\cup
		\{\idx{\arrType_i}\mid i\in [1,n]\} \cup \{\trsig{E_i} \mid i \in [1,n], E_i \notin (\arrays_\Sigma \cup \structsorts_{\Sigma_{n+1}})\} \cup \{\trtype{\sigma}\mid
		\sigma\in\structsorts_{\Sigma_{n+1}}\cup\arrays_{\Sigma}\}$.

	If $\tau = \idx{\arrType_i}$ for some $i \in [1,n]$, then since $\idx{\arrType_i} \in \sorts{\newSig{\Sigma}{\formulaProof}} \setminus (\structsorts_{\Sigma_{n+1}^\formulaProof} \cup \arrays_{\newSig{\Sigma}{\formulaProof}})$, by the assumption over $\nextM$ we have $|\idx{\arrType_i}^{\nextM}| = \infty$, which is a contradiction.

	Hence, $\tau = \trsig{\sigma}$ for some $\sigma \in \sorts{\Sigma_{n+1}}$. By \Cref{dec:finite-implies-struct}, we have $\sigma \in \structsorts_{\Sigma_{n+1}}$.

	\paragraph{\textbf{Step 2. Constructing a matching subgraph for $\sigma$.}}

	Consider the graph $\sigG{\Sigma}$ (see \Cref{def:cyclic-sorts}).
	Since $|\trsig{\sigma}^{\nextM}| < \infty$, \Cref{lem:trsig:infinite-depth-eventually-cyclic} shows that $\sigma$ is not eventually cyclic.
	Consider the subgraph of $\sigG{\Sigma}$ reachable from $\sigma$, denoted $\sigG{\Sigma}' = (V',E')$.
	Since $\sigma$ is not eventually cyclic, $\sigG{\Sigma}'$ is a DAG.

	Let $\rho \in V'$. Since $\rho$ is reachable from $\sigma$ in $\sigG{\Sigma}$, \Cref{lem:cardinality-partial-order-trsig} implies that $|\trsig{\rho}^{\nextM}| \leq |\trsig{\sigma}^{\nextM}| < \infty$. By \Cref{dec:finite-implies-struct}, $\rho \in \structsorts_{\Sigma_{n+1}}$.

	Hence, $V' \subseteq \structsorts_{\Sigma_{n+1}}$.

	\paragraph{\textbf{Step 3. Showing that $|\sigma^{\model'}| < \infty$ for every $T_{n+1}$-interpretation $\model'$.}}

	Let $\model'$ be any $T_{n+1}$-interpretation. We use induction on the length of the longest path starting from a node in the graph $\sigG{\Sigma}'$ to show that $\forall \rho \in V'. |\rho^{\model'}| < \infty$.
	\begin{itemize}
		\item \textbf{Base Case (longest path length = 0):} Let $\rho \in V'$ be a node such that the longest path starting from $\rho$ has length 0. This means $\rho$ has no outgoing edges in $\sigG{\Sigma}'$. Since $\rho \in \structsorts_{\Sigma_{n+1}}$, any constructor in $\constructors_{\Sigma_{n+1}}$ with target sort $\rho$ must take no arguments (otherwise there would be an outgoing edge from $\rho$).
		      Therefore, $\rho^{\model'}$ is the finite set of all constant constructors in $\constructors_{\Sigma_{n+1}}$ of sort $\rho$, which is finite since $\Sigma_{n+1}$ is a datatype signature (see \Cref{def:IDT_sig}).
		\item \textbf{Inductive Step:}
		      Assume the claim holds for every node $\rho' \in V'$ where the longest path starting from $\rho'$ has length at most $k$.
		      Let $\rho \in V'$ be a node where the longest path starting from $\rho$ has length $k+1$.
		      Since $\rho \in \structsorts_{\Sigma_{n+1}}$, for every constructor $c:\tau_1 \times \ldots \times \tau_m \to \rho$, and for each argument sort $\tau_i$ (where $i \in \{1, \ldots, m\}$), there is an edge from $\rho$ to $\tau_i$ in $\sigG{\Sigma}$.
		      Since the longest path from $\rho$ has length $k+1$, the longest path from any of its successors $\tau_i$ has length at most $k$.
		      By the induction hypothesis, $|\tau_i^{\model'}| < \infty$ for each $i$.
		      Therefore, due to \Cref{datatypes-axioms} we get that $\rho^{\model'}$ is finite, since it is the finite union of finite sets.
	\end{itemize}
	Thus, for all $\rho \in V'$,
	including $\sigma$,
	we have $|\rho^{\model'}| < \infty$.

	\paragraph{\textbf{Step 4. $\sigma \notin \sorts{\Sigma_{n+1}} \cap \big( \sorts\{\Sigma_1\} \cup \ldots \cup \sorts\{\Sigma_n\}\big)$.}}

	By the assumption of \Cref{thm:decidability-of-new-theory}, the theory $T_{n+1}$ is stably infinite with respect to $\sorts{\Sigma_{n+1}} \cap \big( \sorts\{\Sigma_1\} \cup \ldots \cup \sorts\{\Sigma_n\}\big)$.

	This means that there is a $T_{n+1}$ interpretation in which every sort in $\sorts{\Sigma_{n+1}} \cap \big( \sorts\{\Sigma_1\} \cup \ldots \cup \sorts\{\Sigma_n\}\big)$ has an infinite domain.

	However, by Step~3, the interpretation of $\sigma$ is finite for every $T_{n+1}$-interpretation $\model'$.
	Therefore,
	\begin{equation}
		\label{decidability-not-shared}
		\sigma \notin \sorts{\Sigma_{n+1}} \cap \big( \sorts\{\Sigma_1\} \cup \ldots \cup \sorts\{\Sigma_n\}\big).
	\end{equation}

	\paragraph{\textbf{Step 5. Concluding that $\tau \notin \widehat{S}$}}

	According to Step 1, we have $\tau = \trsig{\sigma}$ for some $\sigma \in \structsorts_{\Sigma_{n+1}}$ and by our assumption over $\tau$ we know that $\tau \in \sorts{\Sigma_{n+1}^\formulaProof}$.

	We now show that $\tau$ does not appear in any of the other signatures (i.e. $\trsig{\Sigma_1} \ldots \trsig{\Sigma_n}$ or $\Sigma_{n+2}$) and thus $\tau \notin \widehat{S}$.
	First, because of Step 3, we know that $\sigma \in \structsorts_{\Sigma_{n+1}}$ and by \Cref{decidability-not-shared} it holds that $\sigma \notin \sorts{\Sigma_i}$ for every $i \in [1,n]$.
	According to \Cref{tab:new-signature-1} we have $\sorts{\trsig{\Sigma_i}} = \{\trsig{I_i}, \trsig{E_i}, \trarr{\arrType_i}\}$,
	it follows that $\tau = \trsig{\sigma} \notin \sorts{\trsig{\Sigma_i}}$ for every $i \in [1,n]$.

	Second, we have established that $\sigma \notin \sorts{\Sigma_i}$ for every $i \in [1,n]$ and thus $\sigma \notin \arrays_{\Sigma}$.
	By \Cref{fig:one-to-one-sig}, we have $\sorts{\Sigma_{n+2}} = \{\trtype{\arrType_i} \mid 1 \leq i \leq n\} \cup \{\trarr{\arrType_i} \mid 1 \leq i \leq n\}$.
	Therefore, $\tau = \trsig{\sigma} \notin \sorts{\Sigma_{n+2}}$.

	Thus, $\tau = \trsig{\sigma}\notin \sorts{\trsig{\Sigma_1}} \cup \ldots \cup \sorts{\trsig{\Sigma_n}} \cup \sorts{\Sigma_{n+2}}$, which implies $\tau \notin \widehat{S}$.
\end{proof}

Now we can restate and prove \Cref{thm:decidability-of-new-theory}.

\Decidability*

\begin{table}[t]
	\newcolumntype{M}[1]{>{\centering\arraybackslash}m{#1}}
	\centering
	\begin{tabular}{|M{2cm}|M{11cm}|}\hline
		\shortstack{
			Domains
		}          &
		\shortstack{
			$\forall \sigma \in \sorts{\newSig{\Sigma}{\formulaProof}}. \sigma^{\newM_1} = \begin{cases}
					\sigma^{\newM_0}                                & \text{if } \sigma \in \sorts{\Sigma_{n+1}^\formulaProof},             \\[4pt]
					\trsig{I_i}^{\newM_1} \ra \trsig{E_i}^{\newM_1} & \text{if } \sigma = \trarr{\arrType_i} \text{ for some } i \in [1,n], \\[4pt]
					\{c_\sigma\}                                    & \text{else }
				\end{cases}$
		}
		\\\hline
		Functions  &
		\shortstack{
			$\select{}$ and $\store{}$ are defined as \\
			read and update
			operations on functions.                  \\\\
			All functions in $\func{\Sigma_{n+1}^\formulaProof}$ are interpreted as in $\newM_0$.
			\\\\
			For every $i \in [1,n]$ we fix $f_{\arrType_i}^{\newM_1}$ and
			$g_{\arrType_i}^{\newM_1}$ to some arbitrary functions.
		}
		\\\hline
		Predicates &
		\shortstack{
			All predicates are interpreted as in $\newM_0$.
		}
		\\\hline
		Variables  &
		\shortstack{
			All variables are interpreted as in $\newM_0$.
		}
		\\\hline
	\end{tabular}
	\caption{A $\newT{\formulaProof}$-interpretation $\newM_1$. $c_\sigma$ is an arbitrary but fixed element for each sort $\sigma \in \sorts{\newSig{\Sigma}{\formulaProof}} \setminus \big(\sorts{\Sigma_{n+1}^\formulaProof} \cup \{\trarr{\arrType_i} \mid i \in [1,n]\}\big)$.
	}
	\label{tab:decidability-new-interpretation}
\end{table}

\begin{proof}
	\paragraph{\textbf{Step 1. Reducing to stable infiniteness of $T_{n+1}^\formulaProof$.}}
	Using
	\Cref{ex:stably-infinite}
	observe that the theories $\trsig{T_1}, \ldots, \trsig{T_n}$ are stably infinite with respect to all their sorts, since they are array theories.
	Similarly, $T_{n+2}$ is stably infinite with respect to all its sorts, since it is the theory of uninterpreted functions.

	Therefore, by \Cref{thm:nelson-oppen-n}, the combination $\trsig{T_1} \oplus \ldots \oplus \trsig{T_n} \oplus T_{n+2}$ is decidable and stably infinite with respect to all its sorts.

	To apply~\Cref{thm:nelson-oppen} to $\trsig{T_1} \oplus \ldots \oplus \trsig{T_n} \oplus T_{n+2}$ and $T_{n+1}^\formulaProof$, it remains to prove that $T_{n+1}^\formulaProof$ is stably infinite with respect to $\widehat{S} = \sorts{\Sigma^\formulaProof_{n+1}} \cap \big(\sorts{\trsig{\Sigma_1}} \cup \ldots \cup \sorts{\trsig{\Sigma_n}} \cup \sorts{\Sigma_{n+2}}\big)$.

	Let $\theta$ be a $T_{n+1}^{\formulaProof}$-satisfiable $\Sigma_{n+1}^\formulaProof$-formula. Then there exists a $T_{n+1}^\formulaProof$-interpretation $\newM_0$ such that $\newM_0 \models \theta$.

	\paragraph{\textbf{Step 2. Constructing a $\newT{\formulaProof}$-interpretation of $\theta$.}}
	Since $\newSig{\Sigma}{\formulaProof} = \trsig{\Sigma_1} \cup \ldots \cup \trsig{\Sigma_n} \cup \Sigma_{n+1}^\formulaProof \cup \Sigma_{n+2}$,
	we can view $\theta$ as a formula in the signature $\newSig{\Sigma}{\formulaProof}$.
	We define an interpretation $\newM_1$ over the signature $\newSig{\Sigma}{\formulaProof}$ using $\newM_0$ (see \Cref{tab:decidability-new-interpretation}).


	We now show that $\newM_1$ is a $\newT{\formulaProof}$-interpretation:

	\begin{itemize}
		\item For every $i \in [1,n]$, $\newM_1^{\trsig{\Sigma_i}}$ is an array interpretation: the array sort is interpreted as the set of functions $\trarr{\arrType_i}^{\newM_1} = \trsig{I_i}^{\newM_1} \ra \trsig{E_i}^{\newM_1}$, and the symbols $\select{\trarr{\arrType_i}}$ and $\store{\trarr{\arrType_i}}$ are interpreted as the usual read and update operations on arrays.
		\item $\newM_1^{\Sigma_{n+1}^\formulaProof}$ is a $T_{n+1}^\formulaProof$-interpretation, since all domains, functions, and predicates in $\Sigma_{n+1}^\formulaProof$ are interpreted as in $\newM_0$.
		\item $\newM_1^{\Sigma_{n+2}}$ is a $T_{n+2}$-interpretation, since $T_{n+2}$ is the theory of uninterpreted functions and thus every interpretation over the signature $\Sigma_{n+2}$ is a $T_{n+2}$-interpretation.
	\end{itemize}

	Note that $\newM_1$ satisfies $\theta$, since $\newM_1^{\Sigma_{n+1}^\formulaProof}$ is interpreted as $\newM_0$, which satisfies $\theta$.

	By \Cref{lem:infinitely-more-elements}, there exists a $\newT{\formulaProof}$-interpretation $\nextM$ such that $\nextM \models \theta$ and:
	\begin{equation}
		\label{decidability-infinite}
		\forall \sigma \in \sorts{\newSig{\Sigma}{\formulaProof}} \setminus (\structsorts_{\Sigma_{n+1}^\formulaProof} \cup \arrays_{\newSig{\Sigma}{\formulaProof}}). |\sigma^{\nextM}| = \infty
	\end{equation}

	\paragraph{\textbf{Step 3. $T_{n+1}^\formulaProof$ is stably infinite w.r.t. $\widehat{S}$.}}

	Define the $T_{n+1}^\formulaProof$-interpretation $\newM = \nextM^{\Sigma_{n+1}^\formulaProof}$.
	By construction, $\newM \models \theta$, and due to \Cref{decidability-infinite} and \Cref{dec:S-hat-is-infinite}, $|\tau^{\newM}| = |\tau^{\nextM}| = \infty$ for every $\tau \in \widehat{S}$.

	Since $\theta$ was an arbitrary satisfiable
	$T_{n+1}^\formulaProof$-formula, we
	conclude that $T_{n+1}^\formulaProof$ is stably infinite with respect to
	$\widehat{S}$.

\end{proof}



\section{Additional Material}
\label{app:moreproofs}

\subsection{A concrete interpretation for \Cref{ex:NO}}
\label{app:ex:co-datatype-model}
In \Cref{ex:NO}, satisfiability of $\formulaExUnsat$ from \Cref{tab:sigs} was proven indirectly,
using an argument that relates to the Nelson-Oppen method.
In this section, we provide a concrete
$\tpndt{\exsigwf}$-interpretation
that satisfies the formula $\formulaExUnsat$.
Define an interpretation $\model$ by setting
$I^\model = \{-1\}$,
$E^\model = \{\emp\} \cup \{\container(-1, n) \mid n \in \mathbb{N}\cup \{\infty\}\}$,
and
$\arrType^\model = \mathbb{N} \cup \{\infty\}$.
The interpretation of the variables is given by
$a^\model = \infty$,
$i^\model = -1$, and
$x^\model = \container(-1, \infty)$.
To interpret the symbols, we define
a bijection $m: \mathbb{N} \cup \{\infty\} \to E^\model$
by setting
$m(0) = \emp$,
$m(n) = \container(-1, n-1)$ for every $n \in \mathbb{N} \setminus \{0\}$,
and
$m(\infty) = \container(-1, \infty)$.
Finally,
${\select{}}^{\model}(\alpha, \beta) = m(\alpha)$
${\store{}}^{\model}(\alpha, \beta, \gamma) = m^{-1}(\gamma)$, and
the remaining function symbols
of $\Sigma_2$ are defined according to \Cref{datatypes-axioms}.

It can be shown that $\model$ is indeed a $\tpndt{\Sigma}$-interpretation
that satisfies $\formulaExUnsat$.

First, let's show that $\model^{\Sigma_1}$ is an arrays $\Sigma_1$-structure
by showing it satisfies the needed axioms from \Cref{fig:arrax}.

\begin{figure}[t]
	\centering
	\begin{tabular}{l}
		$\row_1$ : $\forall a,i,e. \select{\arrType} (\store{\arrType} (a,i,e), i) = e$                                                 \\
		$\row_2$ : $\forall a,i,e,j. \neg (j = i) \implies \select{\arrType} (\store{\arrType} (a, i, e), j) = \select{\arrType}(a, j)$ \\
		$\ext$ : $\forall a_1, a_2. \neg (a_1 = a_2) \implies \exists j: \neg ( \select{\arrType}(a_1, j) = \select{\arrType}(a_2, j))$
	\end{tabular}
	\caption{Axioms for arrays.}
	\label{fig:arrax}
\end{figure}

\begin{enumerate}
	\item $\row_1$ :
	      \[
		      \begin{aligned}
			       & \ \forall \alpha \in \arrType^\model.
			      \forall \beta \in I^\model.
			      \forall \gamma \in E^\model.                             \\
			       & \select{}^\model
			      \big(\store{}^\model (\alpha, \beta, \gamma), \beta\big) \\
			       & = \select{}^\model
			      \big(m^{-1}(\gamma), \beta\big)                          \\
			       & = m\big(m^{-1}(\gamma)\big) = \gamma.
		      \end{aligned}
	      \]

	\item $\row_2$: We need to show that:
	      \[
		      \begin{aligned}
			       & \ \forall \alpha \in \arrType^\model. \beta \in I^\model. \gamma \in E^\model. \delta \in I^\model. \\
			       & \neg (\delta = \beta) \implies
			      \select{}^\model
			      \big(\store{}^\model (\alpha, \beta, \gamma), \delta\big)                                              \\
			       & = \select{}^\model(\alpha, \delta).
		      \end{aligned}
	      \]
	      And since $|I^\model| = 1$, the condition $\neg (\delta = \beta)$ is never satisfied and the axiom is trivially true.

	\item $\ext$ : Let there be $\alpha_1, \alpha_2 \in \arrType$ such that $\alpha_1 \neq \alpha_2$. Then, we need to show that there exists $\beta \in I$ such that $\select{}^\model(\alpha_1, \beta) \neq \select{}^\model(\alpha_2, \beta)$.
	      Then, since $m$ is a bijection, we have:
	      \[
		      \begin{aligned}
			      \select{}^\model(\alpha_1, -1) & = m(\alpha_1) \neq m(\alpha_2) = \select{}^\model(\alpha_2, -1).
		      \end{aligned}
	      \]
\end{enumerate}

Now, we show that $\model^{\Sigma_2}$ is a datatype $\Sigma_2$-structure generated by $\elemsorts$.
\begin{enumerate}
	\item $E^\model$ is indeed the $T_{E}\bigl(\consig{\Sigma}, \{I^\model, \arrType^\model\}\bigr)$,
	\item $c^{\model}(t_1 \til t_n) = c(t_1 \til t_n)$ for every $c\in\constructors_{\Sigma_2}$ of
	      arity $(\sigma_{1}\times\ldots\times\sigma_{n})\ra\sigma$ and $t_1 \in \sigma_{1}^{\model} \til t_n \in \sigma_{n}^{\model}$,
	\item $\sel{c}{i}^{\model}( c(t_1 \til t_n) )=t_i$
	      for every $c\in\constructors_{\Sigma_2}$ of
	      arity $(\sigma_{1}\times\ldots\times\sigma_{n})\ra\sigma$, $t_1 \in \sigma_{1}^{\model} \til t_n \in\sigma_{n}^{\model}$ and $1\leq i\leq n$,
	\item $\is{c}^{\model}=\set{ c(t_1 \til t_n) \mid t_1 \in \sigma_{1}^{\model} \til t_n \in \sigma_{n}^{\model} }$ for every $c\in\constructors_{\Sigma_2}$ of arity
	      $(\sigma_{1}\times\ldots\times\sigma_{n})\ra\sigma$.
\end{enumerate}

Now let's show that $\model \models \formulaExUnsat$:
\begin{enumerate}
	\item $\is{\container}^\model(x^\model)$ holds since $x^\model = \container(-1, \infty)$ is an application of the constructor $\container$.
	\item $\children^\model(x^\model) = \children^\model(\container(-1, \infty)) = \infty = a^\model$.
	\item \(x^\model = \container(-1, \infty) = m(\infty) = \select{}^\model(\infty, -1)
	      = \select{}^\model(a^\model, i^\model)\)
\end{enumerate}

Finally, we verify that indeed \Cref{sec:theory:nested-relation} excludes
the interpretation $\model$ from being an NDT-$\exsigwf$-structure.

\begin{align*}
	\rel{\model} = & \{(\container(-1, n-1), n) \mid n \in \mathbb{N} \setminus \{0\}\}     \\ & \cup \{(\emp, 0), (\container(-1, \infty), \infty)\} \\
	               & \cup \{(-1, \container(-1, n)) \mid n \in \mathbb{N} \cup \{\infty\}\} \\ &
	\cup \{(n, \container(-1, n)) \mid n \in \mathbb{N} \cup \{\infty\}\}
\end{align*}

\noindent
The first two sets come from the select function,
and the last two from constructor applications.
Since $\rel{\model}$ contains the cycle $(a^\model, x^\model, a^\model)$,
we have that $\model$ is not an NDT $\exsigwf$-interpretation.

\subsection{More Details on \Cref{ex:varphi-unsat}}
The formula $\formulaExUnsat$ from
\Cref{tab:sigs}
is $\tndt{\exsigwf}$-unsatisfiable.
Assume otherwise, and
let there be a $\tndt{\exsigwf}$-interpretation $\model$ that satisfies $\formulaExUnsat$. $\model \models \psi_1'$ and therefore $x^\model = \select{}(a,i)$ and so:
\begin{equation}
	\label{eq:half-cycle-1}
	(x^\model, a^\model) \in \rel{\model}.
\end{equation}
Similarly, $\model \models \psi_2$ and thus
$\is{\container}^\model(x^\model)$ and $a^\model = \children(x^\model)$.
\Cref{datatypes-axioms} then suggests that $x^\model =
	\container(a^\model, c)$
for some $b$ and $c$,
and thus:
\begin{equation}
	\label{eq:half-cycle-2}
	(a^\model, x^\model) \in \rel{\model}
\end{equation}
From \Cref{eq:half-cycle-1,eq:half-cycle-2}, we get
a cycle in $\rel{\model}$ which entails an infinite descending sequence,
thus violating well-foundedness.

\subsection{Proof of \Cref{lem:well-foundedness-n+1}}
\label{app:prflemdtsig}
\sigisdtsigtr*

\begin{proof}
	$\Sigma^{\formulaProof}_{n+1}$ is already presented
	using element and struct sorts, as well as constructors,
	selectors and testers.
	In particular,
	$\elemsorts_{\Sigma^{\formulaProof}_{n+1}}$ and
	$\structsorts_{\Sigma^{\formulaProof}_{n+1}}$
	are specified in \Cref{fig:dt-sig-trans}.
	It is left to show that each sort in
	$\structsorts_{\Sigma^{\formulaProof}_{n+1}}$
	is well-founded w.r.t. $\elemsorts_{\Sigma^{\formulaProof}_{n+1}}$
	in $\consig{\Sigma^{\formulaProof}_{n+1}}{\Sigma^{\formulaProof}_{n+1}}$,
	according to \Cref{def:well-sorted-dt}.
	Let $\tau\in\structsorts_{\Sigma^{\formulaProof}_{n+1}}$.
	Then there exists $\sigma\in\structsorts_{\Sigma_{n+1}}\cup\arrays_\Sigma$
	such that $\tau=\trsig{\sigma}$.
	Since \(\Sigma\) is well-founded as an NDT signature, we have that there exists some \(m\in\mathbb{N}\) with
	$
		\sigma\in F_m(\Sigma,\sorts{\Sigma}\,\setminus\,(\arrays_{\Sigma}\,\cup\,
		\structsorts_{\Sigma_{n+1}}))
	$.
	We prove that for every $m\geq 0$, if
	$
		\sigma\in F_m(\Sigma,\sorts{\Sigma}\,\setminus\,(\arrays_{\Sigma}\,\cup\,
		\structsorts_{\Sigma_{n+1}}))
	$
	then
	$\trsig{\sigma} \in G_m(\consig{\Sigma^\formulaProof_{n+1}}{\Sigma^\formulaProof_{n+1}},\elemsorts_{\Sigma^{\formulaProof}_{n+1}})$.
	We do so
	by induction on \(m\).
	As a result, we then get that
	$\tau=\trsig{\sigma}$
	is well-founded w.r.t. $\elemsorts_{\Sigma^{\formulaProof}_{n+1}}$
	in $\consig{\Sigma^{\formulaProof}_{n+1}}{\Sigma^{\formulaProof}_{n+1}}$.

	For the sake of this proof, we abbreviate $F_k(\Sigma,\sorts{\Sigma}\,\setminus\,(\arrays_{\Sigma}\,\cup\, \structsorts_{\Sigma_{n+1}}))$ as $F_k$ and $G_k(\consig{\Sigma^\formulaProof_{n+1}}{\Sigma^\formulaProof_{n+1}},\elemsorts_{\Sigma^{\formulaProof}_{n+1}})$ as $G_k$ for every \(k\in\mathbb{N}\).

	\begin{itemize}
		\item \textbf{Base case (\(m=0\)):} No \(\sigma\in\structsorts_{\Sigma_{n+1}}\cup\arrays_\Sigma\) belongs to
		      \(F_0\), so the claim holds vacuously.
		\item \textbf{Inductive step:}
		      Assume the claim holds for $m$, and let
		      \(
		      \sigma\;\in\;
		      F_{m+1}
		      \).
		      If already
		      $\sigma\in F_m$, then by definition $G_m\subseteq G_{m+1}$ and
		      the inductive hypothesis we are done. Otherwise there are two cases:
		      \begin{enumerate}
			      \item \textbf{Array sort.}
			            Suppose $\sigma=\arrType_i$ for some $i\in[1,n]$.
			            Since $\sigma\in F_{m+1}$, we have that
			            $I_i$
			            and $E_i$ are in
			            $F_m$.
			            \begin{itemize}
				            \item If $E_i\notin\structsorts_{\Sigma_{n+1}}\cup\arrays_\Sigma$, then
				                  $\trsig{E_i}$ is in $\elemsorts_{\Sigma^{\formulaProof}_{n+1}} = G_0$ and clearly so is $\idx{A_i}$, so by the constructor $\arrCons{i}$ we get
				                  $\trsig{A_i}\in G_1\subseteq G_{m+1}$.
				            \item Otherwise $E_i \in \structsorts_{\Sigma_{n+1}}\cup\arrays_\Sigma$ and since $E_i \in F_m$, we have by IH that $\trsig{E_i}\in G_m$, and clearly $\idx{A_i}\in G_0\subseteq G_m$.
				                  Again $\arrCons{i}$ yields
				                  $\trsig{A_i}\in G_{m+1}$.
			            \end{itemize}

			      \item \textbf{Struct sort.}
			            There is a constructor
			            $c : \sigma_1\times\cdots\times\sigma_k \to \sigma$
			            with each $\sigma_j\in F_m$. For every $j\in[1,k]$, if $\sigma_j \in \structsorts_{\Sigma_{n+1}} \cup \arrays_\Sigma$, by IH $\trsig{\sigma_j}\in G_m$, otherwise,
			            $\trsig{\sigma_j} \in G_0\subseteq G_m$. Hence constructor
			            $\trsig{c}$ places $\trsig{\sigma}$ into $G_{m+1}$.
		      \end{enumerate}
	\end{itemize}
\end{proof}

\subsection{More Details on \Cref{ex:preprocessing-tr}}
\label{app:sec:exsattranslationnolemmas}
We hereby define a $\newT{\formulaExUnsat}$-interpretation, $\newM$ that
satisfies $\tra(\formulaExUnsat)$.

For the domains, we set
$\trsig{I}^\newM = \mathbb{N}$, ${\idx{\arrType}}^\newM = \{-1\}$,
and define an $\elemsorts_{\newSig{\exsigwf}{\formulaExUnsat}}$-sorted set $B$ such that $B_\sigma = \sigma^\newM$ for each $\sigma \in \elemsorts_{\newSig{\exsigwf}{\formulaExUnsat}}$.
Then we can set
$\trsig{E}^\newM = T_{\trsig{E}}(\consig{\trsig{\Sigma_2}}{\trsig{\Sigma_2}}, B)$,
$\trsig{\arrType}^\newM = T_{\trsig{\arrType}}(\consig{\trsig{\Sigma_2}}{\trsig{\Sigma_2}}, B)$, and
$\trarr{\arrType}^\newM = \trsig{I}^{\newM} \rightarrow \trsig{E}^\newM$.

The function symbols $\select{}$ and $\store{}$,
as well as the constructors, selectors and testers are interpreted as expected
(for this example, it does not matter how wrongly applied selectors are interpreted).

As for the variables, we set
$\trsig{i}^\newM = 1$,
$\trsig{a}^\newM = \arrCons{1}(-1, \trsig{\emp})$,
$\trsig{x}^\newM = \trsig{\container}(3, \arrCons{1}(-1, \trsig{\emp}))$,
and
$\trarr{a}^\newM = \{(n, \trsig{x}^\newM) \mid n \in \mathbb{N}\}$,
that is,
$\trarr{a}^\newM$ is the constant function that always returns $\trsig{x}^\newM$.
$\newM$ is therefore a $\newT{\formulaExUnsat}$-interpretation that clearly
satisfies every literal of $\tra(\formulaExUnsat)$.


\subsection{Necessity of Stable Infiniteness}
\label{app:sec:stable-infiniteness-necessity}

\label{rem:comp-necessity-of-stable-infiniteness}

\begin{table}[t]
	\newcolumntype{M}[1]{>{\centering\arraybackslash}m{#1}}
	\centering
	{
		\scriptsize
		\begin{tabular}{|M{.8cm}|M{3.6cm}|M{4.3cm}|M{3.7cm}|}\hline
			Name               &
			Sorts              &
			Symbols            &
			Formula                                                     \\\hline\hline
			                   &   &  &                                 \\

			\shortstack{
				$\tilde{\Sigma_1}$
			}                  &
			\shortstack{
				$\sorts{\tilde{\Sigma_1}} = \{\arrType, I, E\}$
			}                  &
			\shortstack{
				$\func{\tilde{\Sigma_1}} = \{\select{\arrType}, \store{\arrType}\}$
			}                  &
			\shortstack{
				$.$
			}
			\\&&&\\\hline
			                   &   &  &                                 \\

			$\tilde{\Sigma_2}$ &
			\shortstack{
				$\sorts{{\tilde{\Sigma_2}}} = \elemsorts_{\tilde{\Sigma_2}} \uplus
					\structsorts_{\tilde{\Sigma_2}}$                            \\
				$\elemsorts_{\tilde{\Sigma_2}} = \{ \arrType \}$            \\
				$\structsorts_{\tilde{\Sigma_2}} = \{ E, S \}$
			}
			                   &
			\shortstack{
				$\func{\tilde{\Sigma_2}} = \constructors_{\tilde{\Sigma_2}} \cup
					\selectors_{\tilde{\Sigma_2}}$                              \\
				$\constructors_{\tilde{\Sigma_2}} = \{ \container, \emp \}$ \\
				$\selectors_{\tilde{\Sigma_2}} = \{\sel{\container}{1} \}$  \\
				$\pred{{\tilde{\Sigma_2}}} = \{ \is{\container}, \is{\emp} \}$
			}
			                   &
			\shortstack{
				$\theta_2:= x \neq y $
			}
			\\&&&\\\hline
		\end{tabular}
		\normalsize
	}
	\caption{%
		Signature $\tilde{\Sigma}=\tilde{\Sigma_{1}}\uplus\tilde{\Sigma_{2}}$ used in the stable–infiniteness counterexample.
		$\tilde{\Sigma_{1}}$ is an arrays signature with sorts $\{\arrType,I,E\}$ and operations
		$\select{\arrType}:\arrType\times I\to E$ and $\store{\arrType}:\arrType\times I\times E\to \arrType$.
		$\tilde{\Sigma_{2}}$ is a datatypes signature whose element sorts are $\{\arrType\}$ and struct sorts are $\{E,S\}$; it provides constructors $\container:\arrType\to S$ and $\emp:E$, selector $\sel{\container}{1}:S\to \arrType$, and testers $\is{\container}:S$ and $\is{\emp}:E$.
		The formula is $\theta := x \neq y$ with $x,y$ of sort $S$.}
	\label{tab:completeness-counter-example}
\end{table}

\begin{table}[t]
	\newcolumntype{M}[1]{>{\centering\arraybackslash}m{#1}}
	\centering
	{
		\scriptsize
		\begin{tabular}{|M{.8cm}|M{3.6cm}|M{4.3cm}|M{3.7cm}|}\hline
			Name              &
			Sorts             &
			Symbols           &
			Formula                                                                                         \\\hline\hline
			                  &   &  &                                                                      \\

			\shortstack{
				$\trsig{\Sigma_1}$
			}                 &
			\shortstack{
				$\sorts{\Sigma_1} = \{\trarr{\arrType}, \trsig{I}, \trsig{E}\}$
			}                 &
			\shortstack{
				$\func{\Sigma_1} = \{\select{\trarr{\arrType}}, \store{\trarr{\arrType}}\}$
			}                 &
			\shortstack{
				$.$
			}
			\\&&&\\\hline
			                  &   &  &                                                                      \\

			$\Sigma_2^\theta$ &
			\shortstack{
				$\sorts{{\Sigma_2^\theta}} = \elemsorts_{\newSig{\Sigma_2}{\theta}} \uplus
					\structsorts_{\newSig{\Sigma_2}{\theta}}$                                                       \\
				$\elemsorts_{\Sigma_2^\theta} = \{ \idx{\arrType} \}$                                           \\
				$\structsorts_{\Sigma_2^\theta} = \{ \trsig{E}, \trsig{S}, \trsig{\arrType} \}$
			}
			                  &
			\shortstack{
				$\func{\Sigma_2^\theta} = \constructors_{\newSig{\Sigma}{\theta}} \cup
					\selectors_{\newSig{\Sigma}{\theta}}$                                                           \\
				$\constructors_{\newSig{\Sigma}{\theta}} = \{ \trsig{\container}, \trsig{\emp}, \arrCons{1} \}$ \\
				$\selectors_{\newSig{\Sigma}{\theta}} = \{\sel{\trsig{\container}}{1}, \sel{\arrCons{1}}{1} \}$ \\
				$\pred{{\Sigma_2^\theta}} = \{ \is{\trsig{\container}}, \is{\trsig{\emp}}, \is{\arrCons{1}} \}$
			}
			                  &
			\shortstack{
				$\trF(\theta):= \trsig{x} \neq \trsig{y} $
			}
			\\&&&\\\hline

			                  &   &  &                                                                      \\

			$\Sigma_3$        &
			\shortstack{
				$\sorts{{\Sigma_3}} = \{\trsig{\arrType}, \trarr{\arrType}\}$
			}
			                  &
			\shortstack{
				$\func{\Sigma_2} = \{f_\arrType, g_{\arrType}\}$
			}
			                  &
			\shortstack{
				$.$
			}
			\\&&&\\\hline
		\end{tabular}
		\normalsize
	}
	\caption{%
		The three‐part translated signature for the counterexample of \Cref{tab:completeness-counter-example}.
		\(\trsig{\Sigma_1}\) is an arrays signature with sorts
		\(\{\trarr{\arrType},\,\trsig{I},\,\trsig{E}\}\) and operations
		\(
		\select{\trarr{\arrType}}:\trarr{\arrType}\times\trsig{I}\to\trsig{E},\quad
		\store{\trarr{\arrType}}:\trarr{\arrType}\times\trsig{I}\times\trsig{E}\to\trarr{\arrType}.
		\)
		\(\Sigma_{2}^{\theta}\) is a datatype signature with element sort
		\(\{\idx{\arrType}\}\), struct sorts
		\(\{\trsig{E},\,\trsig{S},\,\trsig{\arrType}\}\), constructors
		\(
		\trsig{\container}:\trsig{\arrType}\to\trsig{S},\quad
		\trsig{\emp}:\trsig{E},\quad
		\arrCons{1}:\idx{\arrType}\to\trsig{\arrType},
		\) and the matching selectors and testers.
		\(\Sigma_{3}\) adds two UF symbols
		\(f_{\arrType}:\trsig{\arrType}\to\trarr{\arrType}\)
		and
		\(g_{\arrType}:\trarr{\arrType}\to\trsig{\arrType}\).
	}
	\label{tab:completeness-counter-example-translated}
\end{table}

The stable infiniteness condition is necessary
in the right-to-left direction of \Cref{thm:correctness-of-translation}.
Consider the NDT signature
\(\tilde{\Sigma}=\tilde{\Sigma_{1}}\cup\tilde{\Sigma_{2}}\)
described in \Cref{tab:completeness-counter-example},
and let
$\tilde{T_{1}}$ be the $\tilde{\Sigma_{1}}$-theory of arrays and
$\tilde{T_{2}}$ the $\tilde{\Sigma_{2}}$-theory of datatypes.

Consider the $\tilde{\Sigma}$‐formula
$\theta_2 := x\neq y$, such that
\(x,y\) both have sort \(S\).
Let $\nextM$ be a $\tilde{T_2}$-interpretation.
Let $B$ be the $\elemsorts_{\tilde{\Sigma_2}}$-sorted set with
$B_{\rho}=\rho^\nextM$ for each
$\rho\in\elemsorts_{\tilde{\Sigma_2}}$.
Thus:
\(
E^{\nextM} \;=\; T_{E}(\consig{\tilde{\Sigma_2}}{\tilde{\Sigma_2}}, B) \;=\; \{\emp\}
\).

Note that $E$ is a shared sort between $\tilde{\Sigma_1}$ and
$\tilde{\Sigma_2}$ and since $\tilde{T_2}$ enforces that the
interpretation of $E$ is $\{\emp\}$, we know that $\tilde{T_2}$ is
not stably infinite over $\sorts{\tilde{\Sigma_1}} \cap \sorts{\tilde{\Sigma_2}}$.

Assume towards a contradiction that \(\theta_2\) is satisfiable by a
$\tndt{\tilde{\Sigma}}$-interpretation $\model_1$.
Since ${\model_1}^{\tilde{\Sigma_2}}$ is a
$\tilde{T_2}$-interpretation, we can denote as $\tilde{B}$ the
$\elemsorts_{\tilde{\Sigma_2}}$-sorted set with
$\tilde{B}_{\rho}=\rho^{\model_1}$ for each
$\rho\in\elemsorts_{\tilde{\Sigma_2}}$.
Thus:
\[
	S^{\model_1} \;=\; T_{S}(\consig{\tilde{\Sigma_2}}{\tilde{\Sigma_2}}, \tilde{B}) \;=\; \{\container(z)\mid z\in\arrType^{\model_1}\}
	\quad
	E^{\model_1} \;=\; T_{E}(\consig{\tilde{\Sigma_2}}{\tilde{\Sigma_2}}, \tilde{B}) \;=\; \{\emp\}
\]

Hence, $x^{\model_1} = \container(a)$ and
$y^{\model_1} = \container(b)$ for some
$a,b\in\arrType^{\model_1}$.
Since $x^{\model_1} \neq
	y^{\model_1}$, we have $a \neq b$. According to the Extensionality
axiom, we have that there exists $k \in I^{\model_1}$ such that
$\select{\arrType}^{\model_1}(a, k) \neq
	\select{\arrType}^{\model_1}(b, k)$. However, since
$\select{\arrType}^{\model_1}(a, k),
	\select{\arrType}^{\model_1}(b, k) \in E^{\model_1}$, we have that
$\select{\arrType}^{\model_1}(a, k) = \emp =
	\select{\arrType}^{\model_1}(b, k)$, which is a contradiction.

Yet,
its translation \(\trF(\theta_2)\) and its matching signature are given in~\Cref{tab:completeness-counter-example-translated},
and \(\trF(\theta_2)\) is satisfied
by the $\tndt{\newSig{\tilde{\Sigma}}{\theta_2}}$-interpretation
$\newM$ defined in~\Cref{tab:exCompleteness-2}.

\subsection{Discussion on the translation}
We describe various decisions made in our translation-based
approach.

\subsubsection{The sorts $\idx{\arrType}$}
Although $\trF(\formulaProof)$ does not introduce
terms of the sorts
$\idx{\arrType_i}$ for $ i\in[1,n]$,
those sorts are still required in $\newSig{\Sigma}{\formulaProof}$.
If we drop them, the translation may turn a satisfiable
formula into an unsatisfiable one,
thus violating \Cref{thm:correctness-of-translation}.
The following example makes this precise.

Let $\theta$ be defined as follows:
\[
	\theta \;:=\;
	a \neq b
	\;\land\;
	x = \select{\arrType}(a,i)
	\;\land\;
	x = \select{\arrType}(b,i),
\]
Where $a$ and $b$ are array variables of sort $\arrType$, $x$ is a variables of sort $E$ and $i$ is a variable of sort $I$.
$\theta$ is $\tndt{\exsigwf}$-satisfiable.

Assume we construct $\newSig{\exsigwf}{\theta}$ exactly as in~\Cref{fig:dt-sig-trans} but \emph{omit} the index sort $\idx{\arrType}$.
The three component signatures that form
$\newSig{\exsigwf}{\theta}$ are summarised in
\Cref{tab:theta-sigs}.
Because $\idx{\arrType}$ is absent,
the constructor $\arrCons{1}$ now has arity
$\trtype{E}\!\to\!\trtype{\arrType}$ and hence only one selector.

\begin{table}[t]
	\centering
	\renewcommand\arraystretch{1.15}
	\begin{tabular}{|c|p{3.9cm}|p{5.3cm}|p{3.2cm}|}
		\hline
		Signature                                                                        & Sorts & Functions & Predicates \\\hline\hline
		$\trsig{\Sigma_1}$                                                               &
		$\{\trtype{I},\trtype{E},\trarr{\arrType}\}$                                     &
		$\select{\trarr{\arrType}} :
			\trarr{\arrType}\!\times\!\trtype{I}\!\to\!\trtype{E}$\newline
		$\store{\trarr{\arrType}} :
		\trarr{\arrType}\!\times\!\trtype{I}\!\times\!\trtype{E}\!\to\!\trarr{\arrType}$ &
		---                                                                                                               \\\hline
		$\Sigma_2^{\theta}$                                                              &
		$\{\trtype{I},\trtype{E},\trtype{\arrType}\}$                                    &
		$\arrCons{1} : \trtype{E} \to \trtype{\arrType}$\newline
		$\trsig{\container} : \trtype{I}\!\times\!\trtype{\arrType}\!\to\!\trtype{E}$\newline
		$\sel{\arrCons{1}}{1}: \trtype{\arrType}\!\to\!\trtype{E}$\newline
		$\sel{\trsig{\container}}{1}: \trtype{E}\!\to\!\trtype{I}$\newline
		$\sel{\trsig{\container}}{2}: \trtype{E}\!\to\!\trtype{\arrType}$                &
		$\is{\arrCons{1}} : \trtype{\arrType}$\newline
		$\is{\trsig{\container}} : \trtype{E}$                                                                            \\\hline
		$\Sigma_3$                                                                       &
		$\{\trtype{\arrType},\trarr{\arrType}\}$                                         &
		$f_{\arrType} : \trtype{\arrType}\!\to\!\trarr{\arrType}$\newline
		$g_{\arrType} : \trarr{\arrType}\!\to\!\trtype{\arrType}$                        &
		---                                                                                                               \\\hline
	\end{tabular}
	\caption{Component signatures of $\newSig{\exsigwf}{\theta}$ when the
		index sort $\idx{\arrType}$ is \emph{omitted}.}
	\label{tab:theta-sigs}
\end{table}

Then, the translation of~$\theta$ would be
\begin{align*}
	\trF(\theta)=\; &
	\bigl(
	\trsig{a}\neq\trsig{b}
	\land
	\trsig{x}= \select{\trarr{\arrType}}(\trarr{a},\trsig{i})
	\land
	\trsig{x}= \select{\trarr{\arrType}}(\trarr{b},\trsig{i})
	\bigr)                                                         \\
	                & \land
	\bigl(
	\tempvar{a}{1}= \select{\trarr{\arrType}}(\trarr{a},\trsig{i})
	\land
	\tempvar{b}{1}= \select{\trarr{\arrType}}(\trarr{b},\trsig{i}) \\
	                & \land
	\tempvar{a}{1}= \sel{\arrCons{1}}{1}(\trsig{a})
	\land
	\tempvar{b}{1}= \sel{\arrCons{1}}{1}(\trsig{b})                \\
	                & \qquad\;\land
	\trarr{a}=f_{\arrType}(\trsig{a})
	\land
	\trarr{b}=f_{\arrType}(\trsig{b})
	\land
	\trsig{a}=g_{\arrType}(\trarr{a})
	\land
	\trsig{b}=g_{\arrType}(\trarr{b})
	\bigr).
\end{align*}

Since $\arrCons{1}$ is the only constructor whose target sort is $\trsig{\arrType}$,
and it has only one argument (which is of sort $\trtype{E}$),
the four middle conjuncts imply
\(
\sel{\arrCons{1}}{1}(\trsig{a})=
\sel{\arrCons{1}}{1}(\trsig{b}),
\)
hence $\trsig{a}=\trsig{b}$.
This contradicts $\trsig{a}\neq\trsig{b}$,
the first conjunct,
so $\trF(\theta)$ is $\newT{\theta}$-\emph{unsatisfiable}.

On the other hand,
when the index sort $\idx{\arrType}$ is retained,
$\arrCons{1}$ has sort
$\idx{\arrType}\times\trtype{E}\to\trtype{\arrType}$.
Hence
$\sel{\arrCons{1}}{2}(\trsig{a})$ and
$\sel{\arrCons{1}}{2}(\trsig{b})$ may be equal without forcing
$\trsig{a}=\trsig{b}$.

\subsubsection{On The Third Lemma}

The motivation for $\Lthree{\formulaProof}$ is to ensure that
for every array sort $A$,
there is a one to one mapping between the interpretations of free variables of sort $\trarr{A}$ and those of sort $\trtype{A}$.
$f_{\arrType}$ and $g_{\arrType}$ are asserted to be inverses of one another
when restricted to the relevant domain, thus ensuring the 1-1 correspondence.

\begin{example}
	\label{ex:l3}
	Consider the $\exsigwf$-formula $\formulaExUnsat$ from \Cref{tab:sigs}. Then,
	\[
		\Lthree{\formulaExUnsat} = \Bigl\{
		\trarr{a} = f_{\arrType}(\trsig{a}),
		\quad
		\trsig{a} = g_{\arrType}(\trarr{a})
		\Bigr\}.
	\]
	Here, \(\trsig{a}\) and \(\trarr{a}\) are new variables of sort
	\(\trarr{\arrType}\),
	\(f_{\arrType}\) is a function symbol of arity
	\(\trtype{\arrType} \to \trarr{\arrType}\),
	and \(g_{\arrType}\) is a function symbol of arity
	\(\trarr{\arrType} \to \trtype{\arrType}\).
\end{example}

\begin{table}[t]
	\newcolumntype{M}[1]{>{\centering\arraybackslash}m{#1}}
	\centering
	\begin{tabular}{|M{3cm}|M{10cm}|}\hline
		$\widehat{\formulaExUnsat}$                                                                                       &
		$
			a \neq b \;\land\;
			x = \select{}(a,i)\;\land\;
			b = \store{}(a, i, x)
		$
		\\\hline
		$\tra(\widehat{\formulaExUnsat}) \wedge \Lone{\widehat{\formulaExUnsat}} \wedge \Ltwo{\widehat{\formulaExUnsat}}$ &
		\shortstack{
			$\trsig{a} \neq \trsig{b} \land \trsig{x} = \select{}(\trarr{a}, \trsig{i}) \land \trarr{b} = \store{}(\trarr{a}, \trsig{i}, \trsig{x}) \land$ \\
			$\tempvar{a}{1} = \select{}(\trarr{a}, \trsig{i}) \land \tempvar{a}{1} = \sel{\arrCons{1}}{2}(\trsig{a}) \land$                                \\
			$\tempvar{b}{1} = \select{}(\trarr{b}, \trsig{i}) \land \tempvar{b}{1} = \sel{\arrCons{1}}{2}(\trsig{b})$
		}
		\\\hline
	\end{tabular}
	\caption{A $\tndt{\exsigwf}$-formula and its translation.
		Note that $\newSig{\exsigwf}{\formulaExSat} = \newSig{\exsigwf}{\formulaExUnsat}$,
		given in \Cref{tab:sigs-tran}.
	}
	\label{tab:necessity-l3}
\end{table}

\begin{example}
	\label{ex:necessity-l3}
	In order to demonstrate the necessity of $\Lthree{\formulaProof}$, we can consider the following $\exsigwf$-formula $\widehat{\formulaExUnsat}$ in \Cref{tab:necessity-l3}.
	Notice that $\newSig{\exsigwf}{\widehat{\formulaExUnsat}}$ is the same
	as $\newSig{\exsigwf}{\formulaExUnsat}$ (given in \Cref{tab:sigs-tran}), and
	$\tra(\widehat{\formulaExUnsat}) \wedge \Lone{\widehat{\formulaExUnsat}} \wedge \Ltwo{\widehat{\formulaExUnsat}}$ is also given in \Cref{tab:necessity-l3}.

	Now, $\widehat{\formulaExUnsat}$ is not $\tndt{\exsigwf}$-satisfiable,
	which can be seen by pure array reasoning.
	However, its translation with the first two sets of lemmas is
	$\newT{\widehat{\formulaExUnsat}}$-satisfiable, since there is nothing
	preventing $\trsig{a} \neq \trsig{b}$ while $\trarr{a} = \trarr{b}$, despite
	the fact that both $\trsig{a}$ and $\trarr{a}$ represent $a$ and $\trsig{b}$
	and $\trarr{b}$ represent $b$.
	The issue is that $a$ and $b$ may differ on indices that do not occur in $\widehat{\formulaExUnsat}$,
	and this is lost in translation, unless we use $\Lthree{\formulaProof}$.
\end{example}

\subsubsection{Optimized Lemmas}
The lemmas from \Cref{sec:lemmas} are defined in a way that
$\Lone{\formulaProof}$ and
$\Ltwo{\formulaProof}$ collect
the lemmas of
$\Lzero{a}{i}{\formulaProof}$ for specific terms $a$ and $i$.
A simpler alternative is to discard the delicate
construction of \( \Lone{\formulaProof} \) and \( \Ltwo{\formulaProof} \) and instead add \emph{every} lemma
\( \Lzero{a}{i}{\formulaProof} \) for each array variable \( a\in\fv{\arrType_j}{\formulaProof} \)
and every index \( i\in\indexSet{\arrType_j}{\formulaProof} \).
Formally, we could have replaced $\reduced{\formulaProof}$ (see \Cref{ex:L0,ex:l3}) by:
\[
	\reduced{\formulaProof}'
	\;=\;
	\bigl(
	\bigcup_{1\le j\le n}
	\,
	\bigcup_{a\in\fv{\arrType_j}{\formulaProof}}
	\,
	\bigcup_{i\in\indexSet{\arrType_j}{\formulaProof}}
	\Lzero{a}{i}{\formulaProof}
	\bigr)
	\;\cup\;
	\Lthree{\formulaProof},
\]
and clearly \( \reduced{\formulaProof}' \supseteq \Lone{\formulaProof}\cup\Ltwo{\formulaProof} \).


The reason to bother with $\reduced{\formulaProof}$ rather than using the simpler
$\reduced{\formulaProof}'$ is not correctness, but efficiency.
Let \( k \) be the number of distinct variables of sort \( \arrType_j \) in $\formulaProof$,
\( l \) the number of \texttt{select} literals \( x=\select{\arrType_j}(a,i) \) in $\formulaProof$,
\( m \) the number of \texttt{store} literals
\( b=\store{\arrType_j}(a,i,v) \) in $\formulaProof$,
and \( t = \indexCar{\arrType_j}{\formulaProof}=|\indexSet{\arrType_j}{\formulaProof}| \).
Because \( \Lzero{a}{i}{\formulaProof} \) contains a constant (two) literals,
\(\reduced{\formulaProof}\) contributes \( O(l) \)
literals from the \texttt{select}s
and \( O(m\,t) \) from the \texttt{store}s, while never exceeding
\( O(k\,t) \).
Hence
\(
|\reduced{\formulaProof}| = O\bigl(\min\{\,l + m\,t,\; k\,t\,\}\bigr).
\)
By contrast the naive set \( \reduced{\formulaProof}' \) always contains
\( k\,t \) such blocks, so
\( |\reduced{\formulaProof}'| = O(k\,t) \).

Consequently, whenever \( l \ll k\,t \land m \ll k\), the original
definition saves a factor of
\( O\!\bigl(k\,t/(l + m\,t)\bigr) \) literals, whereas in the worst
case the two approaches are equivalent.

For instance, consider the case where $m$ is constant and $l,t,k \in O\!\bigl(n\bigr)$ for some $n$, then $\reduced{\formulaProof} \in O\!\bigl(n\bigr)$ while $\reduced{\formulaProof}' \in O\!\bigl(n^2\bigr)$.

\begin{example}
	Let there be $t \in \mathbb{N}$, consider the following formula in the signature $\exsigwf$:

	\[
		\theta = \bigwedge_{j=1}^t \bigl(x_j = \select{\arrType}(a_j, i_j)\bigr).
	\]

	Where $x_j$ are variables of sort $E$, $a_j$ are variables of sort $\arrType$, and $i_j$ are variables of sort $I$.

	Then, $\reduced{\theta} = \bigcup_{j=1}^t \Lzero{a_j}{i_j}{\theta} \cup \Lthree{\theta}$, which contains $4t$ literals, while $\reduced{\theta}' = \bigcup_{j=1}^t \bigcup_{k=1}^t \Lzero{a_j}{i_k}{\theta} \cup \Lthree{\theta}$, which contains $2t^2 + 2t$ literals.
\end{example}

\end{document}

%% file: macros.tex
\usepackage{xparse}

\spnewtheorem{fact}{Fact}{\bfseries}{\itshape}
\spnewtheorem{notation}{Notation}{\bfseries}{\itshape}

\newcommand{\elemsorts}{{\bf Elem}}

\newcommand{\arrays}{{\bf {Array}}}

\newcommand{\select}[1]{%
	\textbf{select}_{#1}%
}
\newcommand{\greensat}{{\color{Green} sat}}

\newcommand{\redsat}{{\color{Red} sat}}
\newcommand{\greenunsat}{{\color{Green} unsat}}
\newcommand{\children}{{\mathbf{chld}}}
\newcommand{\id}{{\mathbf{id}}}
\newcommand{\head}{{\mathbf{hd}}}
\newcommand{\tail}{{\mathbf{tl}}}
\newcommand{\assignTail}[1]{\mathbf{tail}_{#1}}
\newcommand{\trComp}{\Psi}

\newcommand{\sigG}[1]{G_{#1}}
\newcommand{\minIx}[1]{\text{idx}(#1)}
\newcommand{\dtIx}[1]{\text{dtIdx}(#1)}

\newcommand{\cons}{{\mathbf{cons}}}
\newcommand{\container}{{\mathbf{c}}}

\newcommand{\exsigdt}{\Sigma_{0}}
\newcommand{\exsignwf}{\Sigma_{\mathbf{ndt}}^{-}}
\newcommand{\exsigwf}{\Sigma_{\mathbf{ndt}}}

\newcommand{\store}[1]{%
	\textbf{store}_{#1}%
}
\makeatletter
\newcommand{\indOrder}[2]{%
	o%
	\if\relax\detokenize{#1#2}\relax
	\else
		_{(#1, #2)}%
	\fi
}
\makeatother

\newcommand{\tdt}[1]{{\Ta}_{{#1}}^{{\mathit DT}}}
\newcommand{\tpndt}[1]{{\Ta}_{{#1}}^{{\mathit PNDT}}}
\newcommand{\tndt}[1]{{\Ta}_{{#1}}^{{\mathit NDT}}}

\newcommand{\nil}{{\mathbf{nil}}}
\newcommand{\emp}{{\mathbf{emp}}}
\newcommand{\arrCons}[1]{c_{#1}}

\newcommand{\subdomain}[2]{#1^{#2}}
\newcommand{\applyStar}{\Gamma_\text{DT}}
\newcommand{\applyPartStar}[1]{\Gamma_\text{DT}^{#1}}
\newcommand{\applyTag}{\Gamma_\text{Arr}}

\newcommand{\tempvar}[2]{u_{(#1, #2)}}

\newcommand{\indexSet}[2]{{\mathbb I}_{(#1,#2)}}

\newcommand{\indexCar}[2]{n_{(#1,#2)}}
\newcommand{\newSig}[2]{#1^{#2}}
\newcommand{\newT}[1]{T_{\text{new}}^{#1}}
\newcommand{\dtSig}{\Sigma_{\text{DT}}}

\newcommand{\structsorts}{{\bf Struct}}

\newcommand{\smtlib}{SMT-LIB~2\xspace}
\newcommand{\cvcfive}{\mathbf{cvc5}\xspace}
\newcommand{\zzz}{\mathbf{Z3}\xspace}

\newcommand{\constArr}[1]{\textit{const}({#1})}
\newcommand{\constDomArr}[2]{\textit{const}_{{#1}}({#2})}

\newcommand{\height}{{\it h}}

\newcommand{\intTerms}[1]{S_{#1}}
\newcommand{\storeElem}[2]{\text{SE}_{#1}({#2})}
\newcommand{\nonStoreElem}[2]{\text{NSE}_{#1}({#2})}

\newcommand{\defelem}[1]{\mathbf{d}_{#1}}
\newcommand{\almostC}[2]{\text{AC}({#1}, {#2})}
\newcommand{\constructors}{{\Ca}{\Oa}}
\newcommand{\consig}[2]{{#1}_{\mid \constructors_{#2}}}
\newcommand{\selectors}{{\Sa}{\Ea}}

\newcommand{\arrType}{A}
\newcommand{\idx}[1]{{{#1}^\circ}}

\newcommand{\row}{{\cal R}{\Oa}{\cal W}}
\newcommand{\ext}{{\cal E}{\cal X}{\cal T}}
\newcommand{\rel}[1]{{\cal R}_{#1}}
\newcommand{\model}{\Aa}
\newcommand{\modelS}{\Aa_0}

\newcommand{\newM}{\Ba}
\newcommand{\newMS}{\Ba_0}
\newcommand{\newMC}{\Ba_1}
\newcommand{\nextM}{\Ca}

\newcommand{\trtype}[1]{{#1^*}}
\newcommand{\trarr}[1]{#1'}
\newcommand{\trsig}[1]{{#1^*}}
\newcommand{\tra}{\bf {tr}}
\newcommand{\literal}{l}
\newcommand{\trF}{\bf {trFormula}}
\newcommand{\predecessors}{{\Pa}}

\newcommand{\sel}[2]{s_{(#1, #2)}}
\newcommand{\is}[1]{\text{is}_{#1}}

\newcommand{\powerset}[1]{{{\cal P}({#1})}}

\newcommand{\fv}[2]{{\it vars}_{#1}({#2})}

\newcommand{\sorts}[1]{{\cal S}_{#1}}
\newcommand{\func}[1]{{\cal F}_{#1}}

\newcommand{\pred}[1]{{\cal P}_{#1}}

\newcommand{\Ta}{{\cal T}}

\newcommand{\Aa}{{\cal A}}
\newcommand{\Ba}{{\cal B}}
\newcommand{\Ca}{{\cal C}}
\newcommand{\Sa}{{\cal S}}

\newcommand{\Oa}{{\cal O}}
\newcommand{\Ea}{{\cal E}}
\newcommand{\Pa}{{\cal P}}
\newcommand{\set}[1]{\left\{#1\right\}}

\def\S{{\cal S}}
\newcommand{\T}{{\cal T}}

\newcommand{\ra}{\rightarrow}
\newcommand{\bi}{\begin{itemize}}
		\newcommand{\ei}{\end{itemize}}
\newcommand{\be}{\begin{enumerate}}
		\newcommand{\ee}{\end{enumerate}}
\newcommand{\bd}{\begin{description}}
		\newcommand{\ed}{\end{description}}

\newcommand{\g}{\Gamma}

\newcommand{\suq}{\subseteq}

\newcommand{\til}{,\dots,}

\newcommand{\allArr}[2]{{\bf L^{#2}_0}({#1})}
\newcommand{\LzeroF}{{\mathbf{L_0}}}
\newcommand{\Lzero}[3]{\mathbf{L_0^{#3}}({#1, #2})}
\newcommand{\Lone}[1]{\mathbf{L_1}({#1})}
\newcommand{\Ltwo}[1]{\mathbf{L_2}({#1})}
\newcommand{\Lthree}[1]{\mathbf{L_3}({#1})}

\newcommand{\reduced}[1]{{\bf L}({#1})}

\newcommand{\zeroForm}{{\varphi_0}}
\newcommand{\formulaExUnsat}{{{\varphi}}}
\newcommand{\formulaExSat}{{{\psi}}}
\newcommand{\formulaProof}{{{\phi}}}
\newcommand{\formulaEvalGen}{{\psi}}
\newcommand{\formulaEval}[1]{\formulaEvalGen({#1})}

\newcommand{\formulaSatEval}[1]{\formulaEvalGen'({#1})}

\newcommand{\associatedIndices}[2]{\mathit{AssocInd}(#1, #2)}

\newcommand{\boxedcomment}[3]{\begin{mdframed}{{\textcolor{#1}{{#2}: {#3}}}}\end{mdframed}}

\newcommand{\axname}[1]{{\mathit ({#1})}}
\newcommand{\axtestpos}{{\axname{test^+}}}
\newcommand{\axtestneg}{{\axname{test^-}}}
\newcommand{\axsel}{{\axname{sel}}}
\newcommand{\axcyc}{{\axname{cyc}}}

\newcommand{\yoni}[1]{\boxedcomment{red}{Yoni}{#1}}
\newcommand{\tomer}[1]{\boxedcomment{blue}{Tomer}{#1}}

\newcommand{\yoniFixed}[1]{\boxedcomment{orange}{Yoni}{#1}}
\newcommand{\tomerFixed}[1]{\boxedcomment{blue}{Tomer}{#1}}